\documentclass[11pt]{article}
\usepackage{epsf}
\usepackage{amsmath}
\usepackage{amssymb}
\usepackage{amsfonts}
\usepackage{graphicx}
\usepackage{rotating}
\usepackage{bm}
\usepackage[usenames]{color}
\usepackage[scale={0.8,0.85}]{geometry}
\usepackage{framed}
\definecolor{shadecolor}{rgb}{0.9,0.9,0.9}
\definecolor{darkblue}{rgb}{0.0,0.0,0.55}
\definecolor{darkred}{rgb}{0.55,0.0,0.0}
\definecolor{electricgreen}{rgb}{0.0, 1.0,0.0}
\definecolor{forestgreen}{rgb}{0.13, 0.55,0.13}
\definecolor{crimsonglory}{rgb}{0.75,0.0,0.2}

\newtheorem{theorem}{Theorem}[section]
\newtheorem{proposition}[theorem]{Proposition}
\newtheorem{corollary}[theorem]{Corollary}
\newtheorem{lemma}[theorem]{Lemma}
\newtheorem{claim}[theorem]{Claim}

\newtheorem{definition}{Definition}[section]
\newtheorem{observation}{Observation}[section]
\newtheorem{openproblem}{Open Problem}[section]
\newtheorem{conjecture}{Conjecture}[section]

\newcommand{\qedsymb}{\hfill{\rule{2mm}{2mm}}}
\newenvironment{proof}{\begin{trivlist}
\item[\hspace{\labelsep}{\bf\noindent Proof: }]
}{\qedsymb\end{trivlist}}

\newcommand{\remove}[1]{}
\begin{document}

\title{{\bf The Complexity of Computational Problems 
            about Nash Equilibria
            in Symmetric Win-Lose Games}\thanks{This work was partially supported by 
                                                                          the Italian Ministry of Education, Universities and Research
                                                                          (PRIN 2010--2011 research project ARS TechnoMedia: 
                                                                           ``Algorithmics for Social Technological Networks"),
                                                                          and by research funds
                                                at the University of Cyprus.
                                                This work extends and improves results
                                                that appeared in preliminary form
                                                in the authors' paper
                                                "The Complexity of Decision Problems 
                                                 about Nash Equilibria
                                                 in Win-Lose Games,"
                                                {\it Proceedings of the
                                                     5th International Symposium
                                                     on Algorithmic Game Theory,}
                                                pp.\ 37--48,
                                                Vol.~7615,     
                                                Lecture Notes in Computer Science,
                                                Springer-Verlag,
                                                October 2012. 
                                               }
      }

\author{{\sl Vittorio Bil\`{o}}\thanks{Department of Mathematics and Physics 
                                       "Ennio De Giorgi", 
                                       University of Salento,
                                       73100 Lecce, Italy.
                                       Part of the work of this author
                                       was done while visiting
                                       the Department of Computer Science,
                                       University of Cyprus.
                                       Email {\tt vittorio.bilo@unisalento.it}
                                      }
        \and                              
        {\sl Marios Mavronicolas}\thanks{Department of Computer Science,
                                         University of Cyprus,
                                         Nicosia CY-1678,
                                         Cyprus.
                                         Email
                                         {\tt mavronic@cs.ucy.ac.cy}
                                        }
     }

\date{{\sc (\today)}}

\maketitle
\pagenumbering{arabic}
\thispagestyle{empty}

\begin{abstract}
We revisit the complexity of deciding,
given a {\it bimatrix game,} 
whether it has a {\it Nash equilibrium}
with certain natural properties; 
such decision problems were early known 
to be ${\mathcal{NP}}$-hard~\cite{GZ89}. 
We show that
${\mathcal{NP}}$-hardness
still holds under two significant restrictions
in simultaneity:
the game is
{\it win-lose}
(that is,
all {\it utilities} 
are $0$ or $1$)
and {\it symmetric}.
To address the former restriction,
we design
win-lose {\it gadgets}
and a win-lose reduction;
to accomodate the latter restriction,
we employ and analyze 
the classical {\it ${\mathsf{GHR}}$-symmetrization}~\cite{GHR63}
in the win-lose setting.
Thus,
{\it symmetric win-lose bimatrix games} 
are as complex as general bimatrix games 
with respect to such decision problems.

As a byproduct of our techniques,
we derive hardness results
for search, counting and parity
problems
about Nash equilibria
in symmetric win-lose bimatrix games.
\end{abstract}

\newpage

\section{Introduction}

\subsection{Framework and Motivation}

\subsubsection{Nash Equilibria, Win-Lose Bimatrix Games and Symmetric Games}

\noindent
Among the most fundamental computational problems 
in {\it Algorithmic Game Theory} 
are those concerning 
the {\it Nash equilibria}~\cite{N50,N51} 
of a {\it game,} 
where no 
{\it player} could unilaterally deviate 
to increase her expected {\it utility}.
Such 
problems
have been studied extensively~\cite{AKV05,BM11,BDL08,CDT09,CTV07,CLR06,CS05,CS08,DGP09,GZ89,MT10,MT10a}
for 2-player games with rational utilities
given by a bimatrix. 
There are two prominent special cases 
of such general 
{\it bimatrix games}:
{\it win-lose} and {\it symmetric}.
\begin{itemize}

\item
Utilities are taken from $\{ 0, 1 \}$
in {\it win-lose bimatrix games,}
originally put forward 
in~\cite{CS05}.

\item
In a {\it symmetric} game~\cite{N50,N51},
players have identical strategy sets
and the utility of a player
is determined by the multiset of strategies
chosen by her and the other players,
with no discrimination.
By a classical result of Nash,
every symmetric game has a {\it symmetric} Nash equilibrium,
where all players are playing
the same {\it mixed strategy}~\cite{N50,N51}.
A {\it symmetrization}
transforms a given 
bimatrix game
into a symmetric one;
the target
is that
a Nash equilibrium
for the original bimatrix game
can be reconstructed efficiently
from a (symmetric) Nash equilibrium
for the {\it symmetric bimatrix game}
(cf.~\cite{BvN50,GKT50,GHR63}).

\end{itemize} 

\subsubsection{The Search Problem}

\noindent
The fundamental theorem of Nash~\cite{N50,N51} 
that a Nash equilibrium is guaranteed to exist 
for a finite game
makes \textcolor{crimsonglory}{its} search problem {\it total};
hence,
the search problem is not
${\mathcal{NP}}$-hard 
unless ${\mathcal{NP}}
        =
        \mbox{co-}
        {\mathcal{NP}}$~\cite[Theorem 2.1]{MP91}.
In a series of breakthrough papers
culminating in~\cite{CDT09,DGP09},
it was established that,
even for bimatrix games,
the search problem is complete
for ${\mathcal{PPAD}}$~\cite{P94}, 
a complexity class capturing the computation of discrete fixed points;
under suitable formulations,
the problem is ${\mathcal{FIXP}}$-complete
for games with more than two players~\cite[Theorem 18]{EY10}.

Abbott, Kane and Valiant~\cite{AKV05} 
presented 
a polynomial time transformation
of a general bimatrix game
into a win-lose bimatrix game,
accompanied with a polynomial time map 
returning a Nash equilibrium
for the general game 
when given one 
for the win-lose game. 
So
the search problem
is ${\mathcal{PPAD}}$-hard
for win-lose bimatrix games,
which suggests that hardness
is not due to the rational utilities.
The search problem for a symmetric Nash equilibrium
in a symmetric bimatrix game
is also ${\mathcal{PPAD}}$-hard, 
thanks to the two symmetrizations from 1950
due to Brown and von Neumann~\cite{BvN50}
and due to Gale, Kuhn and Tucker~\cite{GKT50},
respectively;  
the complexity
of the search problem
for {\em any} Nash equilibrium
in a symmetric bimatrix game
has been mentioned as an open problem
by Papadimitriou~\cite[Section 2.3.1]{P07}.

\subsubsection{Decision, Counting and Parity Problems}
\label{their decision problems}

\noindent
Decision problems
arise naturally 
by twisting the search problem
in simple ways
that deprive it from its existence guarantee
for a Nash equilibrium.
Here is a non-exhaustive list of
such natural decision problems 
(
see Section~\ref{framework decision problems} 
for formal statements): 
given a game, does it have:
\begin{itemize}

\item[{\it (i)}]
At least $k+1$ Nash equilibria for some fixed integer $k \geq 1$?
(The special case with $k=1$ was
introduced in~\cite{GZ89};
the cases with $k > 1$
are considered here for the first time.)

\item[{\it (ii)}]
A Nash equilibrium where each player 
has expected utility at least 
(resp., at most)
a given number?~\cite{GZ89}


\item[{\it (iii)}]
A Nash equilibrium where
the total expected utility of players 
is at least 
(resp., at most)
a given number?~\cite{CS08}


\item[{\it (iv)}]
A Nash equilibrium where the players' supports
contain 
(resp., are contained in)
a given set of strategies?~\cite{GZ89}

\item[{\it (v)}]
A Nash equilibrium where the players' supports
have sizes greater (resp., smaller)
than a given integer?~\cite{GZ89}

\item[{\it (vi)}]
A Nash equilibrium where
each player
is using uniform probabilities
on her support?~\cite{BDL08}

\item[{\it (vii)}]
A Nash equilibrium where
some player
is using non-uniform probabilities
on her support?

\item[{\it (viii)}]
A symmetric Nash equilibrium?

\item[{\it (ix)}]
A {\it non-symmetric} Nash equilibrium? (cf.~\cite[Section 2.3.1]{P07})

\end{itemize}
Some of these
problems 
were first shown ${\mathcal{NP}}$-hard 
for symmetric bimatrix games
in the seminal paper
by Gilboa and Zemel~\cite[Section 1.2]{GZ89}.
Later
Conitzer and Sandholm~\cite[Section 3]{CS08} 
gave a unifying polynomial time reduction
from {\sf CNF SAT}
to show in one shot
${\mathcal{NP}}$-hardness results,
encompassing those from~\cite{GZ89},
for symmetric bimatrix games;
their reduction
yields games
with Nash equilibria mirroring
satisfiability parsimoniously. 
McLennan and Tourky~\cite[Theorem 1]{MT10}
refined some of these 
${\mathcal{NP}}$-hardness results
for {\it imitation bimatrix games,}
where the utility of the {\it imitator}
is $1$
if and only if she chooses
the same strategy
as the {\it mover}~\cite{MT10a}.
The problems {\it (vii)}
and {\it (viii)}
are considered here
for the first time;
{\it (viii)} is trivial for symmetric games.

\noindent
The present authors
introduced the decision problem
about the equivalence of the sets of Nash equilibria
of two given games,
which they proved $\mbox{co-${\mathcal{NP}}$}$-hard~\cite[Theorem 1]{BM11}:
\begin{itemize}

\item[{\it (x)}]
Do the two given games
have the same sets of Nash equilibria?

\end{itemize}
\noindent
An additional decision problem,
which is trivial 
for bimatrix games,
becomes ${\mathcal{NP}}$-hard 
already for 3-player games~\cite[Theorem 2]{BM11}:
\begin{itemize}

\item[{\it (xi)}]
A Nash equilibrium where
all probabilities are rational?~\cite{BM11}

\end{itemize}
\noindent
It is natural to ask whether
these problems remain ${\cal NP}$-hard
when restricted to win-lose games. 
To the best of our knowledge,
this important question has been  addressed
only in~\cite{BDL08,CS05}.
It was shown
in~\cite[Theorem 1]{CS05}
that the decision problem {\it (i)} with $k = 1$
is ${\mathcal{NP}}$-hard
for win-lose bimatrix games;
there so is 
a variant of {\it (ii)}  
for imitation win-lose bimatrix games.
The decision problem {\it (vi)}
was shown ${\mathcal{NP}}$-hard
for imitation win-lose bimatrix games 
in~\cite[Theorem 1]{BDL08}.

The {\it counting problem}
and the {\it parity problem}
for Nash equilibria
ask for the number 
and for the parity of the number of Nash equilibria,
respectively,
for a given game.
To each decision problem
there corresponds
a counting problem
and a parity problem,
asking for the number 
and the parity of the number of Nash equilibria
with the corresponding property,
respectively.
By the parsimonious property of the reduction in~\cite{CS08},
these counting (resp., parity) problems
are $\# {\mathcal{P}}$-hard
(resp., $\oplus {\mathcal{P}}$-hard)
for general bimatrix games;
the $\# {\mathcal{P}}$-hardness
(resp., $\oplus {\mathcal{P}}$-hardness)
is inherited from the
$\# {\mathcal{P}}$-hardness~\cite{V79}
(resp.,
$\oplus {\mathcal{P}}$-hardness~\cite{PZ83})
of computing the number
(resp., the parity of the number)
of satisfying assignments
for a {\sf CNF SAT} formula.

\subsubsection{State-of-the-Art and Statement of Results}

\noindent
The polynomial time transformation
of a general bimatrix game
into a win-lose bimatrix game from~\cite{AKV05} 
gave no guarantee 
on the preservation of properties
of Nash equilibria;
so, 
it had no implication on
the 
complexity
of deciding the properties 
for win-lose bimatrix games. 
Thus, 
the composition of
a polynomial time reduction 
from an ${\mathcal{NP}}$-hard problem 
to a decision problem about Nash equilibria 
for general bimatrix games 
(cf.~\cite{CS08,GZ89})
with the polynomial time transformation from~\cite{AKV05} 
does not yield 
a polynomial time reduction 
from the ${\mathcal{NP}}$-hard problem 
to the decision problems
for win-lose bimatrix games, 
and their complexity 
remained open.

{\em In this work, 
         we settle the complexity of 
         the decision problems 
         about Nash equilibria {\em \cite{BM11,BDL08,CS05,CS08,GZ89,MT10,MT10a,MVY15}}
         for symmetric win-lose bimatrix games.}
Our main result is 
that these 
problems 
are ${\cal NP}$-hard 
for symmetric win-lose bimatrix games
(Theorems~\ref{mainextended} and~\ref{maintheorem symmetric}). 
Further,
deciding the existence
of a symmetric Nash equilibrium
is ${\mathcal{NP}}$-hard
for win-lose bimatrix games
(Theorem~\ref{second last minute theorem}),
and 
of a rational Nash equilibrium~\cite{BM11} 
is ${\mathcal{NP}}$-hard
for win-lose 3-player games
(Theorem~\ref{last minute theorem}).
As a byproduct,
we derive,
for symmetric win-lose bimatrix games,
the ${\mathcal{PPAD}}$-hardness
of the search problem
(Theorem~\ref{new result}),
the $\# {\mathcal{P}}$-hardness
of the counting problem
and
the $\oplus {\mathcal{P}}$-hardness
of the parity problem
(Theorem~\ref{sharp pi complete symmetric win-lose}),
and the $\# {\mathcal{P}}$-hardness
and the $\oplus {\mathcal{P}}$-hardness
of the counting and the parity version,
respectively,
of each decision problem
(Theorem~\ref{mainextended}).

\subsection{The Techniques}
\label{tools}

\noindent
We combine three powerful techniques:
{\sf (1)}
{\it Gadget games}
(Section~\ref{gadget games}).
{\sf (2)}
A {\it win-lose reduction}
(Section~\ref{winlose reduction}).
{\sf (3)}
The classical {\it ${\mathsf{GHR}}$-symmetrization}
~\cite{GHR63},
with a new analysis tailored to win-lose bimatrix games
(Section~\ref{winlose gkt symmetrization}).
All three techniques involve the
{\it positive utility property}:
the property is enjoyed
by the gadget games,
which form 
a key component of the
reduction,
and it is required
for the new analysis of the ${\mathsf{GHR}}$-symmetrization.
The property 
requires
that each player may,
in response to the 
choices
of the other players,
always choose a strategy
to make her utility strictly positive; 
it is strictly weaker than 
the {\it strictly positive utilities property,}
assumed 
for the analysis of the ${\mathsf{GHR}}$-symmetrization
in~\cite{JPT86}. 
We prove that
computing a Nash equilibrium
for a win-lose bimatrix game
with the 
property
is as hard as computing
a Nash equilibrium
for a win-lose bimatrix game
(Proposition~\ref{pup is not restrictive});
hence,
it is ${\mathcal{PPAD}}$-hard.
So
assuming the positive utility property
for 
symmetric win-lose bimatrix games
does not simplify the search problem.

\subsubsection{Gadget Games}
\label{introduction gadget games}

\noindent
A {\it gadget game}
is a {\em fixed} 
game with a small number of players,
which is void by design
of the property
associated with some 
decision problem
about Nash equilibria
(from 
Section~\ref{framework decision problems}).
We present 
win-lose gadget games
covering all such properties.
For example,
the win-lose 3-player irrational game
${\widehat{{\mathsf{G}}}}_{2}$
(Section~\ref{irrational game})
has a single irrational Nash equilibrium,
dismatching the 
problem {\it (x)};
the win-lose non-uniform game
${\widehat{{\mathsf{G}}}}_{3}$
(Section~\ref{nonuniform game})
has no uniform Nash equilibrium,
dismatching the 
problem {\it (vi)};
the win-lose non-symmetric game
${\widehat{{\mathsf{G}}}}_{4}$
(Section~\ref{nonsymmetric game})
has no symmetric Nash equilibrium,
dismatching the 
problem {\it (viii)};
for a given integer $k \geq 1$,
the win-lose diagonal game
${\widehat{{\mathsf{G}}}}_{5}[k]$
(Section~\ref{diagonal game})
has exactly $k$ Nash equilibria,
dismatching the 
problem {\it (i)}.

\subsubsection{The Win-Lose Reduction}
\label{introduction winlose reduction}

\noindent
The technical backbone of the complexity results
for decision, counting and parity problems
is a {\it win-lose reduction} we design,
taking a {\em fixed} win-lose gadget game
${\widehat{{\mathsf{G}}}}$
with the positive utility property
as a parameter
(Section~\ref{winlose reduction}).
The reduction transforms
a given {\sf 3SAT} formula ${\mathsf{\phi}}$
into a win-lose game
${\mathsf{G}} 
  =
  {\mathsf{G}}({\widehat{{\mathsf{G}}}},
                                          {\mathsf{\phi}})$;
${\mathsf{G}}$ and ${\widehat{{\mathsf{G}}}}$
have the same number of players.
We prove that 
each Nash equilibrium for ${\widehat{{\mathsf{G}}}}$
is always a Nash equilibrium
for ${\mathsf{G}}({\widehat{{\mathsf{G}}}},
                              {\mathsf{\phi}})$
(Lemma~\ref{f is PNE} (Condition {\sf (1)})),                         
while 
${\mathsf{G}}({\widehat{{\mathsf{G}}}},
                         {\mathsf{\phi}})$
has additional  Nash equilibria
if and only if
${\mathsf{\phi}}$
is satisfiable;
those
are related parsimoniously
to the satisfying assignments of ${\mathsf{\phi}}$.
More important,
each additional
Nash equilibrium
enjoys properties
that do not depend on
${\widehat{{\mathsf{G}}}}$
(Propositions~\ref{if unsatisfied}
and~\ref{final lemma}).
So
a decision problem
associated with some particular property
reduces to deciding the inequivalence
of a win-lose game with a {\em fixed}
win-lose gadget game
dismatching the property,
and their equivalence
is co-${\mathcal{NP}}$-hard.

\subsubsection{The ${\mathsf{GHR}}$-Symmetrization}
\label{the ghr symmetrization}

\noindent
To extend hardness results
from win-lose to symmetric win-lose bimatrix games,
we \textcolor{black}{seek}
{\it win-lose} symmetrizations
transforming a win-lose bimatrix game 
into a symmetric one. 
The {\it ${\mathsf{GHR}}$-symmetrization}~\cite{GHR63}
is the single win-lose symmetrization we know of;
for emphasis,
we shall call it the {\it win-lose ${\mathsf{GHR}}$-symmetrization.}\footnote{The symmetrization
                                                                                               due to Brown and von Neumann~\cite{BvN50}
                                                                                               involves the sum
                                                                                               of the two matrices,
                                                                                               which may increase
                                                                                               the number of utility values.
                                                                                               The 
                                                                                               symmetrization due to Gale, Kuhn and Tucker~\cite{GKT50}
                                                                                               introduces 
                                                                                               utilities $-1$.
                                                                                               So they both result in
                                                                                               more than two
                                                                                               values for the utilities,
                                                                                               and none of the symmetrizations
                                                                                               is win-lose.} 
We use tools from~\cite{JPT86}
to provide a new analysis
of the ${\mathsf{GHR}}$-symmetrization,
tailored to win-lose bimatrix games
with the positive utility property
(Section~\ref{winlose gkt symmetrization}),
which yields a tight characterization
of the Nash equilibria 
for 
the resulting symmetric win-lose bimatrix game 
${\widetilde{{\mathsf{G}}}}$:
they may only result,
albeit in a non-parsimonious way,
as {\it balanced mixtures}~\cite{JPT86}
involving Nash equilibria for ${\mathsf{G}}$
(Theorems~\ref{frombasictosymmetric}
and~\ref{fromsymmetrictobasic2}).\footnote{The balanced mixture was introduced
                                                                           in~\cite{JPT86}
                                                                           for the analysis of the
                                                                           ${\mathsf{GHR}}$-symmetrization;
                                                                           it was later used in~\cite{JJPT92},
                                                                           where it was called 
                                                                           the {\it ${\mathsf{GKT}}$-product,}
                                                                           for the analysis of the ${\mathsf{GKT}}$-symmetrization~\cite{GKT50}.}
By Propositions~\ref{if unsatisfied}
and~\ref{final lemma},
the set of balanced mixtures
is determined by
the satisfiability of ${\mathsf{\phi}}$; 
hence,
the ``cascade'' of the win-lose reduction
and the win-lose ${\mathsf{GHR}}$ symmetrization
is a non-parsimonious reduction
from {\sf 3SAT}
to decision problems
about Nash equilibria in
symmetric win-lose bimatrix games.

\subsection{Three-Steps Plan of the ${\mathcal{NP}}$-Hardness Proof}
\label{introduction complexity results}

\noindent

\begin{itemize}

\item
\underline{{\it Step 1:}}
For a property
of Nash equilibria 
for symmetric win-lose bimatrix games,
fix a gadget game:
a win-lose bimatrix game
${\widehat{{\mathsf{G}}}}$
whose Nash equilibria
dismatch the property. 

\item
\underline{{\it Step 2:}}
Apply the win-lose reduction
with parameter ${\widehat{{\mathsf{G}}}}$
on the formula ${\mathsf{\phi}}$
to get the win-lose bimatrix game
${\mathsf{G}}
  :=
  {\mathsf{G}}({\widehat{{\mathsf{G}}}}, 
                        {\mathsf{\phi}})$.
\begin{center}
\fbox{
\begin{minipage}{5.5in}
\begin{itemize}

\item
When ${\mathsf{\phi}}$
is unsatisfiable,
the Nash equilibria for 
${\mathsf{G}}({\widehat{{\mathsf{G}}}}, {\mathsf{\phi}})$
are those for ${\widehat{{\mathsf{G}}}}$
(Proposition~\ref{if unsatisfied});
hence,
they dismatch the property.

\item
When ${\mathsf{\phi}}$ is satisfiable,
there are additional Nash equilbria
for ${\mathsf{G}}({\widehat{{\mathsf{G}}}}, {\mathsf{\phi}})$,
which satisfy the property
(Proposition~\ref{final lemma}).
(It follows that
the properties of the Nash equilibria
for ${\widehat{{\mathsf{G}}}}$
and of the additional ones
for ${\mathsf{G}}({\widehat{{\mathsf{G}}}}, {\mathsf{\phi}})$
are ``conflicting''.)

\item[$\Rightarrow$]
${\mathsf{G}}$ 
has a Nash equilibrium
matching the property
if and only if
${\mathsf{\phi}}$
is satisfiable.
\begin{itemize}

\item[$\Rightarrow$]
The associated decision problem
is ${\mathcal{NP}}$-hard
for win-lose bimatrix games.

\end{itemize}

\end{itemize}
\end{minipage}
}
\end{center}
\noindent
Note that {\it Step 2}
only allows proving
the ${\mathcal{NP}}$-hardness
of decision problems
associated with properties
matched by the additional Nash equilibria
for the win-lose game
${\mathsf{G}}({\widehat{{\mathsf{G}}}},
                        {\mathsf{\phi}})$.

\item
\underline{{\it Step 3:}} 
Apply the win-lose ${\mathsf{GHR}}$-symmetrization 
on ${\mathsf{G}}$
to get the symmetric win-lose bimatrix game
${\widetilde{{\mathsf{G}}}}
  :=
  {\mathsf{GHR}}( {\mathsf{G}}({\widehat{{\mathsf{G}}}}, 
                                                            {\mathsf{\phi}})
                            )$,
whose Nash equilibria                             
may only result
as balanced mixtures
involving Nash equilibria
for ${\mathsf{G}}$.
Some balanced mixtures preserve the properties
of the Nash equilibria for ${\mathsf{G}}$,
while other may dismatch them.
This allows incorporating
properties either matched or dismatched
by the additional Nash equilibria for ${\mathsf{G}}$.
We show that
the associated decision problem
is ${\mathcal{NP}}$-hard
for symmetric win-lose bimatrix games
by establishing a suitable equivalence
to the satisfiability of ${\mathsf{\phi}}$
as follows:

\begin{center}
\fbox{
\begin{minipage}{5.5in}
There are three possible cases
in achieving
the equivalence between
the existence of a Nash equilibrium
matching the property
for ${\widetilde{{\mathsf{G}}}}$
and the satisfiablity of ${\mathsf{\phi}}$:
\begin{itemize}

\item[{\sf (1)}]
The existence
of a Nash equilibrium
matching the property
for the win-lose game ${\mathsf{G}}$
is equivalent to
the satisfiability of ${\mathsf{\phi}}$,
and two extra conditions hold:
\begin{itemize}

\item
When ${\mathsf{\phi}}$ is unsatisfiable,
every balanced mixture
dismatches the property.

\item
When $\phi$ is satisfiable,
there is a balanced mixture
matching the property.

\end{itemize}
In this case,
the equivalence
is preserved by the win-lose ${\mathsf{GHR}}$-symmetrization. 

\item[{\sf (2)}]
There is no Nash equilibrium
for ${\mathsf{G}}$
matching the property,
no matter whether ${\mathsf{\phi}}$
is satisfiable or not,
and two extra conditions hold:
\begin{itemize}

\item
When ${\mathsf{\phi}}$ is unsatisfiable,
every balanced mixture
dismatches the property.

\item
When ${\mathsf{\phi}}$ is satisfiable,
there is a balanced mixture
matching the property.

\end{itemize}
Hence,
the existence of a Nash equilibrium
matching the property
for 
${\widetilde{{\mathsf{G}}}}$
is equivalent to
the satisfiability of ${\mathsf{\phi}}$. 

\item[{\sf (3)}]
When ${\mathsf{\phi}}$ is unsatisfiable,
there is a balanced mixture
matching the property.
When ${\mathsf{\phi}}$
is satisfiable,
there is an additional  balanced mixture
matching the property.
Hence,                     
this excludes the equivalence
between
the existence
of a Nash equilibrium
matching the property
for ${\widetilde{{\mathsf{G}}}}$
and the satisfiability of ${\mathsf{\phi}}$.

We extend ${\widetilde{{\mathsf{G}}}}$ to
${\widetilde{{\mathsf{G}}}}
  \parallel
  {\widehat{{\mathsf{G}}}}_{5}[k]$
by ``embedding'' 
the symmetric win-lose gadget 
${\widehat{{\mathsf{G}}}}_{5}[k]$
as a subgame
so that the balanced mixtures
arising when ${\mathsf{\phi}}$ is unsatisfiable
are ``destroyed',
while the balanced mixtures
arising when ${\mathsf{\phi}}$
is satisfiable
are not ``destroyed''.
This induces the equivalence of
the existence
of a Nash equilibrium
matching the property for
${\widetilde{{\mathsf{G}}}}
  \parallel
  {\widehat{{\mathsf{G}}}}_{5}[k]$
and the satisfiability of ${\mathsf{\phi}}$.

\item[$\Rightarrow$]
Since equivalence 
to the satisfiability of ${\mathsf{\phi}}$
holds in all cases,
the ${\mathcal{NP}}$-hardness 
of the associated decision problem
for symmetric win-lose games
follows.

\end{itemize}

\end{minipage}
}
\end{center}
\end{itemize}


\subsection{The Complexity Results and Significance}
\label{contribution}

\subsubsection{Decision Problems}

The three-steps proof plan
from Section~\ref{introduction complexity results}
is used to yield,
as our main result,
the ${\mathcal{NP}}$-hardness
of deciding a handful of properties 
for symmetric win-lose bimatrix games 
(Theorems~\ref{mainextended} and~\ref{maintheorem symmetric}).
These complexity results imply that
symmetric win-lose bimatrix games
are as complex
as general bimatrix games
with respect to the handful of decision problems 
about Nash equilibria
considered before in~\cite{BM11,BDL08,CS05,CS08,GZ89,MT10,MT10a,MVY15}.

While the win-lose reduction applies
to games with any number $r \geq 2$
of players,
the ${\mathsf{GHR}}$-symmetrization 
is specific to bimatrix games.
Hence,
{\it Step 1} and {\it Step 2}
suffice on their own
for proving
the ${\mathcal{NP}}$-hardness of decision problems
about Nash equilibria
which either remain trivial for symmetric games
(such as deciding the existence
of a symmetric Nash equilibrium)
or become non-trivial for win-lose games
with more than two players
(such as deciding the existence
of a rational Nash equilibrium):
\begin{itemize}

\item
Choosing ${\widehat{{\mathsf{G}}}}$
as the win-lose bimatrix game ${\widehat{{\mathsf{G}}}}_{5}$
with no symmetric Nash equilibrium
yields the ${\mathcal{NP}}$-hardness
of deciding the existence
of a symmetric Nash equilibrium
for win-lose bimatrix games
(Theorem~\ref{second last minute theorem}).
So
win-lose bimatrix games
are as complex as general bimatrix games
with respect to deciding the existence
of a symmetric Nash equilibrium.

\item
Choosing ${\widehat{{\mathsf{G}}}}$
as the win-lose 3-player game ${\widehat{{\mathsf{G}}}}_{2}$
with a single irrational Nash equilibrium
yields the ${\mathcal{NP}}$-hardness
of deciding the existence 
of a rational Nash equilibrium
for win-lose 3-player games
(Theorem~\ref{last minute theorem}).
So
win-lose 3-player games
are as complex as general 3-player games
with respect to deciding the 
existence of a rational Nash equilibrium.

\end{itemize}
These results
represent an analog
of the earlier result 
that win-lose bimatrix games 
are as complex as general bimatrix games 
with respect to the search problem
for a Nash equilibrium~\cite{AKV05}.

\subsubsection{Search, Counting and Parity Problems}

We show that computing a Nash equilibrium
for a symmetric win-lose bimatrix game
is ${\mathcal{PPAD}}$-hard
(Theorem~\ref{new result});
the proof appeals to the characterization
of the Nash equilibria for the win-lose
${\mathsf{GHR}}$-symmetrization
of a win-lose bimatrix game
(Theorem~\ref{fromsymmetrictobasic2}).

Recall that
the reduction used
for the ${\mathcal{NP}}$-hardness proof
is non-parsimonious.
Hence,
the counting and parity problems
about the number and the parity of the number
of Nash equilibria
for a symmetric win-lose bimatrix game,
as well as the counting and parity versions
of decision problems
about Nash equilibria,
are not immediately $\# {\mathcal{P}}$-hard
and $\oplus {\mathcal{P}}$-hard,
respectively.
Nevertheless,
the ${\mathcal{NP}}$-hardness proof
yields $\# {\mathcal{P}}$-hardness
and $\oplus {\mathcal{P}}$-hardness
results as well:
\begin{itemize}

\item
We show that computing 
the number 
(resp., the parity of the number)
of Nash equilibria
for a symmetric win-lose bimatrix game
is $\# {\mathcal{P}}$-hard
(resp., $\oplus {\mathcal{P}}$-hard)
(Theorem~\ref{sharp pi complete symmetric win-lose}).
The proof 
draws from simple formulas
for the numbers of Nash equilibria
for ${\mathsf{G}}$
and ${\widetilde{{\mathsf{G}}}}$
in terms of the numbers of Nash equilibria for
${\widehat{{\mathsf{G}}}}$
and 
of satisfying assignments for ${\mathsf{\phi}}$,
denoted as $\# {\mathsf{\phi}}$;
the formulas follow from properties
of the gadget games,
the win-lose reduction
and the win-lose
${\mathsf{GHR}}$-symmetrization.
Solving the formula for $\# {\mathsf{\phi}}$
yields the $\# {\mathcal{P}}$-hardness; 
computing the parity of $\# {\mathsf{\phi}}$,
denoted as $\oplus {\mathsf{\phi}}$\textcolor{crimsonglory}{,}
from the formula
yields the $\oplus {\mathcal{P}}$-hardness.

\item
We examine 
the balanced mixtures
of Nash equilibria for ${\mathsf{G}}$
matching each property
to derive simple 
formulas
for the numbers of Nash equilibria
for ${\widetilde{{\mathsf{G}}}}$
matching the property.
Hence, 
the counting versions of these decision problems
are $\# {\mathcal{P}}$-hard 
(Theorem~\ref{mainextended}).

Furthermore,
all but two
of the formulas
are {\it parity-preserving}:
the parity of the number of Nash equilibria
for ${\widetilde{{\mathsf{G}}}}$
matching the property
yields $\oplus {\mathsf{\phi}}$;
this implies the $\oplus {\mathcal{P}}$-hardness 
of all but \textcolor{crimsonglory}{one} of the parity versions
(Theorem~\ref{mainextended}),
the exception made by a trivial parity problem. 
In this sense,
the ``cascade'' of the win-lose reduction
and the win-lose ${\mathsf{GHR}}$-symmetrization
is {\it parity-preserving}:
it yields the parity
of the number of satisfying assignments.
As far as we know,
this is the {\em first} non-parsimonious,
yet parity-preserving,
reduction to decision problems about Nash equilibria.

\end{itemize}

\subsection{Evaluation and Comparison to Related Work}
\label{related work and significance}

\noindent
None of the works~\cite{BM11,BDL08,CS05,CS08,GZ89,MT10,MT10a,MVY15}
on the complexity of decision and counting problems
about Nash equilibria
in bimatrix games
considered the two restrictions 
to win-lose bimatrix and symmetric bimatrix games
in simultaneity;
neither did 
the works~\cite{AKV05,CDT09,CTV07,DGP09}
on the complexity of the search problem. 
This work encompasses {\em all} of the decision problems,
together with their counting and parity versions,
in the common framework composed
of the gadget games,
the win-lose reduction
and the win-lose ${\mathsf{GHR}}$-symmetrization.
So, 
problem-specific reductions and techniques,
such as the {\it regular subgraphs} technique from~\cite{BDL08}
or the {\it good assignments} technique from~\cite{CS05},
are not necessary.

\subsubsection{Complexity Results}

Theorems~\ref{mainextended},
\ref{maintheorem symmetric}
and~\ref{last minute theorem}
improve and extend previous ${\mathcal{NP}}$-hardness
and $\mbox{co-}{\mathcal{NP}}$-hardness results
for the decision problems
from~\cite{BM11,BDL08,CS05,CS08,GZ89,MT10,MT10a,MVY15}
as follows:
\begin{itemize}

\item
The ${\mathcal{NP}}$-hardness
of the handful of decision problems
improves the results in~\cite{CS08,GZ89},
which applied to general symmetric bimatrix games,
and extends their refinements
in~\cite{MT10,MT10a},
which applied to imitation bimatrix games.

\item
The ${\mathcal{NP}}$-hardness
of deciding the existence
of a uniform Nash equilibrium
for a symmetric win-lose bimatrix game
extends~\cite[Theorem 1]{BDL08}, 
which applied to imitation win-lose bimatrix games.

\item
The ${\mathcal{NP}}$-hardness
of deciding the existence of $k+1$ Nash equilibria,
with $k \geq 1$,
for a symmetric win-lose bimatrix game
improves~\cite[Theorem 1]{CS05},
which addressed the special case $k=1$
and applied to win-lose bimatrix games.

\item
The ${\mathcal{NP}}$-hardness
of deciding the existence
of a non-symmetric Nash equilibrium
for a symmetric win-lose bimatrix game
improves~\cite[Theorem 3]{MVY15},
which applied to general symmetric bimatrix games.

\item
The $\mbox{co-}{\mathcal{NP}}$-hardness of deciding
the {\it Nash-equivalence}
of a given symmetric win-lose lose bimatrix game
with the ${\mathsf{GHR}}$-symmetrization
of a fixed win-lose gadget game
improves~\cite[Theorem 1]{BM11},
which established the Nash-equivalence
of a given general bimatrix game
with a fixed general gadget game.

\item
The ${\mathcal{NP}}$-hardness
of deciding the existence
of a rational Nash equilibrium
for a win-lose 3-player game
improves~\cite[Theorem 2]{BM11},
which applied to general 3-player games.

\end{itemize}
\noindent
Corresponding extensions and improvements follow
for the $\# {\mathcal{P}}$-hardness
of the counting versions of these decision problems.
In particular:
\begin{itemize}

\item
The $\# {\mathcal{P}}$-hardness
of counting the number of non-symmetric Nash equilibria
for a symmetric win-lose bimatrix game
improves~\cite[Theorem 5]{MVY15},
which applied to general symmetric bimatrix games.

\end{itemize}

Of particular interest is the
$\# {\mathcal{P}}$-hardness
(resp., $\oplus {\mathcal{P}}$-hardness)
of computing,
for a given symmetric win-lose bimatrix game,
the number 
(resp., the parity of the number)
of symmetric Nash equilibria;
recall that the corresponding decision problem,
asking for the existence of a symmetric Nash equilibrium
for a given symmetric game,
is trivial by the early result of Nash~\cite{N50,N51}.

To the best of our knowledge, 
the proof in~\cite{V05}
that computing the parity of
the number of satisfying assignments
for a {\it read-twice} formula,
where each variable occurs at most twice,
is $\oplus {\mathcal{P}}$-hard
is the only proof of $\oplus {\mathcal{P}}$-hardness
employing a reduction
from a corresponding ${\mathcal{NP}}$-completeness proof
which, although non-parsimonious,
is parity-preserving.

Figure~\ref{table1} tabulates
the presented complexity results
for decision problems
in comparison to those in~\cite{BM11,BDL08,CS05,CS08,GZ89,MVY15}.
Theorem~\ref{new result}
improves the ${\mathcal{PPAD}}$-hardness
of the search problem for win-lose bimatrix games
from~\cite{AKV05}
to symmetric win-lose bimatrix games.
The $\# {\mathcal{P}}$-hardness
of computing the number of Nash equilibria
for a symmetric win-lose bimatrix game
in Theorem~\ref{sharp pi complete symmetric win-lose}
improves~\cite[Corollary 12]{CS08},
which applied to general symmetric bimatrix games.

\subsubsection{The Win-Lose Reduction}

The unifying reduction from {\sf CNF SAT},
introduced in~\cite{CS08} and further developed in~\cite{BM11,MM16},
is inadequate to cover win-lose games.
(See Section~\ref{conitzer sandholm inadequacy} for a technical discussion.)
New ideas and technical constructs,
such as {\it pair variables,}
were needed for the win-lose reduction,
which thus improves vastly
in yielding a win-lose game
while still preserving the equivalence
between 
the existence
of additional Nash equilibria
for the game
${\mathsf{G}}\left( {\widehat{{\mathsf{G}}}},
                                 {\mathsf{\phi}}
                       \right)$     
and the satisfiability of ${\mathsf{\phi}}$.
We needed the finer structure
of a {\sf 3SAT} formula
in order to guarantee
that certain deviations
in the game constructed 
by the win-lose reduction
are non-profitable.
Specifically,
the proofs
for Propositions~\ref{if unsatisfied}
and~\ref{final lemma}
rely on choosing ${\mathsf{\phi}}$ as
a {\sf 3SAT} formula
in an essential way.

The win-lose reduction
generalizes the reduction from~\cite{CS08},
which yielded a bimatrix game,
to yield an $r$-player game,
with $r \geq 2$;
it is this generalization that has enabled
showing complexity results about decision problems,
such as deciding the existence of a rational Nash equilibrium
(Theorem~\ref{last minute theorem}),
which are trivial for bimatrix games
but become ${\mathcal{NP}}$-hard for $3$-player games.

We note that the reduction in~\cite{CS08}
yielded the
${\mathcal{NP}}$-hardness
of approximate versions
of some of the decision problems
over general bimatrix games.
We anticipate that the presented composition of the win-lose reduction
and the win-lose ${\mathsf{GHR}}$-symmetrization
yields corresponding ${\mathcal{NP}}$-hardness results
over symmetric win-lose bimatrix games.

\subsubsection{The ${\mathsf{GHR}}$-Symmetrization}

\noindent
Our analysis of the
${\mathsf{GHR}}$-symmetrization~\cite{GHR63}
yields the {\em first} complete characterization
of the Nash equilibria 
for the symmetric game
resulting from a win-lose symmetrization. 
For its previous analysis in~\cite{JPT86},
it was assumed that all utilities
in the original bimatrix game
were strictly positive
in order to guarantee
that the $0$ utilities
added for the ${\mathsf{GHR}}$-symmetrization
are strictly less than any utility;
thus, 
the original bimatrix game
were not a win-lose game.
Recall that shifting and scaling the utilities
does not alter the set of Nash equilibria for a game.
Thus,
a bimatrix game with at most {\em two} values 
for the (strictly positive) utilities
could be transformed into an equivalent win-lose game
by shifting and scaling the utilities.
(Note that the case where
all utilities in the original bimatrix games are equal
is degenerate.)
But then
it is no more the case that
the $0$ utilities
added for the ${\mathsf{GHR}}$-symmetrization
are strictly less
than both $0$ and $1$,
and the analysis of the ${\mathsf{GHR}}$-symmetrization from~\cite{JPT86}
is no more applicable.

To circumvent this difficulty,
we need a different assumption
on the utilities
in the original bimatrix game
which
is not too restrictive to make
the algorithmic problems easier,
while it is strong enough
to yield a tight characterization
of the Nash equilibria
for the ${\mathsf{GHR}}$-symmetrization;
this is the positive utility property.
We analyze the ${\mathsf{GHR}}$-symmetrization
for win-lose bimatrix games
(Section~\ref{winlose gkt symmetrization})\textcolor{black}{,
replacing}
the assumption of strictly positive utilities from~\cite{JPT86}
with the positive utility property. 
The resulting characterization of Nash equilibria
for the ${\mathsf{GHR}}$-symmetrization
of a win-lose bimatrix game
with the positive utility property
is similar to the characterization
of Nash equilibria
for the ${\mathsf{GHR}}$-symmetrization
of a bimatrix game
with strictly positive utilities
in~\cite{JPT86}:
Cases {\sf (C'.2)}
and {\sf (C'.3)}
from Theorem~\ref{fromsymmetrictobasic2}
correspond to the cases addressed
in~\cite[Theorem 4.1]{JPT86};
\cite[Theorem 4.2]{JPT86}
and~\cite[Theorem 4.3]{JPT86}
correspond to 
Theorems~\ref{frombasictosymmetric}
and~\ref{fromsymmetrictobasic2},
respectively.

\subsubsection{Tractable Cases}

\noindent
Bil\`{o} and Fanelli~\cite[Section 4]{BF10}
consider a Linear Programming formulation,
denoted as {\it LR,}
of Nash equilibria in bimatrix games,
which is a relaxation of a corresponding
Quadratic Programming formulation
they propose.
They show~\cite[Theorem 1]{BF10}
that any feasible solution to {\it LR}
is a Nash equilibrium when 
the bimatrix game is {\it regular}~\cite[Definition 4]{BF10}:
the sum of corresponding entries
in the two matrices
remains constant either accross rows
or across columns;
so,
regular bimatrix games
are the {\it rank-1} bimatrix games,
encompassing {\it zero-sum} bimatrix games.
As a consequence of~\cite[Theorem 1]{BF10},
a Nash equilibrium optimizing any objective function
involving the players' utilities
and meeting any set of constraints
expressible through Linear Programming
can be computed in polynomial time
(through solving {\it LR}).
Since nearly all decision problems
about Nash equilibria 
studied in this work are so expressible,
it follows that they are polynomial time solvable
when restricted to regular bimatrix games.
(A notable exception is problem {\it (v)}
                                                                                     from Section~\ref{their decision problems},
                                                                                     which remains ${\mathcal{NP}}$-hard
                                                                                     even when restricted to zero-sum games~\cite{GZ89}.) 
This positive result
stands in contrast
to the established ${\mathcal{NP}}$-hardness
for symmetric win-lose bimatrix games.
Other positive results on the search problem
for a Nash equilibrium
in {\it planar,} {\it sparse} 
and {\it $K_{3,3}$}
and 
{\it $K_{5}$ 
minor-free}
and
{\it minor-closed}
win-lose bimatrix games
appear in~\cite{AOV07},~\cite{CLR06} and~\cite{DK11},
respectively.

\subsection{Paper Organization}
\label{road map}

\noindent
Section~\ref{mathematical preliminaries}
introduces the game-theoretic framework.
Section~\ref{winlose bimatrix with pup}
considers the positive utility property
for win-lose bimatrix games.
The win-lose gadget games
are presented
in Section~\ref{gadget games}.
Section~\ref{winlose reduction}
treats the win-lose reduction
and its properties.
The win-lose ${\mathsf{GHR}}$-symmetrization
and its properties are
analyzed in Section~\ref{winlose gkt symmetrization}.
Section~\ref{complexity results}
presents the complexity results.
We conclude,
in Section~\ref{epilogue},
with a discussion of the results
and some open problems.

\section{Framework and Preliminaries}
\label{mathematical preliminaries}

\noindent
Games,
Nash equilibria
and their decision problems
are treated in
Sections~\ref{framework games},~\ref{framework equilibria}
and~\ref{framework decision problems},
respectively.

\subsection{Games}
\label{framework games}

\noindent
A {\it game}
is a triple
${\mathsf{G}}
 =
 \langle [r],
         \{ {\sf \Sigma}_{i} \}_{i \in [r]},
         \{ {\sf U}_{i} \}_{i \in [r]}
 \rangle$,
where:
{\it (i)}
$[r] = \{ 1, \ldots, r \}$
is a finite set
of {\it players}
with $r \geq 2$,
and
{\it (ii)}
for each player $i \in [r]$,
${\sf \Sigma}_{i}$ is the set of {\it strategies}
for player $i$,
and
the {\it utility function}
${\sf U}_{i}$
is a function
${\sf U}_{i}: \times_{k \in [r]}
                {\sf \Sigma}_{k}
              \rightarrow
              {\mathbb{R}}$
for player $i$.
The game ${\mathsf{G}}$
is
{\it win-lose} 
(resp., {\it general})
if for each player $i\in [r]$, 
the utility function
is a function
${\sf U}_{i}: \times_{k \in [r]}
                {\sf \Sigma}_{k}
              \rightarrow
              \{0,1\}$
(resp.,
${\sf U}_{i}: \times_{k \in [r]}
                       {\sf \Sigma}_{k}
                     \rightarrow
                     {\mathbb{Q}}$).              
For each player $i\in [r]$, denote
${\sf \Sigma}_{-i}
 =
 \times_{k \in [r]\setminus\{i\}}
   {\sf \Sigma}_{k}$;
denote
${\sf \Sigma}
 =
 \times_{k \in [r]}
   {\sf \Sigma}_{k}$.
For an integer $r \geq 2$,
$r$-${\cal G}$
is the set of
{\it $r$-player games};
so,
${\cal G}
 =
 \bigcup_{r \geq 2}
   r\mbox{-}{\cal G}$
is the set of all games.

A {\it profile}
is a tuple
${\bf s} \in {\sf \Sigma}$
of $r$ strategies,
one for each player.
For a profile ${\bf s}$,
the vector
${\sf U}({\bf s})
 =
 \left\langle {\sf U}_{1}({\bf s}),
              \ldots,
              {\sf U}_{r}({\bf s})
 \right\rangle$
is called the {\it utility vector}.
For a profile ${\bf s}$
and a strategy $t \in {\sf \Sigma}_{i}$
of player $i$,
denote as ${\bf s}_{-i} \diamond t$
the profile obtained
by substituting $t$
for $s_{i}$ in ${\bf s}$.
A {\it partial profile}
${\bf s}_{-i}
 \in
 {\sf\Sigma}_{-i}$
is a tuple of $r-1$ strategies,
one for each player other than $i$.
We define:

\begin{definition}
\label{positive utility property definition}
The game ${\mathsf{G}}$
has the {\it positive utility property} 
if for each player $i \in [r]$ 
and each partial profile 
${\bf s}_{-i}
 \in
 {\mathsf{\Sigma}}_{-i}$, 
there is a strategy $t = t({\bf s}_{-i})
                                \in
                               {\mathsf{\Sigma}}_{i}$ 
such that
${\mathsf{U}}_{i}\left( {\bf s}_{-i}
                        \diamond
                        t
                 \right)
 > 0$.
\end{definition}

\noindent
For win-lose bimatrix games, 
the positive utility property
means that
the utility matrix 
of the row (resp., column) player 
has no all-zeros column (resp., row).

A {\it bimatrix game}
is a 2-player game with
player 1,
or {\it row player,}
and
player 2,
or {\it column player,}
with ${\mathsf{\Sigma}}_{1}
      =
      {\mathsf{\Sigma}}_{2}
      =
      [n]$,\footnote{The assumption that the two players
                               have equal numbers of strategies
                               is without loss of generality
                               since equality can be achieved,
                               without altering the positive utility property,
                               by adding "dummy" strategies,
                               which are never played
                               by a utility-maximizing player.}    
which is represented
as the pair of matrices
$\left\langle {\mathsf{R}},
              {\mathsf{C}}
 \right\rangle$,
where for each profile
$\left\langle s_{1},
              s_{2}
 \right\rangle
 \in
 {\mathsf{\Sigma}}$,
${\mathsf{U}}_{1}\left( \left\langle s_{1}, s_{2}
                        \right\rangle
                 \right)
 =                
 {\mathsf{R}}[s_{1}, s_{2}]$
and
${\mathsf{U}}_{2}\left( \left\langle s_{1}, s_{2}
                        \right\rangle
                 \right)
 =
 {\mathsf{C}}[s_{1}, s_{2}]$.
The game
${\mathsf{G}}
 =
 \left\langle {\mathsf{R}},
              {\mathsf{C}}
 \right\rangle$
is {\it symmetric}
if
${\mathsf{U}}_{1}\left( \left\langle s_{1}, s_{2}
                        \right\rangle
                 \right)
 =
 {\mathsf{U}}_{2}\left( \left\langle s_{2}, s_{1}
                        \right\rangle
                 \right)$:
exchanging players and strategies
preserves utilities;
so,
${\mathsf{R}}
 =
 {\mathsf{C}}^{\mbox{{\rm T}}}$.
For a constant $c$,
the game ${\mathsf{G}}$
is {\it $c$-sum}
if for each profile
$\left\langle s_{1},
              s_{2}
 \right\rangle$,
${\mathsf{U}}_{1}\left( \left\langle s_{1}, s_{2}
                        \right\rangle
                 \right)
 +
 {\mathsf{U}}_{2}\left( \left\langle s_{1}, s_{2}
                        \right\rangle
                 \right)
  =
  c$;
so,
${\mathsf{R}}[s_{1}, s_{2}]  
 +
 {\mathsf{C}}[s_{1}, s_{2}]  
 =
 c$.
For a player $i \in [2]$
in the game ${\mathsf{G}}$,
we shall denote as 
$\overline{i}$
the player other than $i$;
so 
$\overline{i}
 \in
 [2] \setminus \{ i \}$.

A {\it mixed strategy}
for player $i \in [r]$
is a probability distribution
$\sigma_{i}$
on her strategy set ${\sf \Sigma}_{i}$:
a function
$\sigma_{i}: {\sf \Sigma}_{i} \rightarrow [0, 1]$
such that
$\sum_{s \in {\sf \Sigma}_{i}}
   \sigma_{i}(s)
 =
 1$.
Denote as ${\sf Supp}(\sigma_{i})$
the set of strategies
$s \in {\sf \Sigma}_{i}$
such that
$\sigma_{i}(s)
 >
 0$.
The mixed strategy
$\sigma_{i}: {\sf \Sigma}_{i} \rightarrow [0, 1]$
is {\it pure,}
and player $i$ is {\it pure,}
if ${\mathsf{Supp}}(\sigma_{i})
 =
 \{ s \}$ 
for some strategy
$s \in {\mathsf{\Sigma}}_{i}$;
player $i$ is {\it mixed}
if she is not pure.
The mixed strategy
$\sigma_{i}: {\sf \Sigma}_{i} \rightarrow [0, 1]$
is {\it fully mixed,}
and player $i$ is 
{\it fully mixed,}
if ${\sf Supp}(\sigma_{i})
    =
    {\sf \Sigma}_{i}$;
so,
player $i$ puts non-zero probability
on all strategies
in her set of strategies 
${\sf \Sigma}_{i}$.
The mixed strategy
$\sigma_{i}: {\sf \Sigma}_{i} \rightarrow [0, 1]$
is {\it uniform}
if for each pair of strategies
$s, t \in {\sf Supp}(\sigma_{i})$,
${\sf \sigma}_{i}(s)
 =
 {\sf \sigma}_{i}(t)$.
The mixed strategy
$\sigma_{i}$
is {\it rational}
if all values of $\sigma_{i}$
are rational numbers.

A {\it mixed profile}
$\bm{\sigma}
 =
 (\sigma_{i})_{i \in [r]}$
is a tuple of mixed strategies,
one for each player.
So, a profile is 
the degenerate case of a mixed profile 
where all probabilities are either $0$ or $1$.
A {\it partial mixed profile}
$\bm{\sigma}_{-i}$
is a tuple of $r-1$ mixed strategies,
one for each player
other than $i$.
For a mixed profile $\bm{\sigma}$
and a mixed strategy $\tau_{i}$
of player $i \in [r]$,
denote as
$\bm{\sigma}_{-i} \diamond \tau_{i}$
the mixed profile obtained
by substituting $\tau_{i}$
for $\sigma_{i}$
in $\bm{\sigma}$.
A mixed profile
is {\it uniform}
(resp., {\it fully mixed})
if all of its mixed strategies
are uniform
(resp., fully mixed);
else it is {\it non-uniform}
(resp.,
{\it non-fully mixed}).
A mixed profile
is {\it symmetric}
if all mixed strategies are identical;
else it is {\it non-symmetric}.
A mixed profile
is {\it rational}
if all of its mixed strategies
are rational;
else it is {\it irrational}.

The mixed profile $\bm{\sigma}$
induces a probability measure
${\mathbb P}_{\bm{\sigma}}$
on the set of profiles
in the natural way;
so,
for a profile ${\bf s}$,
${\mathbb P}_{\bm{\sigma}}({\bf s})
 =
 \prod_{k \in [r]}
    \sigma_{k} (s_{k})$.
Say that
the profile ${\bf s}$
is {\it supported}
in the mixed profile $\bm{\sigma}$,
and write ${\bf s}\sim{\bm{\sigma}}$,
if ${\mathbb P}_{\bm{\sigma}}({\bf s})
    >
    0$.
Under the mixed profile
$\bm{\sigma}$,
the utility
of each player becomes a random variable.
So,
associated with the mixed profile
$\bm{\sigma}$
is the
{\it expected utility}
${\sf U}_{i}(\bm{\sigma})
 =
 {\mathbb E}_{{\bf s}
              \sim
              \bm{\sigma}}
              \left( {\sf U}_{i}\left( {\bf s}
                                 \right)
              \right)$
for each player $i \in [r]$:
the expectation
according to $\mathbb{P}_{\bm{\sigma}}$
of her utility;
so,
{
\small
\begin{eqnarray*}
{\mathsf{U}}_{i}(\bm{\sigma})
& = &
 \sum_{{\bf s}
       \in
       {\sf \Sigma}}
       \left( \prod_{k \in [r]}
                 \sigma_{k} (s_{k})
       \right)
       \cdot
       {\sf U}_{i}({\bf s})\, .
\end{eqnarray*}
}
Recall that
in a $c$-sum bimatrix game ${\mathsf{G}}$,
for each mixed profile $\bm{\sigma}$,
${\mathsf{U}}_{1}\left( \bm{\sigma}
                 \right)
 +                
 {\mathsf{U}}_{2}\left( \bm{\sigma}
                 \right)
 =
 c$.
Also,
in a symmetric bimatrix game,
the {\it mixed exchangeability property} holds:
exchanging players and mixed strategies preserves expected utilities
in the sense that
for any mixed profile
$\langle \sigma_{1}, \sigma_{2}
  \rangle$,
${\mathsf{U}}_{1}(\langle s_{1}, s_{2}
                                \rangle)
  =
  {\mathsf{U}}_{2}(\langle s_{2}, s_{1}
                                \rangle)$.\footnote{Indeed,
                                                                                                           ${\mathsf{U}}_{1}(\langle \sigma_{1}, \sigma_{2}
                                                                                                                                           \rangle)
                                                                                                             =
                                                                                                             \sum_{s_{1}, s_{2} \in [n]}
                                                                                                                \sigma_{1}(s_{1})\,
                                                                                                                \sigma_{2}(s_{2})\,
                                                                                                                \cdot
                                                                                                                {\mathsf{R}}[s_{1}, s_{2}]
                                                                                                             =
                                                                                                             \sum_{s_{1}, s_{2} \in [n]}
                                                                                                                \sigma_{1}(s_{1})\,
                                                                                                                \sigma_{2}(s_{2})\,
                                                                                                                \cdot
                                                                                                                {\mathsf{C}}^{{\rm T}}[s_{1}, s_{2}]
                                                                                                              =
                                                                                                              \sum_{s_{1}, s_{2} \in [n]}
                                                                                                                \sigma_{1}(s_{1})\,
                                                                                                                \sigma_{2}(s_{2})\,
                                                                                                                \cdot
                                                                                                                {\mathsf{C}}[s_{2}, s_{1}]
                                                                                                             =
                                                                                                             {\mathsf{U}}_{2}(\langle \sigma_{2}, \sigma_{1}
                                                                                                                                           \rangle)$.}

\subsection{Nash Equilibria}
\label{framework equilibria}

\noindent
A {\it pure Nash equilibrium}
is a profile
${\bf s} \in {\mathsf{\Sigma}}$
such that
for each player
$i \in [r]$ and
for each strategy
$t
 \in
 {\sf \Sigma}_{i}$,
${\sf U}_{i}\left( {\bf s}
            \right)
 \geq
 {\sf U}_{i}\left( {\bf s}_{-i}
                   \diamond
                   t
            \right)$.
A {\it mixed Nash equilibrium,} 
or {\it Nash equilibrium} for short,
is a mixed profile
$\bm{\sigma}$
such that
for each player $i \in [r]$ and
for each mixed strategy
$\tau_{i}$,
${\mathsf{U}}_{i}(\bm{\sigma})
 \geq
 {\mathsf{U}}_{i}(\bm{\sigma}_{-i} \diamond \tau_{i})$.
Note that 
the mixed exchangeability property
of symmetric bimatrix games
implies that
(mixed) Nash equilibria are preseved in a symmetric bimatrix game:
$\langle \sigma_{1}, \sigma_{2} \rangle$
is a Nash equilibrium
if and only if
$\langle \sigma_{2}, \sigma_{1} \rangle$ is.   
Denote as 
${\mathcal{NE}}({\mathsf{G}})$
(resp., ${\mathcal{SNE}}({\mathsf{G}})$)
the set of Nash equilibria
(resp., symmetric Nash equilibria)
for the game ${\mathsf{G}}$.
For each game ${\mathsf{G}}$,
${\mathcal{NE}}({\mathsf{G}})
  \neq
  \emptyset$~\cite{N50,N51};
for each symmetric game ${\mathsf{G}}$,
${\mathcal{SNE}}({\mathsf{G}})
  \neq
  \emptyset$~\cite{N50,N51}.  
Two $r$-player games $\widehat{\sf G}$ and $\sf G$ 
are {\it Nash-equivalent}~\cite{BM11} 
if ${\mathcal{NE}}(\widehat{\mathsf{G}})
    =
    {\mathcal{NE}}({\sf G})$;
that is, 
they have the same set of Nash equilibria. 
We shall make extensive use
of the following
basic property
of Nash equilibria.

\begin{lemma}
\label{basic property of mixed nash equilibria}
A mixed profile
$\bm{\sigma}$
is a Nash equilibrium
if and only if
for each player $i \in [r]$,
{\sf (1)}
for each strategy 
$t \in {\mathsf{Supp}}(\sigma_{i})$,
${\mathsf{U}}_{i}(\bm{\sigma})
 =
 {\mathsf{U}}_{i}(\bm{\sigma}_{-i} \diamond t)$,
and
{\sf (2)}
for each strategy
$t \not\in {\mathsf{Supp}}(\sigma_{i})$,
${\mathsf{U}}_{i}(\bm{\sigma})
 \geq
 {\mathsf{U}}_{i}(\bm{\sigma}_{-i} \diamond t)$.
\end{lemma}

\noindent
Given a partial mixed profile
$\bm{\sigma}_{-i}$
for some player $i \in [r]$,
a {\it best-response}
for player $i$
to $\bm{\sigma}_{-i}$
is a pure strategy
$s \in {\mathsf{\Sigma}}_{i}$
such that
${\mathsf{U}}_{i}\left( \bm{\sigma}_{-i}
                        \diamond
                        s
                 \right)
 =
 \max_{t 
            \in
           {\mathsf{\Sigma}}_{i}                 
          }
   {\mathsf{U}}_{i}\left( \bm{\sigma}_{-i}
                                      \diamond
                                      t
                              \right)$.
Lemma~\ref{basic property of mixed nash equilibria}
(Condition {\sf (2)})
immediately implies that
in a Nash equilibrium $\bm{\sigma}$,
for each player $i \in [r]$
and strategy
$s \in {\mathsf{\Sigma}}_{i}$,
$s \in {\mathsf{Supp}}(\sigma_{i})$
only if
$s$ is a best-response
for player $i$
to $\bm{\sigma}_{-i}$.                   
Lemma~\ref{basic property of mixed nash equilibria}
and the positive utility property
immediately imply:

\begin{lemma}
\label{park kafe}
Fix a win-lose game ${\mathsf{G}}$
with the positive utility property.
Then,
in a Nash equilibrium $\bm{\sigma}$,
for each player $i \in [r]$,
${\mathsf{U}}_{i}(\bm{\sigma})
 >
 0$.
\end{lemma}

\begin{proof}
Assume,
by way of contradiction,
that ${\mathsf{U}}_{i}(\bm{\sigma})
      =
      0$.
Choose a partial profile
${\bf s}_{-i}$
supported in $\bm{\sigma}_{-i}$;
so,
${\mathbb P}_{\bm{\sigma}_{-i}}({\bf s}_{-i})
    >
    0$.
By the positive utility property,
there is a strategy
$t({\bf s}_{-i}) 
 \in
 {\mathsf{\Sigma}}_{i}$
with
${\mathsf{U}}_{i}\left( {\bf s}_{-i}
                        \diamond
                        t({\bf s}_{-i})  
                 \right)
 >
 0$.                
Since ${\mathsf{G}}$ is win-lose,
{
\small
\begin{eqnarray*}
       {\mathsf{U}}_{i}\left( \bm{\sigma}_{-i}
                              \diamond
                              t({\bf s}_{-i})
                       \right)  
\geq & {\mathbb P}_{\bm{\sigma}_{-i}}({\bf s}_{-i})
       \cdot
       {\mathsf{U}}_{i}\left( \bm{\sigma}_{-i}
                                           \diamond
                                           t({\bf s}_{-i})
                       \right)\ \ 
>\ \
0\ \ 
=\ \
{\mathsf{U}}_{i}(\bm{\sigma})\, .
\end{eqnarray*}
}
A contradiction
to Lemma~\ref{basic property of mixed nash equilibria}.
\end{proof}

\noindent
We conclude with a simple property
of Nash equilibria
for win-lose games.

\begin{lemma}
\label{if zero utility}
Fix a win-lose bimatrix game 
$\langle {\mathsf{R}},
              {\mathsf{C}}
 \rangle$ 
with a Nash equilibrium
${\bm{\sigma}}
  \in
  {\mathcal{NE}}(\langle {\mathsf{R}},
                                        {\mathsf{C}}
                            \rangle)$
such that
${\mathsf U}_{i}({\bm{\sigma}})
  =
  0$
for some player $i \in [2]$.
Then,
$\langle {\mathsf{R}},
             {\mathsf{C}}
  \rangle$
has a pure Nash equilibrium.
\end{lemma}

\begin{proof}
For simpler notation,
assume, 
without loss of generality, 
that $i=1$. 
Denote 
\begin{eqnarray*}
\overline{s}
&  := &
  \textrm{argmax}_{s_{2}
                                 \in
                                 {\mathsf{Supp}}(\sigma_{2})}
 \max_{s_{1}                                 
            \in
            {\mathsf{Supp}}(\sigma_{1}) }
  \left\{ {\mathsf{U}}_{2}
                            (\langle s_{1}, s_{2}
                             \rangle)
  \right\}\, ;
\end{eqnarray*}  
so,
$\overline{s}$
is the strategy in ${\mathsf{Supp}}(\sigma_{2})$
that incurs the maximum possible utility to player $2$
over all strategies $s_{1}$ from ${\mathsf{Supp}}(\sigma_{1})$
chosen by player $1$.  
Choose an arbitrary strategy
$s^{\ast}
  \in
  {\mathsf{Supp}}(\sigma_{1})$.
We shall prove that
the profile $\langle s^{\ast},
                               \overline{s}
                   \rangle$ 
is a pure Nash equilibrium for 
$\langle {\mathsf{R}},
              {\mathsf{C}}
  \rangle$.
Since 
${\bm{\sigma}}$
is a Nash equilibrium
for 
$\langle {\mathsf{R}},
              {\mathsf{C}}
  \rangle$,
Lemma~\ref{basic property of mixed nash equilibria} (Condition {\sf (2)})
implies that
for each strategy 
$s \in {\mathsf{\Sigma}}_{1}$, 
${\mathsf{U}}_{1}
                 ({\bm{\sigma}}_{-1} \diamond s)
  \leq
  {\mathsf{U}}_{1}({\bm{\sigma}}) = 0$.
Since 
$\langle {\mathsf{R}},
              {\mathsf{C}}
  \rangle$ is win-lose, 
it follows that
${\mathsf{U}}_{1}
                 ({\bm{\sigma}}_{-1} \diamond s)
  =
  0$,
which implies that 
for each strategy 
$s \in {\mathsf{\Sigma}}_{1}$,
${\mathsf{U}}_{1}(\langle s,  s_{2}
                               \rangle)
  =
  0$
for all strategies
$s_{2}
  \in
  {\mathsf{Supp}}(\sigma_{2})$. 
Since 
$\overline{s}
  \in
  {\mathsf{Supp}}(\sigma_{2})$,
it follows that
for each strategy 
$s \in {\mathsf{\Sigma}}_{1}$, 
${\mathsf U}_{1}(\langle s, \overline{s}
                             \rangle)
  =
  0$.
In particular,
${\mathsf{U}}_{1}
                (\langle s^{\ast}, \overline{s}
                \rangle)
  = 0$,
so that player $1$ cannot improve 
by deviating from $s^{\ast}$. 
By the definition of $\overline{s}$, 
player $2$ cannot improve either 
by deviating from $\overline{s}$.
Hence,
$\langle s^{\ast}, \overline{s}
  \rangle$
is a pure Nash equilibrium 
for $\langle {\mathsf{R}},
                   {\mathsf{C}}
       \rangle$.
\end{proof}


\subsection{Decision Problems about Nash Equilibria}
\label{framework decision problems}

\noindent
We shall assume
some basic familiarity
of the reader
with the basic notions
of ${\cal NP}$-completeness,
as outlined, for example,
in~\cite{GJ79}. 
All decision problems
will be stated 
in the style of~\cite{GJ79},
where {\sc I.} and {\sc Q.}
stand for {\sc Instance} and {\sc Question},
respectively.
They are categorized
in six groups.

\noindent
\underline{{\it Group I} ---}
Question
about cardinality:

\vspace{0.2cm}
\noindent
{\sf $\exists$ $k+1$ NASH} (with $k \geq 1$)

\begin{tabular}{lp{13.7cm}l}
\hline
{\sc I.:} & A game ${\sf G}$.  \\
{\sc Q.:} & Does $\sf G$ have at least $k+1$ Nash equilibria?  \\
\hline
\end{tabular}

\noindent
\underline{{\it Group II} ---}
Questions about the expected utilities:

\vspace{0.2cm}
\noindent
{\sf $\exists$ NASH WITH LARGE UTILITIES}

\begin{tabular}{lp{13.7cm}l}
\hline
{\sc I.:} & A game ${\sf G}$ and a number $u$.                                    \\
{\sc Q.:} & Is there a Nash equilibrium $\bm\sigma$ 
            such that for each player $i\in [r]$, 
            ${\sf U}_i({\bm\sigma})\geq u$?                                        \\
\hline
\end{tabular}

\vspace{0.2cm}
\noindent
{\sf $\exists$ NASH WITH SMALL UTILITIES}

\begin{tabular}{lp{13.7cm}l}
\hline
{\sc I.:} & A game ${\sf G}$ and a number $u$.                          \\
{\sc Q.:} & Is there a Nash equilibrium $\bm\sigma$ 
            such that
            for each player $i\in [r]$, 
            ${\sf U}_i({\bm\sigma})\leq u$?                              \\
\hline
\end{tabular}

\vspace{0.2cm}
\noindent
{\sf $\exists$ NASH WITH LARGE TOTAL UTILITY}

\begin{tabular}{lp{13.7cm}l}
\hline
{\sc I.:} & A game ${\sf G}$ and a number $u$.             \\
{\sc Q.:} & Is there a Nash equilibrium $\bm\sigma$ 
            such that
            $\sum_{i\in [r]}{\sf U}_i({\bm\sigma})\geq u$?  \\
\hline
\end{tabular}

\vspace{0.2cm}
\noindent
{\sf $\exists$ NASH WITH SMALL TOTAL UTILITY}

\begin{tabular}{lp{13.7cm}l}
\hline
{\sc I.:} & A game ${\sf G}$ and a number $u$.\\
{\sc Q.:} & Is there a Nash equilibrium $\bm\sigma$ 
            such that
            $\sum_{i\in [r]}{\sf U}_i({\bm\sigma})\leq u$?  \\
\hline
\end{tabular}

\noindent
\underline{{\it Group III} ---}
Questions about the supports:

\vspace{0.2cm}
\noindent
{\sf $\exists$ NASH WITH LARGE SUPPORTS}

\begin{tabular}{lp{13.7cm}l}
\hline
{\sc I.:} & A game ${\sf G}$ and an integer $k\geq 1$.\\
{\sc Q.:} & Is there a Nash equilibrium $\bm\sigma$ 
            such that for each player $i\in [r]$, 
            $|{\sf Supp}(\sigma_i)|\geq k$?  \\
\hline
\end{tabular}

\vspace{0.2cm}
\noindent
{\sf $\exists$ NASH WITH SMALL SUPPORTS}

\begin{tabular}{lp{13.7cm}l}
\hline
{\sc I.:} & A game ${\sf G}$ and an integer $k\geq 1$.\\
{\sc Q.:} & Is there a Nash equilibrium $\bm\sigma$ 
            such that
            for each player $i\in [r]$,
            $|{\sf Supp}(\sigma_i)|\leq k$?            \\
\hline
\end{tabular}

\vspace{0.2cm}
\noindent
{\sf $\exists$ NASH WITH RESTRICTING SUPPORTS}

\begin{tabular}{lp{13.7cm}l}
\hline
{\sc I.:} & A game ${\sf G}$ 
            and a subset of strategies 
            ${\sf T}_i\subseteq{\sf\Sigma}_i$ 
            for each player $i\in [r]$.                 \\
{\sc Q.:} & Is there a Nash equilibrium $\bm\sigma$ 
            such that
            for each player $i\in [r]$, 
            ${\sf T}_i\subseteq{\sf Supp}(\sigma_i)$?  \\
\hline
\end{tabular}

\vspace{0.2cm}
\noindent
{\sf $\exists$ NASH WITH RESTRICTED SUPPORTS}

\begin{tabular}{lp{13.7cm}l}
\hline
{\sc I.:} & A game ${\sf G}$ 
            and a subset of strategies 
            ${\sf T}_i\subseteq{\sf\Sigma}_i$ 
            for each player $i\in [r]$.                \\
{\sc Q.:} & Is there a Nash equilibrium $\bm\sigma$ 
            such that
            for each player $i\in [r]$, 
            ${\sf Supp}(\sigma_i)\subseteq{\sf T}_i$?  \\
\hline
\end{tabular}

\noindent
\underline{{\it Group IV} ---}
Questions about refinements of Nash equilibrium:

\noindent
A Nash equilibrium $\bm\sigma$ is 
{\it Pareto-Optimal} 
if for each mixed profile
$\widehat{\bm\sigma}$
where there is a player $i \in [r]$
with
${\sf U}_{i}(\widehat{\bm\sigma})
 >
 {\sf U}_{i}({\bm\sigma})$,
there is a player $j \in [r]$
such that 
${\sf U}_{j}(\widehat{\bm\sigma})
 <
 {\sf U}_{j}({\bm\sigma})$;
so,
loosely speaking,
there is no other mixed profile
where
at least one player
is strictly better off
and every player
is at least as well off.
Denote as 
${\mathsf{Diff}} \left( \bm{\sigma},
                        \widehat{\bm{\sigma}}
                 \right)       
 :=
 \{ i \in [r]:
    \sigma_i \neq \widehat{\sigma}_i
 \}$ 
the set of players 
with different mixed strategies
in the mixed profiles ${\bm\sigma}$ and $\widehat{\bm\sigma}$.
A Nash equilibrium $\bm\sigma$ is 
{\it Strongly Pareto-Optimal} 
if for each mixed profile $\widehat{\bm\sigma}$ 
with a player $i \in [r]$
such that
${\sf U}_{i}(\widehat{\bm\sigma})
 >
 {\sf U}_{i}({\bm\sigma})$, 
there is a player
$j \in
  {\mathsf{Diff}}({\bm\sigma},\widehat{\bm\sigma})$
such that 
${\sf U}_{j}(\widehat{\bm\sigma})
 \leq
 {\sf U}_j({\bm\sigma})$;
so,
loosely speaking,
there is no other mixed profile
where
at least one player
is strictly better off
and every player
with a different mixed strategy
is strictly better off.

\vspace{0.2cm}
\noindent
{\sf $\exists$ PARETO-OPTIMAL NASH}

\begin{tabular}{lp{13.7cm}l}
\hline
{\sc I.:} & A game ${\sf G}$.\\
{\sc Q.:} & Is there a Pareto-Optimal Nash equilibrium?  \\
\hline
\end{tabular}

\vspace{0.2cm}
\noindent
{\sf $\exists$ $\neg$ PARETO-OPTIMAL NASH}

\begin{tabular}{lp{13.7cm}l}
\hline
{\sc I.:} & A game ${\sf G}$.\\
{\sc Q.:} & Is there a Nash equilibrium which is not Pareto-Optimal?  \\
\hline
\end{tabular}

\vspace{0.2cm}
\noindent
{\sf $\exists$ STRONGLY PARETO-OPTIMAL NASH}

\begin{tabular}{lp{13.7cm}l}
\hline
{\sc I.:} & A game ${\sf G}$.\\
{\sc Q.:} & Is there a Strongly Pareto-Optimal Nash equilibrium?  \\
\hline
\end{tabular}

\vspace{0.2cm}
\noindent
{\sf $\exists$ $\neg$ STRONGLY PARETO-OPTIMAL NASH}

\begin{tabular}{lp{13.7cm}l}
\hline
{\sc I.:} & A game ${\sf G}$.\\
{\sc Q.:} & Is there a Nash equilibrium which is not Strongly Pareto-Optimal?  \\
\hline
\end{tabular}

\noindent
\underline{{\it Group V} ---}
Questions about the probabilities:

\vspace{0.2cm}
\noindent
{\sf $\exists$ NASH WITH SMALL PROBABILITIES}

\begin{tabular}{lp{13.7cm}l}
\hline
{\sc I.:} & A game ${\mathsf{G}}$.                                 \\
{\sc Q.:} & Is there a Nash equilibrium $\bm{\sigma}$
            such that 
            for each player $i \in [r]$, 
            $\max_{s \in {\mathsf{\Sigma}}_{i}}
                 \sigma_{i} (s) 
              \leq
              \frac{\textstyle 1}
                      {\textstyle 2}$?                                  \\
\hline
\end{tabular}

\vspace{0.2cm}
\noindent
{\sf $\exists$ UNIFORM NASH}

\begin{tabular}{lp{13.7cm}l}
\hline
{\sc I.:} & A game ${\mathsf{G}}$.                                \\
{\sc Q.:} & Does ${\mathsf{G}}$ have a uniform Nash equilibrium?  \\
\hline
\end{tabular}

\vspace{0.2cm}
\noindent
{\sf $\exists$ $\neg$ UNIFORM NASH}

\begin{tabular}{lp{13.7cm}l}
\hline
{\sc I.:} & A game ${\mathsf{G}}$.                                \\
{\sc Q.:} & Does ${\mathsf{G}}$ have a non-uniform Nash equilibrium?  \\
\hline
\end{tabular}

\vspace{0.2cm}
\noindent
{\sf $\exists$ SYMMETRIC NASH}

\begin{tabular}{lp{13.7cm}l}
\hline
{\sc I.:} & A game ${\mathsf{G}}$.                                \\
{\sc Q.:} & Does ${\mathsf{G}}$ have a symmetric Nash equilibrium?  \\
\hline
\end{tabular}

\vspace{0.2cm}
\noindent
{\sf $\exists$ $\neg$ SYMMETRIC NASH}

\begin{tabular}{lp{13.7cm}l}
\hline
{\sc I.:} & A game ${\mathsf{G}}$.                                \\
{\sc Q.:} & Does ${\mathsf{G}}$ have a non-symmetric Nash equilibrium?  \\
\hline
\end{tabular}

\vspace{0.2cm}
\noindent
{\sf $\exists$ RATIONAL NASH}

\begin{tabular}{lp{13.7cm}l}
\hline
{\sc I.:} & A game ${\mathsf{G}}$.                                 \\
{\sc Q.:} & Does ${\mathsf{G}}$ have a rational Nash equilibrium?  \\
\hline
\end{tabular}

\noindent
Restricted to general bimatrix games, 
all previous decision problems 
belong to $\cal NP$:
given the polynomial-length supports
for a mixed profile,
it is polynomial time verifiable
that it is a Nash equilibrium
satisfying the property
associated with the decision problem.
{\sf $\exists$ RATIONAL NASH}
is trivial
for general bimatrix games,
and
{\sf $\exists$ SYMMETRIC NASH}
is trivial for symmetric games.
Note that
{\sf $\exists$ UNIFORM NASH}
and
{\sf $\exists$ RATIONAL NASH}
belong to ${\mathcal{NP}}$
for any number of players. 

Each of the previous decision problems
has a corresponding {\it cardinality} or {\it counting} version:
the problem
of computing the number of Nash equilibria
witnessing the validity of {\sc Question}.\footnote{{\sf $\exists$ $k+1$ NASH}
                                                                      has no counting version
                                                                      as it is already defined with a 
                                                                      cardinality question.} 
Note that for general bimatrix games,
the counting versions of
all previous decision problems
belong to $\# {\mathcal{P}}$
(thanks again to the efficient verifiability property
of Nash equilibria).
Furthermore,
each of the decision problems
has a {\it parity} version:
the problem of computing the parity
of the number of Nash equilibria
witnessing the validity of {\sc Question}.
The parity versions of all previous decision problems
belong to $\oplus {\mathcal{P}}$.\footnote{$\oplus {\mathcal{P}}$, 
                                                                               read as {\it Parity $\oplus {\mathcal{P}}$,}
                                                                               is the complexity class formalizing the question of the parity
                                                                               of the number of solutions
                                                                               to a combinatorial problem.
                                                                               Formally,
                                                                               $\oplus {\mathcal{P}}$
                                                                               is the class of sets $S$ such that
                                                                               there is a nondeterministic Turing machine 
                                                                               which, on input $x$,
                                                                               has an odd number of accepting computations
                                                                               if and only if $x \in S$.
                                                                               We shall adopt a definition
                                                                               of $\oplus {\mathcal{P}}$-completeness
                                                                               using polynomial-time many-to-one reductions.
                                                                               The development of the theory
                                                                               of $\oplus {\mathcal{P}}$-complete parity problems
                                                                               is rather limited --- see the discussion in~\cite{V05}.}
We shall use $\#$ (resp., $\oplus$)
in the place of $\exists$
to denote the counting (resp., parity) versions;
for example,
{\sf $\#$ RATIONAL NASH}
(resp.,
{\sf $\oplus$ RATIONAL NASH})
denotes the counting (resp., parity)
version of
{\sf $\exists$ RATIONAL NASH}. 
Clearly,
                                        by the mixed exchangeability property of symmetric bimatrix games,
                                        the number of non-symmetric Nash equilibria
                                        in a symmetric bimatrix game
                                        is even;
                                        this is because
                                        $\langle \sigma_{1}, \sigma_{2} 
                                          \rangle$
                                        is a non-symmetric Nash equilibrium
                                        if and only if
                                        $\langle \sigma_{2}, \sigma_{1} 
                                          \rangle$ is.    
                                        Hence,
                                        {\sf $\oplus \neg$ SYMMETRIC NASH}
                                        for symmetric bimatrix games
                                        is in ${\mathcal{P}}$.

\noindent
\underline{{\it Group VI} ---}
Equivalence property:

\vspace{0.2cm}
\noindent
{\sf NASH-EQUIVALENCE}

\begin{tabular}{lp{13.7cm}l}
\hline
{\sc I.:} & Two games
            $\widehat{{\sf G}}$
            and
            ${\sf G}$
            from $r$-${\cal G}$,
            for some integer $r \geq 2$.                                             \\
{\sc Q.:} & Are $\widehat{{\sf G}}$
            and
            ${\sf G}$
            Nash-equivalent?                                                         \\
\hline
\end{tabular}

\noindent
For a fixed game $\widehat{{\sf G}}$, 
called the {\it gadget game,} 
a parameterized restriction of {\sf NASH-EQUIVALENCE} 
with a single input (the game ${\mathsf{G}}$) results.

\vspace{0.2cm}
\noindent
{\sf NASH-EQUIVALENCE}($\widehat{{\sf G}}$)

\begin{tabular}{lp{13.7cm}l}
\hline
{\sc I.:} & A game
            ${\sf G}$
            from $r$-${\cal G}$ (where $\widehat{\sf G}$ is
            from $r$-${\cal G}$).                                                    \\
{\sc Q.:} & Are $\widehat{{\sf G}}$
            and ${\sf G}$
            Nash-equivalent?                                                         \\
\hline
\end{tabular}

\noindent
Restricted to general bimatrix games,
{\sf NASH-EQUIVALENCE}
and
{\sf NASH-EQUIVALENCE}$({\widehat{{\mathsf{G}}}})$
belong to co-${\mathcal{NP}}$:
given the polynomial-length supports
for a mixed profile,
it is polynomial time verifiable that
it is a Nash equilibrium
for exactly one of the two games.

\noindent
Matching ${\mathcal{NP}}$-
and co-${\mathcal{NP}}$-hardness results
for the decision problems above
are \textcolor{black}{later} summarized in
Figure~\ref{table1}.


\section{Win-Lose Bimatrix Games with the Positive Utility Property}
\label{winlose bimatrix with pup}

\noindent
We now prove that
assuming the positive utility property
does not simplify
the search problem
for a win-lose bimatrix game.
Fix a win-lose bimatrix game 
$\langle {\mathsf{R}},
              {\mathsf{C}}
  \rangle$.
We start with a preliminary definition.   
For a player $i \in [2]$,  
a strategy 
$\overline{s}
  \in
  {\mathsf{\Sigma}}_{\overline{i}}$ 
is an {\it all-zeros counter-strategy} 
against player $i$ 
if for every profile
$\langle s_{1}, s_{2}
 \rangle
 \in
 {\mathsf{\Sigma}}$
with
$s_{\overline{i}}=\overline{s}$, 
${\mathsf U}_{i}(\langle s_{1}, s_{2}
                            \rangle)
  =
  0$;
that is, 
an all-zeros counter-strategy against the row (resp., column) player
is a column (resp., row) 
of ${\mathsf{R}}$ 
(resp., ${\mathsf{C}}$) 
made up only of zero entries.
Clearly,
the positive utility property
excludes all-zeros counter-strategies.
We prove:

\begin{proposition}
\label{pup is not restrictive}
Fix a win-lose bimatrix game 
$\langle {\mathsf{R}},
             {\mathsf{C}}
  \rangle$. 
Then, 
one of two conditions holds:
\begin{enumerate}

\item[{\sf (C.1)}]
$\langle{\mathsf{R}},
            {\mathsf{C}}
  \rangle$
has a pure Nash equilibrium.

\item[{\sf (C.2)}]
There is a polynomial time constructible
win-lose bimatrix game
\textcolor{black}{$\langle   {\bar{{\mathsf{R}}}},
                                                       {\bar{{\mathsf{C}}}}
                                          \rangle$}
with the positive utility property 
such that 
${\mathcal{NE}}(\langle {\bar{{\mathsf{R}}}},
                                        {\bar{{\mathsf{C}}}}
                            \rangle)
  =
  {\mathcal{NE}}(\langle {\mathsf{R}},
                                        {\mathsf{C}}
                           \rangle)$.
\end{enumerate}
\end{proposition}

\begin{proof}
If Condition {\sf (C.1)} holds,
then we are done.
So assume that
Condition {\sf (C.1)}
does not hold. 
Note that if at least one
of ${\mathsf{R}}$ and ${\mathsf{C}}$
is the null matrix,
then
$\langle {\mathsf{R}},
              {\mathsf{C}}
  \rangle$
has a pure Nash equilibrium
and Condition {\sf (C.1)} holds.
It follows that none of
${\mathsf{R}}$ and ${\mathsf{C}}$
is the null matrix.
Denote as $c\geq 0$ and $r \geq 0$
the number of all-zeros counter-strategies 
against the row and the column player, respectively,
in the game
$\langle {\mathsf{R}},
              {\mathsf{C}}
  \rangle$. 
Note that
$\langle {\mathsf{R}},
              {\mathsf{C}}
  \rangle$
has the positive utility property
if and only if 
$r = c = 0$.
Renumber now the players' strategies 
so that the first $r$ strategies for player $1$ 
are the all-zeros counter-strategies against the column player,
and the first $c$ strategies for player $2$ 
are the all-zeros counter-strategies 
against the row player.
Construct from
$\langle {\mathsf{R}},
              {\mathsf{C}}
  \rangle$
a win-lose bimatrix game
$\langle {\bar{{\mathsf{R}}}},
              {\bar{{\mathsf{C}}}}
  \rangle$
as follows:

\begin{center}
\fbox{
\begin{minipage}{6.0in}
\begin{itemize}

\item[{\sf (0)}]
If $r = c = 0$,
then
$\langle {\bar{{\mathsf{R}}}},
              {\bar{{\mathsf{C}}}}
  \rangle
  :=
  \langle {\mathsf{R}},
              {\mathsf{C}}
  \rangle$.

\item[{\sf (1)}]
Otherwise,
add a new strategy $n+1$ 
to the strategy set of each player
and set:
\begin{itemize}

\item[{\sf (1.0)}] 
$\textcolor{black}{{\bar{{\mathsf{U}}}}}(\langle n+1, 
                                  n+1
                     \rangle)
  =
  \langle 0, 0
  \rangle$.

\item[{\sf (1.1)}] 
$\textcolor{black}{{\bar{{\mathsf{U}}}}}(\langle n+1, s_{2}
                        \rangle)
  =
  \langle 1, 0
  \rangle$ 
for $s_{2} \leq c$,
or
$\langle 0, 1 \rangle$
for
$c+1 \leq s_{2} \leq n$.

\item[{\sf (1.2)}] 
$\textcolor{black}{{\bar{{\mathsf{U}}}}}(\langle s_{1}, n+1
                        \rangle)
  =
  \langle 0,1
  \rangle$ 
for $s_{1} \leq r$,
or
$\langle 1, 0 \rangle$
for $r+1 \leq s_{1} \leq n$.
\end{itemize}
\end{itemize}
\end{minipage}
}
\end{center}

\noindent
Clearly,
this is a polynomial time construction.
We prove:

\begin{lemma}
\label{ictcs lecce}
The game
$\langle {\bar{{\mathsf{R}}}},
                                                      {\bar{{\mathsf{C}}}}
  \rangle$
has the positive utility property.
\end{lemma}

\begin{proof}
Note that
in Case {\sf (0)},
the game 
$\langle {\bar{{\mathsf{R}}}},
              {\bar{{\mathsf{C}}}}
  \rangle$
has the positive utility property
by definition.
So assume we are in Case {\sf (1)}.
Consider the row matrix
${\bar{{\mathsf{R}}}}$;
we need to prove that
it does not contain
an all-zeros column.
By Step {\sf (1.2)}
in the construction,
the column $n+1$ of
${\bar{{\mathsf{R}}}}$
has an entry equal to $1$
if and only if
$r+1 \leq n$;
the latter holds since
${\mathsf{R}}$ is not the null matrix.
So it remains to consider
the columns of
${\bar{{\mathsf{R}}}}$
that are inherited from
${\mathsf{R}}$.
By the definition of $c$,
every column $j$ of ${\mathsf{R}}$
with
$c+1 \leq j \leq n$
is not all-zeros;
for a column $j \in [c]$
of ${\mathsf{R}}$,
${\bar{{\mathsf{U}}}}_{1}(n+1, j)
  =
  1$
by Step (1.1)
in the construction,
and we are done.  
We use corresponding arguments
to prove that
the column matrix
${\bar{{\mathsf{C}}}}$
does not contain
an all-zeros row.
\end{proof}

\noindent
Furthermore,
Condition {\sf (C.2)}
holds vacuously in Case {\sf (0)}.
Note also that
in Case {\sf (1)},
$\langle {\bar{{\mathsf{R}}}},
              {\bar{{\mathsf{C}}}}
  \rangle$
has no pure Nash equilibrium 
in which some player chooses strategy $n+1$. 
We proceed by case analysis.
Assume first
that there is
a Nash equilibrium 
${\bar{{\bm{\sigma}}}}$ for
$\langle {\bar{{\mathsf{R}}}},
              {\bar{{\mathsf{C}}}}
  \rangle$
such that
at least one player 
$\overline{i} \in [2]$ plays
with positive probability 
an all-zeros counter-strategy $\overline{s}$
against player $i$. 
By Lemma~\ref{basic property of mixed nash equilibria},
choosing the strategy $\overline{s}$
with probability $1$
is a best-response of player $\overline{i}$
to ${\bar{{\bm{\sigma}}}}_{-\overline{i}}$;
so, 
player $\overline{i}$
cannot improve her utility
when choosing the strategy $\overline{s}$ with probability $1$.
Now,
by the definition of an all-zeros counter-strategy,
player $i$ gets utility $0$
in the mixed profile 
${\bar{{\bm{\sigma}}}}_{-\overline{i}} \diamond \overline{s}$,
which she could not improve by switching to another
mixed strategy.  
Hence,
${\bar{{\bm{\sigma}}}}_{-\overline{i}} \diamond \overline{s}$
is a Nash equilibrium with
${\bar{{\mathsf U}}}_{i}({\bar{{\bm{\sigma}}}}_{-\overline{i}} \diamond \overline{s})
  =
  0$. 
By Lemma~\ref{if zero utility}, 
it follows that
$\langle {\bar{{\mathsf{R}}}},
              {\bar{{\mathsf{C}}}}
  \rangle$
has a pure Nash equilibrium ${\bf s}$. 
Since
$\langle {\bar{{\mathsf{R}}}},
              {\bar{{\mathsf{C}}}}
  \rangle$
has no pure Nash equilibrium
in which some player plays strategy $n+1$, 
${\bf s}$ is also a pure Nash equilibrium for
$\langle {\mathsf{R}},
              {\mathsf{C}}
  \rangle$.
A contradiction.

Assume now that 
in each Nash equilibrium
${\bar{{\bm{\sigma}}}}$ for
$\langle {\bar{{\mathsf{R}}}},
              {\bar{{\mathsf{C}}}}
  \rangle$, 
no player $\overline{i}$
plays an all-zeros counter-strategy against player $i$;
that is, 
${\mathsf{Supp}}({\bar{\sigma}}_{1})
  \subseteq
  \{ r+1, \ldots, n+1 \}$
and
${\mathsf{Supp}}({\bar{\sigma}}_{2})
  \subseteq
  \{ c+1,\ldots ,n+1 \}$. 
We first prove that
no Nash equilibrium for
$\langle {\mathsf{R}},
              {\mathsf{C}}
  \rangle$
is ``destroyed''
due to adding
strategy $n+1$.

\begin{lemma}
\label{august 15}
It holds that
${\mathcal{NE}}( \langle {\mathsf{R}},
                                        {\mathsf{C}}
                                    \rangle
                          )          
  \subseteq
  {\mathcal{NE}}( \langle {\bar{{\mathsf{R}}}},
                                         {\bar{{\mathsf{C}}}}
                             \rangle
                          )$.
\end{lemma}

\begin{proof}
Assume,
by way of contradiction,
that there is a Nash equilibrium
${\bm{\sigma}}
  \in
  {\mathcal{NE}}\left( \langle {\mathsf{R}},
                                               {\mathsf{C}}
                                    \rangle
                          \right)$ 
such that
${\bm{\sigma}}
  \notin
  {\mathcal{NE}}( \langle {\bar{{\mathsf{R}}}},
                                        {\bar{{\mathsf{C}}}}
                            \rangle
                          $.
By the construction of the game
$\langle {\bar{{\mathsf{R}}}},
              {\bar{{\mathsf{C}}}}
 \rangle$,
this implies that 
there is player $i \in [2]$
with 
${\bar{{\mathsf{U}}}}_{i}({\bm{\sigma}}_{-i}
                              \diamond 
                              (n+1))
   >
   {\mathsf{U}}_{i}({\bm{\sigma}})$. 
Since the game
$\langle {\bar{{\mathsf{R}}}},
              {\bar{{\mathsf{C}}}}
 \rangle$
is win-lose, 
this implies that
${\bar{{\mathsf{U}}}}_{i}({\bm{\sigma}}_{-i}
                              \diamond 
                              (n+1))
   >
  0$.
Since
${\mathsf{Supp}}(\sigma_{1})
  \subseteq
  \{ r+1, \ldots, n \}$
and 
${\mathsf{Supp}}(\sigma_{2})
  \subseteq
  \{ c+1, \ldots, n \}$,
it follows that
any profile supported in 
${\bm{\sigma}}_{-i}\diamond (n+1)$ 
falls into either Case ({\sf 1.1}) with $i=1$
and $c+1 \leq s_{2} \leq n$ 
or into Case ({\sf 1.2}) with $i=2$
and $r+1 \leq s_{1} \leq n$.
By construction,
this implies that
${\bar{{\mathsf{U}}}}_{i}({\bm{\sigma}}_{-i} \diamond (n+1))
  =
  0$ 
for each player $i \in [2]$. 
A contradiction.
\end{proof}

\noindent
We now prove that
no new Nash equilibrium for
\textcolor{black}{$\langle {\bar{{\mathsf{R}}}},
              {\bar{{\mathsf{C}}}}
 \rangle$}
is ``created''
due to adding strategy $n+1$.

\begin{lemma}
\label{august 16}
It holds that
${\mathcal{NE}}( \langle {\bar{{\mathsf{R}}}},
                                         {\bar{{\mathsf{C}}}}
                             \rangle
                          )
  \setminus
  {\mathcal{NE}}( \langle {\mathsf{R}},
                                         {\mathsf{C}}
                             \rangle
                          )          
  =
  \emptyset$.
\end{lemma}

\begin{proof}
Assume, 
by way of contradiction, 
that there is a Nash equilibrium
${\bar{{\bm{\sigma}}}}
  \in
 {\mathcal{NE}}( \langle {\bar{{\mathsf{R}}}},
                                        {\bar{{\mathsf{C}}}}
                                   \rangle
                          )
  \setminus
  {\mathcal{NE}}( \langle {\mathsf{R}},
                                               {\mathsf{C}}
                                    \rangle
                          )$.            
Since
${\bar{{\bm{\sigma}}}}
  \not\in
  {\mathcal{NE}}\left( \langle {\mathsf{R}},
                                               {\mathsf{C}}
                                    \rangle
                          \right)$,  
it must be that 
${\bar{\sigma}}_{1} (n+1)
                                          +
                                          {\bar{\sigma}}_{2} (n+1) > 0$. 
Assume
that
${\bar{\sigma}}_{1}(n+1)>0$
and consider player $1$ switching to strategy $n+1$. 
Since ${\mathsf{Supp}}({\bar{{\sigma}}}_{2})
            \subseteq
            \{ c+1, \ldots, n+1 \}$,
it follows that
only profiles
falling into Case {\sf (0)}
or Case {\sf (1.1)}
with $c+1 \leq s_{2} \leq n$
are supported by 
${\bar{{\bm{\sigma}}}}_{-1} \diamond (n+1)$,
which implies,
by the construction,
that
${\bar{{\mathsf U}}}_{1}\left( {\bar{{\bm{\sigma}}}}_{-1} \diamond (n+1)
                            \right)   
  =
  0$.
Since ${\bar{{\bm{\sigma}}}}$ is a Nash equilibrium,
Lemma~\ref{basic property of mixed nash equilibria}
(Condition {\sf (1)})
implies that
${\bar{{\mathsf{U}}}}_{1}({\bar{{\bm\sigma}}}) = 0$. 
Hence,
Lemma \ref{if zero utility} implies that 
$\langle {\bar{{\mathsf{R}}}},
              {\bar{{\mathsf{C}}}}
  \rangle$
has a pure Nash equilibrium $\langle s_{1}, s_{2}
                                              \rangle$. 
Since
$\langle {\bar{{\mathsf{R}}}},
              {\bar{{\mathsf{C}}}}
  \rangle$
has no pure Nash equilibrium
where some player chooses
the strategy $n+1$, 
it follows that
$s_{1}, s_{2} \in [n]$.
Hence,
$\langle s_{1}, s_{2}
  \rangle$
is a pure Nash equilibrium for
$\langle {\mathsf{R}},
              {\mathsf{C}}
  \rangle$.
A contradiction.
A corresponding argument applies 
to yield a contradiction
when
$\bar{\sigma}_{2} (n+1)>0$.
\end{proof}

\noindent
By Lemmas~\ref{august 15} and~\ref{august 16},
it follows that
${\mathcal{NE}}(\langle {\bar{{\mathsf{R}}}},
                                        {\bar{{\mathsf{C}}}}
                            \rangle)
  =
  {\mathcal{NE}}(\langle {\mathsf{R}},
                                       {\mathsf{C}}
                           \rangle)$,
which 
completes the proof of
Condition {\sf (C.2)}.
\end{proof}

\noindent
Since polynomial time
suffices for deciding the existence of
and computing a pure Nash equilibrium
for a bimatrix game, 
Proposition~\ref{pup is not restrictive}
implies that computing a Nash equilibrium
for a win-lose bimatrix game with the positive utility property
is as hard
as computing a Nash equilibrium
for a win-lose bimatrix game.
Hence,
computing a Nash equilibrium
for a win-lose bimatrix game
with the positive utility property
is ${\mathcal{PPAD}}$-hard.
(The ${\mathcal{PPAD}}$-hardness
will be extended later
to symmetric win-lose bimatrix games
with the positive utility property.)

\section{Win-Lose Gadgets}
\label{gadget games}

\noindent
Here we collect together
the win-lose games
that will be used as gadgets
in later proofs.


\subsection{The Cyclic Game $\widehat{{\mathsf{G}}}_{1}[h]$}
\label{cyclic game}

\noindent
For an integer $h \geq 1$,
define the win-lose bimatrix game
$\widehat{{\mathsf{G}}}_{1}[h]
 :=
 \left\langle [2],
              \{ {\widehat{{\mathsf{\Sigma}}}}_{i}
              \}_{i \in [2]},
              \{ {\widehat{{\mathsf{U}}}}_{i}
              \}_{i \in [2]}
 \right\rangle$
with
${\widehat{{\mathsf{\Sigma}}}}_{1}
 =
 {\widehat{{\mathsf{\Sigma}}}}_{2}
 =
 \{ s_{0}, \ldots, s_{h-1}
 \}$;
${\widehat{{\mathsf{U}}}}_{1}({\bf s})
 =
 1$
if and only if
${\bf s} = \langle s_{i}, s_{i}
           \rangle$
and
${\widehat{{\mathsf{U}}}}_{2}({\bf s})
 =
 1$
if and only if
${\bf s} = \langle s_{i}, s_{i+1}
           \rangle$
(with addition taken modulo $h$).
So,
roughly speaking,
player $1$ wins if and only if
the two players concur, 
while player $2$ wins 
if and only if
the two players choose successive strategies
with player $2$ following.
Clearly,
$\widehat{{\mathsf{G}}}_{1}[h]$
has the positive utility property.
We prove:

\begin{proposition}
\label{gadget1}
$\widehat{\sf G}_{1}[h]$ has a single 
Nash equilibrium
$\bm{\sigma}$,
which is fully mixed and uniform
with
${\widehat{{\mathsf{U}}}}_{1}(\bm{\sigma})
 =
 {\widehat{{\mathsf{U}}}}_{2}(\bm{\sigma})
 =
 \frac{\textstyle{1}}
      {\textstyle{h}}$;
for $h=1$,
$\bm{\sigma}$
is Pareto-Optimal
and
Strongly Pareto-Optimal.      
\end{proposition}

\begin{proof}
Fix a Nash equilibrium $\bm{\sigma}$.
We first prove that
$\bm{\sigma}$ is fully mixed.
Consider first player $1$
and a strategy
$s \in {\widehat{{\mathsf{\Sigma}}}}_{1}$
with
$\sigma_{1} (s)
 >
 0$.
Then,
by Lemma~\ref{basic property of mixed nash equilibria}
(Condition {\sf (1)}),
it follows that
 ${\widehat{{\mathsf{U}}}}_{1}({\bm\sigma})
 =
 {\widehat{{\mathsf{U}}}}_{1}({\bm\sigma}_{-1}\diamond s)$.
Assume, by way of contradiction,
that
$\sigma_{1} (s-1)
 =
 0$.
If $\sigma_{2}(s)=0$,
then
${\widehat{{\mathsf{U}}}}_{1}({\bm\sigma}_{-1}\diamond s)
 =
 0$,
so that
${\widehat{{\mathsf{U}}}}_{1}({\bm\sigma})
 = 
 0$,
a contradiction to Lemma~\ref{park kafe}. 
Hence,
$\sigma_{2}(s) > 0$.
By Lemma~\ref{basic property of mixed nash equilibria}
(Condition {\sf (1)}), 
it follows that
${\widehat{{\sf U}}}_{2}({\bm\sigma})
 =
 {\widehat{{\sf U}}}_{2}({\bm\sigma}_{-2}\diamond s)
 =
 0$,
a contradiction to Lemma~\ref{park kafe}.
A similar argument proves that 
$\sigma_{2}(s)>0$ for each $s \in {\widehat{{\mathsf{\Sigma}}}}_{2}$.

To prove that
$\bm{\sigma}$ is uniform,
assume, by way of contradiction,
that there are indices
$j,k \in {\widehat{{\mathsf{\Sigma}}}}_{1}$
with
$\sigma_{1}(j)
 \neq
 \sigma_{1}(k)$. 
Since ${\mathsf{Supp}}(\sigma_{2}) 
           =
           {\widehat{{\mathsf{\Sigma}}}}_{2}$,
it follows,
by Lemma~\ref{basic property of mixed nash equilibria}
(Condition {\sf (1)}),
that
${\widehat{{\sf U}}}_{2}({\bm\sigma})
 =
 {\widehat{{\sf U}}}_{2}({\bm\sigma}_{-2}\diamond (j+1))
 =
 {\widehat{{\sf U}}}_{2}({\bm\sigma}_{-2}\diamond (k+1))$,
so that
$\sigma_{1}(j)
 =
 \sigma_{1}(k)$,
a contradiction. 
A similar argument 
proves that for all indices
$j, k \in {\widehat{{\mathsf{\Sigma}}}}_{2}$,
$\sigma_{2}(j)
 =
 \sigma_{2}(k)$,
and this implies that there is
a single Nash equilibrium.
Thus, 
${\widehat{{\mathsf{U}}}}_{1}(\bm{\sigma})
 =
 {\widehat{{\mathsf{U}}}}_{2}(\bm{\sigma})
 =
 1
 \cdot
 h
 \cdot
 \frac{\textstyle{1}}
      {\textstyle{h}}
 \cdot
 \frac{\textstyle{1}}
      {\textstyle{h}}     
 =
 \frac{\textstyle{1}}
      {\textstyle{h}}$.     
For $h=1$,
$|{\widehat{{\mathsf{\Sigma}}}}_{1}|
 =
 |{\widehat{{\mathsf{\Sigma}}}}_{2}|
 =
 1$
and
$\bm{\sigma}$ is pure;
clearly,
${\bm{\sigma}}$
is also
both Pareto-Optimal
and
Strongly Pareto-Optimal.       
\end{proof}

\noindent
By Proposition~\ref{gadget1},
the win-lose game ${\widehat{{\mathsf{G}}}}_{1}[h]$
is a negative instance for
{\sf $\exists$ UNIFORM NASH};
when $h=1$,
it is also a negative instance
for
{\sf $\exists$ $\neg$ PARETO-OPTIMAL NASH}
and for
{\sf $\exists$ $\neg$ STRONGLY PARETO-OPTIMAL NASH}.

\subsection{The Irrational Game $\widehat{{\mathsf{G}}}_{2}$}
\label{irrational game}

\noindent
Define the win-lose bimatrix game
$\widehat{\mathsf{G}}_2
 =
 \left\langle [3],
              \{ {\widehat{{\mathsf{\Sigma}}}}_{i}
              \}_{i \in [3]},
              \{ {\widehat{{\mathsf{U}}}}_{i}
              \}_{i \in [3]}
 \right\rangle$
with
${\widehat{{\mathsf{\Sigma}}}}_{1}
 =
 {\widehat{{\mathsf{\Sigma}}}}_{2}
 =
 \{ 0, 1
 \}$
and
${\widehat{{\mathsf{\Sigma}}}}_{3}
 =
 \{ 0, 1, 2
 \}$;
the utility functions are depicted below:
\begin{center}
\begin{small}
\begin{tabular}{||c|c||c|c||c|c||}
\hline
  {\sf Profile} ${\bf s}$  & {\sf Utility vector} $\widehat{{\sf U}}({\bf s})$  
& {\sf Profile} ${\bf s}$  & {\sf Utility vector} ${\widehat{{\mathsf{U}}}}({\bf s})$ 
& {\sf Profile} ${\bf s}$  &  {\sf Utility vector} ${\widehat{{\mathsf{U}}}}({\bf s})$ \\
\hline
\hline
$\langle 0, 0, 0
 \rangle$                & $\langle 1, 0, 1
                                    \rangle$            &  $\langle 0, 1, 1
                                                                    \rangle$                & $\langle 1, 0, 0
                                                                                                      \rangle$                  & $\langle 1, 0, 2
                                                                                                                                           \rangle$                 & $\langle 1, 0, 0
                                                                                                                                                                              \rangle$                                                                      \\
\hline
$\langle 0, 0, 1
 \rangle$                & $\langle 1, 1, 0
                                   \rangle$               & $\langle 0, 1, 2
                                                                     \rangle$                & $\langle 0, 0, 1
                                                                                                        \rangle$               &  $\langle 1, 1, 0
                                                                                                                                          \rangle$                 & $\langle 1, 1, 0
                                                                                                                                                                              \rangle$                                                                   \\
\hline
$\langle 0, 0, 2
 \rangle$                & $\langle 0, 1, 0
                                    \rangle$             & $\langle 1, 0, 0
                                                                    \rangle$                & $\langle 0, 0, 1
                                                                                                      \rangle$                 & $\langle 1, 1, 1
                                                                                                                                          \rangle$                  & $\langle 0, 0, 1
                                                                                                                                                                               \rangle$                                                                    \\
\hline
$\langle 0, 1, 0
 \rangle$                & $\langle 0, 1, 0
                                    \rangle$              & $\langle 1, 0, 1
                                                                     \rangle$                & $\langle 0, 1, 1
                                                                                                        \rangle$              & $\langle 1, 1, 2
                                                                                                                                         \rangle$                   & $\langle 1, 1, 0
                                                                                                                                                                               \rangle$                                                                  \\
\hline
\hline
\end{tabular}
\end{small}
\end{center}

\noindent
Clearly,
$\widehat{\mathsf{G}}_{2}$
has the positive utility property.
We prove:

\begin{proposition}
\label{gadget2}
$\widehat{\mathsf{G}}_{2}$
has a single Nash equilibrium,
which is irrational.
\end{proposition}

\begin{proof}
The following table 
establishes that
$\widehat{{\mathsf{G}}}_{2}$
has no pure Nash equilibrium:

\begin{center}
\begin{small}
\begin{tabular}{||c|c|c||c|c|c||}
\hline
{\sf Profile} ${\bf s}$  & Player: & Improves by switching to: &
{\sf Profile} ${\bf s}$  & Player: & Improves by switching to: \\
\hline
\hline
$\langle 0, 0, 0
 \rangle$                & $2$       & $1$                       &
$\langle 1, 0, 0
 \rangle$                & $1$       & $0$                       \\      
\hline
$\langle 0, 0, 1
 \rangle$                & $3$       & $0$                       &
$\langle 1, 0, 1
 \rangle$                & $1$       & $0$                       \\
\hline
$\langle 0, 0, 2
 \rangle$                & $1$       & $1$                       &
$\langle 1, 0, 2
 \rangle$                & $3$       & $1$                       \\
\hline
$\langle 0, 1, 0
 \rangle$                & $3$       & $2$                       &
$\langle 1, 1, 0
 \rangle$                & $3$       & $1$                       \\   
\hline
$\langle 0, 1, 1
 \rangle$                & $3$       & $2$                       &
$\langle 1, 1, 1
 \rangle$                & $1$       & $0$                       \\ 
\hline
$\langle 0, 1, 2
 \rangle$                & $1$       & $1$                       &
$\langle 1, 1, 2
 \rangle$                & $3$       & $1$                       \\                          
\hline
\hline
\end{tabular}
\end{small}
\end{center}

\noindent
We now turn to a mixed Nash equilibrium
$\bm{\sigma}$,
which shall be represented
with the vector of probabilities
$\langle \sigma_{1}, \sigma_{2}, \sigma^0_{3}, \sigma^1_{3} \rangle$,
where
for each player $i \in [2]$,
$\sigma_{i} := \sigma_{i} (0)$,
while
$\sigma_{3}^{s}
 :=
 \sigma_{3} (s)$
for $s \in \{ 0, 1 \}$;
thus,
for each player $i \in [2]$,
$\sigma_{i} (1) = 1-\sigma_{i}$,
while
$\sigma_{3} (2)
 =
 1 - 
 \sigma_{3}^{0}
 - 
 \sigma_{3}^{1}$.
We 
calculate
the conditional expected utilities
${\widehat{{\mathsf{U}}}}_{i}\left( {\bf s}_{-i}
                        \diamond
                        s
                 \right)$
for each player $i \in [3]$
and strategy
$s \in \textcolor{black}{{\widehat{{\mathsf{\Sigma}}}}}_{i}$: 
\begin{eqnarray*}
           {\widehat{{\sf U}}}_{1}(\bm{\sigma}_{-1}\diamond 0)
& = & \sigma_2\cdot\sigma_3^0
      +
      \sigma_2 \cdot \sigma_3^1
      +
      (1-\sigma_2)\cdot\sigma_3^1                               \\
& = & \sigma_2\cdot\sigma_3^0
      +
      \sigma_3^1\, ,                                                                 \\
          {\widehat{{\sf U}}}_{1}(\bm{\sigma}_{-1} \diamond 1)
& = & \sigma_2\cdot(1-\sigma_3^0-\sigma_3^1)
      +
      (1-\sigma_2)\cdot\sigma_3^0
      +
      (1-\sigma_2)\cdot (1-\sigma_3^0-\sigma_3^1)               \\
& = & 1-\sigma_3^1 -\sigma_2\cdot\sigma_3^0\, ,                 \\
      {\widehat{{\sf U}}}_{2}(\bm{\sigma}_{-2}\diamond 0)
& = & \sigma_1\cdot\sigma_3^1
      +
      \sigma_1\cdot(1-\sigma_3^0-\sigma_3^1)
      +
      (1-\sigma_1)\cdot\sigma_3^1                               \\
& = & \sigma_1\cdot(1-\sigma_3^0-\sigma_3^1)+\sigma_3^1\, ,     \\
      {\widehat{{\sf U}}}_{2}(\bm{\sigma}_{-2}\diamond 1)
& = & \sigma_1\cdot\sigma_3^0
      +
      (1-\sigma_1)\cdot\sigma_3^0
      +
      (1-\sigma_1)\cdot (1-\sigma_3^0-\sigma_3^1)               \\
& = & \sigma_3^0
      +(1-\sigma_1)\cdot (1-\sigma_3^0-\sigma_3^1)\, ,          \\
      {\widehat{{\sf U}}}_{3}(\bm{\sigma}_{-3}\diamond 0)
& = & \sigma_1\cdot\sigma_2
      +
      (1-\sigma_1)\cdot\sigma_2                                 \\
& = & \sigma_2\, ,                                              \\
      {\widehat{{\sf U}}}_{3}(\bm{\sigma}_{-3}\diamond 1)
& = & (1-\sigma_1)\cdot\sigma_2
      +
      (1-\sigma_1)\cdot (1-\sigma_2)                            \\
& = & 1-\sigma_1\, ,                                            \\
      {\widehat{{\sf U}}}_{3}(\bm{\sigma}_{-3}\diamond 2)
& = & \sigma_1\cdot (1-\sigma_2)\, .
\end{eqnarray*}
We now examine all possible cases for $\bm{\sigma}$. 
We shall derive a contradiction
for all but one case,
where we shall establish
that $\bm{\sigma}$ is irrational.

\begin{itemize}

\item
\underline{{\it Only player 1 is mixed.}}
Then, 
$\sigma_{2}
 \in
 \{ 0, 1 \}$
and
$\left( \sigma_{3}^{0},
        \sigma_{3}^{1}
 \right)
 \in
 \{ (0,0),(0,1),(1,0)
 \}$.
Since $\bm{\sigma}$
is a Nash equilibrium
where player 1 is mixed,
Lemma~\ref{basic property of mixed nash equilibria}
(Condition {\sf (1)})
implies that
${\widehat{{\mathsf{U}}}}_{1}\left( \bm{\sigma}_{-1}
                        \diamond
                        0
                 \right)
 =
 {\widehat{{\mathsf{U}}}}_{1}\left( \bm{\sigma}_{-1}
                        \diamond
                        1
                 \right)$,
or
$\sigma_{2} \cdot \sigma_{3}^{0}
 +
 \sigma_{3}^{1}
 =
 1
 - \sigma_{3}^{1} 
 - \sigma_{2} \cdot \sigma_{3}^{0}$,
or
$2 \cdot \left( \sigma_{2} \cdot \sigma_{3}^{0}
                +
                \sigma_{3}^{1}
         \right)
 =
 1$.
Hence,
for $\sigma_{2} = 0$
(resp., $\sigma_{2} = 1$),
$\sigma_{3}^{1}
 =
 \frac{\textstyle 1}
      {\textstyle 2}$
(resp.,
$\sigma_{3}^{0}
 +
 \sigma_{3}^{1}
 =
 \frac{\textstyle 1}
      {\textstyle 2}$).
A contradiction.

\item
\underline{{\it Only player 2 is mixed.}}
Then, 
$\sigma_{1}
 \in
 \{ 0, 1 \}$
and
$\left( \sigma_{3}^{0},
        \sigma_{3}^{1}
 \right)
 \in
 \{ (0,0),(0,1),(1,0)
 \}$.
Since $\bm{\sigma}$
is a Nash equilibrium where
player $2$ is mixed,
Lemma~\ref{basic property of mixed nash equilibria}
(Condition {\sf (1)})
implies that
${\widehat{{\mathsf{U}}}}_{2}\left( \bm{\sigma}_{-2}
                        \diamond
                        0
                 \right)
 =
 {\widehat{{\mathsf{U}}}}_{2}\left( \bm{\sigma}_{-2}
                        \diamond
                        1
                 \right)$,
or
$\left( 2 \sigma_{1} - 1
 \right)
 \cdot
 \left( 1 - \sigma_{3}^{0} - \sigma_{3}^{1}
 \right)
 =
 \sigma_{3}^{0} - \sigma_{3}^{1}$.
Hence,
for $\sigma_{1} = 0$
(resp., $\sigma_{1} = 1$),
$1 - \sigma_{3}^{0} - \sigma_{3}^{1} 
 =
 \sigma_{3}^{1} - \sigma_{3}^{0}$
(resp.,
$1 - \sigma_{3}^{0} - \sigma_{3}^{1} 
 =
 \sigma_{3}^{0} - \sigma_{3}^{1}$),
or
$\sigma_{3}^{1}
 =
 \frac{\textstyle 1}
      {\textstyle 2}$  
(resp.,
$\sigma_{3}^{0}
 =
 \frac{\textstyle 1}
      {\textstyle 2}$).
A contradiction.

\item
\underline{{\it Only player 3 is mixed.}}
Then, 
$\left( \sigma_{1},
        \sigma_{2}
 \right)
 \in
 \{ (0,0),(0,1),(1,0),(1,1)
 \}$.
Since $\bm{\sigma}$ is a Nash equilibrium
where player $3$ is mixed,
Lemma~\ref{basic property of mixed nash equilibria}
(Condition {\sf (1)})
implies that
there are strategies
$s, t \in \{ 0, 1, 2 \}$
such that 
${\widehat{{\mathsf{U}}}}_{3}\left( \bm{\sigma}_{-3}
                                       \diamond
                                        s
                               \right)
 =
 {\widehat{{\mathsf{U}}}}_{3}\left( \bm{\sigma}_{-3}
                                      \diamond
                                     t
                             \right)$.
By Lemma~\ref{park kafe},
it follows that
${\widehat{{\mathsf{U}}}}_{3}\left( \bm{\sigma}_{-3}
                                       \diamond
                                        s
                               \right)
 =
 {\widehat{{\mathsf{U}}}}_{3}\left( \bm{\sigma}_{-3}
                                      \diamond
                                     t
                             \right)
 >
 0$.
Inspecting the following table
yields that
$\sigma_{1} = 0$,
$\sigma_{2} = 1$
and
$\sigma_{3}^{2} = 0$:
\begin{center}
\begin{small}
\begin{tabular}{||c|c|c|c|c||}
\hline
\hline
$\sigma_{1}$                                         &
$\sigma_{2}$                                         & 
${\widehat{{\mathsf{U}}}}_{3}\left( \bm{\sigma}_{-3}
                        \diamond
                        0
                 \right)
 =
 \sigma_{2}$                                         &
${\widehat{{\mathsf{U}}}}_{3}\left( \bm{\sigma}_{-3}
                        \diamond
                        1
                 \right)
 =
 1 - \sigma_{1}$                                     &
${\widehat{{\mathsf{U}}}}_{3}\left( \bm{\sigma}_{-3}
                        \diamond
                        2
                 \right)
 =
 \sigma_{1} \cdot
 \left( 1 - \sigma_{2}
 \right)$                                            \\
\hline
\hline
0 & 0 & 0 & 1 & 0                                    \\
\hline
0 & 1 & 1 & 1 & 0                                    \\
\hline
1 & 0 & 0 & 0 & 1                                    \\
\hline
1 & 1 & 1 & 0 & 0                                    \\
\hline
\hline
\end{tabular}
\end{small}
\end{center}                        
Hence,
$\sigma_{3}^{0} +
 \sigma_{3}^{1}
 =
 1$.
Since $\bm{\sigma}$
is a Nash equilibrium
and $\sigma_{1} = 0$,
Lemma~\ref{basic property of mixed nash equilibria}
(Condition {\sf (2)})
implies that
${\widehat{{\mathsf{U}}}}_{1}\left( \bm{\sigma}_{-1}
                        \diamond
                        1
                 \right) 
 \geq                
 \textcolor{black}{{\widehat{{\mathsf{U}}}}}_{1}\left( \bm{\sigma}_{-1}
                        \diamond
                        0
                 \right)$.
Since $\sigma_{2} = 1$,
it follows that
$1 - \sigma_{3}^{1} - \sigma_{3}^{0}
 \geq
 \sigma_{3}^{0} + \sigma_{3}^{1}$,
or
$\sigma_{3}^{0} + \sigma_{3}^{1} 
 \leq
 \frac{\textstyle 1}
      {\textstyle 2}$.
A contradiction.

\item
\underline{{\it Only players 2 and 3 are mixed.}}
Since player $1$ is pure, 
$\sigma_{1} 
 \in
 \{ 0, 1 \}$. 
Assume, by way of contradiction,
that $\sigma_{1} = 0$.
Then,
${\widehat{{\mathsf{U}}}}_{3}\left( \bm{\sigma}_{-3}
                        \diamond
                        0
                 \right) 
 =
 \sigma_{2}$,
${\widehat{{\mathsf{U}}}}_{3}\left( \bm{\sigma}_{-1}
                        \diamond
                        1
                 \right)
 =
 1$
and
${\widehat{{\mathsf{U}}}}_{3}\left( \bm{\sigma}_{-1}
                        \diamond
                        2
                 \right)
 =
 0$.
Since $\bm{\sigma}$ is a Nash equilibrium,
Lemma~\ref{basic property of mixed nash equilibria}
(Condition {\sf (2)})
implies that
$2 \not\in {\mathsf{Supp}}(\sigma_{3})$.
Since $\bm{\sigma}$ is a Nash equilibrium
where player $3$ is mixed,
Lemma~\ref{basic property of mixed nash equilibria}
(Condition {\sf (1)})
implies that
${\widehat{{\mathsf{U}}}}_{3}\left( \bm{\sigma}_{-3}
                        \diamond
                        0
                 \right) 
 =
 {\widehat{{\mathsf{U}}}}_{3}\left( \bm{\sigma}_{-1}
                        \diamond
                        1
                 \right)$,
which implies that
$\sigma_{2} = 1$,
so that player $2$ is pure.
A contradiction.
It follows that
$\sigma_{1} = 1$.                                  
Then,
${\widehat{{\mathsf{U}}}}_{2}\left( \bm{\sigma}_{-2}
                        \diamond
                        0
                 \right) 
 =
 1 - \sigma_{3}^{0}$
and
${\widehat{{\mathsf{U}}}}_{2}\left( \bm{\sigma}_{-2}
                        \diamond
                        1
                 \right)
 =
 \sigma_{3}^{0}$.

Since $\bm{\sigma}$ is a Nash equilibrium
where player $2$ is mixed,
Lemma~\ref{basic property of mixed nash equilibria}
(Condition {\sf (1)})
implies that
${\widehat{{\mathsf{U}}}}_{2}\left( \bm{\sigma}_{-2}
                        \diamond
                        0
                 \right) 
 =
{\widehat{{\mathsf{U}}}}_{2}\left( \bm{\sigma}_{-2}
                        \diamond
                        1
                 \right)$,
or
$\sigma_{3}^{0}
 =
 \frac{\textstyle 1}
      {\textstyle 2}$.
Also,
${\widehat{{\mathsf{U}}}}_{3}\left( \bm{\sigma}_{-3}
                        \diamond
                        0
                 \right) 
 =
 \sigma_{2}$,
${\widehat{{\mathsf{U}}}}_{3}\left( \bm{\sigma}_{-3}
                        \diamond
                        1
                 \right)
 =
 0$
and
${\widehat{{\mathsf{U}}}}_{3}\left( \bm{\sigma}_{-3}
                        \diamond
                        2
                 \right)
 =
 1 - \sigma_{2}$.
Since player $2$ is mixed,
$\sigma_{2} > 0$.
Since $\bm{\sigma}$ is a Nash equilibrium,
Lemma~\ref{basic property of mixed nash equilibria}
(Condition {\sf (2)})
implies that
$1 \not\in {\mathsf{Supp}}(\sigma_{3})$,
or
$\sigma_{3}^{1}
 =
 0$.
Since $\bm{\sigma}$ is a Nash equilibrium
where player $3$ is mixed,
Lemma~\ref{basic property of mixed nash equilibria}
(Condition {\sf (1)})
implies that
${\widehat{{\mathsf{U}}}}_{3}\left(  \bm{\sigma}_{-3}
                                        \diamond
                                        0
                              \right)
 =
 {\widehat{{\mathsf{U}}}}_{3}\left( \bm{\sigma}_{-3}
                                      \diamond
                                      2
                              \right)$.
It follows that
$\sigma_{2} = \frac{\textstyle 1}
                                {\textstyle 2}$.

Since $\sigma_{1} = 1$
and $\bm{\sigma}$ is a Nash equilibrium,
Lemma~\ref{basic property of mixed nash equilibria}
(Condition {\sf (2)})
implies that
${\widehat{{\mathsf{U}}}}_{1}\left( \bm{\sigma}_{-1}
                        \diamond
                        0
                 \right)
 \geq
 {\widehat{{\mathsf{U}}}}_{1}\left( \bm{\sigma}_{-1}
                        \diamond
                        1
                 \right)$,
or
$\sigma_{2} \cdot \sigma_{3}^{0} + \sigma_{3}^{1}
 \geq
 1 - \sigma_{3}^{1} - \sigma_{2} \cdot \sigma_{3}^{1}$.
Hence,
$\frac{\textstyle 1}
      {\textstyle 2}                  
 \cdot
 \frac{\textstyle 1}
      {\textstyle 2}  
 +
 0
 \geq
 1 - 0 - \frac{\textstyle 1}
              {\textstyle 2}                  
         \cdot
         \frac{\textstyle 1}
              {\textstyle 2}$.
A contradiction.

\item
\underline{{\it Only players 1 and 3 are mixed.}}
Since player $2$ is pure,
$\sigma_{2} 
 \in 
 \{ 0, 1 \}$. 
Assume,
by way of contradiction,
that $\sigma_{2} = 1$.
Then,
\textcolor{black}{${\widehat{{\mathsf{U}}}}_{3}\left( \bm{\sigma}_{-3}
                        \diamond
                        0
                 \right) 
 =
 1$,}
\textcolor{black}{${\widehat{{\mathsf{U}}}}_{3}\left( \bm{\sigma}_{-3}
                        \diamond
                        1
                 \right)
 =
 1 - \sigma_{1}$}
and
$\textcolor{black}{{\widehat{{\mathsf{U}}}}}_{3}\left( \bm{\sigma}_{-3}
                        \diamond
                        2
                 \right)
 =
 0$.
Since $\bm{\sigma}$ is a Nash equilibrium,
Lemma~\ref{basic property of mixed nash equilibria}
(Condition {\sf (2)})
implies that
$2 \not\in {\mathsf{Supp}}(\sigma_{3})$,
or
$\sigma_{3}^{2}
 =
 0$.
Since $\bm{\sigma}$ is a Nash equilibrium
where player $3$ is mixed,
Lemma~\ref{basic property of mixed nash equilibria}
(Condition {\sf (1)})
implies that
\textcolor{black}{${\widehat{{\mathsf{U}}}}_{3}\left( \bm{\sigma}_{-3}
                        \diamond
                        0
                 \right)
 =                 
 {\widehat{{\mathsf{U}}}}_{3}\left( \bm{\sigma}_{-3}
                        \diamond
                        1
                 \right)$,}
which implies that
$\sigma_{1}
 =
 0$,
so that 
player $1$ is pure.
A contradiction.
It follows that
$\sigma_{2} = 0$.

Then,
\textcolor{black}{${\widehat{{\mathsf{U}}}}_{3}\left( \bm{\sigma}_{-3}
                        \diamond
                        0
                 \right)
 =
 0$,}                 
\textcolor{black}{${\widehat{{\mathsf{U}}}}_{3}\left( \bm{\sigma}_{-3}
                        \diamond
                        1
                 \right)
 =
 1 - \sigma_{1}$}
and                 
\textcolor{black}{${\widehat{{\mathsf{U}}}}_{3}\left( \bm{\sigma}_{-3}
                        \diamond
                        2
                 \right)
 =                 
 \sigma_{1}$}.
Since player $1$ is mixed,
$\sigma_{1} > 0$.
Since $\bm{\sigma}$ is a Nash equilibrium,
Lemma~\ref{basic property of mixed nash equilibria}
(Condition {\sf (2)})
implies that
$0 \not\in {\mathsf{Supp}}(\sigma_{3})$,
or
$\sigma_{3}(0) = 0$.

Since $\bm{\sigma}$ is a Nash equilibrium
where player $1$ is mixed,
Lemma~\ref{basic property of mixed nash equilibria}
(Condition {\sf (1)})
implies that
\textcolor{black}{${\widehat{{\mathsf{U}}}}_{1}\left( \bm{\sigma}_{-1}
                        \diamond
                        0
                 \right)
 =                 
 {\widehat{{\mathsf{U}}}}_{1}\left( \bm{\sigma}_{-1}
                        \diamond
                        1
                 \right)$,}
or
$\sigma_{3}^{1} 
 =
 1 - \sigma_{3}^{1}$,
so that
$\sigma_{3}^{1}
 =
 \frac{\textstyle 1}
      {\textstyle 2}$.

Since $\bm{\sigma}$ is a Nash equilibrium where
player $3$ is mixed,
Lemma~\ref{basic property of mixed nash equilibria}
(Condition {\sf (1)})
implies that
\textcolor{black}{${\widehat{{\mathsf{U}}}}_{3}\left( \bm{\sigma}_{-3}
                        \diamond
                        1
                 \right)
 =                
 {\widehat{{\mathsf{U}}}}_{3}\left( \bm{\sigma}_{-3}
                        \diamond
                        2
                 \right)$,}
or
$ 1 - \sigma_{1}
 =
 \sigma_{1}$,
or
$\sigma_{1} 
 =
 \frac{\textstyle 1}
      {\textstyle 2}$.
Since $\sigma_{2} = 0$
and $\bm{\sigma}$
is a Nash equilibrium,
Lemma~\ref{basic property of mixed nash equilibria}
(Condition {\sf (2)})
implies that
\textcolor{black}{${\widehat{{\mathsf{U}}}}_{2}\left( \bm{\sigma}_{-2}
                        \diamond
                        0
                 \right)
 \leq
 {\widehat{{\mathsf{U}}}}_{2}\left( \bm{\sigma}_{-2}
                        \diamond
                        1
                 \right)$,}
or
$\frac{\textstyle 1}
      {\textstyle 2}
 \cdot
 \left( 1 - 0 - \frac{\textstyle 1}
                     {\textstyle 2}
 \right)
 +
 \frac{1}
      {2}
 \leq
 0
 +
 \frac{\textstyle 1}
      {\textstyle 2}
 \cdot
 \left( 1 - 0 - \frac{\textstyle 1}
                     {\textstyle 2}
 \right)$
or
$\frac{\textstyle 1}
      {\textstyle 4}
 \leq
 - \frac{\textstyle 1}
        {\textstyle 4}$.
A contradiction.

\item
\underline{{\it Only players 1 and 2 are mixed.}}
Since player $3$ is pure,
$\left( \sigma_{3}^{0},
        \sigma_{3}^{1}
 \right)
 \in
 \{ (0,0),(0,1),(1,0)
 \}$.
Since $\bm{\sigma}$ is a Nash equilibrium
where player $1$ is mixed,
Lemma~\ref{basic property of mixed nash equilibria}
(Condition {\sf (1)})
implies that
\textcolor{black}{${\widehat{{\mathsf{U}}}}_{1}\left( \bm{\sigma}_{-1}
                        \diamond
                        0
                 \right)
=
 {\widehat{{\mathsf{U}}}}_{1}\left( \bm{\sigma}_{-1}
                        \diamond
                        1
                 \right)$,}
or
$\sigma_{2} \cdot \sigma_{3}^{0}
 +
 \sigma_{3}^{1}
 =
 1
 - \sigma_{3}^{1}
 - \sigma_{2} \cdot \sigma_{3}^{0}$
or
$2 \cdot \left( \sigma_{2} \cdot \sigma_{3}^{0}
                +
                \sigma_{3}^{1}
         \right)
 =
 1$.
This implies that
$\left( \sigma_{3}^{0},
        \sigma_{3}^{1}
 \right)
 \not\in
 \{ (0,0),(0,1)
 \}$;
so,
$\left( \sigma_{3}^{0},
        \sigma_{3}^{1}
 \right) 
 =
 (1, 0)$.

Since $\bm{\sigma}$ is a Nash equilibrium
where player $2$ is mixed,
Lemma~\ref{basic property of mixed nash equilibria}
(Condition {\sf (1)})
implies that
$\textcolor{black}{{\widehat{{\mathsf{U}}}}}_{2}\left( \bm{\sigma}_{-2}
                        \diamond
                        0
                 \right)
 =
 \textcolor{black}{{\widehat{{\mathsf{U}}}}}_{2}\left( \bm{\sigma}_{-2}
                        \diamond
                        1
                 \right)$,
or
$\sigma_{1} \cdot
 \left( 1 - \sigma_{3}^{0} - \sigma_{3}^{1}
 \right)
 +
 \sigma_{3}^{1}
 =
 \sigma_{3}^{0}
 +
 \left( 1 - \sigma_{1}
 \right)
 \cdot
 \left( 1 - \sigma_{3}^{0} - \sigma_{3}^{1}
 \right)$
or
$\sigma_{1} \cdot 0
 +
 0
 =
 1 +
 \left( 1 - \sigma_{1}
 \right)
 \cdot
 0$
or $0 = 1$.
A contradiction.

\item
\underline{{\it All players are mixed.}}
We proceed by case analysis.
Assume first that
$\sigma_{3}^{0} = 0$.
Since $\bm{\sigma}$ is a Nash equilibrium
where player $1$ is mixed,
Lemma~\ref{basic property of mixed nash equilibria}
(Condition {\sf (1)})
implies that
\textcolor{black}{${\widehat{{\mathsf{U}}}}_{1}\left( \bm{\sigma}_{-1}
                        \diamond
                        0
                 \right)
 =
 {\widehat{{\mathsf{U}}}}_{1}\left( \bm{\sigma}_{-1}
                        \diamond
                        1
                 \right)$,}
or
$\sigma_{3}^{1}
 =
 1 - \sigma_{3}^{1}$,
so that
$\sigma_{3}^{1}
 =
 \frac{\textstyle 1}
      {\textstyle 2}$.

Since $\bm{\sigma}$ is a Nash equilibrium
where player $2$ is mixed,
Lemma~\ref{basic property of mixed nash equilibria}
(Condition {\sf (1)})
implies that
\textcolor{black}{${\widehat{{\mathsf{U}}}}_{2}\left( \bm{\sigma}_{-2}
                        \diamond
                        0
                 \right)
 =
 {\widehat{{\mathsf{U}}}}_{2}\left( \bm{\sigma}_{-2}
                        \diamond
                        1
                 \right)$,}
or
$\sigma_{1}
 \cdot
 \frac{\textstyle 1}
      {\textstyle 2}
 +
 \frac{\textstyle 1}
      {\textstyle 2}
 =
 \left( 1 - \sigma_{1}
 \right)
 \cdot
 \frac{\textstyle 1}
      {\textstyle 2}$
or
$\sigma_{1} = 0$,
so that
player $1$ is pure.                
A contradiction.

Assume now that
$\sigma_{3}^{1} = 0$.
Since $\bm{\sigma}$ is a Nash equilibrium
where player $3$ is mixed,
Lemma~\ref{basic property of mixed nash equilibria}
(Condition {\sf (1)})
implies that
\textcolor{black}{${\widehat{{\mathsf{U}}}}_{3}\left( \bm{\sigma}_{-3}
                        \diamond
                        0
                 \right)
 =
 {\widehat{{\mathsf{U}}}}_{3}\left( \bm{\sigma}_{-3}
                        \diamond
                        2
                 \right)$,}
or
$\sigma_{2}
 =
 \sigma_{1}
 \cdot
 \left( 1 - \sigma_{2}
 \right)$.
Since $\bm{\sigma}$ is a Nash equilibrium
where player $1$ is mixed,
Lemma~\ref{basic property of mixed nash equilibria}
(Condition {\sf (1)})
implies that
\textcolor{black}{${\widehat{{\mathsf{U}}}}_{1}\left( \bm{\sigma}_{-1}
                        \diamond
                        0
                 \right)
 =
 {\widehat{{\mathsf{U}}}}_{1}\left( \bm{\sigma}_{-1}
                        \diamond
                        1
                 \right)$,}
or 
$\sigma_{2} \cdot \sigma_{3}^{0}
 =
 1 - \sigma_{2} \cdot \sigma_{3}^{0}$
so that
$2\, \sigma_{2} \cdot \sigma_{3}^{0} = 1$.
Since $\bm{\sigma}$ is a Nash equilibrium
where player $2$ is mixed,
Lemma~\ref{basic property of mixed nash equilibria}
(Condition {\sf (1)})
implies that
\textcolor{black}{${\widehat{{\mathsf{U}}}}_{2}\left( \bm{\sigma}_{-2}
                        \diamond
                        0
                 \right)
 =
 {\widehat{{\mathsf{U}}}}_{2}\left( \bm{\sigma}_{-2}
                        \diamond
                        1
                 \right)$, }
or
$\sigma_{1} \cdot
 \left( 1 - \sigma_{3}^{0}
 \right)
 =
 \sigma_{3}^{0}
 +
 \left( 1 - \sigma_{1}
 \right)
 \cdot
 \left( 1 - \sigma_{3}^{0}
 \right)$
or
$2 \sigma_{1} - 1
 =
 \frac{\textstyle \sigma_{3}^{0}}
      {\textstyle 1 - \sigma_{3}^{0}}$.
Hence,
it follows that
$\left( 2 \sigma_{1} - 1      
 \right)
 \cdot
 \left( 2 \sigma_{2} - 1
 \right)
 =
 1$.
Thus,
$\left( 2 \frac{\textstyle \sigma_{2}}
               {\textstyle \sigma_{2} - 1}
        -
        1
 \right)                    
 \cdot
 \left( 2 \sigma_{2} - 1
 \right)
 =
 1$,
or
$\left( -1 + 3 \sigma_{2}
 \right)
 \cdot
 \left( 2 \sigma_{2} - 1
 \right)
 =
 1 - \sigma_{2}$,
which yields
$\sigma_{2} = \frac{\textstyle 2}
                   {\textstyle 3}$.
Hence,
$\sigma_{1}
 =
 \frac{\textstyle \sigma_{2}}
      {\textstyle 1 - \sigma_{2}}
 =
 2$.
A contradiction.

Assume finally that
$\sigma_{3}^{2}
 =
 0$,
so that
$\sigma_{3}^{0} +
 \sigma_{3}^{1}
 =
 1$.
Since $\bm{\sigma}$ is a Nash equilibrium
where player $1$ is mixed,
Lemma~\ref{basic property of mixed nash equilibria}
(Condition {\sf (1)})
implies that
\textcolor{black}{${\widehat{{\mathsf{U}}}}_{1}\left( \bm{\sigma}_{-1}
                        \diamond
                        0
                 \right)
 =
 {\widehat{{\mathsf{U}}}}_{1}\left( \bm{\sigma}_{-1}
                        \diamond
                        1
                 \right)$,}
or   
$\sigma_{2}
 \cdot
 \sigma_{3}^{0}
 +
 \sigma_{3}^{1}
 =
 1 - \sigma_{3}^{1}
 -
 \sigma_{2} \cdot \sigma_{3}^{0}$
or
$2\,
 \sigma_{2}
 \sigma_{3}^{0}
 +
 2
 \cdot
 \left( 1 - \sigma_{3}^{0}
 \right)
 =
 1$,
so that
$2\, \sigma_{3}^{0}
 \cdot
 \left( 1 - \sigma_{2}
 \right)
 =
 1$. 
Since $\bm{\sigma}$ is a Nash equilibrium
where player $2$ is mixed,
Lemma~\ref{basic property of mixed nash equilibria}
(Condition {\sf (1)})
implies that
\textcolor{black}{${\widehat{{\mathsf{U}}}}_{2}\left( \bm{\sigma}_{-2}
                        \diamond
                        0
                 \right)
 =
 {\widehat{{\mathsf{U}}}}_{2}\left( \bm{\sigma}_{-2}
                        \diamond
                        1
                 \right)$,}
or   
$\sigma_{3}^{1}
 =
 \sigma_{3}^{0}$.
It follows that
$\sigma_{3}^{0}
 =
 \frac{\textstyle 1}
      {\textstyle 2}$,
so that
$\sigma_{2} = 0$
and player $2$ is pure.
A contradiction.
It follows from the case analysis
that
$0 < \sigma_{3}^{0}, \sigma_{3}^{1}, \sigma_{3}^{2}
   < 1$,
and player $3$ is fully mixed.

Since  $\bm{\sigma}$ is a Nash equilibrium
where player $3$ is fully mixed,
Lemma~\ref{basic property of mixed nash equilibria}
(Condition {\sf (1)})
implies that
\textcolor{black}{${\widehat{{\mathsf{U}}}}_{3}\left( \bm{\sigma}_{-3}
                        \diamond
                        0
                 \right)
 =
 {\widehat{{\mathsf{U}}}}_{3}\left( \bm{\sigma}_{-3}
                        \diamond
                        1
                 \right)
 =
 {\widehat{{\mathsf{U}}}}_{3}\left( \bm{\sigma}_{-3}
                        \diamond
                        2
                 \right)$,}
or
$\sigma_{2}
 =
 1 - \sigma_{1}
 =
 \sigma_{1}
 \cdot
 \left( 1 - \sigma_{2}
 \right)$.
Hence,
$\left( 1 - \sigma_{2}
 \right)^{2}
 =
 \sigma_{2}$,
so that
$\sigma_{2}
 =
 \frac{\textstyle 3 - \sqrt{5}}
      {\textstyle 2}$
and
$\sigma_{1}
 =
 \frac{\textstyle \sqrt{5} - 1}
      {\textstyle 2}$.
Since $\bm{\sigma}$ is a Nash equilibrium
where player $1$ is mixed,
Lemma~\ref{basic property of mixed nash equilibria}
(Condition {\sf (1)})
implies that
\textcolor{black}{${\widehat{{\mathsf{U}}}}_{1}\left( \bm{\sigma}_{-1}
                        \diamond
                        0
                 \right)
 =
 {\widehat{{\mathsf{U}}}}_{1}\left( \bm{\sigma}_{-1}
                        \diamond
                        1
                 \right)$,}
or   
$2\, \sigma_{2}
 \cdot
 \sigma_{3}^{0}
 +
 2\,
 \sigma_{3}^{1}
 =
 1$.
Since $\bm{\sigma}$ is a Nash equilibrium
where player $2$ is mixed,
Lemma~\ref{basic property of mixed nash equilibria}
(Condition {\sf (1)})
implies that
\textcolor{black}{${\widehat{{\mathsf{U}}}}_{2}\left( \bm{\sigma}_{-2}
                        \diamond
                        0
                 \right)
 =
 {\widehat{{\mathsf{U}}}}_{2}\left( \bm{\sigma}_{-2}
                        \diamond
                        1
                 \right)$,}
or   
$\left( 2 \sigma_{1} - 1
 \right)
 \cdot
 \left( 1 - \sigma_{3}^{0} - \sigma_{3}^{1}
 \right)
 =
 \sigma_{3}^{0}
 -
 \sigma_{3}^{1}$
or
$2\, \left( 1 - \sigma_{1}
     \right)
 \cdot
 \sigma_{3}^{1}
 -
 2\, \sigma_{1} \cdot 
 \sigma_{3}^{0}
 =
 1 - 2 \sigma_{1}$.     
Combining we obtain
$\left( 2\, \sigma_{1}
        +
        2\, (1 - \sigma_{1})
            \cdot
            \sigma_{2}
 \right)
 \cdot
 \sigma_{3}^{1}
 =
 \left( 1 - 2 \sigma_{1}
 \right)
 \cdot
 \sigma_{2}
 +
 \sigma_{1}$,
which yields
$\sigma_{3}^{1}
 =
 \frac{\textstyle 5 - \sqrt{5}}
      {\textstyle 8}$.
Now,
$\sigma_{3}^{0}
 =
 \frac{\textstyle 1 - 2 \sigma_{3}^{1}}
      {\textstyle 2 \sigma_{2}}
 =
 \frac{\textstyle 1 + \sqrt{5}}
      {\textstyle 8}$,
and
$\sigma_{3}^{2}
 =
 1 - \sigma_{3}^{0} - \sigma_{3}^{1}
 =
 \frac{\textstyle 1}
      {\textstyle 4}$.

\end{itemize}
\noindent
The claim now follows.
\end{proof}

\noindent
By Proposition~\ref{gadget2},
the win-lose 3-player game
${\widehat{{\mathsf{G}}}}_{2}$
is a negative instance for
{\sf $\exists$ RATIONAL NASH}.
Although ${\widehat{{\mathsf{G}}}}_{2}$
is also a negative instance for
{\sf $\exists$ UNIFORM NASH}
and for
{\sf $\exists$ SYMMETRIC NASH}
we \textcolor{black}{shall} present \textcolor{black}{next}
the win-lose bimatrix games
${\widehat{{\mathsf{G}}}}_{3}$
and
${\widehat{{\mathsf{G}}}}_{4}$
as negative instances for
{\sf $\exists$ UNIFORM NASH}
and for
{\sf $\exists$ SYMMETRIC NASH},
respectively,
since we \textcolor{black}{seek to establish}
that these
two decision problems
are ${\mathcal{NP}}$-hard
for win-lose bimatrix games.

\subsection{The Non-Uniform Game $\widehat{{\mathsf{G}}}_{3}$}
\label{nonuniform game}

\noindent
Define the win-lose game
\textcolor{black}{$\widehat{{\sf G}}_{3}
 =
 \langle
   [2],
   \{ {\widehat{{\sf\Sigma}}}_i\}_{i \in [2]},
   \{ {\widehat{{\sf U}}}_i\}_{i \in [2]}
 \rangle$}
with
\textcolor{black}{${\widehat{{\mathsf{\Sigma}}}}_{1}
 =
 {\widehat{{\mathsf{\Sigma}}}}_{2}
 =
 [4]$};
the utility functions are depicted
below:
\begin{center}
\begin{small}
\begin{tabular}{||c|c||c|c||c|c||c|c||}
\hline
{\sf Profile} ${\bf s}$  & $\widehat{{\mathsf{U}}}({\bf s})$  &
{\sf Profile} ${\bf s}$  & $\widehat{{\mathsf{U}}}({\bf s})$  &
{\sf Profile} ${\bf s}$  & $\widehat{{\mathsf{U}}}({\bf s})$  &
{\sf Profile} ${\bf s}$  & $\widehat{{\mathsf{U}}}({\bf s})$  \\
\hline
\hline
$\langle 1, 1 \rangle$ & $\langle 0, 1 \rangle$ & 
$\langle 1, 2 \rangle$ & $\langle 0, 1 \rangle$ & 
$\langle 1, 3 \rangle$ & $\langle 1, 0 \rangle$ &
$\langle 1, 4 \rangle$ & $\langle 1, 0 \rangle$ \\
\hline
$\langle 2, 1 \rangle$ & $\langle 0, 1 \rangle$ & 
$\langle 2, 2 \rangle$ & $\langle 1, 0 \rangle$ & 
$\langle 2, 3 \rangle$ & $\langle 1, 0 \rangle$ &
$\langle 2, 4 \rangle$ & $\langle 0, 1 \rangle$ \\
\hline
$\langle 3, 1 \rangle$ & $\langle 0, 1 \rangle$ & 
$\langle 3, 2 \rangle$ & $\langle 1, 0 \rangle$ & 
$\langle 3, 3 \rangle$ & $\langle 0, 1 \rangle$ &
$\langle 3, 4 \rangle$ & $\langle 1, 0 \rangle$ \\
\hline
$\langle 4, 1 \rangle$ & $\langle 1, 0 \rangle$ & 
$\langle 4, 2 \rangle$ & $\langle 0, 1 \rangle$ & 
$\langle 4, 3 \rangle$ & $\langle 0, 1 \rangle$ &
$\langle 4, 4 \rangle$ & $\langle 0, 1 \rangle$ \\
\hline
\hline
\end{tabular}
\end{small}
\end{center}
Clearly,
$\widehat{{\mathsf{G}}}_3$
has the positive utility property.
Note that
for each profile
${\bf s}$,
exactly one of
\textcolor{black}{${\widehat{{\mathsf{U}}}}_{1}\left( {\bf s}
                              \right)
 =
 1$}
and
\textcolor{black}{${\widehat{{\mathsf{U}}}}_{2}\left( {\bf s}
                              \right)
 =
 1$}
holds,
and $\widehat{{\mathsf{G}}}_3$
is $1$-sum.
Define the {\it shape}
of a profile ${\bf s}$ to be
{\it square}
(resp., 
{\it circle})
when 
${\widehat{{\mathsf{U}}}}_{1}\left( {\bf s}
                              \right)
 =
 1$
(resp.,
${\widehat{{\mathsf{U}}}}_{2}\left( {\bf s}
                              \right)
 =
 1$).
Figure~\ref{shapes} presents
a bimatricial representation
of the set of profiles for
$\widehat{{\mathsf{G}}}_{3}$,
where each profile
is drawn as either a square or a circle;
rows and columns
correspond to strategies of player $1$ and $2$,
respectively.

\begin{figure}[tp]
\begin{center}
\setlength{\unitlength}{1cm}
\begin{small}
\begin{picture}(5.5,5.5)(0,0)
\put(0,0){\framebox(1,1)}
\put(2,0.5){\circle{1}}
\put(3.5,0.5){\circle{1}}
\put(5,0.5){\circle{1}}

\put(0.5,2){\circle{1}}
\put(1.5,1.5){\framebox(1,1)}
\put(3.5,2){\circle{1}}
\put(4.5,1.5){\framebox(1,1)}

\put(0.5,3.5){\circle{1}}
\put(1.5,3){\framebox(1,1)}
\put(3,3){\framebox(1,1)}
\put(5,3.5){\circle{1}}

\put(0.5,5){\circle{1}}
\put(2,5){\circle{1}}
\put(3,4.5){\framebox(1,1)}
\put(4.5,4.5){\framebox(1,1)}

\put(0.16,0.43){$\langle 4,1\rangle$}
\put(1.66,0.43){$\langle 4,2\rangle$}
\put(3.16,0.43){$\langle 4,3\rangle$}
\put(4.66,0.43){$\langle 4,4\rangle$}

\put(0.16,1.93){$\langle 3,1\rangle$}
\put(1.66,1.93){$\langle 3,2\rangle$}
\put(3.16,1.93){$\langle 3,3\rangle$}
\put(4.66,1.93){$\langle 3,4\rangle$}

\put(0.16,3.43){$\langle 2,1\rangle$}
\put(1.66,3.43){$\langle 2,2\rangle$}
\put(3.16,3.43){$\langle 2,3\rangle$}
\put(4.66,3.43){$\langle 2,4\rangle$}

\put(0.16,4.93){$\langle 1,1\rangle$}
\put(1.66,4.93){$\langle 1,2\rangle$}
\put(3.16,4.93){$\langle 1,3\rangle$}
\put(4.66,4.93){$\langle 1,4\rangle$}

\end{picture}
\end{small}
\caption{A bimatricial representation of all profiles
         in $\widehat{{\sf G}}_3$
         as shapes.}
\label{shapes}         
\end{center}
\end{figure}
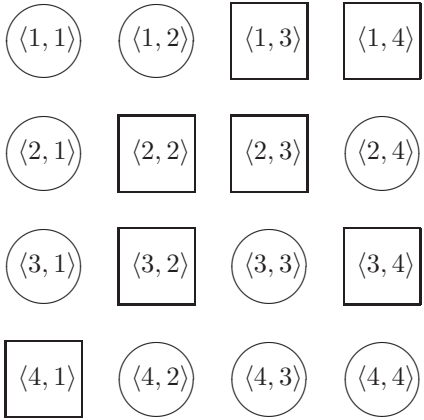
\noindent

\noindent
Note that
in a uniform Nash equilibrium $\bm{\sigma}$,
a best-response of player $2$
(resp., player $1$)
to $\sigma_{1}$
(resp., $\sigma_{2}$)
is a strategy
\textcolor{black}{$s \in {\widehat{{\mathsf{\Sigma}}}}_{2}$}
(resp.,
\textcolor{black}{$s \in {\widehat{{\mathsf{\Sigma}}}}_{1}$})
whose corresponding column
(resp., row),
restricted to the rows
(resp., the columns)
corresponding to the strategies
in ${\mathsf{Supp}}(\sigma_{1})$
(resp.,
in ${\mathsf{Supp}}(\sigma_{2})$),
contains a maximum number of circles
(resp., squares).
In particular,
all columns
(resp., rows)
corresponding to strategies
that are best-responses of player $2$
(resp., player $1$)
to $\sigma_{1}$
(resp., $\sigma_{2}$),
restricted to the rows
(resp., the columns)
corresponding to the strategies
in ${\mathsf{Supp}}(\sigma_{1})$
(resp.,
in ${\mathsf{Supp}}(\sigma_{2})$),
contain the same number of circles
(resp., squares).
We prove:

\begin{proposition}
\label{barrage}
$\widehat{{\sf G}}_3$ 
has no uniform Nash equilibrium.
\end{proposition}

\begin{proof}
Note first that
$\widehat{{\mathsf{G}}}_{3}$
has no pure Nash equilibrium.
We now prove:

\begin{lemma}[No Pure Player]
\label{barrage 1}
Fix a Nash equilibrium 
$\bm{\sigma}$.
Then,
for each player $i \in [2]$,
$|{\mathsf{Supp}}(\sigma_{i})|
 \geq
 2$.
\end{lemma}

\begin{proof}
Assume,
by way of contradiction, 
that there is
a player $i \in [2]$
with $|{\mathsf{Supp}}(\sigma_{i})|
      =
      1$.
Since $\bm{\sigma}$ is mixed,
it follows that      
player 
$\overline{i}
  \in
  [2] \setminus \{ i \}$ is mixed.
Since $\bm{\sigma}$ is a Nash equilibrium
where player $\overline{i}$ is mixed,
Lemma~\ref{basic property of mixed nash equilibria}
(Condition {\sf (1)})
implies that
for each strategy 
$s \in {\mathsf{Supp}}(\sigma_{\overline{i}})$,
\textcolor{black}{${\widehat{{\sf U}}}_{\overline{i}}
      ({\bm\sigma})
 =     
 {\widehat{{\mathsf{U}}}}_{\overline{i}}
          \left( \bm{\sigma}_{- \overline{i}} 
                   \diamond
                   s
          \right)$}.
So,
all profiles
supported in
$\bm{\sigma}_{- \overline{i}}
  \diamond
  s$
have the same shape. 
Since
$\widehat{\mathsf{G}}_{3}$
is a 1-sum, win-lose game, 
either 
\textcolor{black}{${\widehat{{\mathsf{U}}}}_{i}({\bm{\sigma}})
 =
 0$}
or
\textcolor{black}{${\widehat{{\mathsf{U}}}}_{\overline{i}}
          ({\bm{\sigma}})
 =
 0$}.
Since
$\widehat{{\mathsf{G}}}_{3}$ 
has the positive utility property,
Lemma~\ref{park kafe}
implies that both
\textcolor{black}{${\widehat{{\mathsf{U}}}}_{i}({\bm{\sigma}})
 >
 0$}
and
\textcolor{black}{${\widehat{{\mathsf{U}}}}_{\overline{i}}
          ({\bm{\sigma}})
 >
 0$}.
A contradiction. 
\end{proof}

\noindent
To complete the proof,
we shall establish that
there is no uniform Nash equilibrium $\bm\sigma$
with $2 \leq |{\mathsf{Supp}}(\sigma_{1})|
        \leq 4$.
We first prove:

\begin{lemma}
\label{barrage 5}
Fix a uniform Nash equilibrium 
$\bm{\sigma}$.
Then,
$|{\mathsf{Supp}}(\sigma_{1})|
 \neq
 4$.
\end{lemma}

\begin{proof}
Assume,
by way of contradiction,
that
$|{\mathsf{Supp}}(\sigma_{1})|
 =
 4$.
In this case, 
strategy $1$ is
the unique
best-response
of player 2
to $\sigma_{1}$
since column 1
is the only column
containing the maximum number of circles,
which is 3.
So,
${\mathsf{Supp}}(\sigma_{2})
 =
 \{ 1 \}$.
A contradiction
to Lemma~\ref{barrage 1}.
\end{proof}

\noindent
We now prove a simple structural property:

\begin{lemma}[No Column Crossing All Supported Rows Only in Circles]  
\label{barrage structure}
Fix a Nash equilibrium $\bm{\sigma}$.
Then,
there is no strategy
\textcolor{black}{$s \in {\widehat{{\mathsf{\Sigma}}}}_{2}$}
such that
all profiles supported in
$\bm{\sigma}_{-2}
 \diamond 
 s$
are circles.
\end{lemma}

\begin{proof}
Assume, 
by way of contradiction,
that there is
a strategy 
\textcolor{black}{$s \in {\widehat{{\mathsf{\Sigma}}}}_{2}$}
such that
all profiles supported in 
${\bm\sigma}_{-2}
 \diamond
 s$
are circles.
This implies that
\textcolor{black}{${\widehat{{\sf U}}}_{2}\left( \bm{\sigma}_{-2}
                   \diamond
                   s
            \right)
 =
 1$}.
Since $\bm{\sigma}$ is a Nash equilibrium
where player $2$ is mixed, 
Lemma~\ref{basic property of mixed nash equilibria}
(Condition {\sf (2)})
implies that
\textcolor{black}{${\widehat{{\sf U}}}_{2}(\bm{\sigma})
 \geq
 {\widehat{{\sf U}}}_{2}\left( \bm{\sigma}_{-2}
                   \diamond
                   s
            \right)$}.
Since $\widehat{{\mathsf{G}}}_{3}$
is win-lose,
this implies that
\textcolor{black}{${\widehat{{\sf U}}}_{2}(\bm{\sigma})
 =            
 1$}.
Since $\widehat{{\mathsf{G}}}_{3}$
is $1$-sum,
it follows that
\textcolor{black}{${\widehat{{\mathsf{U}}}}_{1}(\bm{\sigma})
 =            
 0$}.
Since $\widehat{{\mathsf{G}}}_{3}$
has the positive utility property,
Lemma~\ref{park kafe} implies that
\textcolor{black}{${\widehat{{\mathsf{U}}}}_{1}(\bm{\sigma})
 >            
 0$}.
A contradiction.
\end{proof}

\noindent
We next prove:

\begin{lemma}
\label{barrage 2}
Fix a uniform Nash equilibrium 
$\bm{\sigma}$.
Then,
$|{\mathsf{Supp}}(\sigma_{1})|
 \neq
 2$.
\end{lemma}

\begin{proof}
Assume,
by way of contradiction,
that
$|{\mathsf{Supp}}(\sigma_{1})|
 =
 2$.
Say that two distinct strategies
\textcolor{black}{$s, s^{\prime}
 \in
 {\widehat{{\mathsf{\Sigma}}}}_{1}$}
for player 1
{\em collide on}
strategy 
\textcolor{black}{$t \in {\widehat{{\mathsf{\Sigma}}}}_{2}$}
if both profiles 
$\left\langle s, t
 \right\rangle$
and
$\left\langle s^{\prime}, t
 \right\rangle$
are circles.
Lemma~\ref{barrage structure}
implies that
the two strategies in 
${\mathsf{Supp}}(\sigma_{1})$
do not collide
on any strategy in 
\textcolor{black}{${\widehat{{\mathsf{\Sigma}}}}_{2}$}.
Since column 1 contains only circles
except for
the entry corresponding to row $4$, 
Lemma~\ref{barrage structure}
implies that
$4 
 \in
 {\mathsf{Supp}}(\sigma_{1})$.
So,
three cases need to be considered,
corresponding to which the other strategy 
in ${\mathsf{Supp}}(\sigma_{1})$ is.
By inspection,
strategies $1$ and $4$ collide on column $2$;
strategies $2$ and $4$ collide on column $4$;
strategies $3$ and $4$ collide on column $3$.
A contradiction.
\end{proof}

\noindent
We finally prove:

\begin{lemma}
\label{barrage 3}
Fix a uniform Nash equilibrium 
$\bm{\sigma}$. 
Then,
$|{\mathsf{Supp}}(\sigma_{1})|
 \neq
 3$.
\end{lemma}

\begin{proof}
Assume,
by way of contradiction,
that
$|{\mathsf{Supp}}(\sigma_{1})|
 =
 3$.
There are four possibilities
corresponding to which strategy $s$
among the four strategies 
in \textcolor{black}{${\widehat{{\mathsf{\Sigma}}}}_{1}$}
is excluded from
${\mathsf{Supp}}(\sigma_{1})$.
We proceed by case analysis.
\begin{enumerate}

\item
\underline{$s=1$ so that
           ${\mathsf{Supp}}\left( \sigma_{1}
                           \right)
            =
            \{ 2, 3, 4 \}$:}
Since $\bm{\sigma}$ is uniform,
strategies in $\{ 1, 3, 4 \}$
are the only three
best-responses of player $2$
to $\sigma_{1}$ 
since each of their corresponding columns,
restricted to rows corresponding
to the strategies in
${\mathsf{Supp}}(\sigma_{1})$,
contains the maximum number of circles,
which is $2$.
By Lemma~\ref{basic property of mixed nash equilibria},
it follows that
${\mathsf{Supp}}(\sigma_{2})
  \subseteq
  \{ 1, 3, 4 \}$.
Observe that the only square profiles
supported in ${\bm{\sigma}}$
whenever
${\mathsf{Supp}}\left( \sigma_{2}
                \right)
 =
 \{ 1, 3, 4 \}$
are 
$\langle 4, 1 \rangle$,
$\langle 2, 3 \rangle$
and
$\langle 3, 4 \rangle$;
so,
there is exactly one square profile
for each possible best-response
of player 2
to $\sigma_{1}$.
This implies that
\textcolor{black}{${\widehat{{\mathsf{U}}}}_{1}\left( \bm{\sigma}
                 \right)
 =
 \frac{\textstyle 1}
         {\textstyle 3}$,}
regardless of ${\mathsf{Supp}}(\sigma_{2})$.          
If $|{\mathsf{Supp}}(\sigma_{2})|
     <
     3$,
then
for each strategy
$s \in {\mathsf{Supp}}(\sigma_{1})$,
\textcolor{black}{${\widehat{{\mathsf{U}}}}_{1}\left( \bm{\sigma}_{-1}
                        \diamond
                        s
               \right)
 \geq
 \frac{\textstyle 1}
         {\textstyle 2}$,}
a contradiction to Lemma~\ref{basic property of mixed nash equilibria}
(Condition {\sf (1)}).              
So,
${\mathsf{Supp}}\left( \sigma_{2}
                             \right)
      =
      \{ 1, 3, 4 \}$.
But then 
\textcolor{black}{${\widehat{{\mathsf{U}}}}_{1}\left( \bm{\sigma}_{-1}
                        \diamond
                        1
               \right)
 =
 \frac{\textstyle 2}
         {\textstyle 3}$,}
a contradiction to Lemma~\ref{basic property of mixed nash equilibria}
(Condition {\sf (1)}).

\item
\underline{$s = 2$
           so that
           ${\mathsf{Supp}}\left( \sigma_{1}
                           \right)
            =
            \{ 1, 3, 4 \}$:}
Since $\bm{\sigma}$ is uniform,
strategies in $\{ 1, 2, 3 \}$
are the only three
best-responses for player $2$
to $\sigma_{1}$ 
since each of their corresponding columns,
restricted to rows corresponding
to the strategies in
${\mathsf{Supp}}(\sigma_{1})$,
contains the maximum number of circles,
which is $2$.
By Lemma~\ref{basic property of mixed nash equilibria},
it follows that
${\mathsf{Supp}}(\sigma_{2})
  \subseteq
  \{ 1, 2, 3 \}$.
Observe that the only square profiles
supported in ${\bm{\sigma}}$
whenever 
${\mathsf{Supp}}\left( \sigma_{2}
                \right)
 =
 \{ 1, 2, 3 \}$
are 
$\langle 4, 1 \rangle$,
$\langle 3, 2 \rangle$;
and
$\langle 1, 3 \rangle$;
so,
there is exactly one square profile
for each possible best-response
of player 2
to $\sigma_{1}$.
This implies that
\textcolor{black}{${\widehat{{\mathsf{U}}}}_{1}\left( \bm{\sigma}
                 \right)
 =
 \frac{\textstyle 1}
         {\textstyle 3}$,}
regardless of ${\mathsf{Supp}}(\sigma_{2})$.          
If $|{\mathsf{Supp}}(\sigma_{2})|
     <
     3$,
then
for each strategy
$s \in {\mathsf{Supp}}(\sigma_{1})$,
\textcolor{black}{${\widehat{{\mathsf{U}}}}_{1}\left( \bm{\sigma}_{-1}
                        \diamond
                        s
               \right)
 \geq
 \frac{\textstyle 1}
         {\textstyle 2}$,}
a contradiction to Lemma~\ref{basic property of mixed nash equilibria}
(Condition {\sf (1)}).              
So,
${\mathsf{Supp}}\left( \sigma_{2}
                             \right)
      =
      \{ 1, 2, 3 \}$.
But then 
\textcolor{black}{${\widehat{{\mathsf{U}}}}_{1}\left( \bm{\sigma}_{-1}
                        \diamond
                        2
               \right)
 =
 \frac{\textstyle 2}
         {\textstyle 3}$,}
a contradiction to Lemma~\ref{basic property of mixed nash equilibria}
(Condition {\sf (1)}).

\item
\underline{$s = 3$
           so that
           ${\mathsf{Supp}}\left( \sigma_{1}
                           \right)
            =
            \{ 1, 2, 4 \}$:}
Since $\bm{\sigma}$ is uniform,
strategies in $\{ 1, 2, 4 \}$
are the only three
best-responses for player $2$
to $\sigma_{1}$ 
since each of their corresponding columns,
restricted to rows corresponding
to the strategies in
${\mathsf{Supp}}(\sigma_{1})$,
contains the maximum number of circles,
which is $2$.
By Lemma~\ref{basic property of mixed nash equilibria},
it follows that
${\mathsf{Supp}}(\sigma_{2})
  \subseteq
  \{ 1, 2, 4 \}$.
Observe that the only square profiles
enabled in ${\bm{\sigma}}$
whenever
${\mathsf{Supp}}\left( \sigma_{2}
                \right)
 =
 \{ 1, 2, 4 \}$
are 
$\langle 4, 1 \rangle$,
$\langle 2, 2 \rangle$
and
$\langle 1, 4\rangle$;
so,
there is exactly one square profile
for each possible best-response
of player 2
to $\sigma_{1}$. 
This implies that
\textcolor{black}{${\widehat{{\mathsf{U}}}}_{1}\left( \bm{\sigma}
                 \right)
 =
 \frac{\textstyle 1}
         {\textstyle 3}$,}
regardless of ${\mathsf{Supp}}(\sigma_{2})$.          
If $|{\mathsf{Supp}}(\sigma_{2})|
     <
     3$,
then
for each strategy
$s \in {\mathsf{Supp}}(\sigma_{1})$,
\textcolor{black}{${\widehat{{\mathsf{U}}}}_{1}\left( \bm{\sigma}_{-1}
                        \diamond
                        s
               \right)
 \geq
 \frac{\textstyle 1}
         {\textstyle 2}$,}
a contradiction to Lemma~\ref{basic property of mixed nash equilibria}
(Condition {\sf (1)}).              
So,
${\mathsf{Supp}}\left( \sigma_{2}
                             \right)
      =
      \{ 1, 2, 4 \}$.
But then 
\textcolor{black}{${\widehat{{\mathsf{U}}}}_{1}\left( \bm{\sigma}_{-1}
                        \diamond
                        3
               \right)
 =
 \frac{\textstyle 2}
         {\textstyle 3}$,}
a contradiction to Lemma~\ref{basic property of mixed nash equilibria}
(Condition {\sf (1)}).

\item
\underline{$s = 4$
           so that
           ${\mathsf{Supp}}\left( \sigma_{1}
                           \right)
            =
            \{ 1, 2, 3 \}$:}
Then,
all profiles supported in
$\bm{\sigma}_{-2}
 \diamond
 1$
are circles.
A contradiction
to Lemma~\ref{barrage structure}.

\end{enumerate}
The proof is now complete.
\end{proof}
\noindent
The claim follows now from
Lemmas~\ref{barrage 1},
\ref{barrage 5},
\ref{barrage 2}
and~\ref{barrage 3}.
\end{proof}

\noindent
By Proposition~\ref{barrage},
the win-lose game
${\widehat{{\mathsf{G}}}}_{3}$
is a negative instance for
{\sf $\exists$ UNIFORM NASH}.

\subsection{The Non-Symmetric Game $\widehat{{\mathsf{G}}}_{4}$}
\label{nonsymmetric game}

\noindent
Define the win-lose bimatrix game
\textcolor{black}{$\widehat{\mathsf{G}}_{4}
 =
 \left\langle [2],
              \{ {\widehat{{\mathsf{\Sigma}}}}_{i}
              \}_{i \in [2]},
              \{ {\widehat{{\mathsf{U}}}}_{i}
              \}_{i \in [2]}
 \right\rangle$}
with
\textcolor{black}{${\widehat{{\mathsf{\Sigma}}}}_{1}
 =
 {\widehat{{\mathsf{\Sigma}}}}_{2}
 =
 \{ 1, 2, 3
 \}$;}
the utility functions are depicted below:
\begin{center}
\begin{small}
\begin{tabular}{||c|c||c|c||c|c||}
\hline
  {\sf Profile} ${\bf s}$  & $\widehat{{\sf U}}({\bf s})$
& {\sf Profile} ${\bf s}$  & $\widehat{{\sf U}}({\bf s})$
& {\sf Profile} ${\bf s}$  & $\widehat{{\sf U}}({\bf s})$ \\
\hline
\hline
$\langle 1, 1
 \rangle$        & $\langle 1, 0
                           \rangle$        &  $\langle 1, 2
                                                      \rangle$        & $\langle 0, 1
                                                                                \rangle$        &  $\langle 1, 3
                                                                                                           \rangle$       & $\langle 1, 0
                                                                                                                                     \rangle$                                                                \\
\hline
$\langle 2, 1
  \rangle$      & $\langle 0, 1
                           \rangle$        &  $\langle 2, 2
                                                       \rangle$       & $\langle 0, 0
                                                                                \rangle$        &  $\langle 2, 3
                                                                                                           \rangle$       & $\langle 1, 0
                                                                                                                                    \rangle$                                                               \\
\hline
$\langle 3, 1
 \rangle$      & $\langle 0, 1
                         \rangle$         &  $\langle 3, 2
                                                     \rangle$         & $\langle 1, 0
                                                                                 \rangle$       &  $\langle 3, 3
                                                                                                           \rangle$       & $\langle 0, 1
                                                                                                                                     \rangle$                                                               \\
\hline
\hline
\end{tabular}
\end{small}
\end{center}
\noindent
Clearly,
$\widehat{\mathsf{G}}_{4}$
has the positive utility property.
We prove:

\begin{proposition}
\label{gadget5}
${\widehat{{\mathsf{G}}}}_{4}$
has no symmetric Nash equilibrium.
\end{proposition}

\begin{proof}
Assume, 
by way of contradiction, 
that $\widehat{\mathsf{G}}_{4}$ has
a symmetric Nash equilibrium $\bm\sigma$.
Clearly,
${\widehat{{\mathsf{G}}}}_{4}$
has no pure Nash equilibrium.
Hence,
four cases \textcolor{black}{need to be examined}:
\begin{itemize}

\item[{\sf (1)}]
\underline{For each player $i \in [2]$,
                 ${\mathsf{Supp}}(\sigma_{i})
                  =
                  \{ 1, 2 \}$:}
Then,
${\widehat{{\mathsf{U}}}}_{1}({\bm{\sigma}}_{-1}
                                                     \diamond 
                                                     2)
  =
  0$.
Since 
$2 \in {\mathsf{Supp}}(\sigma_{1})$, 
Lemma~\ref{basic property of mixed nash equilibria}
(Condition (1)) 
implies that
${\widehat{{\mathsf{U}}}}_{1}({\bm{\sigma}})
  =
  0$.
A contradiction to Lemma~\ref{park kafe}.

\item[{\sf (2)}]
\underline{For each player $i \in [2]$,
                 ${\mathsf{Supp}}(\sigma_{i})
                  =
                  \{ 2, 3 \}$:}
Then, 
${\widehat{{\mathsf{U}}}}_{2}({\bm{\sigma}}_{-2}
                                                     \diamond
                                                     2)
  =0$.
Since 
$2 \in {\mathsf{Supp}}(\sigma_{2})$,
Lemma~\ref{basic property of mixed nash equilibria}
(Condition (1)) 
implies that
${\widehat{{\mathsf{U}}}}_{2}({\bm{\sigma}})
  =
  0$.
A contradiction to Lemma~\ref{park kafe}.

\item[{\sf (3)}]
\underline{For each player $i \in [2]$,
                 ${\mathsf{Supp}}(\sigma_{i})
                  =
                  \{ 1, 3 \}$:}
Then,
${\widehat{{\mathsf{U}}}}_{1}({\bm{\sigma}}_{-1}
                                                     \diamond
                                                     3)
  = 0$.
Since
$3 \in {\mathsf{Supp}}(\sigma_{1})$, 
Lemma~\ref{basic property of mixed nash equilibria}
(Condition (1))
implies that
${\widehat{{\mathsf{U}}}}_{1}({\bm{\sigma}})
  =
  0$.
A contradiction to Lemma~\ref{park kafe}.

\item[{\sf (4)}]
\underline{For each player $i \in [2]$,
                 ${\mathsf{Supp}}(\sigma_{i})
                  =
                  \{ 1, 2, 3 \}$:}
Then,
${\widehat{{\mathsf{U}}}}_{1}({\bm{\sigma}}_{-1}\diamond 1)
  =
  \sigma_{2}(1) + \sigma_{2}(3)$
and
${\widehat{{\mathsf{U}}}}_{1}({\bm\sigma}_{-1}\diamond 2)
  =
  \sigma_{2}(3)$.
Since $\sigma_{2}(1) > 0$,
it follows that
${\widehat{{\mathsf{U}}}}_{1}({\bm{\sigma}}_{-1} \diamond 1)
  >
  {\widehat{{\mathsf{U}}}}_{1}({\bm{\sigma}}_{-1} \diamond 2)$.
Since
$1, 2 \in {\mathsf{Supp}}(\sigma_{1})$, 
Lemma~\ref{basic property of mixed nash equilibria} (Condition (1))
implies that
${\widehat{{\mathsf{U}}}}_{1}({\bm{\sigma}}_{-1} \diamond 1)
  =
  {\widehat{{\mathsf{U}}}}_{1}({\bm{\sigma}}_{-1} \diamond 2)$.
A contradiction.

\end{itemize}
The proof is complete.
\end{proof}

\noindent
\textcolor{black}{By Proposition~\ref{gadget5},
the win-lose game
${\widehat{{\mathsf{G}}}}_{4}$
is a negative instance for
{\sf $\exists$ SYMMETRIC NASH}.}

\subsection{The Diagonal Game ${\widehat{{\mathsf{G}}}}_{5}[k]$}
\label{diagonal game}

\noindent
Fix an integer $k \geq 1$.
Denote as ${\mathsf{D}}_{k}$ 
the $k \times k$ win-lose matrix
such that
${\mathsf{D}}_{k}[i,j]=1$ 
if and only if
$i \leq j$,
with $1 \leq i, j \leq k$;
so,
${\mathsf{D}}_k$ is an upper-diagonal matrix. 
Then,
${\mathsf{D}}_k^{{\rm T}}$ 
is a $k \times k$ win-lose matrix
such that 
${\mathsf{R}}_{k}^{{\rm T}}[i, j] =1$
if and only if 
$j \leq i$,
with $1 \leq i, j \leq k$;
so,
${\mathsf{D}}_{k}^{{\rm T}}$ 
is a lower-diagonal matrix.
Define the win-lose bimatrix game
${\widehat{{\mathsf{G}}}}_{5}[k]
  :=
  \left\langle{\mathsf{D}}_{k},
                   {\mathsf{D}}_{k}^{{\rm T}}
  \right\rangle$,
which is symmetric by construction.  
Clearly,
${\widehat{{\mathsf{G}}}}_{5}[k]$
has the positive utility property.
We prove:

\begin{proposition}
\label{very recent}
Fix an integer $k \geq 1$. 
Then,
${\widehat{{\mathsf{G}}}}_{5}[k]$
has exactly $k$ Nash equilibria,
\textcolor{black}{which are Pareto-Optimal,
Strongly Pareto-Optimal
and symmetric,
each yielding utility $1$
to each player.}
\end{proposition}

\begin{proof}
Note that
${\widehat{{\mathsf{G}}}}_{5}[k]$
has exactly $k$ pure Nash equilibria ${\bf s }^{s}$,
with $s \in [k]$,
where each player 
$i \in [2]$
chooses strategy $s$
in ${\bf s}^{s}$
to get utility $1$. 
Note that each of the $k$ pure equilibria
is Pareto-Optimal,
Strongly Pareto-Optimal
and symmetric. 
Hence,
it remains to prove that
${\widehat{{\mathsf{G}}}}_{5}[k]$
has no mixed Nash equilibrium.
Towards this end,
assume,
by way of contradiction,
that ${\widehat{{\mathsf{G}}}}_{5}[k]$
has a mixed Nash equilibrium ${\bm{\sigma}}$. 
For each player $i \in [2]$, 
set 
$s^{\ast}_{i}
  :=
  \min {\sf Supp}(\sigma_i)$;
so,
$s^{\ast}_{1}$
(resp., $s^{\ast}_{2}$)
is the row (resp., column) with least index
played by the row (resp., column) player
in ${\bm{\sigma}}$.
By the definitions of ${\mathsf{D}}_{k}$ 
and ${\mathsf{D}}^{{\rm T}}_{k}$, 
it follows that the set of best-response strategies 
for the row (resp., column) player
is $[s^{\ast}_{2}]$
(resp., $[s^{\ast}_{1}]$);
thus,
by  Lemma~\ref{basic property of mixed nash equilibria},
${\mathsf{Supp}}(\sigma_{1})
  \subseteq 
  [s^{\ast}_{2}]$
and
${\mathsf{Supp}}(\sigma_{2})
  \subseteq
  [s^{\ast}_{1}]$. 
These imply together that
$s^{\ast}_{1}
  \leq
  s^{\ast}_{2}$ 
and 
$s^{\ast}_2
  \leq
  s^{\ast}_{1}$,
so that
$s^{\ast}_{1}
  =
  s^{\ast}_{2}$. 
 Hence, 
 for each player $i \in [2]$,
 ${\sf Supp}(\sigma_{i})
   \subseteq
   [s^{\ast}_{i}]
   =
   [\min {\mathsf{Supp}}(\sigma_{i})]$,
which implies that
${\mathsf{Supp}}(\sigma_{i})
  =
  \{ s^{\ast}_{i} \}$.
Hence,
${\bm{\sigma}}$ is a pure Nash equilibrium.
A contradiction.
\end{proof}

\noindent
By Proposition~\ref{very recent},
the win-lose 
game
${\widehat{{\mathsf{G}}}}_{5}[k]$
is a negative instance for
{\sf $\exists$ $k+1$ NASH},
with $k \geq 1$,
for {\sf $\exists$ $\neg$ PARETO-OPTIMAL NASH},
for
{\sf $\exists$ $\neg$ STRONGLY PARETO-OPTIMAL NASH}
and for
{\sf $\exists$ $\neg$ SYMMETRIC NASH}.

\section{The Win-Lose Reduction}
\label{winlose reduction}

\noindent
We start with
some preliminary material.
Denote
$I_{n}
 :=
 \{0, 1, \ldots, n-1\}$;
for all arithmetic operations on $I_{n}$,
we shall 
assume a cyclic ordering
on its elements;
so,
$(n-1) + 1 = 0$
and for a given pair
$i, j  \in I_{n}$,
$j-i$
is the number of elements of $I_{n}$
encountered when moving from $j$ (excluded)
to $i$ (included)
in the anticlockwise direction.
\textcolor{black}{A {\sf CNF SAT} formula is}
a boolean formula
${\sf \phi}
 =
 \bigwedge_{i \in [k]}
   \bigvee_{j \in [l_{i}]}
      \ell_{ij}$
in {\it Conjunctive Normal Form,}
where $\ell_{ij}$
is a {\it literal}
(a boolean {\it variable}
or its negation),
we shall denote as
{\it (i)}
${\mathsf{\cal{C}}}
 =
 {\mathsf{\cal{C}}}({\mathsf{\phi}})
 =
 \{ \bigvee_{j \in [l_i]}
      \ell_{ij}
    \mid
    i \in [k]
 \}$,
the set of {\it clauses},
and
{\it (ii)}
${\mathsf{Var}}
 =
 {\mathsf{Var}}({\mathsf{\phi}})
 =
 \{ v_{0}, \ldots, v_{n-1}
 \}$
and
${\mathsf{L}}
 =
 {\mathsf{L}}({\mathsf{\phi}})
 =
 \{ \ell_{0},\overline{\ell}_{0},
    \ldots,
    \ell_{n-1},\overline{\ell}_{n-1}
 \}$
the sets of variables and
literals,
respectively,
with $|{\sf Var}({\sf \phi})|
      =
      n$
and $|{\sf L}({\sf \phi})|
      =
      2n$.
We shall use $\mathsf{c}$
to denote a clause 
from ${\mathsf{\cal{C}}}({\sf\phi})$,
and either $\ell$ or $\overline\ell$ 
to denote a literal 
from ${\mathsf{L}}({\mathsf{\phi}})$.
${\mathsf{\phi}}$ is a {\sf 3SAT} formula
if each clause ${\mathsf{c}}$
has $|{\mathsf{c}}| = 3$.
An {\it assignment}
is a function 
${\mathsf{\gamma}}:
 {\sf Var}({\sf \phi})
 \rightarrow
 \{ 0, 1 \}$;
so,
${\mathsf{\gamma}}$
is represented
by the $n$-tuple of literals
made true.
For a literal
$\ell \in {\mathsf{L}}({\mathsf{\phi}})$,
denote as $I (\ell)$
the index $j \in I_{n}$
such that
$\ell
  \in
  \{ \ell_{j}, \overline{\ell}_{j} \}$.
Denote as ${\mathsf{\phi}}({\mathsf{\gamma}})$
the value of ${\mathsf{\phi}}$
when each variable $v$
takes the value ${\mathsf{\gamma}}(v)$.
${\mathsf{\gamma}}$ is a {\it satisfying assignment}
if ${\mathsf{\phi}}({\mathsf{\gamma}}) = 1$,
and ${\mathsf{\phi}}$ is satisfiable
if it has a satisfying assignment;
deciding the satisfiability of a 
\textcolor{black}{{\sf 3SAT} formula}
${\mathsf{\phi}}$
is ${\mathcal{NP}}$-complete.
Denote as $\# {\mathsf{\phi}}$
\textcolor{black}{and $\oplus {\mathsf{\phi}}$}
the number \textcolor{black}{and the parity of the number}
of satisfying assignments of ${\mathsf{\phi}}$,
\textcolor{black}{respectively;}
\textcolor{black}{when ${\mathsf{\phi}}$
is a {\sf 3SAT} formula,
computing $\# {\mathsf{\phi}}$
is $\# {\mathcal{P}}$-complete~\cite{V79}
and
computing $\oplus {\mathsf{\phi}}$
is $\oplus {\mathcal{P}}$-complete~\cite{PZ83}.}

\subsection{\textcolor{black}{The Reduction}}
\label{reduction definition}

\textcolor{black}{
We shall construct,
given
an $r$-player win-lose gadget game
${\widehat{{\mathsf{G}}}}$,
for an arbitrary integer $r \geq 2$,
and an instance ${\mathsf{\phi}}$ of {\sf 3SAT}
with $n=|{\sf Var}({\sf\phi})|\geq 4$,
the $r$-player win-lose game
${\mathsf{G}}
 =
 {\mathsf{G}}({\widehat{{\mathsf{G}}}},
                        {\sf \phi})
 = \left\langle [r],
                \left\{ {\mathsf{\Sigma}}_{i}
                \right\}_{i \in [r]},
                \left\{ {\mathsf{U}}_{i}
                \right\}_{i \in [r]}
   \right\rangle$.
Players $1$ and $2$ are called
{\it special}.}
\textcolor{black}{For a player $i \in [2]$,
$\overline{i}
 \in
 [2] \setminus \{ i \}$
denotes the special player other than $i$.} 

\subsubsection{\textcolor{black}{Inadequacy of the Reduction in~\cite{CS08} and Getting Around it}}
\label{conitzer sandholm inadequacy}

\textcolor{black}{A crucial step 
in both the win-lose
reduction 
and the reduction 
in~\cite{CS08}
is to guarantee two significant properties:}
\begin{itemize}

\item[\textcolor{black}{{\sf (P.1)}}]
\textcolor{black}{A mixed profile in which 
some special player 
chooses some literals with positive probability 
is a Nash equilibrium 
for the constructed game ${\mathsf{G}}$ 
if and only if 
$\phi$ is satisfiable.}

\item[\textcolor{black}{{\sf (P.2)}}]    
\textcolor{black}{
Each satisfying assignment
induces a Nash equilibrium
where both special players 
are choosing the literals
from the satisfying assignment 
with uniform probabilities.}

\end{itemize}
\noindent       
\textcolor{black}{To guarantee these properties, 
both the win-lose reduction and the reduction in~\cite{CS08} 
have to rule out 
the existence of a Nash equilibrium 
in which either
some special player plays literals
with non-uniform probabilities 
or some pair of complementary literals
($\ell$ and $\overline{\ell}$) 
\textcolor{black}{is not played by some special player}.} 
\textcolor{black}{The reduction
in~\cite{CS08}
achieves these properties
with the following approach:}
\begin{itemize}

\item
\textcolor{black}{By assigning a non-zero utility
to both special players 
when they are choosing
a pair of non-complementary literals.}

\item
\textcolor{black}{By using $n$ auxiliary 
\textcolor{black}{strategies,
called {\it variables,}}
${\mathsf{v}}_{0},
  \ldots,
  {\mathsf{v}}_{n-1}$, 
where, for each index
$j \in I_n$, 
${\mathsf{v}}_j$ 
is associated with 
the pair of complementary literals 
$\{ \ell_j,
     \overline{\ell}_j
  \}$,
as follows:  
When a special player 
$i \in [2]$ 
chooses a literal 
$\ell \in \{\ell_{j}, \overline{\ell}_{j}
             \}$, 
the utility given to player 
$\overline{i}$ 
when choosing 
\textcolor{black}{a variable other than ${\mathsf{v}}_j$} 
is more than the utility
she would get
when choosing any literal 
from ${\sf L}(\phi)$.}

\end{itemize}
\textcolor{black}{Thus,
this approach requires using
at least two different non-zero utilities.
Hence,
the reduction in~\cite{CS08}
is inadequate for win-lose games. 
We resolve this inadequacy
by combining two new ideas
in the win-lose reduction:}

\begin{itemize}

\item
\textcolor{black}{Instead
of giving non-zero utility
to both special players  
when they are choosing
non-complementary literals,
only one of them gets non-zero utility.
This idea is implemented by using
a cyclic ordering among the literals
and extending the construction of
the cyclic game 
${\widehat{{\sf G}}}_{1}[n]$. 
(See Cases {\sf (1)}
through {\sf (4)}
in Figure \ref{table reduction}.) 
We shall prove that this suffices
to rule out the existence of Nash equilibria
in which literals in the support of both special players
are played with non-uniform probabilities.}

\item
\textcolor{black}{To guarantee that 
every pair of complementary literals are chosen
with strictly positive probability, 
we introduce a new kind of auxiliary strategies, 
\textcolor{black}{namely the {\it pair variables}}
${\sf v}_{j,k}$,
for each pair $j, k\in I_n$
with $j\neq k$;
thus,
there are $n(n-1)$ of them.
The \textcolor{black}{pair variables}
provide the special player $\overline{i}$ 
with an improving deviation 
whenever there is a pair of complementary literals
not chosen by the special player $i$;
at the same time,
we define the utility functions
so that
no special player could improve
by switching to a \textcolor{black}{pair variable}
when both special players
are playing the same satisfying assignment
with uniform probabilities. 
(See Cases {\sf (5/a)}
and {\sf (5/b)}
in Figure \ref{table reduction}.)} 
\textcolor{black}{This is the point where the structure of a {\sf 3SAT} formula
shall be needed.}

\end{itemize}
\noindent
\textcolor{black}{The proof of Lemma~\ref{same probability on pairs of literals}
demonstrates how the two new ideas 
combine together.
We now proceed to formally define
the win-lose reduction.}

\subsubsection{\textcolor{black}{Formal Definition}}

\noindent
For each player $i \in [2]$,
${\sf \Sigma}_{i} := \widehat{\mathsf{\Sigma}}_{i} 
                     \cup {\mathsf{L}}({\mathsf{\phi}})
                     \cup
                     {\mathsf{\cal{C}}}({\mathsf{\phi}})
                     \cup
                     {\mathsf{V}}({\mathsf{\phi}})$;
for each player $i \in [r] \setminus [2]$,
${\mathsf{\Sigma}}_{i} := \widehat{\mathsf{\Sigma}}_{i}
                          \cup
                          \{ {\mathsf{\delta}} 
                          \}$,
where 
${\mathsf{V}}({\mathsf{\phi}})
 =
 \{ {\mathsf{v}}_{i,j}
    \mid
    0 \leq i,j < n,
    i \neq j
 \}$
and ${\mathsf{\delta}}$
is a {\it special} strategy.
\textcolor{black}{The strategies in
${\mathsf{V}}({\mathsf{\phi}})$
are called {\it pair variables}.}
\textcolor{black}{The two special players 
$1$ and $2$ are 
the only players
whose sets of strategies are determined
by the formula ${\mathsf{\phi}}$.}
Note that for each player $i \in [2]$,
$|{\mathsf{\Sigma}}_{i}|
 =
 |\widehat{{\mathsf{\Sigma}}}_{i}|
 +
 2n
 +
 |{\mathcal{C}}({\mathsf{\phi}})|
 +
 n(n+1)$,
and for each player $i \in [r] \setminus [2]$,
$|{\mathsf{\Sigma}}_{i}|
  =
  |\widehat{{\mathsf{\Sigma}}}_{i}|
  +
  1$;
so,
the size of the strategy sets in 
${\mathsf{G}}({\widehat{{\mathsf{G}}}},
                        {\mathsf{\phi}})$
is polynomial in the sizes of ${\widehat{{\mathsf{G}}}}$
and ${\mathsf{\phi}}$. 
In the following,
we shall write
${\mathsf{L}}$,
${\mathcal{C}}$
and
${\mathsf{V}}$
for
${\mathsf{L}}({\mathsf{\phi}})$,
${\mathsf{\cal{C}}}({\mathsf{\phi}})$,
and
${\mathsf{V}}({\mathsf{\phi}})$,
respectively.

\noindent
Fix a profile
${\bf s}
 =
 \langle s_{1}, \ldots, s_{r}
 \rangle$
from
${\mathsf{\Sigma}} = {\mathsf{\Sigma}}_{1}
                     \times
                     \ldots
                     \times
                     {\mathsf{\Sigma}}_{r}$.
Use ${\bf s}$
to partition $[r]$ into
$\widehat{\cal P}({\bf s})
 =
 \{ i \in [r]
    \mid
    s_{i} \in {\widehat{{\mathsf{\Sigma}}}}_{i}
 \}$
and
${\cal P}({\bf s})
 =
 \{ i \in [r]
    \mid
    s_{i} \not\in {\widehat{{\mathsf{\Sigma}}}}_{i}
 \}$;
so,
$\widehat{\cal P}({\bf s})$
and
${\cal P}({\bf s})$
are the sets of players
choosing and not choosing,
respectively,
strategies from ${\widehat{{\mathsf{G}}}}$.
The utility functions
are depicted
in Figure~\ref{table reduction}.

\begin{figure}[ht]
\begin{center}
\begin{tabular}{||l|l|l||}
\hline
\hline
{\sf Case}                                                                                                     &
{\sf Condition on the
     profile ${\bf s}
              =
              \langle s_{1},
                      s_{2},\ldots,s_{r}
              \rangle$}                                                                                        &
{\sf Utility vector} ${\mathsf{U}}({\bf s})$                                                                   \\
\hline
\hline
{\sf (1)}                                                                                                      &
${\bf s}
 =
 \langle \ell,
         \overline{\ell},
         {\mathsf{\delta}},
         \ldots,
         {\mathsf{\delta}}
 \rangle$                                                                                                      &
$\langle 0,
         0,1,\ldots,1
 \rangle$                                                                                                      \\
\hline
{\sf (2)}                                                                                                      &
${\bf s}
 =
 \langle \ell,
         \ell^{\prime},
         \delta,\ldots,\delta
 \rangle$
with $I(\ell^{\prime}) - I(\ell) 
          \in \{ 0, 1 \}$
and $\ell \neq \overline{\ell}$                                                                    &
$\langle 1,0,1,\ldots,1
 \rangle$                                                                                                      \\
\hline
{\sf (3)}                                                                                                      &
${\bf s}
 =
 \langle \ell,
             \ell^{\prime},
         {\mathsf{\delta}},
         \ldots,
         {\mathsf{\delta}}
 \rangle$
with $I(\ell^{\prime})
          -
          I(\ell)
          \in
          \{ 2, 3 \}$                                                                                        &
$\langle 0,1,1,\ldots,1
 \rangle$                                                                                                      \\
\hline
{\sf (4)}                                                                                                      &
${\bf s}
 =
 \langle \ell,
             \ell^{\prime},
             {\mathsf{\delta}},
             \ldots,
             {\mathsf{\delta}}
 \rangle$
with $I(\ell^{\prime}) 
          -
          I (\ell)
          > 3$                                                                                                  &
$\langle 0,0,1,\ldots,1
 \rangle$                                                                                                      \\
\hline
\hline
{\sf (5/a)}                                                                                                      &
${\bf s}
 =
 \langle {\sf v}_{ij},
         \ell_k,\delta,\ldots,\delta
 \rangle$
with $k\in\{i,j\}$                                                                                             &
$\langle 1,0,1,\ldots,1
 \rangle$                                                                                                      \\
\hline
\hline
{\sf (5/b)}                                                                                                      &
${\bf s}
 =
 \langle \ell_{k},
             {\sf v}_{ij},
             \delta,\ldots,\delta
 \rangle$
with $k \in \{i,j\}$                                                                                         &
$\langle 0,1,1,\ldots,1
 \rangle$                                                                                                      \\
\hline
\hline
{\sf (6/a)}                                                                                                     &
${\bf s}
 =
 \langle {\sf c},
            \ell,
         {\mathsf{\delta}},
         \ldots,
         {\mathsf{\delta}}
 \rangle$
with
$\overline{\ell}\in {\mathsf{c}}$                                                                              &
$\langle 1, 0, 1,\ldots,1
 \rangle$                                                                                                      \\
\hline
\hline
{\sf (6/b)}                                                                                                     &
${\bf s}
 =
 \langle \ell,
            {\sf c},
         {\mathsf{\delta}},
         \ldots,
         {\mathsf{\delta}}
 \rangle$
with
$\overline{\ell}\in {\mathsf{c}}$                                                                              &
$\langle 0, 1,1,\ldots,1
 \rangle$                                                                                                      \\
\hline
\hline
{\sf (7)}                                                                                                      &
For each $i \in [r]$, 
$s_{i} \in {\widehat{{\mathsf{\Sigma}}}}_{i}$.
(So, $|\widehat{{\cal P}}({\bf s})|
      =
      r$.)                                                                                                      &
${\widehat{{\mathsf{U}}}}(\langle s_1,\ldots,s_r
                          \rangle)$                                                                            \\
\hline
\hline
{\sf (8)}                                                                                                      &
$|{\cal \widehat{P}}({\bf s})|>1$ 
and 
$|{\cal P}({\bf s})|>0$                                                                                        &
${\sf U}_i({\bf s})=1$ 
iff $i\in{\cal \widehat{P}}({\bf s})$                                                                          \\
\hline
{\sf (9)}                                                                                                      &
${\cal \widehat{P}}({\bf s})
 =
 \{ i \}$ 
with $i\in [r]\setminus [2]$                                                                                   &
${\sf U}_{j}({\bf s})
 =
 {\mathsf{\delta}}_{ji}$                                                                                       \\
\hline
{\sf (10)}                                                                                                     &
${\cal \widehat{P}}({\bf s})=\{i\}$ with $i\in [2]$                                                            
\&
$s_{[2]\setminus\{i\}}
 \in
 \{ \ell_0,\overline{\ell}_0,\ell_1,\overline{\ell}_1\}
 \cup
 {\mathcal{C}}
 \cup
 {\mathsf{V}}$                                                                                                 &
${\mathsf{U}}_{j}({\bf s})
 =
 {\mathsf{\delta}}_{ji}$                                                                                       \\
\hline
\hline
{\sf (11)}                                                                                                     &
None of the above.                                                                                               &
$\langle 0,\ldots, 0 \rangle$                                                                                  \\
\hline
\hline
\end{tabular}
\caption{The utility functions for the game ${\mathsf{G}} = {\mathsf{G}}({\widehat{{\mathsf{G}}}},
                                                                                                                       {\mathsf{\phi}}
                                                                                                                      )$.
              ($\mathsf{\delta}_{ij}$ denotes the {\it Kronecker delta}.)
             }
\label{table reduction}
\end{center}
\end{figure}

\subsubsection{\textcolor{black}{Notation}}

\noindent
For a mixed profile $\bm\sigma$
and a special player $i \in [2]$, 
$\sigma_{i}(S)$
denotes the total probability
assigned by $\sigma_{i}$
to strategies
from $S \subseteq {\mathsf{\Sigma}}_{i}$; 
so,
in particular,
$\sigma_{i}({\mathsf{L}})$,
$\sigma_{i}({\cal C}
            \cup
            {\mathsf{V}})$ 
and $\sigma_{i}( \widehat{\mathsf{\Sigma}}_{i}
                         )$
are the total probabilities 
assigned by $\sigma_{i}$ 
to strategies from
${\mathsf{L}}$, 
${\mathcal{C}}
 \cup
 {\mathsf{V}}$
and 
$\widehat{\mathsf{\Sigma}}_i$, 
respectively,
with
$\sigma_{i}({\mathsf{L}})
 +
 \sigma_{i}({\mathcal{C}}
            \cup
            {\mathsf{V}})
 +
 \sigma_{i}(\widehat{\mathsf{\Sigma}}_{i})
 = 1$.
For a player 
$i \in [r] \setminus [2]$,
$\sigma_{i} ( {\widehat{{\mathsf{\Sigma}}}}_{i}
                   )$
and $\sigma_{i}({\mathsf{\delta}})$,
with
$\sigma_{i} ( {\widehat{{\mathsf{\Sigma}}}}_{i}
                   )
  +                  
  \sigma_{i}({\mathsf{\delta}})
  = 1$,
are set correspondingly. 
Sometimes
a strategy from ${\widehat{{\mathsf{G}}}}$
will be called a {\it gadget strategy}. 

\subsection{Properties of Nash Equilibria for ${\mathsf{G}}$}
\label{properties of the game}

\noindent
We now present
a collection of properties
for a Nash equilibrium
$\bm{\sigma}
 \in
 {\mathcal{NE}}({\mathsf{G}})$.
The first property is that
either all players 
are exclusively playing
gadget strategies 
or none \textcolor{black}{is}.

\begin{lemma}\label{first property of best-responses}
Fix a Nash equilibrium
$\bm{\sigma}
 \in
 {\mathcal{NE}}({\mathsf{G}})$
with 
${\mathsf{Supp}}\left( \sigma_{i}
               \right)
 \subseteq
 {\widehat{\mathsf{\Sigma}}}_{i}$
for some player $i \in [r]$.
Then,
for each player
$j \in [r] \setminus \{ i \}$,
${\mathsf{Supp}}\left( \sigma_{j}
               \right)
 \subseteq
 {\widehat{\mathsf{\Sigma}}}_{j}$.
\end{lemma}

\begin{proof}
Fix a player
$j \in [r] \setminus \{ i \}$
and a strategy
$s \in {\mathsf{Supp}}(\sigma_{j})$.
Since $\bm{\sigma}$
is a Nash equilibrium,
$s$ is a best-response
to $\bm{\sigma}_{-j}$.
We shall establish that
$s 
 \in
 {\widehat{{\mathsf{\Sigma}}}}_{j}$.
Consider the mixed profile
$\bm{\sigma}_{-j}
 \diamond
 s$.
Since
${\mathsf{Supp}}\left( \sigma_{i}
               \right)
 \subseteq
 {\widehat{\mathsf{\Sigma}}}_{i}$,
each profile
${\bf s}_{-j}
 \diamond
 s$
supported in 
$\bm{\sigma}_{-j}
 \diamond
 s$
may only come
from Cases
{\sf (7)},
{\sf (8)},
{\sf (9)},
{\sf (10)}
and
{\sf (11)}.
For any such profile
${\bf s}_{-j}
 \diamond
 s$,
it follows,
by the utility functions,
that if
$s \not\in \widehat{{\mathsf{\Sigma}}}_{j}$,
then
${\mathsf{U}}_{j}\left( {\bf s}_{-j}
                        \diamond
                        s            
                 \right)
 =
 0$,
which implies that
${\mathsf{U}}_{j}\left( \bm{\sigma}_{-j}
                        \diamond
                        s            
                 \right)
 =
 0$.
So,
to establish the claim,
it suffices to determine a strategy 
$s \in \widehat{{\mathsf{\Sigma}}}_{j}$                          
with
${\mathsf{U}}_{j}\left( \bm{\sigma}_{-j}
                        \diamond
                        s
                 \right)                                    
 >
 0$. 
We proceed by case analysis.
\begin{itemize}

\item
Assume first that
there is a profile
${\bf s}_{-j}$
supported in $\bm{\sigma}_{-j}$
with
${\bf s}_{-j}
 \in
 \widehat{{\mathsf{\Sigma}}}_{-j}$;
hence,
by the positive utility property
of $\widehat{{\mathsf{G}}}$,
there is a strategy
$t = t({\bf s}_{-j})
 \in
 {\widehat{{\mathsf{\Sigma}}}}_{j}$
with
$\widehat{{\mathsf{U}}}_{j}\left( {\bf s}_{-j}
                                  \diamond
                                  t
                           \right)
 >
 0$.
Set
$s := t$. 
Then,
{
\small
\begin{eqnarray*}                      
         \lefteqn{{\mathsf{U}}_{j}\left( \bm{\sigma}_{-j}
                                                            \diamond
                                                             s
                         \right)}                                                                                                                                                   \\      
\geq & {\mathbb P}_{{\bf s}_{-j}
                                   \sim
                                   \bm{\sigma}_{-j}
                                  }
                  \left( {\bf s}_{-j}
                         \diamond
                         s 
                  \right)
           \cdot
           {\mathsf{U}}_{j}
                  \left( {\bf s}_{-j}
                         \diamond
                         s 
                  \right)      
       &  \mbox{(since ${\mathsf{G}}$ is win-lose)}                                                                                                       \\
=    & \underbrace{{\mathbb P}_{{\bf s}_{-j}
                                  \sim
                                  \bm{\sigma}_{-j}
                                 }
                              \left( {\bf s}_{-j}
                                     \diamond
                                     s 
                              \right)}_{>0}
         \cdot
         \underbrace{\widehat{{\mathsf{U}}}_{j}
                                       \left( {\bf s}_{-j}
                                              \diamond
                                              s 
                                       \right)
                    }_{> 0}
     & \mbox{(by the utility functions)}                                                                                                                  \\            
>  & 0\, .
    &                                   
\end{eqnarray*}
}

\item
Assume now that
there is no profile
${\bf s}_{-j}$
supported in $\bm{\sigma}_{-j}$
with
${\bf s}_{-j}
 \in
 \widehat{{\mathsf{\Sigma}}}_{-j}$.
Set $s$ to any strategy
from
$\widehat{{\mathsf{\Sigma}}}_{j}$.
Thus,
$\{ i, j \}
 \subseteq
 \widehat{{\cal P}}\left( {\bf s}_{-j}
                          \diamond
                          s
                   \right)$
for each profile
${\bf s}_{-j}
 \diamond
 s$
supported in
$\bm{\sigma}_{-j}
 \diamond
 s$;                  
so,
by the utility functions,
only profiles
from Case {\sf (8)}
may be supported in
$\bm{\sigma}_{-j}
 \diamond
 s$.                     
Since
$s \in \widehat{{\mathsf{\Sigma}}}_{j}$,
$j \in \widehat{{\cal P}}\left( {\bf s}_{-j}
                                \diamond
                                s
                         \right)$.
Hence,
by the utility functions,
${\mathsf{U}}_{j}
          \left( {\bf s}_{-j}
                 \diamond
                 s 
          \right)
 =
 1$,
which implies that
${\mathsf{U}}_{j}
          \left( \bm{\sigma}_{-j}
                 \diamond
                 s 
          \right)
 =
 1
 >
 0$.

\end{itemize}
The claim follows.
\end{proof}

\noindent
The next property concerns the
two special players:
if one of them is assigning probability $0$ on literals,
then the other is exclusively playing
gadget strategies.  

\begin{lemma}
\label{second property of best-responses}
Fix a Nash equilibrium
$\bm\sigma
 \in
 {\mathcal{NE}}({\mathsf{G}})$
with
${\mathsf{Supp}}(\sigma_{i})
 \subseteq
 \widehat{\mathsf{\Sigma}}_{i}
 \cup
 {\mathcal{C}}
 \cup
 {\mathsf{V}}$
for some special player 
$i \in [2]$. 
Then,
${\mathsf{Supp}}(\sigma_{\overline{i}})
 \subseteq
 {\widehat{{\mathsf{\Sigma}}}}_{\overline{i}}$.
\end{lemma}

\begin{proof}
Fix a strategy
$s 
 \in
 {\mathsf{Supp}}(\sigma_{\overline{i}})$.
Since $\bm{\sigma}$
is a Nash equilibrium,
$s$ is a best-response
to $\bm{\sigma}_{-\overline{ i}}$.
We shall establish that
$s
 \in
 {\widehat{{\mathsf{\Sigma}}}}_{\overline{i}}$.
Consider the mixed profile
$\bm{\sigma}_{- \overline{i}}
 \diamond
 s$.
Since
${\mathsf{Supp}}(\sigma_{i})
 \subseteq
 \widehat{\mathsf{\Sigma}}_{i}
 \cup
 {\mathcal{C}}
 \cup
 {\mathsf{V}}$,
or
${\mathsf{Supp}}(\sigma_{i})
 \cap
 {\mathsf{L}}
 =
 \emptyset$, 
each profile
${\bf s}_{- \overline{i}}
 \diamond
 s$
supported in
$\bm{\sigma}_{- \overline{i}}
 \diamond
 s$ 
may only come from
Cases 
{\sf (5)},
{\sf (6)},
{\sf (7)},
{\sf (8)},
{\sf (9)},
{\sf (10)}
and
{\sf (11)}.
For any such profile
${\bf s}_{- \overline{i}}
 \diamond
 s$,
it follows,
by the utility functions,
that if
$s
 \notin
 \widehat{\mathsf{\Sigma}}_{\overline{i}}$,
then
${\mathsf{U}}_{\overline{i}}
          \left( {\bf s}_{- \overline{i}}
                 \diamond
                 s  
          \right)
 =
 0$,
which implies that
${\mathsf{U}}_{\overline{i}}
          \left( \bm{\sigma}_{- \overline{i}}
                 \diamond
                 s  
          \right)
 =
 0$.
(Note that 
Cases {\sf (5)} and {\sf (6)}
occur if and only if
player $\overline{i}$
chooses a strategy
from ${\mathsf{L}}$.)
So,
to establish the claim,
it suffices to determine a strategy 
$s 
 \in
 \widehat{\mathsf{\Sigma}}_{\overline{i}}$ 
with
${\mathsf{U}}_{\overline{i}}
          \left( {\bm{\sigma}}_{- \overline{i}}
                 \diamond s
          \right)
 >
 0$.
We proceed by case analysis.

\begin{itemize}

\item
Assume that
there is a profile
${\bf s}$
supported in $\bm{\sigma}$ with
${\bf s}_{- \overline{i}}
 \in
 {\widehat{{\mathsf{\Sigma}}}}_{- \overline{i}}$.
By the positive utility property
of $\widehat{{\mathsf{G}}}$,
there is a strategy
$t = t({\bf s}_{- \overline{i}})
  \in
  {\widehat{{\mathsf{\Sigma}}}}_{{\overline{i}}}$
with
${\mathsf{U}}_{\overline{i}}
          \left( {\bf s}_{- \overline{i}}
                 \diamond
                 t  
          \right)
 >
 0$.
Set
$s
 :=
 t$.
Then,
{
\small
\begin{eqnarray*}                      
         \lefteqn{{\mathsf{U}}_{\overline{i}}\left( \bm{\sigma}_{- \overline{i}}
                                                                             \diamond
                                                                             s
                                                                    \right)}                                                                                                \\       
\geq & {\mathbb P}_{{\bf s}_{- \overline{i}}
                                    \sim
                                    \bm{\sigma}_{- \overline{i}}
                                  }
                  \left( {\bf s}_{- \overline{i}}
                           \diamond
                            s 
                  \right)
           \cdot
           {\mathsf{U}}_{\overline{i}}
                    \left( {\bf s}_{- \overline{i}}
                             \diamond
                             s 
                    \right)
       &  \mbox{(since ${\mathsf{G}}$ is win-lose)}                                                                                               \\
=    & \underbrace{{\mathbb P}_{{\bf s}_{- \overline{i}}
                                                       \sim
                                                       \bm{\sigma}_{- \overline{i}}
                                                      }
                              \left( {\bf s}_{- \overline{i}}
                                     \diamond
                                     s 
                              \right)}_{> 0}
         \cdot
         \underbrace{\widehat{{\mathsf{U}}}_{\overline{i}}
                                       \left( {\bf s}_{- \overline{i}}
                                              \diamond
                                              s 
                                       \right)}_{> 0}
       & \mbox{(by the utility functions)}                                                                                                        \\
>    & 0\, .
       &
\end{eqnarray*}
}

\item
Assume now that
there is no profile
${\bf s}$
supported in $\bm{\sigma}$ with
${\bf s}_{- \overline{i}}
 \in
 {\widehat{{\mathsf{\Sigma}}}}_{- \overline{i}}$.
Define $s$
to be any strategy
from
${\widehat{{\mathsf{\Sigma}}}}_{\overline{i}}$.
So,
$\overline{i}
 \in
 {\widehat{\mathcal{P}}}
                    \left( {\bf s}_{- \overline{i}}
                           \diamond
                           s
                    \right)$
for each profile
${\bf s}_{- \overline{i}}$
supported in
$\bm{\sigma}_{- \overline{i}}
 \diamond
 s$;
it follows that
only profiles 
from Cases 
{\sf (8)}
and
{\sf (10)}
may be supported in
$\bm{\sigma}_{- \overline{i}}
 \diamond
 s$.
Hence, 
by the utility functions,
${\mathsf{U}}_{\overline{i}}
          \left( {\bf s}_{- \overline{i}}
                 \diamond
                 s 
          \right)
 =
 1$,
which implies that
${\mathsf{U}}_{\overline{i}}
          \left( \bm{\sigma}_{- \overline{i}}
                 \diamond
                 s 
          \right)
 =
 1
 >
 0$.

\end{itemize}
The claim follows.
\end{proof}

\noindent
We now prove
that the only Nash equilibria
for ${\mathsf{G}}$ 
where some player
is exclusively playing
gadget strategies 
are the Nash equilibria
for the gadget game
$\widehat{{\mathsf{G}}}$.

\begin{lemma}
\label{f is PNE}
The following conditions hold:
\begin{enumerate}

\item[{\sf (C.1)}]
${\mathcal{NE}}({\widehat{{\mathsf{G}}}})
 \subseteq
 {\mathcal{NE}}({\mathsf{G}})$.

\item[{\sf (C.2)}]
There is no Nash equilibrium 
${\bm{\sigma}}
 \in
 {\mathcal{NE}}({\mathsf{G}})
 \setminus
 {\mathcal{NE}}({\widehat{{\mathsf{G}}}})$
with
${\mathsf{Supp}}(\sigma_{i})
 \subseteq
 {\widehat{{\mathsf{\Sigma}}}}_{i}$,
for some player $i\in [r]$.

\end{enumerate}
\end{lemma}

\begin{proof}
For Condition {\sf (C.1)},
fix a Nash equilibrium
$\bm{\sigma}
 \in
 {\mathcal{NE}}({\widehat{{\mathsf{G}}}})$;
so,
no player $i \in [r]$
can improve by switching
to a strategy 
$s \in {\widehat{{\mathsf{\Sigma}}}}_{i}$. 
Consider now player $i \in [r]$
switching to a strategy
$s 
 \notin
 {\widehat{{\mathsf{\Sigma}}}}_{i}$.
Then,
a profile ${\bf s}$
supported in
${\bm{\sigma}}_{-i}\diamond s$ 
may only come from Cases 
{\sf (8)},
{\sf (9)},
{\sf (10)} 
and {\sf (11)}, 
so that
$i \in {\mathcal{P}}({\bf s})$
and
${\mathsf{U}}_{i}({\bf s}) = 0$,
which implies that
${\mathsf{U}}_{i}
          ({\bm{\sigma}}_{-i}
           \diamond s)
 =
 0$.
Since ${\mathsf{G}}$ is win-lose,
${\mathsf{U}}_{i}({\bm{\sigma}}) \geq 0$. 
It follows that
${\mathsf{U}}_{i}
          ({\bm{\sigma}}_{-i}
           \diamond s)
 \leq
 {\mathsf{U}}_{i}({\bm{\sigma}})$,
and player $i$ cannot improve
by switching to $s$.
So,
$\bm{\sigma}$
is a Nash equilibrium
for ${\mathsf{G}}$,
so that
${\mathcal{NE}}({\widehat{{\mathsf{G}}}})
 \subseteq
 {\mathcal{NE}}({\mathsf{G}})$.

For Condition {\sf (C.2)},
assume,
by way of contradiction,
that there is a Nash equilibrium 
${\bm{\sigma}}
 \in
 {\mathcal{NE}}({\mathsf{G}})
 \setminus
 {\mathcal{NE}}({\widehat{{\mathsf{G}}}})$
with
${\mathsf{Supp}}(\sigma_{i})
 \subseteq
 {\widehat{{\mathsf{\Sigma}}}}_{i}$
for some player $i \in [r]$. 
By Lemma~\ref{first property of best-responses},
for each player
$j \in [r]
       \setminus
       \{ i \}$,
${\mathsf{Supp}}(\sigma_{j})
 \subseteq
 {\widehat{\mathsf{\Sigma}}}_{j}$.
Hence,
$\bm{\sigma}$
is a mixed profile for ${\widehat{{\mathsf{G}}}}$,
which implies that
$\bm{\sigma}
 \in
 {\mathcal{NE}}({\widehat{{\mathsf{G}}}})$.  
A contradiction.
\end{proof}

\noindent
A Nash equilibrium for
${\mathsf{G}}$
\textcolor{black}{coming}
from ${\widehat{\mathsf{G}}}$
will be called
a {\it gadget equilibrium}.
The following properties
characterize
a Nash equilibrium
${\bm{\sigma}}
 \in
 {\mathcal{NE}}({\mathsf{G}})
 \setminus
 {\mathcal{NE}}({\widehat{\mathsf{G}}})$.
Lemmas~\ref{literals are always played},
\ref{same probability on pairs of literals},
\ref{only literals are played}
and~\ref{no negations are played}
concern the two special players;
Lemma~\ref{delta is always played}
concerns the non-special players. 
We first prove that  
each special player
is playing some literal with positive probability.

\begin{lemma}
\label{literals are always played}
Fix a Nash equilibrium
${\bm{\sigma}}
 \in
 {\mathcal{NE}}({\mathsf{G}})
 \setminus
 {\mathcal{NE}}({\widehat{{\mathsf{G}}}})$.
Then,
$\sigma_{1}({\mathsf{L}})
 \cdot
 \sigma_{2}({\mathsf{L}})
 >
 0$.
\end{lemma}

\begin{proof}
Assume,
by way of contradiction,
that there is a player
$i \in [2]$
with
$\sigma_{i}({\mathsf{L}})
 =
 0$.
So,
${\mathsf{Supp}}(\sigma_{i})
 \subseteq
 {\widehat{\mathsf{\Sigma}}}_{i}
 \cup
 {\mathcal{C}}
 \cup
 {\mathsf{V}}$.
Then,   
Lemma~\ref{second property of best-responses}
implies that
${\mathsf{Supp}}(\sigma_{\overline{i}})
 \subseteq
 {\widehat{\mathsf{\Sigma}}}_{\overline{i}}$. 
A contradiction
to Lemma~\ref{f is PNE} (Condition {\sf (2)}).
\end{proof}

\noindent 
We now prove that in a Nash equilibrium
${\bm\sigma}
 \in
 {\mathcal{NE}}({\mathsf{G}})
 \setminus
 {\mathcal{NE}}({\widehat{\mathsf{G}}})$,
each special player
assigns the same positive probability
to each pair of literals 
$\{ \ell, \overline{\ell} \}$.

\begin{lemma}
\label{same probability on pairs of literals}
Fix a Nash equilibrium
${\bm{\sigma}}
 \in
 {\mathcal{NE}}({\mathsf{G}})
 \setminus{\mathcal{NE}}({\widehat{{\mathsf{G}}}})$
and a special player $i \in [2]$. 
Then,
for each literal
$\ell \in {\mathsf{L}}$,
      $\sigma_{i}\left( \left\{ \ell, 
                                          \overline{\ell}
                               \right\}           
                      \right)
= 
\frac{\textstyle \sigma_{i}({\sf L})}
                  {\textstyle n}$. 
\end{lemma}

\begin{proof}
We shall establish
a sequence of claims.
We first prove:

\begin{claim}
\label{5 6 1}
For each special player $i \in [2]$,
there is no index $j \in I_{n}$
with
$\sigma_{i}({\mathsf{L}})
 =
 \sigma_{i}\left( \left\{ \ell_{j},
                                     \overline{\ell}_{j}
                  \right\}
           \right)$.
\end{claim}

\begin{proof}
Assume,
by way of contradiction,
that there is an index $j \in I_{n}$
with
$\sigma_{i}({\mathsf{L}})
 =
 \sigma_{i}\left( \left\{ \ell_{j},
                          \overline{\ell}_{j}
                  \right\}
           \right)$.
Set $i := 1$.
Consider now player $2$
switching to a literal
$\ell$.
By the utility functions
(Cases {\sf (3)} and {\sf (4)}),
player $2$ gets utility $1$
if and only if
$\ell \in {\mathsf{L}}^{\prime}
          :=
          \{ \ell_{j+2},
             \overline{\ell}_{j+2},
             \ell_{j+3},
             \overline{\ell}_{j+3}
          \}$
and
all non-special players
choose ${\mathsf{\delta}}$.
By Lemma~\ref{f is PNE} (Condition {\sf (2)}),
for each non-special player $i^{\prime} \in [r] \setminus [2]$,
${\mathsf{\delta}}
  \in
  {\mathsf{Supp}}(\sigma_{i^{\prime}})$.
Hence,
there is at least one profile from Case {\sf (3)}
supported in
$\bm{\sigma}_{-2} 
  \diamond
  \ell$.  
It follows that
${\mathsf{U}}_{2}\left( \bm{\sigma}_{-2}
                        \diamond
                        \ell_{k}
                 \right)
 >
 {\mathsf{U}}_{2}\left( \bm{\sigma}_{-2}
                        \diamond
                        \ell_{h}
                 \right)$
for any pair of literals
$\ell_{k} \in {\mathsf{L}}^{\prime}$
and
$\ell_{h} \not\in {\mathsf{L}}^{\prime}$.
By Lemma~\ref{basic property of mixed nash equilibria}
(Condition {\sf (2)}),
this implies that
${\mathsf{Supp}}(\sigma_{2})
 \cap
 {\mathsf{L}}
 \subseteq
 {\mathsf{L}}^{\prime}$.

Now,
taking that
${\mathsf{Supp}}(\sigma_{2})
 \cap
 {\mathsf{L}}
 \subseteq
 {\mathsf{L}}^{\prime}$,
consider player $1$ switching to
a literal $\ell$.
By the utility functions
(Cases {\sf (1)}, {\sf (2)} and {\sf (4)}),
if
$\ell \not\in 
          {\mathsf{L}}^{\prime\prime}
          :=
          \{ \ell_{j+1},
             \overline{\ell}_{j+1},
             \ell_{j+2},
             \overline{\ell}_{j+2},
             \ell_{j+3},
             \overline{\ell}_{j+3}
          \}$
then
player $1$ gets utility $0$;
if she chooses some literal
$\ell_{k} \in {\mathsf{L}}^{\prime\prime}$          
and all non-special players
choose ${\mathsf{\delta}}$
then she gets utility $1$          
By Lemma~\ref{f is PNE} (Condition {\sf (2)}),
for each non-special player 
$i^{\prime} \in [r] \setminus [2]$,
${\mathsf{\delta}}
  \in
  {\mathsf{Supp}}(\sigma_{i^{\prime}})$.
Hence,
there is at least one profile
from Case {\sf (2)}
supported in $\bm{\sigma}_{-1} \diamond \ell$.
It follows that  
${\mathsf{U}}_{1}\left( \bm{\sigma}_{-1}
                        \diamond
                        \ell_{k}
                 \right)
 >
 {\mathsf{U}}_{1}\left( \bm{\sigma}_{-1}
                        \diamond
                        \ell_{h}
                 \right)$
for any literal
$\ell_{h} \not\in {\mathsf{L}}^{\prime\prime}$.                 
By Lemma~\ref{basic property of mixed nash equilibria} (Condition {\sf (2)}),
this implies that
${\mathsf{Supp}}(\sigma_{1})
 \cap
 {\mathsf{L}} 
 \subseteq
 {\mathsf{L}}^{\prime\prime}$.                           
Note that
with $n \geq 4$,
this implies 
$\ell_{j}
 \not\in
 {\mathsf{Supp}}(\sigma_{1})
 \cap
 {\mathsf{L}}$.
A contradiction.
We use a corresponding argument
to establish the claim
for player $2$.
\end{proof}

\noindent
We continue to prove:

\begin{claim}
\label{5 6 2}
The following conditions hold
for each index $k \in I_{n}$:
\begin{enumerate}

\item[{\sf (C.1)}]
If 
$\sigma_{1}\left( \left\{ \ell_{k},
                          \overline{\ell}_{k}
                  \right\}
           \right)        
 >
 0$
then
$\sigma_{2}\left( \left\{ \ell_{k},
                          \overline{\ell}_{k}
                  \right\}
           \right)               
 \cdot 
 \sigma_{2}\left( \left\{ \ell_{k+1},
                          \overline{\ell}_{k+1}
                  \right\}
           \right)               
 >
 0$.

\item[{\sf (C.2)}]
If 
$\sigma_{2}\left( \left\{ \ell_{k},
                          \overline{\ell}_{k}
                  \right\}
           \right)               
 >
 0$
then
$\sigma_{1}\left( \left\{ \ell_{k-3},
                          \overline{\ell}_{k-3}
                  \right\}
           \right)               
 \cdot
 \sigma_{1}\left( \left\{ \ell_{k-2},
                          \overline{\ell}_{k-2}
                  \right\}
           \right)               
 >
 0$.

\end{enumerate}
\end{claim}

\begin{proof}
We start with Condition {\sf (C.1)}.
Claim~\ref{5 6 1}
implies that there are
at least two distinct indices
$j, h \in I_{n}$
with
$\sigma_{2}\left( \left\{ \ell_{j},
                          \overline{\ell}_{j}
                  \right\}        
           \right)
 >
 0$
and 
$\sigma_{2}\left( \left\{ \ell_{h},
                          \overline{\ell}_{h}
                  \right\}
           \right)               
 >
 0$.
If $\{ k, k+1 \}
      =
      \{ j, h \}$,
then we are done.
So assume that
$\{ k, k+1 \}
  \neq
  \{ j, h \}$.
This implies that either
$k \neq j$ and $k+1 \neq j$
or $k \neq h$ and $k+1 \neq h$.
Choose
$k \neq h$ and $k+1 \neq h$.

Since
$\sigma_{1}\left( \left\{ \ell_{k},
                          \overline{\ell}_{k}
                  \right\}
           \right)                
 >
 0$,
either $\ell_{k}$
or $\overline{\ell}_{k}$
is a best-response
for player $1$
to $\bm{\sigma}_{-1}$; 
assume,
without loss of generality,
that $\ell_{k}$
is a best-response
for player $1$
to $\bm{\sigma}_{-1}$.
Consider player $1$
switching to $\ell_{k}$.
By the utility functions
(Cases {\sf (1)}, {\sf (2)} and {\sf (4)}),
player $1$ gets utility $1$
if and only if
player $2$
chooses a literal
from
$\{ \ell_{k},
    \ell_{k+1},
    \overline{\ell}_{k+1}
 \}$
and
all non-special players
choose ${\mathsf{\delta}}$.     
Hence,
{
\small
\begin{eqnarray*}
      {\mathsf{U}}_{1}
               \left( \bm{\sigma}_{-1}
                      \diamond
                      \ell_{k}
               \right)
& = & \sigma_{2} \left( \left\{ \ell_{k},
                                \ell_{k+1},
                                \overline{\ell}_{k+1}
                        \right\}
                 \right)
      \cdot
      \prod_{i \in [r]
                  \setminus
                  [2]}                           
        \sigma_{i}\left( {\mathsf{\delta}}
                  \right)\, .
\end{eqnarray*}
}

Consider player $1$
switching to ${\sf v}_{h,k}$.
By the utility functions
(Case {\sf (5)}),
player $1$ gets utility $1$
if and only if
player $2$
chooses a literal
from
$\{ \ell_{h},
    \overline{\ell}_{h}, 
    \ell_{k},
    \overline{\ell}_{k}
 \}$
and
all non-special players
choose ${\mathsf{\delta}}$.    
Hence,
{
\small
\begin{eqnarray*}
      {\mathsf{U}}_{1}
               \left( \bm{\sigma}_{-1}
                      \diamond
                      {\mathsf{v}}_{h,k}
               \right)
& = & \sigma_{2} \left( \left\{ \ell_{h},
                                \overline{\ell}_{h},
                                \ell_{k},
                                \overline{\ell}_{k}
                        \right\}
                 \right)
      \cdot
      \prod_{i \in [r]
                  \setminus
                  [2]}                           
        \sigma_{i}\left( {\mathsf{\delta}}
                  \right)\, .
\end{eqnarray*}
}
Since $\ell_{k}$ is a best-response
for player $1$
to $\bm{\sigma}_{-1}$,
it follows that
{
\small
\begin{eqnarray*}
      \sigma_{2} \left( \left\{ \ell_{k},
                                \ell_{k+1},
                                \overline{\ell}_{k+1}
                        \right\}
                 \right)
      \cdot
      \prod_{i \in [r]
                  \setminus
                  [2]}                           
        \sigma_{i}\left( {\mathsf{\delta}}
                  \right)
&\geq& \sigma_{2} \left( \left\{ \ell_{h},
                                \overline{\ell}_{h},
                                \ell_{k},
                                \overline{\ell}_{k}
                        \right\}
                 \right)
      \cdot
      \prod_{i \in [r]
                  \setminus
                  [2]}                           
        \sigma_{i}\left( {\mathsf{\delta}}
                  \right)\, .
\end{eqnarray*}
}
By Lemma~\ref{f is PNE}
(Condition {\sf (2)}),
$\prod_{i \in [r]
              \setminus
              [2]}                           
   \sigma_{i}\left( {\mathsf{\delta}}
             \right)
 \neq
 0$;
thus, 
it follows that
{
\small
\begin{eqnarray*}
\sigma_{2} \left( \left\{ \ell_{k},
                           \ell_{k+1},
                           \overline{\ell}_{k+1}
                   \right\}
            \right)
& \geq &
 \sigma_{2} \left( \left\{ \ell_{h},
                           \overline{\ell}_{h},
                           \ell_{k},
                           \overline{\ell}_{k}
                   \right\}
            \right)\, ,
\end{eqnarray*} 
}           
which implies that 
{
\small 
\begin{eqnarray*}                      
\sigma_{2} \left( \left\{ \ell_{k+1},
                           \overline{\ell}_{k+1}
                   \right\}
            \right)
& \geq &
 \sigma_{2} \left( \left\{ \ell_{h},
                           \overline{\ell}_{h},
                           \overline{\ell}_{k}
                   \right\}
            \right)\, .
\end{eqnarray*}
}
\noindent
Since
$\sigma_{2} \left( \left\{ \ell_{h},
                           \overline{\ell}_{h},
                   \right\}
            \right)            
 >
 0$,
it follows that
$\sigma_{2} \left( \left\{ \ell_{k+1},
                           \overline{\ell}_{k+1}
                   \right\}
            \right)
 >
 0$.

Consider now player $1$
switching to ${\sf v}_{h,k+1}$.
By the utility functions
(Case {\sf (5)}),
player $1$ gets utility $1$
if and only if
player $2$
chooses a literal
from
$\{ \ell_{h},
    \overline{\ell}_{h}, 
    \ell_{k+1},
    \overline{\ell}_{k+1}
 \}$
and
all non-special players
choose ${\mathsf{\delta}}$.    
Hence,
{
\small
\begin{eqnarray*}
      {\mathsf{U}}_{1}
               \left( \bm{\sigma}_{-1}
                      \diamond
                      {\mathsf{v}}_{h,k+1}
               \right)
& = & \sigma_{2} \left( \left\{ \ell_{h},
                                \overline{\ell}_{h},
                                \ell_{k+1},
                                \overline{\ell}_{k+1}
                        \right\}
                 \right)
      \cdot
      \prod_{i \in [r]
                  \setminus
                  [2]}                           
        \sigma_{i}\left( {\mathsf{\delta}}
                  \right)\, .
\end{eqnarray*}
}
Since $\ell_{k}$ is a best-response
for player $1$
to $\bm{\sigma}_{-1}$,
it follows that
{
\small
\begin{eqnarray*}
      \sigma_{2} \left( \left\{ \ell_{k},
                                \ell_{k+1},
                                \overline{\ell}_{k+1}
                        \right\}
                 \right)
      \cdot
      \prod_{i \in [r]
                  \setminus
                  [2]}                           
        \sigma_{i}\left( {\mathsf{\delta}}
                  \right)
&\geq& \sigma_{2} \left( \left\{ \ell_{h},
                                \overline{\ell}_{h},
                                \ell_{k+1},
                                \overline{\ell}_{k+1}
                        \right\}
                 \right)
      \cdot
      \prod_{i \in [r]
                  \setminus
                  [2]}                           
        \sigma_{i}\left( {\mathsf{\delta}}
                  \right)\, .
\end{eqnarray*}
}
\noindent
By Lemma~\ref{f is PNE}
(Condition {\sf (2)}),
$\prod_{i \in [r]
              \setminus
              [2]}                           
   \sigma_{i}\left( {\mathsf{\delta}}
             \right)
 \neq
 0$;
thus, 
it follows that
{
\small
\begin{eqnarray*}
\sigma_{2} \left( \left\{ \ell_{k},
                           \ell_{k+1},
                           \overline{\ell}_{k+1}
                   \right\}
            \right)
& \geq &
 \sigma_{2} \left( \left\{ \ell_{h},
                           \overline{\ell}_{h},
                           \ell_{k+1},
                           \overline{\ell}_{k+1}
                   \right\}
            \right)\, ,
\end{eqnarray*}
}
\noindent            
which implies that
{
\small
\textcolor{black}{
\begin{eqnarray*}  
\sigma_{2} \left( \left\{ \ell_{k},
                                      \overline{\ell}_{k}
                   \right\}
            \right)
& \geq &                      
\sigma_{2} \left( \left\{ \ell_{k}
                   \right\}
            \right)\ \
 \geq\ \
 \sigma_{2} \left( \left\{ \ell_{h},
                           \overline{\ell}_{h}
                   \right\}
            \right)\, .
\end{eqnarray*}
}
}            
Since
$\sigma_{2} \left( \left\{ \ell_{h},
                           \overline{\ell}_{h}
                   \right\}
            \right)            
 >
 0$,
it follows that
$\sigma_{2} \left( \left\{ \ell_{k},
                             \overline{\ell}_{k}
                   \right\}
            \right)
 >
 0$.
This completes the proof
of Condition {\sf (1)}.

We continue with Condition {\sf (C.2)}.
Claim~\ref{5 6 1}
implies that there are
at least two distinct indices
$j, h \in I_{n}$
with
$\sigma_{1}\left( \left\{ \ell_{j},
                             \overline{\ell}_{j}
                  \right\}        
           \right)
 >
 0$
and 
$\sigma_{1}\left( \left\{ \ell_{h},
                            \overline{\ell}_{h}
                  \right\}
           \right)               
 >
 0$.
If $\{ k-3, k-2 \}
      =
      \{ j, h \}$,
then we are done.
So assume that
$\{ k-3, k-2 \}
  \neq
  \{ j, h \}$.
This implies that either
$k-3 \neq j$ and $k-2 \neq j$
or $k-3 \neq h$ and $k-2 \neq h$.
Choose
$k-3 \neq h$ and $k-2 \neq h$.

Since
$\sigma_{2}\left( \left\{ \ell_{k},
                                       \overline{\ell}_{k}
                            \right\}
                   \right)                
 >
 0$,
either $\ell_{k}$
or $\overline{\ell}_{k}$
is a best-response
for player $2$
to $\bm{\sigma}_{-2}$; 
assume,
without loss of generality,
that $\ell_{k}$
is a best-response
for player $2$
to $\bm{\sigma}_{-2}$.
Consider player $2$
switching to $\ell_{k}$.
By the utility functions
(Cases {\sf (1)}, {\sf (2)} and {\sf (4)}),
player $2$ gets utility $1$
if and only if
player $1$
chooses a literal
from
$\{ \ell_{k-2},
      \overline{\ell}_{k-2},
      \ell_{k-3},
      \overline{\ell}_{k-3}
 \}$
and
all non-special players
choose ${\mathsf{\delta}}$.     
Hence,
{
\small
\begin{eqnarray*}
      {\mathsf{U}}_{2}
               \left( \bm{\sigma}_{-2}
                      \diamond
                      \ell_{k}
               \right)
& = & \sigma_{1} \left( \left\{ \ell_{k-2},
                                                 \overline{\ell}_{k-2},
                                                 \ell_{k-3},
                                                 \overline{\ell}_{k-3}
                             \right\}
                 \right)
      \cdot
      \prod_{i \in [r]
                  \setminus
                  [2]}                           
        \sigma_{i}\left( {\mathsf{\delta}}
                  \right)\, .
\end{eqnarray*}
}

Consider player $2$
switching to ${\sf v}_{h,k-2}$.
By the utility functions
(Case {\sf (5)}),
player $2$ gets utility $1$
if and only if
player $1$
chooses a literal
from
$\{ \ell_{h},
    \overline{\ell}_{h}, 
    \ell_{k-2},
    \overline{\ell}_{k-2}
 \}$
and
all non-special players
choose ${\mathsf{\delta}}$.    
Hence,
{
\small
\begin{eqnarray*}
      {\mathsf{U}}_{2}
               \left( \bm{\sigma}_{-2}
                      \diamond
                      {\mathsf{v}}_{h,k-2}
               \right)
& = & \sigma_{1} \left( \left\{ \ell_{h},
                                                \overline{\ell}_{h},
                                                \ell_{k-2},
                                                \overline{\ell}_{k-2}
                        \right\}
                 \right)
      \cdot
      \prod_{i \in [r]
                  \setminus
                  [2]}                           
        \sigma_{i}\left( {\mathsf{\delta}}
                  \right)\, .
\end{eqnarray*}
}
\noindent
Since $\ell_{k}$ is a best-response
for player $2$
to $\bm{\sigma}_{-2}$,
it follows that
{
\small
\begin{eqnarray*}
      \sigma_{1} \left( \left\{ \ell_{k-2},
                                             \overline{\ell}_{k-2},
                                             \ell_{k-3},
                                             \overline{\ell}_{k-3}
                        \right\}
                 \right)
      \cdot
      \prod_{i \in [r]
                  \setminus
                  [2]}                           
        \sigma_{i}\left( {\mathsf{\delta}}
                  \right)
&\geq& \sigma_{1} \left( \left\{ \ell_{h},
                                                    \overline{\ell}_{h},
                                                    \ell_{k-2},
                                                    \overline{\ell}_{k-2}
                                         \right\}
                                \right)
      \cdot
      \prod_{i \in [r]
                  \setminus
                  [2]}                           
        \sigma_{i}\left( {\mathsf{\delta}}
                  \right)\, .
\end{eqnarray*}
}
\noindent
By Lemma~\ref{f is PNE}
(Condition {\sf (2)}),
$\prod_{i \in [r]
              \setminus
              [2]}                           
   \sigma_{i}\left( {\mathsf{\delta}}
             \right)
 \neq
 0$;
thus, 
it follows that
{
\small
\begin{eqnarray*}
\sigma_{1} \left( \left\{ \ell_{k-2},
                                         \overline{\ell}_{k-2},
                                         \ell_{k-3},
                                         \overline{\ell}_{k-3}
                   \right\}
            \right)
& \geq &
 \sigma_{1} \left( \left\{ \ell_{h},
                                        \overline{\ell}_{h},
                                        \ell_{k-2},
                                        \overline{\ell}_{k-2}
                             \right\}
            \right)\, ,
\end{eqnarray*}
}
\noindent            
which implies that
{
\small
\begin{eqnarray*}                        
\sigma_{1} \left( \left\{ \ell_{k-3},
                                        \overline{\ell}_{k-3}
                   \right\}
            \right)
& \geq &
 \sigma_{1} \left( \left\{ \ell_{h},
                           \overline{\ell}_{h}
                   \right\}
            \right)\, .
\end{eqnarray*}
}            
Since
$\sigma_{1} \left( \left\{ \ell_{h},
                                        \overline{\ell}_{h},
                     \right\}
            \right)            
 >
 0$,
it follows that
$\sigma_{1} \left( \left\{ \ell_{k-3},
                                        \overline{\ell}_{k-3}
                     \right\}
            \right)
 >
 0$.

Consider now player $2$
switching to ${\sf v}_{h,k-3}$.
By the utility functions
(Case {\sf (5)}),
player $2$ gets utility $1$
if and only if
player $1$
chooses a literal
from
$\{ \ell_{h},
    \overline{\ell}_{h}, 
    \ell_{k-3},
    \overline{\ell}_{k-3}
 \}$
and
all non-special players
choose ${\mathsf{\delta}}$.    
Hence,
{
\small
\begin{eqnarray*}
      {\mathsf{U}}_{2}
               \left( \bm{\sigma}_{-2}
                      \diamond
                      {\mathsf{v}}_{h,k-3}
               \right)
& = & \sigma_{1} \left( \left\{ \ell_{h},
                                                \overline{\ell}_{h},
                                                \ell_{k-3},
                                                \overline{\ell}_{k-3}
                        \right\}
                 \right)
      \cdot
      \prod_{i \in [r]
                  \setminus
                  [2]}                           
        \sigma_{i}\left( {\mathsf{\delta}}
                  \right)\, .
\end{eqnarray*}
}
\noindent
Since $\ell_{k}$ is a best-response
for player $2$
to $\bm{\sigma}_{-2}$,
it follows that
{
\small
\begin{eqnarray*}
      \sigma_{1} \left( \left\{ \ell_{k-2},
                                             \overline{\ell}_{k-2},
                                             \ell_{k-3},
                                             \overline{\ell}_{k-3}
                        \right\}
                 \right)
      \cdot
      \prod_{i \in [r]
                  \setminus
                  [2]}                           
        \sigma_{i}\left( {\mathsf{\delta}}
                  \right)
&\geq& \sigma_{1} \left( \left\{ \ell_{h},
                                                    \overline{\ell}_{h},
                                                    \ell_{k-3},
                                                    \overline{\ell}_{k-3}
                                         \right\}
                                \right)
      \cdot
      \prod_{i \in [r]
                  \setminus
                  [2]}                           
        \sigma_{i}\left( {\mathsf{\delta}}
                  \right)\, .
\end{eqnarray*}
}
\noindent
By Lemma~\ref{f is PNE}
(Condition {\sf (2)}),
$\prod_{i \in [r]
              \setminus
              [2]}                           
   \sigma_{i}\left( {\mathsf{\delta}}
             \right)
 \neq
 0$;
thus, 
it follows that
{
\small
\begin{eqnarray*}
\sigma_{1} \left( \left\{ \ell_{k-2},
                                         \overline{\ell}_{k-2},
                                         \ell_{k-3},
                                         \overline{\ell}_{k-3}
                   \right\}
            \right)
& \geq &
 \sigma_{1} \left( \left\{ \ell_{h},
                                        \overline{\ell}_{h},
                                        \ell_{k-3},
                                        \overline{\ell}_{k-3}
                             \right\}
            \right)\, ,
\end{eqnarray*}
}            
which implies that
{
\small
\begin{eqnarray*}                        
\sigma_{1} \left( \left\{ \ell_{k-2},
                                        \overline{\ell}_{k-2}
                   \right\}
            \right)
& \geq &
 \sigma_{1} \left( \left\{ \ell_{h},
                           \overline{\ell}_{h}
                   \right\}
            \right)\, .
\end{eqnarray*}
}            
Since
$\sigma_{1} \left( \left\{ \ell_{h},
                                        \overline{\ell}_{h},
                     \right\}
            \right)            
 >
 0$,
it follows that
$\sigma_{1} \left( \left\{ \ell_{k-2},
                                        \overline{\ell}_{k-2}
                     \right\}
            \right)
 >
 0$,
and this completes the proof. 
\end{proof}

\noindent
We continue to prove:

\begin{claim}
\label{5 6 3}
Fix a player $i \in [2]$.
Then,
for each index $j \in I_{n}$,
$\sigma_{i}\left( \left\{ \ell_{j},
                                      \overline{\ell}_{j}
                          \right\}
                 \right)
 >
 0$.
\end{claim}

\begin{proof} 
Assume, 
by way of contradiction,
that there is an index $j \in I_{n}$
with
$\sigma_{i}\left( \left\{ \ell_{j},
                                      \overline{\ell}_{j}
                            \right\}
                   \right)
 =
 0$.
Without loss of generality, 
take that 
$\sigma_{i}\left( \left\{ \ell_{j+1},
                                      \overline{\ell}_{j+1}
                            \right\}
                   \right)
 >
 0$.
 There are two cases:
\underline{{\sf (1)} $i = 1$.}
By Claim~\ref{5 6 2} (Condition {\sf (C.1)}),
$\sigma_{2}\left( \left\{ \ell_{j+2},
                                       \overline{\ell}_{j+2}
                            \right\}
                   \right)
  >
  0$; 
hence,      
by Claim~\ref{5 6 2} (Condition {\sf (C.2)}),
$\sigma_{1}\left( \left\{ \ell_{j},
                                       \overline{\ell}_{j}
                            \right\}
                   \right)
 >
 0$.
A contradiction.
\underline{{\sf (2)} $i = 2$.}
By Claim~\ref{5 6 2} (Condition {\sf (C.2)}),
$\sigma_{1}\left( \left\{ \ell_{j-1},
                                       \overline{\ell}_{j-1}
                            \right\}
                   \right)
  >
  0$;
hence,       
by Claim~\ref{5 6 2} (Condition {\sf (C.1)}),
$\sigma_{2}\left( \left\{ \ell_{j},
                                       \overline{\ell}_{j}
                            \right\}
                   \right)
 >
 0$.
A contradiction.
\end{proof}

\noindent
To prove that
for each index
$j \in I_{n}$,
$\sigma_{i}\left( \left\{ \ell_{j},
                                     \overline{\ell}_{j}
                           \right\}
                 \right)
 =
 \frac{\textstyle \sigma_{i}({\mathsf{L}})}
         {\textstyle n}$,
assume,
by way of contradiction, 
that there is an index 
$j \in I_{n}$
with 
{
\small
\begin{eqnarray*}
\sigma_{i}\left( \left\{ \ell_{j},
                                       \overline{\ell}_{j}
                            \right\}
                   \right)
& > &
 \sigma_{i}\left( \left\{ \ell_{j+1},
                                      \overline{\ell}_{j+1}
                           \right\}
                  \right)\, .
\end{eqnarray*}
}                  
Fix now $i := 2$.
We proceed by case analysis.
\begin{enumerate}

\item 
Assume first
that there is an index 
$k \in I_{n} 
         \setminus
         \{ j, j+1 \}$ 
with
{
\small
\begin{eqnarray*} 
\sigma_{1} \left( \left\{ \ell_{k},
                                        \overline{\ell}_{k}
                             \right\}
                    \right)
&  > &
 \sigma_{1}\left( \left\{ \ell_{j+1},
                                      \overline{\ell}_{j+1}
                            \right\}
                  \right)\, .
\end{eqnarray*} 
}
\noindent                 
By Claim~\ref{5 6 3},
either $\ell_{j+3}$
or
$\overline{\ell}_{j+3}$
is a best-response for player $2$
to $\bm{\sigma}_{-2}$;
assume,
without loss of generality,
that $\ell_{j+3}$
is such a best-response.
Consider first player $2$
switching to $\ell_{j+3}$.
By the utility functions
(Cases {\sf (3)} and {\sf (4)}),
player $2$ gets utility $1$
if and only if
player $1$ chooses a literal from
$\left\{ \ell_{j},
             \overline{\ell}_{j},
             \ell_{j+1},
             \overline{\ell}_{j+1}
  \right\}$
and all non-special players
choose ${\mathsf{\delta}}$.
Hence,
{
\small
\begin{eqnarray*}
           {\mathsf{U}}_{2}\left( \bm{\sigma}_{-2}
                                                \diamond
                                                \ell_{j+3}
                                       \right)
& = & \sigma_{1}\left( \left\{ \ell_{j},
                                                \overline{\ell}_{j},
                                                \ell_{j+1},
                                                \overline{\ell}_{j+1}
                                     \right\}
                           \right)
          \cdot
          \prod_{i \in [r] \setminus [2]}
              \sigma_{i}({\mathsf{\delta}})\, .
\end{eqnarray*}
}
Consider now player $2$
switching to the variable
${\mathsf{v}}_{j,k}$.
By the utility functions
(Case {\sf (5/b)}),
player $2$
gets utility $1$
if and only if
player $1$ chooses a literal from
$\left\{ \ell_{j},
             \overline{\ell}_{j},
             \ell_{k},
             \overline{\ell}_{k}
  \right\}$
and all non-special players
choose ${\mathsf{\delta}}$.
Hence,
{
\small
\begin{eqnarray*}
           {\mathsf{U}}_{2}\left( \bm{\sigma}_{-2}
                                                \diamond
                                                {\mathsf{v}}_{j,k}
                                       \right)
& = & \sigma_{1}\left( \left\{ \ell_{j},
                                                \overline{\ell}_{j},
                                                \ell_{k},
                                                \overline{\ell}_{k}
                                     \right\}
                           \right)
          \cdot
          \prod_{i \in [r] \setminus [2]}
              \sigma_{i}({\mathsf{\delta}})\, .
\end{eqnarray*}
}
Since $\ell_{j+3}$ is a best-response for player $2$
to $\bm{\sigma}_{-2}$,
it follows that
{
\small
\begin{eqnarray*}
         \sigma_{1}\left( \left\{ \ell_{j},
                                                \overline{\ell}_{j},
                                                \ell_{j+1},
                                                \overline{\ell}_{j+1}
                                     \right\}
                           \right)
          \cdot
          \prod_{i \in [r] \setminus [2]}
              \sigma_{i}({\mathsf{\delta}})
&\geq& \sigma_{1}\left( \left\{ \ell_{j},
                                                \overline{\ell}_{j},
                                                \ell_{k},
                                                \overline{\ell}_{k}
                                     \right\}
                           \right)
          \cdot
          \prod_{i \in [r] \setminus [2]}
              \sigma_{i}({\mathsf{\delta}}).
\end{eqnarray*}
}
By Lemma~\ref{f is PNE}
(Condition {\sf (2)}),
$\prod_{i \in [r] \setminus [2]}
      \sigma_{i}({\mathsf{\delta}})
  \neq
  0$;
thus,  
it follows that
{
\small
\begin{eqnarray*}
\sigma_{1}\left( \left\{ \ell_{j},
                                        \overline{\ell}_{j},
                                        \ell_{j+1},
                                        \overline{\ell}_{j+1}
                            \right\}
                   \right)
&  \geq &
  \sigma_{1}\left( \left\{ \ell_{j},
                                       \overline{\ell}_{j},
                                       \ell_{k},
                                       \overline{\ell}_{k}
                             \right\}
                   \right)\, ,
\end{eqnarray*}
}                   
which implies that
{
\small
\begin{eqnarray*}
\sigma_{1}\left( \left\{  \ell_{j+1},
                                        \overline{\ell}_{j+1}
                            \right\}
                   \right)
&  \geq &
  \sigma_{1}\left( \left\{ \ell_{k},
                                       \overline{\ell}_{k}
                             \right\}
                   \right)\, .
\end{eqnarray*}
}                   
A contradiction.

\item 
Assume now
that there is no index 
$k \in I_{n} 
         \setminus
         \{ j, j+1 \}$ 
with
{
\small
\begin{eqnarray*} 
\sigma_{1} \left( \left\{ \ell_{k},
                                        \overline{\ell}_{k}
                             \right\}
                    \right)
&  > &
 \sigma_{1}\left( \left\{ \ell_{j+1},
                                      \overline{\ell}_{j+1}
                            \right\}
                  \right)\, .
\end{eqnarray*}
}                  
It follows that for each index
$k \in I_{n}
         \setminus
         \{ j, j+1 \}$,
{
\small
\begin{eqnarray*}         
\sigma_{1} \left( \left\{ \ell_{k},
                                        \overline{\ell}_{k}
                             \right\}
                    \right)
& \leq &
 \sigma_{1} \left( \left\{ \ell_{j+1},
                                        \overline{\ell}_{j+1}
                             \right\}
                    \right)\, .
\end{eqnarray*}
}
\noindent                    
Choose $k := j+2$;
so,
{
\small
\begin{eqnarray*}
\sigma_{1} \left( \left\{ \ell_{j+2},
                                        \overline{\ell}_{j+2}
                             \right\}
                    \right)
& \leq &
 \sigma_{1} \left( \left\{ \ell_{j+1},
                                        \overline{\ell}_{j+1}
                             \right\}
                    \right)\, .
\end{eqnarray*}
}
\noindent                    
Consider first player $2$
switching to $\ell_{j+4}$.
By the utility functions
(Cases {\sf (3)} and {\sf (4)}),
player $2$ gets utility $1$
if and only if
player $1$ chooses a literal from
$\left\{ \ell_{j+1},
             \overline{\ell}_{j+1},
             \ell_{j+2},
             \overline{\ell}_{j+2}
  \right\}$
and all non-special players
choose ${\mathsf{\delta}}$.
Hence,
{
\small
\begin{eqnarray*}
           {\mathsf{U}}_{2}\left( \bm{\sigma}_{-2}
                                                \diamond
                                                \ell_{j+4}
                                       \right)
& = & \sigma_{1}\left( \left\{ \ell_{j+1},
                                                \overline{\ell}_{j+1},
                                                \ell_{j+2},
                                                \overline{\ell}_{j+2}
                                     \right\}
                           \right)
          \cdot
          \prod_{i \in [r] \setminus [2]}
              \sigma_{i}({\mathsf{\delta}})\, .
\end{eqnarray*}
}
\noindent
Consider now player $2$
switching to the variable
${\mathsf{v}}_{j,j+1}$.
By the utility functions
(Case {\sf (5)}),
player $2$
gets utility $1$
if and only if
player $1$ chooses a literal from
$\left\{ \ell_{j},
             \overline{\ell}_{j},
             \ell_{j+1},
             \overline{\ell}_{j+1}
  \right\}$.
and all non-special players
choose ${\mathsf{\delta}}$.
Hence,
{
\small
\begin{eqnarray*}
           {\mathsf{U}}_{2}\left( \bm{\sigma}_{-2}
                                                \diamond
                                                {\mathsf{v}}_{j,j+1}
                                       \right)
& = & \sigma_{1}\left( \left\{ \ell_{j},
                                                \overline{\ell}_{j},
                                                \ell_{j+1},
                                                \overline{\ell}_{j+1}
                                     \right\}
                           \right)
          \cdot
          \prod_{i \in [r] \setminus [2]}
              \sigma_{i}({\mathsf{\delta}})\, .
\end{eqnarray*}
}
Since $\ell_{j}$ is a best-response
for player $1$ to $\bm{\sigma}_{-1}$,
it follows that
{
\small
\begin{eqnarray*}
         \sigma_{1}\left( \left\{ \ell_{j+1},
                                                \overline{\ell}_{j+1},
                                                \ell_{j+2},
                                                \overline{\ell}_{j+2}
                                     \right\}
                           \right)
          \cdot
          \prod_{i \in [r] \setminus [2]}
              \sigma_{i}({\mathsf{\delta}})
&\geq& \sigma_{1}\left( \left\{ \ell_{j},
                                                \overline{\ell}_{j},
                                                \ell_{J+1},
                                                \overline{\ell}_{j+1}
                                     \right\}
                           \right)
          \cdot
          \prod_{i \in [r] \setminus [2]}
              \sigma_{i}({\mathsf{\delta}}).
\end{eqnarray*}
}
By Lemma~\ref{f is PNE}
(Condition {\sf (2)}),
$\prod_{i \in [r] \setminus [2]}
      \sigma_{i}({\mathsf{\delta}})
  \neq
  0$;
hence,  
it follows that
{
\small
\begin{eqnarray*}
\sigma_{1}\left( \left\{ \ell_{j+1},
                                        \overline{\ell}_{+1j},
                                        \ell_{j+2},
                                        \overline{\ell}_{j+2}
                            \right\}
                   \right)
&  \geq &
  \sigma_{1}\left( \left\{ \ell_{j},
                                       \overline{\ell}_{j},
                                       \ell_{j+1},
                                       \overline{\ell}_{j+1}
                             \right\}
                   \right)\, ,
\end{eqnarray*}
}                   
which implies that
{
\small
\begin{eqnarray*}
\sigma_{1}\left( \left\{  \ell_{j+2},
                                        \overline{\ell}_{j+2}
                            \right\}
                   \right)
& \geq &
  \sigma_{1}\left( \left\{ \ell_{j},
                                       \overline{\ell}_{j}
                             \right\}
                   \right)\, .
\end{eqnarray*}
}
\noindent                   
It follows that
{
\small
\begin{eqnarray*}
\sigma_{1}\left( \left\{  \ell_{j+1},
                                        \overline{\ell}_{j+1}
                            \right\}
                   \right)
&  \geq  &
  \sigma_{1}\left( \left\{ \ell_{j},
                                       \overline{\ell}_{j}
                             \right\}
                   \right)\, .
\end{eqnarray*}
}                   
A contradiction.

\end{enumerate}
\noindent
We use a corresponding argument
to establish the claim
when
$i = 1$.
\end{proof}

\noindent
We prove that in a Nash equilibrium 
from ${\mathcal{NE}}({\mathsf{G}})
    \setminus
    {\mathcal{NE}}({\widehat{\mathsf{G}}})$, 
each non-special player
chooses $\delta$.

\begin{lemma}
\label{delta is always played}
Fix a Nash equilibrium 
${\bm{\sigma}}
 \in
 {\mathcal{NE}}({\mathsf{G}})
 \setminus
 {\mathcal{NE}}({\widehat{{\mathsf{G}}}})$.
Then,
for each player
$i \in [r] \setminus [2]$, 
$\sigma_{i}({\mathsf{\delta}})
 =
 1$.
\end{lemma}

\begin{proof}
Assume,
by way of contradiction,
that there is a player
$i \in [r] \setminus [2]$
with $\sigma_{i} (\delta) <1$.
Since 
${\mathsf{\Sigma}}_{i}
  =
  \widehat{{\mathsf{\Sigma}}}_{i} \cup \{ {\mathsf{\delta}} \}$,
this implies that
$\sigma_{i} (\widehat{\mathsf{\Sigma}}_{i})
 >
 0$.
Lemma~\ref{f is PNE}
(Condition {\sf (2)})
implies that
$\sigma_{i} (\widehat{\mathsf{\Sigma}}_{i})
 <
 1$.
Since
${\mathsf{\Sigma}}_{i}
  =
  {\widehat{{\mathsf{\Sigma}}}}_{i}
   \cup
   \{ {\mathsf{\delta}} \}$,
it follows that
$\sigma_{i} (\delta)
 >
 0$.

By Lemma~\ref{literals are always played},
there is a literal
$\ell_{j}
 \in
 {\mathsf{Supp}}(\sigma_{1})$.
Consider player $1$
switching to $\ell_{j}$.
By the utility functions,
player $1$ gets utility $1$
if and only if
player $2$ chooses a literal
from $\{ \ell_{j},
         \ell_{j+1},
         \overline{\ell}_{j+1}
      \}$
and each non-special player chooses ${\mathsf{\delta}}$
(Cases {\sf (1)}, {\sf (2)} and {\sf (4)}).
Hence,
{
\small
\begin{eqnarray*}
    \lefteqn{{\mathsf{U}}_{1}\left( {\bm{\sigma}}_{-1}
                                    \diamond
                                    \ell_{j}
                             \right)}                                               \\          
= & \prod_{i' \notin [2] \cup \{ i \}}
       \sigma_{i'} ({\mathsf{\delta}})
    \cdot
    \sigma_{i}({\mathsf{\delta}})
    \cdot
    \sigma_{2}\left( \left\{ \ell_{j},
                             \ell_{j+1},
                             \overline{\ell}_{j+1}
                     \right\}
              \right)
  &                                                                                \\
\leq& \sigma_{i}({\mathsf{\delta}})
      \cdot
      2
      \cdot
      \frac{\textstyle \sigma_{2}({\mathsf{L}})}
           {\textstyle n}
    & \mbox{(by Lemma~\ref{same probability on pairs of literals})}               \\
<   & 2
      \cdot
      \frac{\textstyle \sigma_{2}({\mathsf{L}})}
           {\textstyle n}                  
    & \mbox{(since $\sigma_{i}({\mathsf{\delta}}) < 1$)\, .}                                              
\end{eqnarray*}
}

Consider now player $1$
switching to a strategy
$s 
 \in
 \widehat{{\mathsf{\Sigma}}}_{1}$.
By the utility functions,
player $1$ gets utility $1$
if player $2$
chooses a literal from
$\{ \ell_{0},
    \overline{\ell}_{0},
    \ell_{1},
    \overline{\ell}_{1}
 \}$
(Case {\sf (10)}). 
Then,
by Lemma~\ref{same probability on pairs of literals}, 
{
\small
\begin{eqnarray*} 
         {\sf U}_1({\bm\sigma}_{-1}\diamond s)
& \geq & 2
         \cdot
         \frac{\textstyle \sigma_{2}({\mathsf{L}})}
              {\textstyle n}\, .
\end{eqnarray*}
} 
By Lemma~\ref{literals are always played},
$\sigma_{2}({\mathsf{L}})
 >
 0$. 
This implies that
{
\small
\begin{eqnarray*}
      {\mathsf{U}}_{1}
               ({\bm{\sigma}}_{-1}
                \diamond
                s)
& > & {\mathsf{U}}_{1}
               ({\bm{\sigma}}_{-1}
                \diamond
                \ell_{j})\, .
\end{eqnarray*}
}                
Since
$\bm{\sigma}$ is a Nash equilibrium
and
$\ell_{j}
 \in
 {\mathsf{Supp}}(\sigma_{1})$,
Lemma~\ref{basic property of mixed nash equilibria}
(Condition {\sf (2)})
implies that
{
\small
\begin{eqnarray*}
      {\mathsf{U}}_{1}
               ({\bm{\sigma}}_{-1}
                \diamond
                s)
&\leq& {\mathsf{U}}_{1}
               ({\bm{\sigma}}_{-1}
                \diamond
                \ell_{j})\, .
\end{eqnarray*}     
}
A contradiction.
\end{proof}

\noindent
Lemma~\ref{delta is always played}
implies that
for a Nash equilibrium
${\bm{\sigma}}
  \in  
 {\mathcal{NE}}({\sf G})
 \setminus
 {\mathcal{NE}}({\widehat{{\mathsf{G}}}})$,
the utility of each special player $i \in [2]$
in either ${\bm{\sigma}}$
or ${\bm{\sigma}}_{-i} \diamond s$,
where $s \in {\mathsf{\Sigma}}_{i}$,
is solely determined by the strategies
chosen by the two special players. 
We prove that
each special player
is exclusively playing literals.

\begin{lemma}[Special Players Only Play Literals]
\label{only literals are played}
Fix a Nash equilibrium
${\bm\sigma}
 \in
 {\cal NE}({\mathsf{G}})
 \setminus
 {\cal NE}({\widehat{\mathsf{G}}})$
and a special player $i \in [2]$.
Then,
$\sigma_{i}({\mathsf{L}})
 =
 1$.
\end{lemma}

\begin{proof}
By Lemma~\ref{literals are always played},
there is a literal
$\ell_{j}
 \in
 {\mathsf{Supp}}(\sigma_{\overline{i}})$.
Consider player $\overline{i}$
switching to $\ell_{j}$.
By the utility functions,
Lemma~\ref{delta is always played} 
implies that
player $\overline{i}$
gets utility $1$
if and only if
player $i$ chooses a literal from
$\left\{ \ell_{j}, 
             \ell_{j+1},
             \overline{\ell}_{j+1}
 \right\}$
(Cases {\sf (1)}, {\sf (2)} and {\sf (4)}, 
with $i=2$) 
(resp.,
player $i$
chooses a literal from
$\left\{ \ell_{j-2},
            \overline{\ell}_{j-2},
            \ell_{j-3},
            \overline{\ell}_{j-3}
 \right\}$
(Cases {\sf (3)} and {\sf (4)}, 
with $i=1$)).         
Hence,
for $i=2$,
{
\small
\begin{eqnarray*}
      {\mathsf{U}}_{\overline{i}}
               ({\bm{\sigma}}_{- \overline{i}}
                \diamond
                \ell_{j})   
& = & \sigma_{i}\left( \left\{ \ell_{j},
                                              \ell_{j+1},
                                              \overline{\ell}_{j+1}
                                   \right\}
                          \right)\, ,
\end{eqnarray*}
}
while for $i=1$,
{
\small
\begin{eqnarray*}
      {\mathsf{U}}_{\overline{i}}
               ({\bm{\sigma}}_{- \overline{i}}
                \diamond
                \ell_{j})   
& = & \sigma_{i}\left( \left\{ \ell_{j-2},
                                              \overline{\ell}_{j-2},
                                              \ell_{j-3},
                                              \overline{\ell}_{j-3}
                                   \right\}
                \right)\, ,
\end{eqnarray*}
}
By Lemma~\ref{same probability on pairs of literals},
these imply that
{
\small
\begin{eqnarray*}
         {\mathsf{U}}_{\overline{i}}
                  ({\bm{\sigma}}_{- \overline{i}}
                   \diamond
                   \ell_{j})   
& \leq & 2
         \cdot
         \frac{\textstyle \sigma_{i}({\mathsf{L}})}
              {\textstyle n}\, . 
\end{eqnarray*}
}
Consider now player $\overline{i}$
switching to a strategy
$s 
 \in
 \widehat{{\mathsf{\Sigma}}}_{\overline{i}}$.

Assume,
by way of contradiction,
that
$\sigma_{i}
 \left( {\mathcal{C}}
          \cup
          {\mathsf{V}}
 \right)
 >
 0$.
By the utility functions,
player $\overline{i}$
gets utility $1$
if player $i$ chooses a strategy
$s
 \in
 \left\{ \ell_{0}, 
            \overline{\ell}_{0},
            \ell_{1},
            \overline{\ell}_{1}
 \right\}
 \cup
 {\mathcal{C}}
 \cup
 {\mathsf{V}}$
(Case {\sf (10)}). 
Then,
{
\small
\begin{eqnarray*} 
         {\sf U}_{\overline{i}}
              ({\bm\sigma}_{-i}\diamond s)
\geq & \sigma_{i}\left( \left\{ \ell_{0}, 
                                               \overline{\ell}_{0},
                                               \ell_{1},
                                               \overline{\ell}_{1}
                                     \right\}
                          \right)       
       +
       \sigma_{i}\left( {\mathcal{C}}
                                 \cup
                                {\mathsf{V}}          
                       \right)
     &                                                                    \\
>    & \sigma_{i}\left( \left\{ \ell_{0}, 
                                              \overline{\ell}_{0},
                                              \ell_{1},
                                              \overline{\ell}_{1}
                                   \right\}
                 \right)   
     & \mbox{(since $\sigma_{i}\left( {\mathcal{C}},
                                      {\mathsf{V}} 
                               \right)
                     >
                     0$)}                                                \\             
\geq & 2
       \cdot
       \frac{\textstyle \sigma_{2}({\mathsf{L}})}
            {\textstyle n}\, .
     & \mbox{(by Lemma~\ref{same probability on pairs of literals})}\, .
\end{eqnarray*}
} 
By Lemma~\ref{literals are always played},
$\sigma_{2}({\mathsf{L}})
 >
 0$. 
This implies that
{
\small
\begin{eqnarray*}
      {\mathsf{U}}_{\overline{i}}
               ({\bm{\sigma}}_{- \overline{i}}
                \diamond
                s)
& > & {\mathsf{U}}_{\overline{i}}
               ({\bm{\sigma}}_{- \overline{i}}
                \diamond
                \ell_{j})\, .
\end{eqnarray*}
}                
Since
$\bm{\sigma}$ is a Nash equilibrium
and
$\ell_{j}
 \in
 {\mathsf{Supp}}(\sigma_{\overline{i}})$,
Lemma~\ref{basic property of mixed nash equilibria}
(Condition {\sf (2)})
implies that
{
\small
\begin{eqnarray*}
      {\mathsf{U}}_{\overline{i}}
               ({\bm{\sigma}}_{- \overline{i}}
                \diamond
                s)
&\leq& {\mathsf{U}}_{\overline{i}}
               ({\bm{\sigma}}_{- \overline{i}}
                \diamond
                \ell_{j})\, .
\end{eqnarray*}     
}
A contradiction.     
It follows that
$\sigma_{i}
 \left( {\mathcal{C}}
          \cup
         {\mathsf{V}}
 \right)
 =
 0$.

Assume, 
by way of contradiction, 
that 
$\sigma_{i}
 (\widehat{{\mathsf{\Sigma}}}_{i})
 >
 0$.
By the utility functions,
player $\overline{i}$
gets utility $1$
if player $i$ chooses a strategy
$s
 \in
 \{ \ell_{0}, 
    \overline{\ell}_{0},
    \ell_{1},
    \overline{\ell}_{1}
 \}
 \cup
 {\widehat{{\mathsf{\Sigma}}}}_{i}$
(Cases {\sf (8)} and {\sf (10)}). 
Then,
{
\small
\begin{eqnarray*} 
         {\sf U}_{\overline{i}}
              ({\bm\sigma}_{-i}\diamond s)
\geq & \sigma_{i}\left( \left\{ \ell_{0}, 
                                \overline{\ell}_{0},
                                \ell_{1},
                                \overline{\ell}_{1}
                        \right\}
                 \right)       
       +
       \sigma_{i}\left( {\widehat{{\mathsf{\Sigma}}}_{i}}          
                       \right)
     &                                                                    \\
>    & \sigma_{i}\left( \left\{ \ell_{0}, 
                                \overline{\ell}_{0},
                                \ell_{1},
                                \overline{\ell}_{1}
                        \right\}
                 \right)   
     & \mbox{(since $\sigma_{i}\left( {\widehat{{\mathsf{\Sigma}}}}_{i} 
                                                \right)
                     >
                     0$)}                                                \\             
\geq & 2
       \cdot
       \frac{\textstyle \sigma_{2}({\mathsf{L}})}
            {\textstyle n}\, .
     & \mbox{(by Lemma~\ref{same probability on pairs of literals})}\, .
\end{eqnarray*}
} 
By Lemma~\ref{literals are always played},
$\sigma_{2}({\mathsf{L}})
 >
 0$. 
This implies that
{
\small
\begin{eqnarray*}
      {\mathsf{U}}_{\overline{i}}
               ({\bm{\sigma}}_{- \overline{i}}
                \diamond
                s)
& > & {\mathsf{U}}_{\overline{i}}
               ({\bm{\sigma}}_{- \overline{i}}
                \diamond
                \ell_{j})\, .
\end{eqnarray*}
}                
Since
$\bm{\sigma}$ is a Nash equilibrium
and
$\ell_{j}
 \in
 {\mathsf{Supp}}(\sigma_{\overline{i}})$,
Lemma~\ref{basic property of mixed nash equilibria}
(Condition {\sf (2)})
implies that
{
\small
\begin{eqnarray*}
      {\mathsf{U}}_{\overline{i}}
               ({\bm{\sigma}}_{- \overline{i}}
                \diamond
                s)
&\leq& {\mathsf{U}}_{\overline{i}}
               ({\bm{\sigma}}_{- \overline{i}}
                \diamond
                \ell_{j})\, .
\end{eqnarray*}     
}
A contradiction.     
It \textcolor{black}{follows} that
$\sigma_{i}
 (\widehat{{\mathsf{\Sigma}}}_{i})
 =
 0$.

Since
$\sigma_{i}
 ({\mathcal{C}}
  \cup
  {\mathsf{V}})
 =
 \sigma_{i}
 (\widehat{{\mathsf{\Sigma}}}_{i})
 =
 0$,
it follows that
$\sigma_{i}
 ({\mathsf{L}})
 =
 1$,
and we are done.  
\end{proof}

\noindent
In view of Lemma~\ref{only literals are played},
a Nash equilibrium
${\bm\sigma}
\in
{\mathcal{NE}}({\sf G}) 
\setminus
{\mathcal{NE}}({\widehat{\sf G}})$
will be called a {\it literal equilibrium}.
The next property of such literal equilibria
is that
the two special players are playing literals
in a consistent way:
for each literal,
the two special players
are both playing
either the literal or its negation;
this will imply
that a literal equilibrium
induces an assignment
for ${\mathsf{\phi}}$.

\begin{lemma}
\label{no negations are played}
Fix a Nash equilibrium 
${\bm\sigma}
\in
{\cal NE}({\sf G}) 
\setminus
{\cal NE}({\widehat{\sf G}})$
and a literal $\ell \in {\mathsf{L}}$.
Then,
$\sigma_{1} (\ell)
 \cdot
 \sigma_{2}(\overline{\ell})
 =
 0$.
\end{lemma}

\begin{proof}
Lemmas~\ref{same probability on pairs of literals}
and~\ref{only literals are played}
imply together that
for each player $i\in [2]$
and for each literal
$\ell_{j}
 \in
 {\mathsf{L}}$,
$\sigma_{i}\left( \left\{ \ell_{j},
                                      \overline{\ell}_{j}
                           \right\}
                 \right)                     
 =
 \frac{\textstyle 1}
         {\textstyle n}$.
Assume, 
by way of contradiction,
that
$\sigma_{1}(\ell_{j})
 \cdot
 \sigma_{2}(\overline{\ell}_{j})
 >
 0$.
This implies that
$\sigma_{2}(\ell_{j})
 <
 \frac{\textstyle 1}
      {\textstyle n}$.
Hence,
by the utility functions
(Cases {\sf (1)}, {\sf (2)} and {\sf (4)}), 
{
\small
\begin{eqnarray*}
    \lefteqn{{\mathsf{U}}_{1}({\bm{\sigma}})}                                                   \\
= & {\mathsf{U}}_{1}({\bm{\sigma}}
                     \diamond
                     \ell_{j})
  & \mbox{(by Lemma~\ref{basic property of mixed nash equilibria} (Condition {\sf (1)}))}       \\
= & \sigma_{2}(\ell_{j}) 
    + 
    \sigma_{2}\left( \left\{ \ell_{j+1},
                             \overline{\ell}_{j+1}
                     \right\}
              \right)
  &                                                                                             \\              
< & \frac{\textstyle 2}
           {\textstyle n}\, .
  & 
\end{eqnarray*}
}
Consider now player $1$
switching to
a strategy 
$s
 \in
 \widehat{{\sf\Sigma}}_{1}$.
Then,
by the utility functions
(Case {\sf (10)}),
player $1$ gets utility $1$
if player $2$
chooses a literal from
$\left\{ \ell_{0},
         \overline{\ell}_{0},
         \ell_{1},
         \overline{\ell}_{1}
 \right\}$.
So,      
{
\small
\begin{eqnarray*}
         {\mathsf{U}}_{1}
                  \left( {\bm{\sigma}}_{-1}
                         \diamond
                         s
                  \right)
& \geq & \sigma_{2} \left( \left\{ \ell_{0},
                                                     \overline{\ell}_{0},
                                                     \ell_{1},
                                                    \overline{\ell}_{1}
                           \right\}        
                  \right)\ \
 =\ \ 
 \frac{\textstyle 2}
         {\textstyle n}\, .
\end{eqnarray*}
}
A contradiction.
\end{proof}

\noindent
We finally conclude:

\begin{corollary}
\label{nash is assignment}
Fix a Nash equilibrium 
${\bm\sigma}
\in
{\cal NE}({\sf G}) 
\setminus
{\cal NE}({\widehat{\sf G}})$.
Then,
the following conditions hold:
\begin{enumerate}

\item[{\sf (C.1)}]
$\bm{\sigma}$
induces an assignment
${\mathsf{\gamma}}(\bm{\sigma})$
for ${\mathsf{\phi}}$.

\item[{\sf (C.2)}]
For each player $i \in [2]$,
$\sigma_{i}$
is a uniform distribution
on ${\mathsf{\gamma}}(\bm{\sigma})$.

\item[{\sf (C.3)}]
${\mathsf{U}}_{1}({\bm{\sigma}})
 =
 {\mathsf{U}}_{2}({\bm{\sigma}})
 =
 \frac{\textstyle 2}
      {\textstyle n}$.

\end{enumerate}
\end{corollary}

\begin{proof}
Condition {\sf (C.1)}
follows from Lemma~\ref{no negations are played}.
Condition {\sf (C.2)}
follows from
Condition {\sf (C.1)}
and
Lemma~\ref{same probability on pairs of literals}.
Condition {\sf (C.3)}
follows from
Condition {\sf (C.2)}
and the utility functions
(Cases {\sf (2)}, {\sf (3)} and {\sf (4)}).
\end{proof}

\subsection{Proof of the Reduction}
\label{reduction proof}

\noindent
We are now ready 
to prove:

\begin{proposition}
\label{if unsatisfied}
Assume that 
${\mathsf{\phi}}$
is unsatisfiable.
Then,
${\mathcal{NE}}({\mathsf{G}})
 =
 {\mathcal{NE}}({\widehat{\mathsf{G}}})$.
\end{proposition}

\noindent
We prepare the reader that
the property that each clause ${\mathsf{c}}$ of ${\mathsf{\phi}}$
contains $K \geq 3$ literals
is necessary for the proof
in order to render the deviation to ${\mathsf{c}}$
profitable in a Nash equilibrium.
The assumption that ${\mathsf{\phi}}$
is a {\sf 3SAT} formula
guarantees this property.

\begin{proof}
By Lemma~\ref{f is PNE}
(Condition {\sf (C.1)}), 
${\mathcal{NE}}({\widehat{\mathsf{G}}})
 \subseteq
 {\mathcal{NE}}({\mathsf{G}})$.
So,
assume, by way of contradiction,
that there is a Nash equilibrium 
${\bm\sigma}
 \in
 {\mathcal{NE}}({\mathsf{G}})
 \setminus
 {\mathcal{NE}}({\widehat{\mathsf{G}}})$.
By Corollary~\ref{nash is assignment}
(Conditions {\sf (C.1)} and {\sf (C.3)}),
$\bm{\sigma}$ induces an assignment ${\mathsf{\gamma}}$
for ${\mathsf{\phi}}$,
and
for each player $i \in [2]$,
${\mathsf{U}}_{i}
 ({\bm{\sigma}})
 =
 \frac{\textstyle 2}
         {\textstyle n}$.
Denote as ${\mathsf{c}}$
the clause not satisfied by ${\mathsf{\gamma}}$.
Consider a special player $i \in [2]$
switching to ${\mathsf{c}}$.
Since ${\mathsf{\phi}}$
is a {\sf 3SAT} formula,
there are three literals 
in ${\mathsf{c}}$; 
for each such literal $\ell$,
Corollary~\ref{nash is assignment}
(Condition {\sf (C.2)})
implies that
$\sigma_{\overline{i}}(\overline{\ell})
 =
 \frac{\textstyle 1}
      {\textstyle n}$.
Thus,
${\mathsf{U}}_{i}(\bm{\sigma}_{-i}
                  \diamond
                  {\mathsf{c}})
 =
 \frac{\textstyle 3}
      {\textstyle n}$.
So,
${\mathsf{U}}_{i}(\bm{\sigma}_{-i}
                  \diamond
                  {\mathsf{c}})
 >
 {\mathsf{U}}_{i}({\bm{\sigma}})$.
A contradiction to 
Lemma~\ref{basic property of mixed nash equilibria}
(Condition {\sf (2)}).
\end{proof}

\noindent
We continue to prove:

\begin{proposition}
\label{final lemma}
Assume that ${\mathsf{\phi}}$
is satisfiable.
Then,
for each satisfying assignment
${\mathsf{\gamma}}$,
${\mathsf{G}}$
has a symmetric
Nash equilibrium
${\bm{\sigma}}
  =
  {\bm{\sigma}}({\mathsf{\gamma}})
 \in
 {\mathcal{NE}}({\mathsf{G}})
 \setminus
 {\mathcal{NE}}({\widehat{\mathsf{G}}})$
such that
for each player
$i \in [2]$:
\begin{enumerate}

\item[{\sf (C.1)}]
$|{\mathsf{Supp}}(\sigma_i)| = n$.

\item[{\sf (C.2)}]
${\mathsf{U}}_{i}({\bm{\sigma}})
 =
 \frac{\textstyle 2}
         {\textstyle n}$.

\item[{\sf (C.3)}]
${\mathsf{Supp}}(\sigma_{i})
 \cap
 \widehat{\mathsf{\Sigma}}_{i}
 =
 \emptyset$.

\item[{\sf (C.4)}]
For each literal
$\ell
 \in
 {\mathsf{L}}$
set to true by
${\mathsf{\gamma}}$, 
$\ell 
 \in
 {\mathsf{Supp}}(\sigma_{i})$
with
$\sigma_{i}(\ell)
 =
 \frac{\textstyle 1}
      {\textstyle n}$.

\end{enumerate}
\end{proposition}

\noindent
We prepare the reader that
the property that each clause ${\mathsf{c}}$ of ${\mathsf{\phi}}$
contains $K \leq 3$ literals
is necessary for the proof
in order to render the deviation 
of a player to ${\mathsf{c}}$
non-profitable in a Nash equilibrium.
The assumption that ${\mathsf{\phi}}$
is a {\sf 3SAT} formula
guarantees this property.
Clearly,
the assumption that ${\mathsf{\phi}}$
is a {\sf 3SAT} formula
is the unique assumption 
that simultaneously guarantees
both this property
and the property needed for
the proof of Proposition~\ref{if unsatisfied}.

\begin{proof}
Construct the mixed profile
$\bm{\sigma}$
where,
for each special player
$i \in [2]$,
for each literal
$\ell
 \in
 {\mathsf{L}}$
set to true by
${\mathsf{\gamma}}$, 
$\sigma_{i}(\ell)
 :=
 \frac{\textstyle 1}
      {\textstyle n}$;
for each non-special player
$i \in [r]
       \setminus
       [2]$,
$\sigma_{i}({\mathsf{\delta}})
 :=
 1$.             
Clearly,
Conditions
{\sf (C.1)},
{\sf (C.3)}
and
{\sf (C.4)}
hold by construction;
Condition {\sf (C.2)}
holds by the utility functions
(Cases {\sf (2)}, {\sf (3)} and {\sf (4)}). 
It remains to prove that
$\bm{\sigma}$ is a Nash equilibrium.
There are four possible deviations
for a special player $i \in [2]$:
\begin{itemize}

\item
\underline{To a strategy 
           $s \in
            {\widehat{{\mathsf{\Sigma}}}}$:}
By the utility functions
(Case {\sf (10)}),
player $i$
gets utility $1$
for two possible choices
of literals 
by player $\overline{i}$:
$\ell_{0}$ (resp., $\overline{\ell}_{0}$)
and
$\ell_{1}$ (resp., $\overline{\ell}_{1}$).
Hence,
${\mathsf{U}}_{i}\left( \bm{\sigma}_{-i}
                        \diamond
                        s
                 \right)       
 =
 \frac{\textstyle 2}
        {\textstyle n}$.

\item
\underline{To a literal 
           $\ell \in {\mathsf{L}}$:}
By the utility functions
(Cases {\sf (1)}, {\sf (2)}, {\sf (3)} and {\sf (4)}),
player $i$
gets utility $1$
for at most two possible choices
of literals
by player $\overline{i}$.
Hence,
${\mathsf{U}}_{i}\left( \bm{\sigma}_{-i}
                        \diamond
                        s
                 \right)       
 \leq
 \frac{\textstyle 2}
      {\textstyle n}$.

\item
\underline{To a strategy
           ${\mathsf{v}}_{j,k}
            \in
            {\mathsf{V}}$:}
By the utility functions
(Case {\sf (5)}),
player $i$ gets utility $1$
for two possible choices
of literals
by player $\overline{i}$:
$\ell_{j}$ or $\overline{\ell}_{j}$,
and $\ell_{k}$
or $\overline{\ell}_{k}$.
Hence,
${\mathsf{U}}_{i}\left( \bm{\sigma}_{-i}
                        \diamond
                        {\mathsf{v}}_{j,k}
                 \right)       
 =
 \frac{\textstyle 2}
      {\textstyle n}$.

\item
\underline{To a clause
           ${\mathsf{c}}
            =
            \{ \ell, 
               \ell^{\prime},
               \ell^{\prime\prime}
            \}
            \in
            {\mathcal{C}}$:} 
Since ${\mathsf{\gamma}}$
is a satisfying assignment,
${\mathsf{\gamma}}$
satisfies ${\mathsf{c}}$.
So,
at least one of the literals
from
$\{ \ell, 
    \ell^{\prime},
    \ell^{\prime\prime}
 \}$
is in ${\mathsf{\gamma}}(\bm{\sigma})$.
Since ${\mathsf{\phi}}$
is a {\sf 3SAT} formula,
this implies that
at most two of the literals
from
$\left\{  \overline{\ell}, 
              \overline{\ell}^{\prime},
              \overline{\ell}^{\prime\prime}
 \right\}$ 
are not in ${\mathsf{\gamma}}(\bm{\sigma})$. 
Hence,
by construction,
$\sigma_{\overline{i}}
 \left( \left\{ \overline{\ell}, 
                \overline{\ell}^{\prime},
                \overline{\ell}^{\prime\prime}
        \right\}
 \right)
 \leq
 \frac{\textstyle 2}
      {\textstyle n}$.
By the utility functions
(Case {\sf (6)}),
player $i$ gets utility $1$
if and only if
player $\overline{i}$
chooses a literal from
$\left\{ \overline{\ell}, 
                \overline{\ell}^{\prime},
                \overline{\ell}^{\prime\prime}
        \right\}$.      
It follows that
${\mathsf{U}}_{i}\left( \bm{\sigma}_{-i}
                        \diamond
                        {\mathsf{c}}
                 \right)       
 \leq
 \frac{\textstyle 2}
      {\textstyle n}$.

\end{itemize}
Hence,
player $i \in [2]$
cannot improve by switching.
Consider now a non-special player
$i \in [r] \setminus [2]$.
By construction,
$\sigma_{i}({\mathsf{\delta}})
 =
 1$;
hence, 
it follows,
by the utility functions
(Cases {\sf (2)}, {\sf (3)} and {\sf (4)}),
that
${\mathsf{U}}_{i}\left( \bm{\sigma}
                 \right)
 =
 1$,
and player $i$
cannot improve by switching.                  
So 
$\bm{\sigma}$
is a Nash equilibrium.
\end{proof}

\noindent
Fix a win-lose game $\widehat{\mathsf{G}}$
with the positive utility property.
By Propositions~\ref{if unsatisfied} and~\ref{final lemma},
${\mathsf{G}}({\widehat{{\mathsf{G}}}},
                         {\mathsf{\phi}})$
and
${\widehat{\mathsf{G}}}$
are Nash-equivalent
if and only if
${\mathsf{\phi}}$ is unsatisfiable.
So 
${\mathsf{3SAT}}$
reduces in polynomial time to
$\overline{\mbox{{\sf{NASH-EQUIVALENCE}}$(\widehat{\mathsf{G}})$}}$,
so that,
restricted to win-lose games,
{\sf NASH-EQUIVALENCE}$(\widehat{\mathsf{G}})$
is co-${\mathcal{NP}}$-hard.
We shall later \textcolor{black}{strengthen}
this  co-${\mathcal{NP}}$-hardness result
to the more restricted class
of symmetric win-lose bimatrix games
(Theorem~\ref{maintheorem symmetric}).
By Propositions~\ref{if unsatisfied}
and~\ref{final lemma},
{
\small
\begin{eqnarray*}
               \left| {\mathcal{NE}} \left( {\mathsf{G}}
                                                 \right)
               \right|
&  = &    \left| {\mathcal{NE}} \left( {\widehat{{\mathsf{G}}}}
                                                 \right)
               \right| 
               +
               \# {\mathsf{\phi}}\, .
\end{eqnarray*}
}
Note that
\textcolor{black}{$| {\mathcal{NE}} ( {\widehat{{\mathsf{G}}}}
                              )
 |$}
is a fixed constant
when ${\widehat{{\mathsf{G}}}}$
is independent of ${\mathsf{\phi}}$.   
Since computing $\# {\mathsf{\phi}}$
\textcolor{black}{(resp., $\oplus {\mathsf{\phi}}$)}
for a {\sf 3SAT} formula ${\mathsf{\phi}}$
is $\# {\mathcal{P}}$-hard~\cite{V79}
\textcolor{black}{(resp., $\oplus {\mathcal{P}}$-hard~\cite{PZ83}),}
it follows that
\textcolor{black}{computing} the number 
\textcolor{black}{(resp., the parity of the number)}
of Nash equilibria for
a win-lose bimatrix game
is $\# {\mathcal{P}}$-hard
\textcolor{black}{(resp., $\oplus {\mathcal{P}}$-hard).} 
We shall later \textcolor{black}{strengthen} 
the $\# {\mathcal{P}}$-hardness
\textcolor{black}{and the $\oplus {\mathcal{P}}$-hardness}
to the more restricted class
of symmetric win-lose bimatrix games
(Theorem~\ref{sharp pi complete symmetric win-lose}).

\section{The Win-Lose ${\mathsf{GHR}}$-Symmetrization}
\label{winlose gkt symmetrization}

\noindent
For a given integer $n \geq 2$, 
denote as
${\mathsf{0}}^n$
the $n$-dimensional null function
${\mathsf{0}}^{n}:
  [n]
  \rightarrow
  \{ 0 \}$;
it is represented as the $n$-dimensional null vector
in the natural way.
Denote as ${\mathsf{0}}^{n\times n}$
the $n\times n$ null matrix.
For a positive $n$-dimensional vector
$\bm{\rho} \in [0, 1]^{n}$,
denote as $\bm{\varrho}$
the {\it normalization}
of $\bm{\rho}$
to a probability vector;
so,
for each index $j \in [n]$,
$\varrho (j)
 =
 \frac{\textstyle \rho(j)}
      {\textstyle \sum_{k \in [n]}
                    \rho (k)}$.

\noindent                    
Given the bimatrix game 
${\mathsf{G}}
 =
 \left\langle {\mathsf{R}},
              {\mathsf{C}}
 \right\rangle$
and two functions
${\mathsf{f}}, {\mathsf{g}}: [n]
                             \rightarrow
                             \mathbb{R}$,
define
\begin{eqnarray*}
\label{U hat 1}
       \textcolor{black}{\overline{{\mathsf{U}}}}_{1}({\mathsf{f}}, {\mathsf{g}})
& := & \sum_{s_{1}, s_{2}
             \in
             [n]}
         {\mathsf{f}}(s_{1})
         \cdot
         {\mathsf{g}}(s_{2})
         \cdot
         {\mathsf{R}}[s_{1}, s_{2}]
\end{eqnarray*}
and
\begin{eqnarray*}
\label{U hat 2}
       \textcolor{black}{\overline{{\mathsf{U}}}}_{2}({\mathsf{f}}, {\mathsf{g}})
& := & \sum_{s_{1}, s_{2}
             \in
             [n]}
         {\mathsf{f}}(s_{1})
         \cdot
         {\mathsf{g}}(s_{2})
         \cdot
         {\mathsf{C}}[s_{1}, s_{2}]\, .
\end{eqnarray*}
So,
$\textcolor{black}{\overline{{\mathsf{U}}}}_{1}({\mathsf{f}}, {\mathsf{g}})$
and
$\textcolor{black}{\overline{{\mathsf{U}}}}_{2}({\mathsf{f}}, {\mathsf{g}})$
become the expected utilities
of players 1 and 2,
respectively,
in 
${\mathsf{G}}$
when
${\mathsf{f}}$ and ${\mathsf{g}}$ are mixed strategies.

\subsection{Definition}
\label{the gkt symmetrization}

\begin{definition}
The \textcolor{black}{\emph{${\mathsf{GHR}}$-symmetrization~\cite{GHR63}}}
of the bimatrix game
$\left\langle {\mathsf{R}},
              {\mathsf{C}}
 \right\rangle$,
\textcolor{black}{denoted as
${\mathsf{GHR}}\left(  \left\langle {\mathsf{R}},
                                                        {\mathsf{C}}
                                      \right\rangle
                           \right)$,}  
is the symmetric bimatrix game
\textcolor{black}{$\langle {\mathsf{S}},
                    {\mathsf{S}}^{\mbox{{\rm T}}}
 \rangle$,}
where ${\mathsf{S}}$
is the $2n \times 2n$ matrix
{
\small
\begin{eqnarray*}
      {\mathsf{S}}
& = & \left( \begin{array}{cc}
               {\mathsf{0}}^{n \times n}     & {\mathsf{R}}              \\
               {\mathsf{C}}^{\mbox{{\rm T}}} & {\mathsf{0}}^{n \times n}   
             \end{array}
      \right)\, ;
\end{eqnarray*}
}
so,
{
\small
\begin{eqnarray*}
      {\mathsf{S}}^{\mbox{{\rm T}}}
& = & \left( \begin{array}{cc}
               {\mathsf{0}}^{n \times n}     & {\mathsf{C}}              \\
               {\mathsf{R}}^{\mbox{{\rm T}}} & {\mathsf{0}}^{n \times n}   
             \end{array}
      \right)\, .
\end{eqnarray*}
}
\end{definition}

\noindent
\textcolor{black}{Note that the ${\mathsf{GHR}}$-symmetrization
of
$\left\langle {\mathsf{R}},
                    {\mathsf{C}}
 \right\rangle$ 
is polynomial time computable. 
For shorter notation,
denote
${\widetilde{{\mathsf{G}}}}
  :=
 {\mathsf{GHR}}\left(  \left\langle {\mathsf{R}},
                                                        {\mathsf{C}}
                                     \right\rangle
                           \right)$.}   
\textcolor{black}{Clearly,
when
$\left\langle {\mathsf{R}},
                    {\mathsf{C}}
 \right\rangle$
is a win-lose game,
so also is
${\mathsf{GHR}}\left(  \left\langle {\mathsf{R}},
                                                        {\mathsf{C}}
                                     \right\rangle
                           \right)$.
To emphasize this property,
we shall term
the ${\mathsf{GHR}}$-symmetrization
as the {\it  win-lose ${\mathsf{GHR}}$-symmetrization.}}    
We further observe:

\begin{lemma}
\label{little park kafe}
\textcolor{black}{The win-lose ${\mathsf{GHR}}$-symmetrization
preserves the 
positive utility property.}
\end{lemma}

\noindent
Given a mixed profile
$\bm{\sigma}
 =
 \left\langle \sigma_{1},
              \sigma_{2}
 \right\rangle$
for
\textcolor{black}{${\widetilde{{\mathsf{G}}}}$,}
denote,
for each player $i \in [2]$,
as $\overrightarrow{\sigma_{i}}$
(resp., $\overleftarrow{\sigma_{i}}$)
the restriction of $\sigma_{i}$
to the last (resp., first) $n$ strategies
of player $i$,
so that for each strategy $j \in [n]$,
$\overleftarrow{\sigma_{i}}(j)
  =
  \sigma_{i}(j)$
and
$\overrightarrow{\sigma_{i}}(j)
  =
  \sigma_{i}(n+j)$.
Thus,
$\sigma_{i}
  =
  \overleftarrow{\sigma_{i}}
  \circ
  \overrightarrow{\sigma_{i}}$ is
the {\it concatenation} of   
$\overleftarrow{\sigma_{i}}$
and
$\overrightarrow{\sigma_{i}}$,
\textcolor{black}{which are its first and second {\it component,}
respectively;}
$\bm{\sigma}
 =
 \left\langle  \overleftarrow{\sigma_{1}}
                    \circ
                    \overrightarrow{\sigma_{1}},
                    \overleftarrow{\sigma_{2}}
                    \circ
                    \overrightarrow{\sigma_{2}}
  \right\rangle$,
\textcolor{black}{where
$\overleftarrow{\sigma_{1}}
  \circ
  \overrightarrow{\sigma_{1}}$
and  
$\overleftarrow{\sigma_{2}}
  \circ
  \overrightarrow{\sigma_{2}}$
are the first and second {\it entry,}
respectively,
of ${\bm{\sigma}}$.}                                                      
See Figure~\ref{vittorio}
for a pictorial representation
of \textcolor{black}{the components
$\overleftarrow{\sigma_{1}}$ and
$\overrightarrow{\sigma_{1}}$
(resp., 
$\overleftarrow{\sigma_{2}}$
and
$\overrightarrow{\sigma_{2}}$)}
and their association with
the blocks
${\mathsf{R}}$ and
${\mathsf{C}}^{\mbox{\rm T}}$
(resp.,
 ${\mathsf{C}}$ and
${\mathsf{R}}^{\mbox{\rm T}}$)
of ${\mathsf{S}}$
(resp., ${\mathsf{S}}^{\mbox{\rm T}}$).
We observe:

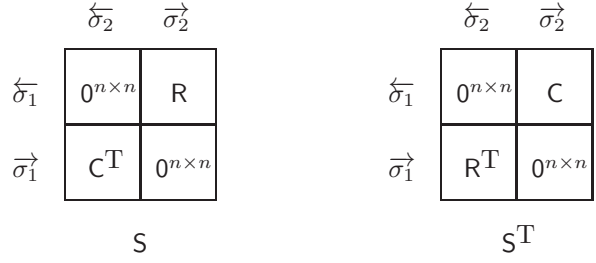
\begin{figure}[tp]
\begin{center}
\setlength{\unitlength}{1cm}
\begin{small}
\begin{picture}(8,4)(0,0)

\put(1,1){\framebox(1,1)}
\put(2,1){\framebox(1,1)}
\put(6,1){\framebox(1,1)}
\put(7,1){\framebox(1,1)}
\put(1,2){\framebox(1,1)}
\put(2,2){\framebox(1,1)}
\put(6,2){\framebox(1,1)}
\put(7,2){\framebox(1,1)}

\put(1.9,0.3){${\mathsf{S}}$}
\put(6.8,0.3){${\mathsf{S}}^{\mbox{\rm{T}}}$}
\put(0.3,1.3){$\overrightarrow{\sigma_{1}}$}
\put(0.3,2.3){$\overleftarrow{\sigma_{1}}$}
\put(5.3,1.3){$\overrightarrow{\sigma_{1}}$}
\put(5.3,2.3){$\overleftarrow{\sigma_{1}}$}
\put(2.3,3.3){$\overrightarrow{\sigma_{2}}$}
\put(1.3,3.3){$\overleftarrow{\sigma_{2}}$}
\put(7.3,3.3){$\overrightarrow{\sigma_{2}}$}
\put(6.3,3.3){$\overleftarrow{\sigma_{2}}$}
\put(1.3,1.3){${\mathsf{C}}^{\mbox{\rm T}}$}
\put(2.2,1.3){${\mathsf{0}}^{n\times n}$}
\put(6.3,1.3){${\mathsf{R}}^{\mbox{\rm T}}$}
\put(7.2,1.3){${\mathsf{0}}^{n\times n}$}
\put(1.2,2.3){${\mathsf{0}}^{n\times n}$}
\put(2.4,2.3){${\mathsf{R}}$}
\put(6.2,2.3){${\mathsf{0}}^{n\times n}$}
\put(7.4,2.3){${\mathsf{C}}$}
\end{picture}
\end{small}
\caption{The bimatrix game
         \textcolor{black}{${\widetilde{{\mathsf{G}}}}
                                                    =
                                                    \langle {\mathsf{S}},
                                                                {\mathsf{S}}^{\mbox{\rm T}}
                                                    \rangle$}
         with the mixed profile
         $\bm{\sigma}$.}
\label{vittorio}
\end{center}
\end{figure}


\begin{observation}
\label{vittorio 1}
The expected utilities
\textcolor{black}{${\widetilde{{\mathsf{U}}}}_{1}({\bm\sigma})$}
and
\textcolor{black}{${\widetilde{{\mathsf{U}}}}_{2}({\bm\sigma})$}  
are:
{
\small
\begin{eqnarray*}
          \mbox{{\sf (Eq. 1)}}\ \ \
          \textcolor{black}{{\widetilde{{\mathsf{U}}}}}_{1}({\bm\sigma}) 
& = & \sum_{s_{1}, s_{2} \in [n]}
             \overleftarrow{\sigma_{1}}(s_{1})\,
             \overrightarrow{\sigma_{2}}(s_{2})
             \cdot
            {\mathsf{R}}[s_{1}, s_{2}]
           +
           \sum_{s_{1}, s_{2} \in [n]}
             \overrightarrow{\sigma_{1}}(s_{1})\,
             \overleftarrow{\sigma_{2}}(s_{2})
             \cdot
             {\mathsf{C}}[s_{2}, s_{1}]\, .
\end{eqnarray*}
}
{
\small
\begin{eqnarray*}
          \mbox{{\sf (Eq. 2)}}\ \ \
          \textcolor{black}{{\widetilde{{\mathsf{U}}}}}_{2}({\bm\sigma}) 
& = & \sum_{s_{1}, s_{2} \in [n]}
        \overleftarrow{\sigma_{1}}(s_{1})\,
        \overrightarrow{\sigma_{2}}(s_{2})
        \cdot
        {\mathsf{C}}[s_{1},s_{2}]
      +
      \sum_{s_{1}, s_{2} \in [n]}
        \overrightarrow{\sigma_{1}}(s_{1})\,
        \overleftarrow{\sigma_{2}}(s_{2})
        \cdot
        {\sf R}[s_{2},s_{1}]\, . 
\end{eqnarray*}
}
\end{observation}

\begin{proof}
It is evident from Figure~\ref{vittorio}
that
{
\small
\begin{eqnarray*}
      \textcolor{black}{{\widetilde{{\mathsf{U}}}}}_{1}({\bm\sigma}) 
& = & \mathbb{E}_{(s_{1},s_{2}) 
                  \sim
                  {\bm\sigma}}
              {\mathsf{S}}[s_{1},s_{2}]                              \\
& = & \mathbb{E}_{(s_1,s_2)
                  \sim
                  (\overleftarrow{\sigma_{1}},
                   \overrightarrow{\sigma_{2}}
                  )
                 }
              {\mathsf R}[s_{1},s_{2}]
      +
      \mathbb{E}_{(s_{1},s_{2})
                  \sim
                  (\overrightarrow{\sigma_{1}},
                   \overleftarrow{\sigma_{2}}
                  )
                 }
              {\mathsf{C}}^{\mbox{\rm{T}}}[s_1,s_2]                  \\
& = & \mathbb{E}_{(s_{1},s_{2})
                  \sim
                  (\overleftarrow{\sigma_{1}},
                   \overrightarrow{\sigma_{2}}
                  )
                 }
              {\mathsf R}[s_{1},s_{2}]
      +
      \mathbb{E}_{(s_{1},s_{2})
                  \sim
                  (\overrightarrow{\sigma_{1}},
                   \overleftarrow{\sigma_{2}}
                  )
                 }
              {\sf C}[s_{2},s_{1}]                                  \\
& = & \sum_{s_{1}, s_{2} \in [n]}
        \overleftarrow{\sigma_{1}}(s_{1})\,
        \overrightarrow{\sigma_{2}}(s_{2})
        \cdot
        {\mathsf{R}}[s_{1}, s_{2}]
      +
      \sum_{s_{1}, s_{2} \in [n]}
        \overrightarrow{\sigma_{1}}(s_{1})\,
        \overleftarrow{\sigma_{2}}(s_{2})
        \cdot
        {\mathsf{C}}[s_{2}, s_{1}]\, ,
\end{eqnarray*}
}
and
{
\small
\begin{eqnarray*}
      \textcolor{black}{{\widetilde{{\sf U}}}}_{2}({\bm\sigma}) 
& = & \mathbb{E}_{(s_{1},s_{2}) 
                  \sim
                  {\bm\sigma}}
              {\mathsf{S}}^{\mbox{\rm{T}}}[s_{1},s_{2}]              \\
& = & \mathbb{E}_{(s_1,s_2)
                  \sim
                  (\overleftarrow{\sigma_{1}},
                   \overrightarrow{\sigma_{2}}
                  )
                 }
              {\mathsf C}[s_{1},s_{2}]
      +
      \mathbb{E}_{(s_{1},s_{2})
                  \sim
                  (\overrightarrow{\sigma_{1}},
                   \overleftarrow{\sigma_{2}}
                  )
                 }
              {\mathsf{R}}^{\mbox{\rm{T}}}[s_1,s_2]                  \\
& = & \mathbb{E}_{(s_{1},s_{2})
                  \sim
                  (\overleftarrow{\sigma_{1}},
                   \overrightarrow{\sigma_{2}}
                  )
                 }
              {\mathsf C}[s_{1},s_{2}]
      +
      \mathbb{E}_{(s_{1},s_{2})
                  \sim
                  (\overrightarrow{\sigma_{1}},
                   \overleftarrow{\sigma_{2}}
                  )
                 }
              {\sf R}[s_{2},s_{1}]                                   \\
& = & \sum_{s_{1}, s_{2} \in [n]}
        \overleftarrow{\sigma_{1}}(s_{1})\,
        \overrightarrow{\sigma_{2}}(s_{2})
        \cdot
        {\mathsf{C}}[s_{1},s_{2}]
      +
      \sum_{s_{1}, s_{2} \in [n]}
        \overrightarrow{\sigma_{1}}(s_{1})\,
        \overleftarrow{\sigma_{2}}(s_{2})
        \cdot
        {\sf R}[s_{2},s_{1}]\, , 
\end{eqnarray*}
}
as needed.
\end{proof}

\noindent
By Observation~\ref{vittorio 1},
it immediately follows:

\begin{observation}               
\label{vittorio 2}
For a given strategy $s \in [2n]$,
the conditional expected utilities are:
{
\small
\begin{eqnarray*}
           \mbox{{\sf (Eq. 1)}}\ \ \
          \textcolor{black}{{\widetilde{{\mathsf{U}}}}}_{1}({\bm\sigma}_{-1} \diamond s) 
& = & \left\{ \begin{array}{ll}
                \sum_{s_{2} \in [n]}
                  \overrightarrow{\sigma_{2}}
                                   (s_{2}) 
                  \cdot
                  {\mathsf{R}}[s, s_{2}]\, ,                 & s \in [n]                     \\
                \sum_{s_{2} \in [n]}
                  \overleftarrow{\sigma_{2}}
                                  (s_{2})
                  \cdot
                  {\mathsf{C}}[s_{2}, s-n]\, ,              & s \in [2n] \setminus [n]
              \end{array}
      \right.\, ,                                                         
\end{eqnarray*}
}
{
\small
\begin{eqnarray*}
       \mbox{{\sf (Eq. 2)}}\ \ \
      \textcolor{black}{{\widetilde{{\mathsf{U}}}}}_{2}({\bm\sigma}_{-2} \diamond s) 
& = & \left\{ \begin{array}{ll}
                \sum_{s_{1} \in [n]}
                  \overrightarrow{\sigma_{1}}
                        (s_{1})
                  \cdot
                  {\mathsf{R}}[s, s_{1}]\, ,              & s \in [n]                      \\
                \sum_{s_{1} \in [n]}
                  \overleftarrow{\sigma_{1}}
                        (s_{1}) 
                  \cdot
                  {\mathsf{C}}[s_{1}, s-n]\, ,                & s \in [2n] \setminus [n]  
              \end{array}
      \right.\, .                      
\end{eqnarray*}
}
\end{observation}

\noindent
We now prove a characterization of
the supports of a Nash equilibrium
for the win-lose ${\mathsf{GHR}}$-symmetrization
of a bimatrix game
with the positive utility property.

\begin{lemma}
\label{noemptystrategies}
Fix a win-lose bimatrix game
\textcolor{black}{${\mathsf{G}}
  =
  \left\langle {\mathsf{R}},
              {\mathsf{C}}
 \right\rangle$}
with the positive utility property,
and consider
its win-lose ${\mathsf{GHR}}$-symmetrization
\textcolor{black}{${\widetilde{{\mathsf{G}}}}$}
with a Nash equilibrium $\bm{\sigma}$.
Then,
there are only three possible cases:
\begin{enumerate}

\item[{\sf (C.1)}]
$\overleftarrow{\sigma_{1}},
 \overrightarrow{\sigma_{1}},
 \overleftarrow{\sigma_{2}},
 \overrightarrow{\sigma_{2}}
 \neq
 {\mathsf{0}}^n$.

\item[{\sf (C.2)}]
$\overleftarrow{\sigma_{1}},
 \overrightarrow{\sigma_{2}}
 =
 {\mathsf{0}}^n$
and
$\overrightarrow{\sigma_{1}},
 \overleftarrow{\sigma_{2}}
 \neq
 {\mathsf{0}}^n$.

\item[{\sf (C.3)}]
$\overrightarrow{\sigma_{1}},
 \overleftarrow{\sigma_{2}}
 =
 {\mathsf{0}}^n$
and
$\overleftarrow{\sigma_{1}},
 \overrightarrow{\sigma_{2}}
 \neq
 {\mathsf{0}}^n$.

\end{enumerate}
\end{lemma}

\begin{proof}
Since $\sigma_{1}$ and
$\sigma_{2}$
are mixed strategies,
the following four implications
hold trivially:
\underline{{\sf{(I.1)}}
$\overleftarrow{\sigma_{1}}
 =
 {\mathsf{0}}^{n}$
$\Rightarrow$
$\overrightarrow{\sigma_{1}}
 \neq
 {\mathsf{0}}^{n}$};
\underline{{\sf{(I.2)}}
$\overrightarrow{\sigma_{1}}
 =
 {\mathsf{0}}^{n}$
$\Rightarrow$
$\overleftarrow{\sigma_{1}}
 \neq
 {\mathsf{0}}^{n}$};
\underline{{\sf{(I.3)}}
$\overleftarrow{\sigma_{2}}
 =
 {\mathsf{0}}^{n}$
$\Rightarrow$
$\overrightarrow{\sigma_{2}}
 \neq
 {\mathsf{0}}^{n}$};
\underline{{\sf{(I.4)}}
$\overrightarrow{\sigma_{2}}
 =
 {\mathsf{0}}^{n}$
$\Rightarrow$
$\overleftarrow{\sigma_{2}}
 \neq
 {\mathsf{0}}^{n}$}.

We now prove:
\underline{{\sf{(I.5)}}
$\overleftarrow{\sigma_{1}}
 =
 {\mathsf{0}}^{n}$
$\Rightarrow$
$\overleftarrow{\sigma_{2}}
 \neq
 {\mathsf{0}}^{n}$:}
Assume,
by way of contradiction,
that
$\overleftarrow{\sigma_{1}}
 =
 {\mathsf{0}}^{n}$
and
$\overleftarrow{\sigma_{2}}
 =
 {\mathsf{0}}^{n}$.
Then,
no profile 
$\left\langle s_{1}, s_{2}
 \right\rangle$
is supported in
$\left\langle \overleftarrow{\sigma_{1}},
              \overrightarrow{\sigma_{2}}
 \right\rangle$
and
no profile 
$\left\langle s_{1}, s_{2}
 \right\rangle$
is supported in
$\left\langle \overrightarrow{\sigma_{1}},
              \overleftarrow{\sigma_{2}}
 \right\rangle$.
By the formula for
$\textcolor{black}{{\widetilde{{\mathsf{U}}}}}_{1}(\bm{\sigma})$
in Observation~\ref{vittorio 1} {\sf (Eq.\,1)},
these imply that
$\textcolor{black}{{\widetilde{{\mathsf{U}}}}}_{1}(\bm{\sigma})
 =
 0$.
By Lemmas~\ref{park kafe}
and~\ref{little park kafe},
$\textcolor{black}{\widetilde{{{\mathsf{U}}}}}_{1}(\bm{\sigma})
 >
 0$.
A contradiction.
In a corresponding way,
the following three implications are proved:
\underline{{\sf{(I.6)}}
$\overrightarrow{\sigma_{1}}
 =
 {\mathsf{0}}^{n}$
$\Rightarrow$
$\overrightarrow{\sigma_{2}}
 \neq
 {\mathsf{0}}^{n}$};
\underline{{\sf{(I.7)}}
$\overleftarrow{\sigma_{2}}
 =
 {\mathsf{0}}^{n}$
$\Rightarrow$
$\overleftarrow{\sigma_{1}}
 \neq
 {\mathsf{0}}^{n}$}; 
\underline{{\sf{(I.8)}}
$\overrightarrow{\sigma_{2}}
 =
 {\mathsf{0}}^{n}$
$\Rightarrow$
$\overrightarrow{\sigma_{1}}
 \neq
 {\mathsf{0}}^{n}$}.

We finally prove:
\underline{{\sf{(I.9)}}
$\overleftarrow{\sigma_{1}}
 =
 {\mathsf{0}}^{n}$
$\Rightarrow$
$\overrightarrow{\sigma_{2}}
 =
 {\mathsf{0}}^{n}$}.
Assume,
by way of contradiction,
that
$\overleftarrow{\sigma_{1}}
 =
 {\mathsf{0}}^{n}$
and
$\overrightarrow{\sigma_{2}}
 \neq
 {\mathsf{0}}^{n}$.
Then,
there is a strategy
$s_{2} \in [2n] \setminus [n]$
with
$\overrightarrow{\sigma_{2}}(s_{2} - n)
 >
 0$;
so,
$s_{2}
 \in
 {\mathsf{Supp}}(\sigma_{2})$. 
Since
$\overleftarrow{\sigma_{1}}
 =
 {\mathsf{0}}^{n}$,
it follows from
the formula for
$\textcolor{black}{{\widetilde{{\mathsf{U}}}}}_{2}\left( \bm{\sigma}
                 \right)$
in Observation~\ref{vittorio 2} {\sf (Eq. 2)}
with $\bm{\sigma} =
          \bm{\sigma}_{-2} \diamond s_{2}$                 
that
$\textcolor{black}{{\widetilde{{\mathsf{U}}}}}_{2}\left( \bm{\sigma}_{-2}
                        \diamond
                        s_{2}
                 \right)
 =
 0$.
Since $\bm{\sigma}$ is a Nash equilibrium
and
$s_{2}
 \in
 {\mathsf{Supp}}(\sigma_{2})$,
Lemma~\ref{basic property of mixed nash equilibria}
(Condition {\sf (1)})
implies that
$\textcolor{black}{{\widetilde{{\mathsf{U}}}}}_{2}\left( \bm{\sigma}
                 \right)
 =
 0$.
By Lemmas~\ref{park kafe} and~\ref{little park kafe},
$\textcolor{black}{{\widetilde{{\mathsf{U}}}}}_{2}\left( \bm{\sigma}
                 \right)
 >
 0$.
A contradiction.
In a corresponding way,
the following three implications are proved:
\underline{{\sf{(I.10)}}
$\overrightarrow{\sigma_{1}}
 =
 {\mathsf{0}}^{n}$
$\Rightarrow$
$\overleftarrow{\sigma_{2}}
 =
 {\mathsf{0}}^{n}$};
\underline{{\sf{(I.11)}}
$\overleftarrow{\sigma_{2}}
 =
 {\mathsf{0}}^{n}$
$\Rightarrow$
$\overrightarrow{\sigma_{1}}
 =
 {\mathsf{0}}^{n}$}; 
\underline{{\sf{(I.12)}}
$\overrightarrow{\sigma_{2}}
 =
 {\mathsf{0}}^{n}$
$\Rightarrow$
$\overleftarrow{\sigma_{1}}
 =
 {\mathsf{0}}^{n}$}.

\noindent
The claim follows now
from the twelve implications.
\end{proof}

\subsection{The Balanced Mixture}
\label{the gkt product}

\noindent
Fix a win-lose bimatrix game
$\left\langle {\mathsf{R}},
                    {\mathsf{C}}
  \right\rangle$
with the positive utility property. 
We revisit
a composition
among either
{\it (i)} 
a pair of Nash equilibria
for $\left\langle {\mathsf{R}},
                  {\mathsf{C}}
     \right\rangle$
or
{\it (ii)}
a Nash equilibrium
for $\left\langle {\mathsf{R}},
                  {\mathsf{C}}
     \right\rangle$
and the tuple of null vectors
$\langle {\mathsf{0}}^{n},
              {\mathsf{0}}^{n}
 \rangle$
from~\cite{JPT86}. 

\begin{center}
\fbox{
\begin{minipage}{6.0in}
\begin{definition}
\label{gkt product}
Given
${\bm{\rho}}
 =
 \left\langle \rho_{1},
                   \rho_{2}
 \right\rangle
 \in
 {\cal NE}(\left\langle {\mathsf{R}},
                        {\mathsf{C}}
           \right\rangle)
 \cup
 \{ \langle {\mathsf{0}}^{n},
                 {\mathsf{0}}^{n}
    \rangle
 \}$
and
\textcolor{black}{${\bm{\tau}}
 =
 \left\langle \tau_{1},
                   \tau_{2}
 \right\rangle                  
 \in
 {\cal NE}\left( \left\langle {\mathsf{R}},
                              {\mathsf{C}}
                 \right\rangle
          \right)$,}
the \textcolor{black}{{\em {\emph{\textbf{balanced mixture}}}}}
of $\bm{\rho}$
and 
\textcolor{black}{$\bm{\tau}$}
is the pair
{
\small
\textcolor{black}{
\begin{eqnarray*}
           {\bm\rho} \ast {\bm\tau}
& := & \left\langle \frac{\overline{{\mathsf{U}}}_{1}({\bm{\tau}})
                                      }
                                     {\overline{{\mathsf{U}}}_{1}({\bm{\tau}})
                                       +
                                       \overline{{\mathsf{U}}}_{2}({\bm{\rho}})
                                      }\,
                              \rho_{1}
                              \circ
                              \frac{\overline{{\mathsf{U}}}_{2}({\bm{\rho}})
                                      }
                                     {\overline{{\mathsf{U}}}_{1}({\bm{\tau}})
                                      +
                                       \overline{{\mathsf{U}}}_{2}({\bm{\rho}})
                                      }\, 
                              \tau_{2}\, ,     
                              \frac{\overline{{\mathsf{U}}}_{1}({\bm{\rho}})
                                      }
                                     {\overline{{\mathsf{U}}}_{1}({\bm{\rho}})
                                      +
                                       \overline{{\mathsf{U}}}_{2}({\bm{\tau}})
                                      }\,
                              \tau_{1}
                              \circ     
                              \frac{\overline{{\mathsf{U}}}_{2}({\bm{\tau}})
                                      }
                                      {\overline{{\mathsf{U}}}_{1}({\bm{\rho}})
                                       +
                                       \overline{{\mathsf{U}}}_{2}({\bm{\tau}})
                                      }\,
                              \rho_{2}     
            \right\rangle\, ,
\end{eqnarray*}
}
}
and
{
\small
\textcolor{black}{
\begin{eqnarray*}
           {\bm\tau} \ast {\bm\rho}
& := & \left\langle \frac{\overline{{\mathsf{U}}}_{1}({\bm{\rho}})
                                      }
                                      {\overline{{\mathsf{U}}}_{1}({\bm{\rho}})
                                       +
                                       \overline{{\mathsf{U}}}_{2}({\bm{\tau}})
                                      }\,
                              \tau_{1}
                              \circ  
                              \frac{\overline{{\mathsf{U}}}_{2}({\bm{\tau}})
                                      }
                                     {\overline{{\mathsf{U}}}_{1}({\bm{\rho}})
                                      +
                                      \overline{{\mathsf{U}}}_{2}({\bm{\tau}})
                                     }\, 
                             \rho_{2}\, ,     
                            \frac{\overline{{\mathsf{U}}}_{1}({\bm{\tau}})
                                    }
                                   {\overline{{\mathsf{U}}}_{1}({\bm{\tau}})
                                    +
                                    \overline{{\mathsf{U}}}_{2}({\bm{\rho}})
                                   }\,
                            \rho_{1}
                            \circ     
                            \frac{\overline{{\mathsf{U}}}_{2}({\bm{\rho}})
                                    }
                                   {\overline{{\mathsf{U}}}_{1}({\bm{\tau}})
                                    +
                                     \overline{{\mathsf{U}}}_{2}({\bm{\rho}})
                                    }\,
                           \tau_{2}     
       \right\rangle\, .
\end{eqnarray*}
}
}
\end{definition}
\end{minipage}
}
\end{center}

\noindent
Roughly speaking,
the balanced mixture
doubles the strategy set of each player
by concatenating it to itself;
it uses 
a suitable \textcolor{black}{combination}
of the probabilities
in ${\bm{\rho}}$ and ${\bm{\tau}}$
and concatenates together
the resulting mixed strategies
without modifying the players' supports
in ${\bm{\rho}}$ and ${\bm{\tau}}$.
Since
$\left\langle {\mathsf{R}},
                    {\mathsf{C}}
 \right\rangle$
has the positive utility property,
Lemma~\ref{park kafe} implies that
both
$\textcolor{black}{\overline{{\mathsf{U}}}}_{1}({\bm{\tau}})
 >
 0$
and
$\textcolor{black}{\overline{{\mathsf{U}}}}_{2}({\bm{\tau}})
 >
 0$;
since
$\left\langle {\mathsf{R}},
                     {\mathsf{C}}
 \right\rangle$ 
is win-lose,
both
$\textcolor{black}{\overline{{\mathsf{U}}}}_{1}({\bm{\rho}})
 \geq
 0$
and
$\textcolor{black}{\overline{{\mathsf{U}}}}_{2}({\bm{\rho}})
 \geq
 0$.
It follows that both
${\bm\rho} \ast{\bm\tau}$
and
${\bm\tau} \ast {\bm\rho}$
are well-defined.
It is straightforward to see that
${\bm\rho} \ast{\bm\tau}$
and
${\bm\tau} \ast {\bm\rho}$
are mixed profiles
for the 
\textcolor{black}{win-lose ${\mathsf{GHR}}$-symmetrization.} 
Thus,
a balanced mixture maps 
\textcolor{black}{the pair
of ${\bm{\rho}}$
and
${\bm{\tau}}$}
to the pair of mixed profiles      
${\bm{\rho}} \ast {\bm{\tau}}$
and
${\bm{\tau}} \ast {\bm{\rho}}$
\textcolor{black}{for the win-lose ${\mathsf{GHR}}$-symmetrization.} 
(We shall provide an inverse map
in Theorem~\ref{fromsymmetrictobasic2}.)

Note that
the balanced mixture is an injective map
as long as the supports 
(but not the probabilities)
are concerned.
Note that when
${\bm{\rho}} = \left\langle 0^{n}, 0^{n}
                          \right\rangle$,
then
${\bm{\rho}} \ast {\bm{\tau}}
  =
  \left\langle 0^{n} \circ \tau_{2},
                    \tau_{1} \circ 0^{n}
  \right\rangle$
and
${\bm{\tau}} \ast {\bm{\rho}}
  =
  \left\langle \tau_{1} \circ 0^{n},
                    0^{n} \circ \tau_{2}
  \right\rangle$;
so,
in 
${\bm{\rho}} \ast {\bm{\tau}}$
and
${\bm{\tau}} \ast {\bm{\rho}}$,
the mixed strategies $\tau_{1}$ and $\tau_{2}$
are swapped between the two players,
while each is concatenated
with $0^{n}$
from left and right,
respectively.
By abuse of notation,
we shall often refer to each of
${\bm\rho} \ast{\bm\tau}$
and
${\bm\tau} \ast {\bm\rho}$
as a balanced mixture.
Note that                                              
$\bm{\rho} \ast \bm{\tau}$ 
together with
$\bm{\tau} \ast \bm{\rho}$
form a {\it symmetric} pair of mixed profiles:
setting
$\bm{\phi} := \bm{\rho}
                       \ast
                       \bm{\tau}$
with $\bm{\phi} =
         \langle \phi_{1}, \phi_{2}
         \rangle$
yields that
$\bm{\tau} 
  \ast
  \bm{\rho}
  =
  \langle \phi_{2}, \phi_{1}
  \rangle$.
Clearly,
${\bm{\rho}}
  \ast
  {\bm{\tau}}$
is uniform
(resp., symmetric)
if and only if
${\bm{\tau}}
  \ast
  {\bm{\rho}}$ is.
We observe a simple uniformity property
of the balanced mixture: 

\begin{observation}
\label{facebook observation}
If 
\textcolor{black}{{\sf (1)}
${\bm{\rho}}
  =
  \langle {\mathsf{0}}^{n},
              {\mathsf{0}}^{n}
  \rangle$,
or
{\sf (2)}}              
either 
${\bm{\rho}}$
or ${\bm{\tau}}$
is non-uniform,
then both
${\bm{\rho}} \ast {\bm{\tau}}$
and
${\bm{\tau}} \ast {\bm{\rho}}$
are non-uniform;
thus,
if ${\bm{\rho}} \ast {\bm{\tau}}$
and
${\bm{\tau}} \ast {\bm{\rho}}$
are uniform,
then both 
${\bm{\rho}}$
and
${\bm{\tau}}$
are uniform.
\end{observation}

We further prove
a simple symmetry property
of the balanced mixture:

\begin{lemma}
\label{symmetry observation}
A balanced mixture is symmetric
if and only if
it is the balanced mixture
${\bm{\tau}} \ast {\bm{\tau}}$
for some
Nash equilibrium
${\bm{\tau}}
  \in
 {\cal NE}(\left\langle {\mathsf{R}},
                                   {\mathsf{C}}
           \right\rangle)$.
\end{lemma}

\begin{proof}
By Definition~\ref{gkt product},
the balanced mixture
${\bm{\tau}}
  \ast
  {\bm{\tau}}$
is symmetric
for
${\bm{\tau}}
  \in
  {\mathcal{NE}}(\left\langle {\mathsf{R}},
                                              {\mathsf{C}}
                            \right\rangle)$.
In the \textcolor{black}{other} direction,
consider the symmetric balanced mixture
${\bm{\rho}}
  \ast
  {\bm{\tau}}$,
with
${\bm{\rho}}
  \in
  {\cal NE}(\left\langle {\mathsf{R}},
                                    {\mathsf{C}}
                 \right\rangle)
  \cup
  \left\{ \left\langle {\mathsf{0}}^{n},
                               {\mathsf{0}}^{n}
            \right\rangle
  \right\}$                             
and
${\bm{\tau}}
  \in
  {\cal NE}(\left\langle {\mathsf{R}},
                                    {\mathsf{C}}
                 \right\rangle)$;  
we shall prove that
${\bm{\rho}}
  =
  {\bm{\tau}}
  \in
  {\cal NE}(\left\langle {\mathsf{R}},
                                    {\mathsf{C}}
                 \right\rangle)$.
Since ${\bm{\rho}}
            \ast
            {\bm{\tau}}$
is symmetric,
we have that
{
\small
\textcolor{black}{
\begin{eqnarray*}
\frac{\textstyle \overline{{\mathsf{U}}}_{1}({\bm{\tau}})
          }
         {\textstyle \overline{{\mathsf{U}}}_{1}({\bm{\tau}})
                           +
                            \overline{{\mathsf{U}}}_{2}({\bm{\rho}})
         }\,
  \rho_{1}
&  = &
  \frac{\textstyle \overline{{\mathsf{U}}}_{1}({\bm{\rho}})
          }
         {\textstyle \overline{{\mathsf{U}}}_{1}({\bm{\rho}})
                           +
                            \overline{{\mathsf{U}}}_{2}({\bm{\tau}})
          }\,
  \tau_{1}
\end{eqnarray*}  
}
}
and
{
\small
\textcolor{black}{
\begin{eqnarray*}  
\frac{\textstyle \overline{{\mathsf{U}}}_{2}({\bm{\rho}})
          }
          {\textstyle \overline{{\mathsf{U}}}_{1}({\bm{\tau}})
                           +
                           \overline{{\mathsf{U}}}_{2}({\bm{\rho}})
           }\, 
 \tau_{2}     
& = &                             
 \frac{\textstyle \overline{{\mathsf{U}}}_{2}({\bm{\tau}})
         }
        {\textstyle \overline{{\mathsf{U}}}_{1}({\bm{\rho}})
                          +
                          \overline{{\mathsf{U}}}_{2}({\bm{\tau}})
         }\,
 \rho_{2}\, .            
\end{eqnarray*}     
}
}
This implies that
$\rho_{1} \neq {\mathsf{0}}^{n}$
and
$\rho_{2} \neq {\mathsf{0}}^{n}$,
so that
${\bm{\rho}}
  \neq
  \left\langle {\mathsf{0}}^{n},
                    {\mathsf{0}}^{n}
  \right\rangle$,
which implies that
${\bm{\rho}}
  \in
  {\cal NE}(\left\langle {\mathsf{R}},
                                    {\mathsf{C}}
                 \right\rangle)$.  
Since $\rho_{1}$, $\rho_{2}$, $\tau_{1}$ and $\tau_{2}$
are mixed strategies
with each summing up its probabilities to $1$,
it follows that
$\rho_{1} = \tau_{1}$
and
$\rho_{2} = \tau_{2}$;
so,
${\bm{\rho}}
  =
  {\bm{\tau}}$,
as needed.
An identical argument applies to
the balanced mixture
${\bm{\tau}}
  \ast
  {\bm{\rho}}$.  
\end{proof}

\noindent
\textcolor{black}{Finally, we observe:}

\begin{lemma}
\label{vittorio simplification for group iv}
\textcolor{black}{
Consider the pair of Nash equilibria
${\bm{\tau}},
  {\bm{\rho}}
  \in
  {\mathcal{NE}}\left( \left\langle {\mathsf{R}},
                              {\mathsf{C}}
                 \right\rangle
          \right)$.
Then,          
for each player $i \in [2]$,
\textcolor{black}{${\widetilde{{\mathsf{U}}}}_{i}\left( {\bm{\tau}}
                                       \ast
                                      {\bm{\rho}}
                             \right)
  <
  \max \{ {\overline{{\mathsf{U}}}}_{1}({\bm{\rho}}),
               {\overline{{\mathsf{U}}}}_{1}({\bm{\tau}}),
               {\overline{{\mathsf{U}}}}_{2}({\bm{\rho}}),
               {\overline{{\mathsf{U}}}}_{2}({\bm{\tau}})
         \}$.
}}                                         
\end{lemma}

\begin{proof}
\textcolor{black}{First set $i := 1$.
By Observation~\ref{vittorio 1},
{
\small
\textcolor{black}{
\begin{eqnarray*}
          {\widetilde{{\mathsf{U}}}}_{1}\left( {\bm{\tau}}
                                                \ast
                                               {\bm{\rho}}
                                       \right)                                                                                                                                                     
& = & \sum_{s_{1}, s_{2} \in [n]}
              \frac{\overline{{\mathsf{U}}}_{1}({\bm{\tau}})
                      }
                      {\overline{{\mathsf{U}}}_{1}({\bm{\tau}})
                       +
                        \overline{{\mathsf{U}}}_{2}({\bm{\rho}})
                      }
             \rho_{1}(s_{1})\,
             \cdot
             \frac{\overline{{\mathsf{U}}}_{2}({\bm{\tau}})
                     }
                     {\overline{{\mathsf{U}}}_{1}({\bm{\rho}})
                      +
                       \overline{{\mathsf{U}}}_{2}({\bm{\tau}})
                     }\,
             \rho_{2}(s_{2})\,
             \cdot
             {\mathsf{R}}[s_{1}, s_{2}]                                                                                                                                                                  \\   
&   & +
          \sum_{s_{1}, s_{2} \in [n]}
                              \frac{\overline{{\mathsf{U}}}_{2}({\bm{\rho}})
                                      }
                                     {\overline{{\mathsf{U}}}_{1}({\bm{\tau}})
                                      +
                                       \overline{{\mathsf{U}}}_{2}({\bm{\rho}})
                                      }
                              \tau_{2}(s_{1})\,
                              \cdot
                              \frac{\overline{{\mathsf{U}}}_{1}({\bm{\rho}})
                                      }
                                     {\overline{{\mathsf{U}}}_{1}({\bm{\rho}})
                                      +
                                       \overline{{\mathsf{U}}}_{2}({\bm{\tau}})
                                      }
                              \tau_{1}(s_{2})\,
                              \cdot
                              {\mathsf{C}}[s_{2}, s_{1}]    
\end{eqnarray*}
}
}
Note that
\textcolor{black}{
$\frac{\textstyle \overline{{\mathsf{U}}}_{1}({\bm{\tau}})
                                      }
         {\textstyle \overline{{\mathsf{U}}}_{1}({\bm{\tau}})
                                      +
                                       \overline{{\mathsf{U}}}_{2}({\bm{\rho}})
                                      },
  \frac{\textstyle \overline{{\mathsf{U}}}_{2}({\bm{\tau}})
          }
          {\textstyle \overline{{\mathsf{U}}}_{1}({\bm{\rho}})
                      +
                       \overline{{\mathsf{U}}}_{2}({\bm{\tau}})
                     },              
  \frac{\textstyle \overline{{\mathsf{U}}}_{2}({\bm{\rho}})
                                      }
          {\textstyle \overline{{\mathsf{U}}}_{1}({\bm{\tau}})
                            +
                            \overline{{\mathsf{U}}}_{2}({\bm{\rho}})
                                      },                                    
     \frac{\textstyle \overline{{\mathsf{U}}}_{1}({\bm{\rho}})
             }
             {\textstyle \overline{{\mathsf{U}}}_{1}({\bm{\rho}})
                               +
                               \overline{{\mathsf{U}}}_{2}({\bm{\tau}})
             }
     <
     1$,}
with
{
\small
\textcolor{black}{
\begin{eqnarray*}
\frac{\overline{{\mathsf{U}}}_{1}({\bm{\sigma}})
                                      }
                                     {\overline{{\mathsf{U}}}_{1}({\bm{\tau}})
                                      +
                                       \overline{{\mathsf{U}}}_{2}({\bm{\rho}})
                                      }
  \cdot                                    
  \frac{\overline{{\mathsf{U}}}_{2}({\bm{\tau}})
                     }
                     {\overline{{\mathsf{U}}}_{1}({\bm{\rho}})
                      +
                       \overline{{\mathsf{U}}}_{2}({\bm{\tau}})
                     }
  +                                 
  \frac{\overline{{\mathsf{U}}}_{2}({\bm{\rho}})
                                      }
                                     {\overline{{\mathsf{U}}}_{1}({\bm{\tau}})
                                      +
                                       \overline{{\mathsf{U}}}_{2}({\bm{\rho}})
                                      }                                    
   \cdot  
   \frac{\overline{{\mathsf{U}}}_{1}({\bm{\rho}})
                                      }
                                     {\overline{{\mathsf{U}}}_{1}({\bm{\rho}})
                                      +
                                       \overline{{\mathsf{U}}}_{2}({\bm{\tau}})
                                      }
&     <   & 1\, .
\end{eqnarray*}
}     
}
\noindent
It follows that
{
\small
\textcolor{black}{
\begin{eqnarray*}
{\widetilde{{\mathsf{U}}}}_{1}\left( {\bm{\tau}}
                                       \ast
                                      {\bm{\rho}}
                             \right)
& < & \max \{ \sum_{s_{1}, s_{2} \in [n]}
                          \rho_{1}(s_{1})\,
                           \cdot
                          \rho_{2}(s_{2})\,
                          \cdot
                          {\mathsf{R}}[s_{1}, s_{2}],
                         \sum_{s_{1}, s_{2} \in [n]}
                            \tau_{2}(s_{1})\,
                             \cdot
                             \tau_{1}(s_{2})\,
                             \cdot
                             {\mathsf{C}}[s_{1}, s_{2}]
                    \}                                                                                                                                                                                  \\               
& = & \max \left\{ {\overline{{\mathsf{U}}}}_{1}({\bm{\rho}}),
                              {\overline{{\mathsf{U}}}}_{2}({\bm{\tau}})
                   \right\}\, .
\end{eqnarray*}
}
}
\noindent
In a corresponding way,
we establish that
\textcolor{black}{${\widetilde{{\mathsf{U}}}}_{2}\left( {\bm{\tau}}
                                       \ast
                                      {\bm{\rho}}
                             \right)
   <
   \max \{ {\overline{{\mathsf{U}}}}_{2}({\bm{\rho}}),
                {\overline{{\mathsf{U}}}}_{1}({\bm{\tau}})
             \}$}.                          
From these together,
the claim follows.
}
\end{proof}

\subsection{Characterization of Nash Equilibria}
\label{symmetrization nash equilibrium}

\noindent
Loosely speaking,
Theorem~\ref{frombasictosymmetric}
establishes that the balanced mixture
yields Nash equilibria
for the win-lose ${\mathsf{GHR}}$-symmetrization,
which are thus induced
by given Nash equilibria
of the bimatrix game
\textcolor{black}{${\mathsf{G}}
  =
  \left\langle {\mathsf{R}},
                    {\mathsf{C}}
  \right\rangle$}
with the positive utility property;
Theorem~\ref{fromsymmetrictobasic2}
establishes that every Nash equilibrium
\textcolor{black}{for}
the win-lose ${\mathsf{GHR}}$-symmetrization
is the balanced mixture
of either a pair of Nash equilibria
for the bimatrix game
$\left\langle {\mathsf{R}},
                    {\mathsf{C}}
  \right\rangle$
or a Nash equilibrium with the null vector. 
We first prove
that the balanced mixture
maps to the set of Nash equilibria
for
\textcolor{black}{${\widetilde{{\mathsf{G}}}}$}.

\begin{theorem}
\label{frombasictosymmetric}
Consider the bimatrix game
\textcolor{black}{${\mathsf{G}}
  =
  \langle {\mathsf{R}},
              {\mathsf{C}}
 \rangle$}
with the positive utility property
and its win-lose ${\mathsf{GHR}}$-symmetrization 
\textcolor{black}{${\widetilde{{\mathsf{G}}}}$}.
Fix
${\bm\rho}
 \in
 {\mathcal{NE}}(\langle {\mathsf{R}},
                                       {\mathsf{C}}
                \rangle)
 \cup
 \left\{ \langle {\mathsf{0}}^{n},
                       {\mathsf{0}}^{n}
           \rangle
 \right\}$
and
${\bm\tau}
 \in
 {\mathcal{NE}}(\langle {\mathsf{R}},
                        {\mathsf{C}}
                \rangle)$.
Then,                
\textcolor{black}{${\bm\rho} 
 \ast
 {\bm\tau},
 {\bm{\tau}}
 \ast
 {\bm{\rho}} 
 \in
 {\mathcal{NE}}( {\widetilde{{\mathsf{G}}}}
                         )$.}
\end{theorem}

\begin{proof}
We prove that
\textcolor{black}{${\bm\rho} 
 \ast
 {\bm\tau}
 \in
 {\mathcal{NE}}( {\widetilde{{\mathsf{G}}}}
                         )$}.
Denote
$\bm{\phi}
 =
 \left\langle \phi_{1}, 
                   \phi_{2}
 \right\rangle                  
 :=
 {\bm{\rho}} 
 \ast
 {\bm{\tau}}$.                
Consider a strategy
$j \in [n]$ with
$j \in
 {\mathsf{Supp}}(\phi_{1})$.
By the definition of
the balanced mixture,
this implies that
$j \in
 {\mathsf{Supp}}(\rho_{1})$.
It follows that
${\bm\rho}
 \in
 {\mathcal{NE}}\left( \left\langle {\mathsf{R}},
                                                     {\mathsf{C}}
                                   \right\rangle
                         \right)$.
Hence,
{
\small
\textcolor{black}{
\begin{eqnarray*}
      \lefteqn{\overline{{\mathsf{U}}}_{1}
                                                    ({\bm\rho})}                                                                                                                                             \\ 
= & \overline{{\mathsf{U}}}_{1}
                                      ({\bm\rho}_{-1}
                                        \diamond
                                        j)
    & \mbox{(by Lemma~\ref{basic property of mixed nash equilibria} (Condition {\sf (1)}))}                                                                   \\                                                                                                                                                                                                             
= & \sum_{s_{2} \in [n]}
          \rho_{2}(s_{2})
          \cdot
          {\mathsf{R}}[j, s_{2}]
    & \mbox{(from the definition of $\langle {\mathsf{R}},
                                                                     {\mathsf{C}}
                                                        \rangle$)}                                                                                                                                            \\
= & \max_{\ell \in [n]}
           \left\{ \sum_{s_{2} \in [n]}
                         \rho_{2}(s_{2})
                         \cdot
                         {\mathsf{R}}[\ell, s_{2}]
          \right\}
    &  \mbox{(by Lemma~\ref{basic property of mixed nash equilibria} (Condition {\sf (2)}))}\, .       
\end{eqnarray*}
}
}
Hence,
we get
{
\small
\textcolor{black}{
\begin{eqnarray*}
    \lefteqn{{\widetilde{{\mathsf{U}}}}_{1}({\bm\phi}_{-1}
                              \diamond
                              j)}                                                                                                                                   \\            
= & \sum_{s_{2} \in [n]}
         \overrightarrow{\phi_{2}}(s_{2})
         \cdot
         {\mathsf{R}}[j,s_2]
    & \mbox{(by Observation~\ref{vittorio 2} {\sf (Eq. 1)})}                                                                    \\
= & \frac{\textstyle \overline{{\mathsf{U}}}_{2}
                                       ({\bm\tau})}
         {\textstyle \overline{{\mathsf{U}}}_{1}
                                       ({\bm\rho})
                     +
                     \overline{{\mathsf{U}}}_{2}
                                       ({\bm\tau})}\,
    \sum_{s_{2} \in [n]}
      \rho_{2}(s_{2})
      \cdot
      {\mathsf{R}}[j, s_2]
  & \mbox{(from the definition of $\bm{\phi}$)}                                                                                 \\
= & \frac{\textstyle \overline{{\mathsf{U}}}_{2}
                                                                ({\bm\tau})}
               {\textstyle \overline{{\mathsf{U}}}_{1}
                                                                ({\bm\rho})
                                +
                                \overline{{\mathsf{U}}}_{2}
                                                               ({\bm\tau})}\,
      \underbrace{\max_{\ell \in [n]}
                             \left\{ \sum_{s_2 \in [n]}
                                           \rho_{2}(s_{2})
                                           \cdot
                                           {\mathsf{R}}[\ell, s_{2}]
                             \right\}}_{= \overline{{\mathsf{U}}}_{1}
                                                                              ({\bm\rho})}
  &                                                                                                                                                   \\
= & \max_{\ell \in [n]}
    \left\{ \sum_{s_2 \in [2n] \setminus [n]}
              \overrightarrow{\phi_{2}}(s_{2})
              \cdot
              {\mathsf{R}}[\ell, s_{2}] 
    \right\}
  & \mbox{(from the definition of $\bm{\phi}$)}                           \\
= & \max_{\ell \in [n]}
          \left\{ {\widetilde{{\mathsf{U}}}}_1({\bm\phi}_{-1}
                     \diamond
                     \ell)
         \right\}
    &  \mbox{(by Observation~\ref{vittorio 2} {\sf (Eq. 1)})}\, .    
\end{eqnarray*}
}
}
Hence,
restricted to strategies from $[n]$,
strategy $j \in [n]$
is a best-response for player 1
in $\bm\phi$,
with
{
\small
\textcolor{black}{
\begin{equation}
\label{first equation}
      {\widetilde{{\mathsf{U}}}}_1({\bm\phi}_{-1}
                     \diamond
                     j)\ \
=\ \ 
\frac{\textstyle \overline{{\mathsf{U}}}_{2}
                                                                    ({\bm\tau})
                       \cdot
                       \overline{{\mathsf{U}}}_{1}
                                                       ({\bm\rho})}
           {\textstyle \overline{{\mathsf{U}}}_{1}
                                                            ({\bm\rho})
                       +
                       \overline{{\mathsf{U}}}_{2}
                                                      ({\bm\tau})}\, .
\end{equation}
}
}
Consider now a strategy
$k \in [2n] \setminus [n]$
with
$k \in {\mathsf{Supp}}(\phi_{1})$.
By the definition 
of the balanced mixture,
this implies that
$k - n 
 \in
 {\mathsf{Supp}}(\tau_{2})$.
It follows that
${\bm\rho}
 \in
 {\mathcal{NE}}(\langle {\mathsf{R}},
                        {\mathsf{C}}
                \rangle)$.
Hence,       
{
\small
\textcolor{black}{
\begin{eqnarray*}
       \lefteqn{\overline{{\mathsf{U}}}_{2}
                                                    ({\bm\tau})}                                                                             \\                                      
= & \overline{{\mathsf{U}}}_{1}
                        ({\bm\tau}_{-2} 
                         \diamond 
                         (k-n))                                                   
    & \mbox{(by Lemma~\ref{basic property of mixed nash equilibria} (Condition {\sf (1)}))}      \\                                                                                                                                                                                                                                
= & \sum_{s_{1} \in [n]}
          \tau_{1}(s_{1})
          \cdot
          {\mathsf{C}}[s_{1}, k-n]
    &  \mbox{(from the definition of $\langle {\mathsf{R}},
                                                                      {\mathsf{C}}
                                                         \rangle$)}                                                                              \\
= & \max_{\ell \in [2n] \setminus [n]}
      \left\{ \sum_{s_{1} \in [n]}
                    \tau_{1}(s_{1})
                    \cdot
                   {\mathsf{C}}[s_{1}, \ell - n]
      \right\}
    & \mbox{(by Lemma~\ref{basic property of mixed nash equilibria} (Condition {\sf (2)}))}      \\                                                                                                                                                                                                                                      
\end{eqnarray*}
}
}
Hence,
we get
{
\small
\textcolor{black}{
\begin{eqnarray*}
      \lefteqn{{\widetilde{{\mathsf{U}}}}_{1}
                                    ({\bm\phi}_{-1}
                                      \diamond
                                      k)}                                                                                                                \\
= & \sum_{s_{1} \in [n]}
          \overleftarrow{\phi_2}(s_{1})
          \cdot
          {\mathsf{C}}[s_{1}, k-n]
    & \mbox{(by Observation~\ref{vittorio 2} {\sf (Eq. 1)})}                                                          \\
= & \frac{\textstyle \overline{{\mathsf{U}}}_{1}
                                       ({\bm\rho})}
         {\textstyle \overline{{\mathsf{U}}}_{1}
                                       ({\bm\rho})
                     +
                     \overline{{\mathsf{U}}}_{2}
                                       ({\bm\tau})}
       \sum_{s_{1} \in [n]}
          \tau_{1}(s_{1})
          \cdot
          {\mathsf{C}}[s_{1}, k-n]
  & \mbox{(from the definition of $\bm{\phi}$)}                                                                            \\
= & \frac{\textstyle \overline{{\mathsf{U}}}_{1}
                                       ({\bm\rho})}
         {\textstyle \overline{{\mathsf{U}}}_{1}
                                       ({\bm\rho})
                     +
                     \overline{{\mathsf{U}}}_{2}
                                       ({\bm\tau})}
      \underbrace{\max_{\ell \in [2n] \setminus [n]}
                               \left\{ \sum_{s_{1} \in [n]}
                                             \tau_{1} (s_1)
                                              \cdot
                                             {\mathsf{C}}[s_1, \ell-n]
                                \right\}}_{= \overline{{\mathsf{U}}}_{2}
                                                                                 ({\bm\tau})}
  &                                                                                                                                                   \\
= & \max_{\ell \in [2n] \setminus [n]}
      \left\{ \sum_{s_{1} \in [n]}
                \overleftarrow{\phi_{2}}(s_{1})
                \cdot
                {\mathsf{C}}[s_{1}, \ell-n]
      \right\}
  & \mbox{(from the definition of $\bm{\phi}$)}                                                                           \\
= & \max_{\ell \in [2n] \setminus [n]}
      \left\{ {\widetilde{{\mathsf{U}}}}_1
                       ({\bm\phi}_{-1}
                        \diamond
                        \ell)
      \right\}
  &  \mbox{(by Observation~\ref{vittorio 2} {\sf (Eq. 1)})}, .        
\end{eqnarray*}
}
}
Hence,
restricted to strategies from
$[2n] \setminus [n]$,
strategy $k \in [2n] \setminus [n]$
is a best-response
for player $1$
in $\bm\phi$,
with
{
\small
\textcolor{black}{
\begin{equation}
\label{second equation}
      {\widetilde{{\mathsf{U}}}}_{1}({\bm\phi}_{-1} \diamond k)\ \
=\ \ 
\frac{\textstyle \overline{{\mathsf{U}}}_{1}
                                         ({\bm\rho})
                       \cdot
                       \overline{{\mathsf{U}}}_{2}
                                         ({\bm\tau})}
           {\textstyle \overline{{\mathsf{U}}}_{1}
                                         ({\bm\rho})
                       +
                       \overline{{\mathsf{U}}}_{2}
                                         ({\bm\tau})}\, .
\end{equation}
}
}
By~(\ref{first equation}) and (\ref{second equation}),
\textcolor{black}{${\widetilde{{\mathsf{U}}}}_{1}({\bm\phi}_{-1} \diamond j)
  =
   {\widetilde{{\mathsf{U}}}}_{1}({\bm\phi}_{-1} \diamond k)$,}
which implies that
all strategies played by player $1$
in $\phi_1$
are best-responses to $\phi_2$. 
Similarly, 
we prove that
all strategies played by player $2$
in $\phi_{2}$
are best-responses to $\phi_1$.
Hence,
${\bm\phi}
  =
  {\bm{\rho}} \ast {\bm{\tau}}
 \in {\mathcal{NE}}( {\widetilde{{\mathsf{G}}}}
                              )$.
Since
${\widetilde{{\mathsf{G}}}}$
is a symmetric game,
and
${\bm\rho} \ast {\bm\tau}$
together with
${\bm\tau} \ast {\bm\rho}$
form a symmetric pair of mixed profiles,
it follows that
${\bm\tau} \ast {\bm\rho}
  \in
  {\mathcal{NE}}( {\widetilde{{\mathsf{G}}}}
                           )$.
\end{proof}

\noindent
We now prove that 
a Nash equilibrium for
${\widetilde{{\mathsf{G}}}}$
is either 
the balanced mixture
of two Nash equilibria 
for $\left\langle {\mathsf{R}},
                          {\mathsf{C}}
     \right\rangle$,
or the balanced mixture
of a Nash equilibrium
for $\left\langle {\mathsf{R}},
                  {\mathsf{C}}
     \right\rangle$
with $\langle {\mathsf{0}}^{n},
                     {\mathsf{0}}^{n}
         \rangle$.

\begin{theorem}
\label{fromsymmetrictobasic2}
Consider a bimatrix game
${\mathsf{G}}
  =
  \left\langle {\mathsf{R}},
                    {\mathsf{C}}
 \right\rangle$
with the positive utility property,
and its win-lose
${\mathsf{GHR}}$-symmetrization
${\widetilde{{\mathsf{G}}}}$
with a Nash equilibrium
$\bm\phi$.
Then, 
exactly one 
of the following conditions hold:
\begin{enumerate}

\item[{\sf (C'.1)}]
$\overleftarrow{\phi_{1}},
 \overrightarrow{\phi_{1}},
 \overleftarrow{\phi_{2}},
 \overrightarrow{\phi_{2}}
 \neq
 {\mathsf{0}}^n$,
with
${\bm\phi}
 =
 \langle \overleftarrow{\varphi_{1}},
             \overrightarrow{\varphi_{2}}
 \rangle             
  \ast
 \langle \overleftarrow{\varphi_{2}},
            \overrightarrow{\varphi_{1}}
 \rangle$
and
$\langle \overleftarrow{\varphi_{2}},
              \overrightarrow{\varphi_{1}}
  \rangle,
 \langle \overleftarrow{\varphi_{1}},
             \overrightarrow{\varphi_{2}}
 \rangle            
 \in
 {\mathcal{NE}}(\langle {\mathsf{R}},
                        {\mathsf{C}}
                \rangle)$.

\item[{\sf (C'.2)}]
$\overleftarrow{\phi_{1}}
 =
 \overrightarrow{\phi_{2}}
 =
 {\sf 0}^n$, 
with 
${\bm\phi}
 =
 \langle {\mathsf{0}}^{n},
             {\mathsf{0}}^{n}
 \rangle            
  \ast
  \langle \overleftarrow{\varphi_{2}},
             \overrightarrow{\varphi_{1}}
  \rangle$
and
$\langle \overleftarrow{\varphi_{2}},
             \overrightarrow{\varphi_{1}}
  \rangle           
  \in
  {\mathcal{NE}}(\langle {\mathsf{R}},
                         {\mathsf{C}}
                 \rangle)$.

\item[{\sf (C'.3)}]
$\overrightarrow{\phi_{1}}
 =
 \overleftarrow{\phi_{2}}
 =
 {\mathsf{0}}^n$,
with
${\bm\phi}
 =
 \langle {\mathsf{0}}^{n},
             {\mathsf{0}}^{n}
 \rangle            
 \ast
 \langle \overleftarrow{\varphi_{1}},
             \overrightarrow{\varphi_{2}}
 \rangle$
and
$\langle \overleftarrow{\varphi_{1}},
              \overrightarrow{\varphi_{2}}
  \rangle             
  \in
  {\mathcal{NE}}(\langle {\mathsf{R}},
                         {\mathsf{C}}
                 \rangle)$.

\end{enumerate}

\end{theorem}

\begin{proof}
Note that each
of the Conditions
{\sf (C'.1)},
{\sf (C'.2)}
and
{\sf (C'.3)}
refines Condition
{\sf (C.1)},
{\sf (C.2)}
and
{\sf (C.3)},
respectively,
in Lemma~\ref{noemptystrategies},
of which exactly one holds.
We first prove that
Condition {\sf (C.3)}
in Lemma~\ref{noemptystrategies}
implies
Condition {\sf (C'.3)}.
So assume
$\overleftarrow{\phi_{1}},
 \overrightarrow{\phi_{1}},
 \overleftarrow{\phi_{2}},
 \overrightarrow{\phi_{2}}
 \neq
 {\mathsf{0}}^n$.
Consider a strategy
$j
 \in
 {\mathsf{Supp}}(\overleftarrow{\phi_1})$;
so, 
$j\in [n]$.
By Observation~\ref{vittorio 2} {\sf (Eq. 1)},
{
\small
\begin{eqnarray*}
\label{conditionrow}
      \textcolor{crimsonglory}{{\widetilde{{\mathsf{U}}}}}_1({\bm\phi}_{-1}
                     \diamond
                     j) 
& = & \sum_{s_{2} \in [n]}
        \overrightarrow{\phi_{2}}(s_{2})
        \cdot
        {\mathsf{R}}[j,s_2]\, .
\end{eqnarray*}
}
Since $j$ is
a best-response  
for player $1$ to $\phi_{2}$,
Lemma~\ref{basic property of mixed nash equilibria}
implies that
{
\small
\begin{eqnarray*}
      \sum_{s_{2} \in [n]}
        \overrightarrow{\phi_{2}}(s_{2})
        \cdot
        {\mathsf{R}}[j, s_{2}]
& = & \max_{\ell \in [n]}
        \left\{ \sum_{s_{2} \in [n]}
                  \overrightarrow{\phi_{2}}(s_{2})
                  \cdot
                  {\mathsf{R}}[\ell,s_{2}]
        \right\}\, .
\end{eqnarray*}
}
\noindent
Clearly,
$\overleftarrow{\varphi_{1}}$
is a mixed strategy
for player $1$ in the game 
$\langle {\mathsf{R}},
         {\mathsf{C}}\rangle$
with
${\mathsf{Supp}}(\overleftarrow{\varphi_{1}})
 =
 {\mathsf{Supp}}(\overleftarrow{\phi_{1}})$.
Thus,
by normalizing $\overrightarrow{\phi_2}$,
it follows that
for each strategy
$j
 \in
 {\mathsf{Supp}}(\overleftarrow{\varphi_{1}})$,
{
\small
\begin{equation}
\label{equation 3}
      \sum_{s_2 \in [n]}
        \overrightarrow{\varphi_{2}}(s_{2})
        \cdot
        {\mathsf{R}}[j, s_{2}]\ \
=\ \
\max_{\ell \in [n]}
        \left\{\sum_{s_{2} \in [n]}
                 \overrightarrow{\varphi_{2}}(s_{2})
                 \cdot
                 {\mathsf{R}}[\ell, s_{2}]
        \right\}\, .
\end{equation}
}
\noindent
Consider now a strategy
$k
 \in
 {\mathsf{Supp}}(\overrightarrow{\phi_{2}})$;
so,
$k \in [n]$.
By Observation~\ref{vittorio 2} {\sf (Eq. 2)},
{
\small
\begin{eqnarray*}  
           \textcolor{crimsonglory}{{\widetilde{{\mathsf{U}}}}}_2({\bm\phi}_{-2}\diamond (k+n)) 
& = & \sum_{s_{1} \in [n]}
        \overleftarrow{\phi_{1}}(s_{1})
        \cdot
        {\mathsf{C}}[s_{1}, k]\, .
\end{eqnarray*}
}
Since $k$ is a best-response
for player $2$ to $\phi_{1}$,
it follows that
{
\small
\begin{eqnarray*}
      \sum_{s_{1} \in [n]}
        \overleftarrow{\phi_{1}}(s_{1})
        \cdot
        {\mathsf{C}}[s_{1}, k]
& = & \max_{\ell \in [n]}
        \left\{ \sum_{s_{1} \in [n]}
                  \overleftarrow{\phi_{1}}(s_{1})
                  \cdot
                  {\mathsf{C}}[s_{1}, k]
        \right\}\, .
\end{eqnarray*}
}
\noindent
Clearly,
$\overrightarrow{\varphi_{2}}$
is a mixed strategy for player $2$
in the game
$\langle {\mathsf{R}},
              {\mathsf{C}}
  \rangle$
with     
${\mathsf{Supp}}(\overrightarrow{\varphi_{2}})
 =
 {\mathsf{Supp}}(\overrightarrow{\phi_{2}})$.
Thus,
by normalizing
$\overleftarrow{\phi_1}$,
it follows that
for each strategy
$k
 \in
 {\mathsf{Supp}}(\overrightarrow{\phi_{2}})$,
{
\small
\begin{equation}
\label{equation 4}
      \sum_{s_{1} \in [n]}
        \overleftarrow{\varphi_{1}}(s_{1})
        \cdot
        {\mathsf{C}}[s_{1}, k]\ \
 =\ \
 \max_{\ell \in [n]}
        \left\{ \sum_{s_{1} \in [n]}
                  \overleftarrow{\varphi_{1}}(s_{1})
                  \cdot
                  {\mathsf{C}}[s_{1}, k]
        \right\}\, .
\end{equation}
}
By~(\ref{equation 3}) and (\ref{equation 4}),
it follows that
$\langle \overleftarrow{\varphi_{1}},
              \overrightarrow{\varphi_{2}}
  \rangle            
  \in
  {\mathcal{NE}}(\langle {\mathsf{R}},
                         {\mathsf{C}}
  \rangle)$.
Similarly, 
we prove that
$\langle \overleftarrow{\varphi_{2}},
             \overrightarrow{\varphi_{1}}
 \rangle            
  \in
  {\mathcal{NE}}(\langle {\mathsf{R}},
                         {\mathsf{C}}
  \rangle)$.
Now it remains to prove that
$\bm{\phi}
 =
 \langle \overleftarrow{\varphi_{1}},
            \overrightarrow{\varphi_{2}}
 \rangle           
 \ast
 \langle \overleftarrow{\varphi_2},
            \overrightarrow{\varphi_1}
 \rangle$.

Fix now strategies
$j, k 
 \in
 {\mathsf{Supp}}(\phi_{1})$
with
$j \in [n]$
and
$k \in [2n] \setminus [n]$.
Then,
by Observation~\ref{vittorio 2} {\sf (Eq. 1)},
{
\small
\begin{eqnarray*}
      \textcolor{black}{{\widetilde{{\mathsf{U}}}}}_{1}({\bm\phi}_{-1}
                                      \diamond
                                      j)
& = & \sum_{s_{2} \in [n]}
        \overrightarrow{\phi_{2}}(s_{2})
        \cdot
        {\mathsf{R}}[j,s_{2}]\, ,
\end{eqnarray*}
}
and
{
\small
\begin{eqnarray*}
      \textcolor{black}{{\widetilde{{\mathsf{U}}}}}_{1}({\bm\phi}_{-1}
                                     \diamond
                                     k)
& = & \sum_{s_{1} \in [n]}
        \overleftarrow{\phi_{2}}(s_{1})
        \cdot
        {\mathsf{C}}[s_{1},k-n]\, ;
\end{eqnarray*}
}
also,
for any fixed strategies
$j, k \in {\mathsf{Supp}}(\phi_{2})$
with
$j \in [n]$
and
$k \in [2n] \setminus [n]$,
by Observation~\ref{vittorio 2} {\sf (Eq. 2)},
{
\small
\begin{eqnarray*}
      \textcolor{black}{{\widetilde{{\mathsf{U}}}}}_{2}({\bm\phi}_{-2}
                                     \diamond
                                     k)
& = & \sum_{s_{1} \in [n]}
        \overleftarrow{\phi_{1}}(s_{1})
        \cdot
        {\mathsf{C}}[s_{1}, k-n]\, ,
\end{eqnarray*}
}
and
{
\small
\begin{eqnarray*}
      \textcolor{black}{{\widetilde{{\mathsf{U}}}}}_{2}({\bm\phi}_{-2}
                                    \diamond
                                    j)
& = & \sum_{s_{2} \in [n]}
        \overrightarrow{\phi_{1}}(s_{2})
        \cdot
        {\mathsf{R}}[j, s_{2}]\, .
\end{eqnarray*}
}
\noindent
Since $\bm{\phi}$ is a Nash equilibrium
for \textcolor{black}{${\widetilde{{\mathsf{G}}}}$,}
Lemma \ref{basic property of mixed nash equilibria}
(Condition {\sf (1)}) 
implies that
$\textcolor{black}{{\widetilde{{\mathsf{U}}}}}_{1}({\bm{\phi}}_{-1} \diamond j)
  =
  \textcolor{black}{{\widetilde{{\mathsf{U}}}}}_{1}({\bm{\phi}}_{-1} \diamond k)$
and
$\textcolor{black}{{\widetilde{{\mathsf{U}}}}}_{2}({\bm{\phi}}_{-2} \diamond  j)
  =
  \textcolor{black}{{\widetilde{{\mathsf{U}}}}}_{2}({\bm{\phi}}_{-2} \diamond k)$.
Hence,
{
\small
\begin{equation}
\label{first from vittorio}
      \sum_{s_{2} \in [n]}
        \overrightarrow{\phi_{2}}(s_{2})
        \cdot
        {\mathsf{R}}[j,s_{2}]\ \
=\ \ 
\sum_{s_{1} \in [n]}
        \overleftarrow{\phi_{2}}(s_{1})
        \cdot
        {\mathsf{C}}[s_{1},k-n]\,
\end{equation}
}
and
{
\small
\begin{equation}
\label{second from vittorio}
      \sum_{s_{1} \in [n]}
        \overleftarrow{\phi_{1}}(s_{1})
        \cdot
        {\mathsf{C}}[s_{1}, k-n]\ \
=\ \ 
\sum_{s_{2} \in [n]}
        \overrightarrow{\phi_{1}}(s_{2})
        \cdot
        {\mathsf{R}}[j, s_{2}]\, ,
\end{equation}
}
respectively.
First, note that
{
\small
\begin{eqnarray*}
    \lefteqn{\textcolor{black}{\overline{{\mathsf{U}}}}_{1}
                      \left( \left\langle \overleftarrow{\phi_{1}},
                                                 \overrightarrow{\phi_{2}}
                               \right\rangle
                     \right)}                                                                                \\
= & \sum_{s_{1}, s_{2} \in [n]}
      \overleftarrow{\phi_{1}}(s_{1})
      \overrightarrow{\phi_{2}}(s_{2})\,
      {\mathsf{R}}[s_{1}, s_{2}] 
  &                                                                                  \\
= & \sum_{s_{1} \in [n]}
      \overleftarrow{\phi_{1}}(s_{1})        
      \left( \sum_{s_{2} \in [n]}
               \overrightarrow{\phi_{2}}(s_{2})
               {\mathsf{R}}[s_{1}, s_{2}]
      \right)
  &                                                                                 \\
= & \left( \sum_{s_{1} \in [n]}
             \overleftarrow{\phi_{1}}(s_{1})                    
    \right)
    \cdot
    \underbrace{\left( 
                  \sum_{s_{2} \in [n]}
                    \overrightarrow{\phi_{2}}(s_{2})\,
                    {\mathsf{R}}[s_{1}, s_{2}]
                \right)}_{\mbox{for any strategy $s_{1} \in [n]$
                                with $\overleftarrow{\phi_{1}}(s_{1}) > 0$
                               }
                         }
  & \mbox{(by Lemma~\ref{basic property of mixed nash equilibria} (Condition {\sf (1)})),}                       
\end{eqnarray*}
}
and
{
\small
\begin{eqnarray*}
    \lefteqn{\textcolor{black}{\overline{{\mathsf{U}}}}_{2}
                      \left( \left\langle \overleftarrow{\phi_{2}},
                                                 \overrightarrow{\phi_{1}}
                               \right\rangle
                     \right)}                                                                                            \\
= & \sum_{s_{1}, s_{2} \in [n]}
      \overleftarrow{\phi_{2}}(s_{1})
      \overrightarrow{\phi_{1}}(s_{2})\,
      {\mathsf{C}}[s_{1}, s_{2}] 
  &                                                                                  \\
= & \sum_{s_{2} \in [n]}
      \overrightarrow{\phi_{1}}(s_{2})        
      \left( \sum_{s_{1} \in [n]}
               \overleftarrow{\phi_{2}}(s_{1})
               {\mathsf{C}}[s_{1}, s_{2}]
      \right)
  &                                                                                 \\
= & \left( \sum_{s_{2} \in [n]}
             \overrightarrow{\phi_{1}}(s_{2})                    
    \right)
    \cdot
    \underbrace{\left( 
                  \sum_{s_{1} \in [n]}
                    \overleftarrow{\phi_{2}}(s_{1})\,
                    {\mathsf{C}}[s_{1}, s_{2}]
                \right)}_{\mbox{for any strategy $s_{2} \in [n]$
                                with $\overrightarrow{\phi_{1}}(s_{2}) > 0$
                               }
                         }
  & \mbox{(by Lemma~\ref{basic property of mixed nash equilibria} (Condition {\sf (1)})).}                       
\end{eqnarray*}
}
By (\ref{first from vittorio}),
it follows that
{
\small
\textcolor{black}{
\begin{equation}
\label{equation 5}
      {\textstyle \textcolor{black}{\overline{{\mathsf{U}}}}_{1}
                                         \left( \left\langle \overleftarrow{\phi_{1}},
                                                                    \overrightarrow{\phi_{2}}
                                                 \right\rangle
                                        \right)}
      \cdot
      {\textstyle \sum_{s_{2} \in [n]}
                             \overrightarrow{\phi_{1}}(s_{2})}\ \                                  
=\ \
      {\textstyle \textcolor{crimsonglory}{\overline{{\mathsf{U}}}}_{2}
                                             \left( \left\langle \overleftarrow{\phi_{2}},
                                                                        \overrightarrow{\phi_{1}}
                                                     \right\rangle
                                            \right)}
       \cdot
      {\textstyle \sum_{s_{1} \in [n]}
                         \overleftarrow{\phi_{1}}(s_{1})}\, .                                                    
\end{equation}
}
}
\noindent
Second, 
note that
{
\small
\begin{eqnarray*}
    \lefteqn{\textcolor{black}{\overline{{\mathsf{U}}}}_{1}
                      \left( \left\langle \overleftarrow{\phi_{2}},
                                                 \overrightarrow{\phi_{1}}
                               \right\rangle
                      \right)}                                                                                                                                                         \\
= & \sum_{s_{1}, s_{2} \in [n]}
      \overleftarrow{\phi_{2}}(s_{1})
      \overrightarrow{\phi_{1}}(s_{2})\,
      {\mathsf{R}}[s_{1}, s_{2}] 
  &                                                                                  \\
= & \sum_{s_{1} \in [n]}
      \overleftarrow{\phi_{2}}(s_{1})        
      \left( \sum_{s_{2} \in [n]}
               \overrightarrow{\phi_{1}}(s_{2})
               {\mathsf{R}}[s_{1}, s_{2}]
      \right)
  &                                                                                 \\
= & \left( \sum_{s_{1} \in [n]}
             \overleftarrow{\phi_{2}}(s_{1})                    
    \right)
    \cdot
    \underbrace{\left( 
                  \sum_{s_{2} \in [n]}
                    \overrightarrow{\phi_{1}}(s_{2})\,
                    {\mathsf{R}}[s_{1}, s_{2}]
                \right)}_{\mbox{for any strategy $s_{1} \in [n]$
                                with $\overleftarrow{\phi_{2}}(s_{1}) > 0$
                               }
                         }
  & \mbox{(by Lemma~\ref{basic property of mixed nash equilibria} (Condition {\sf (1)}))\, ,}                       
\end{eqnarray*}
}
and
{
\small
\begin{eqnarray*}
    \lefteqn{\textcolor{black}{\overline{{\mathsf{U}}}}_{2}
                      \left( \left\langle \overleftarrow{\phi_{1}},
                                                 \overrightarrow{\phi_{2}}
                                \right\rangle
                      \right)}                                                                                                                                                  \\
= & \sum_{s_{1}, s_{2} \in [n]}
      \overleftarrow{\phi_{1}}(s_{1})
      \overrightarrow{\phi_{2}}(s_{2})\,
      {\mathsf{C}}[s_{1}, s_{2}] 
  &                                                                                  \\
= & \sum_{s_{2} \in [n]}
      \overrightarrow{\phi_{2}}(s_{2})        
      \left( \sum_{s_{1} \in [n]}
               \overleftarrow{\phi_{1}}(s_{1})
               {\mathsf{C}}[s_{1}, s_{2}]
      \right)
  &                                                                                 \\
= & \left( \sum_{s_{2} \in [n]}
             \overrightarrow{\phi_{2}}(s_{2})                    
    \right)
    \cdot
    \underbrace{\left( 
                  \sum_{s_{1} \in [n]}
                    \overleftarrow{\phi_{1}}(s_{1})\,
                    {\mathsf{C}}[s_{1}, s_{2}]
                \right)}_{\mbox{for any strategy $s_{2} \in [n]$
                                with $\overrightarrow{\phi_{2}}(s_{2}) > 0$
                               }
                         }
  & \mbox{(by Lemma~\ref{basic property of mixed nash equilibria} (Condition {\sf (1)}))\, .}                       
\end{eqnarray*}
}
By (\ref{second from vittorio}),
it follows that
{
\small
\textcolor{black}{
\begin{equation}
\label{equation 6}
      {\textstyle \textcolor{black}{\overline{{\mathsf{U}}}}_{1}
                                         \left( \left\langle \overleftarrow{\phi_{2}},
                                                                    \overrightarrow{\phi_{1}}
                                                  \right\rangle
                                         \right)}
      \cdot
      {\textstyle \sum_{s_{2} \in [n]}
                           \overrightarrow{\phi_{2}}(s_{2})}\ \ 
=\ \
      {\textstyle \textcolor{black}{\overline{{\mathsf{U}}}}_{2}
                                         \left( \left\langle \overleftarrow{\phi_{1}},
                                                                    \overrightarrow{\phi_{2}}
                                                  \right\rangle
                                         \right)}
      \cdot
      {\textstyle \sum_{s_{1} \in [n]}
                         \overleftarrow{\phi_{2}}(s_{1})}\, .
\end{equation}
}
}
\noindent
Set 
$\overleftarrow{p_{1}}
 :=
 \sum_{s_{1} \in [n]}
   \overleftarrow{\phi_{1}}(s_1)$;
$\overrightarrow{p_{1}}$,
$\overleftarrow{p_{2}}$
and
$\overrightarrow{p_{2}}$
are defined in a corresponding way.
Note that
for each $i \in [2]$, 
$\overleftarrow{p_{i}}
 +
 \overrightarrow{p_{i}}=1$.
By the definition of the balanced mixture,
the first component of
the first entry of
$\langle \overleftarrow{\varphi_{1}},
             \overrightarrow{\varphi_{2}}
  \rangle
 \ast
 \langle \overleftarrow{\varphi_{2}},
             \overrightarrow{\varphi_{1}}
 \rangle$
is
{
\small
\textcolor{black}{
\begin{eqnarray*}
      \lefteqn{\frac{\textstyle \overline{{\mathsf{U}}}_{1}
                                         \left( \left\langle \overleftarrow{\varphi_{2}},
                                                                    \overrightarrow{\varphi_{1}}
                                                 \right\rangle
                                         \right)
           }
           {\textstyle \overline{{\mathsf{U}}}_{1}
                                         \left( \left\langle \overleftarrow{\varphi_{2}},
                                                                    \overrightarrow{\varphi_{1}}
                                                 \right\rangle
                                        \right)
                       +
                       \overline{{\mathsf{U}}}_{2}
                                         \left( \left\langle \overleftarrow{\varphi_{1}},
                                                                    \overrightarrow{\varphi_{2}}
                                                 \right\rangle
                                        \right)
           }
      \cdot     
      \overleftarrow{\varphi_{1}}}                                                                                                                                                                                                      \\     
= & \frac{\textstyle \overline{{\mathsf{U}}}_{1}
                                         \left( \left\langle \overleftarrow{\phi_{2}},
                                                                    \overrightarrow{\phi_{1}}
                                                  \right\rangle
                                        \right)
           }
           {\textstyle \overleftarrow{p_{2}}
                       \cdot
                       \overrightarrow{p_{1}}
                       \cdot
                       \left( \frac{\textstyle \overline{{\mathsf{U}}}_{1}
                                                                 \left( \left\langle \overleftarrow{\phi_{2}},
                                                                                            \overrightarrow{\phi_1}
                                                                         \right\rangle
                                                                \right)
                                   }
                                   {\textstyle \overleftarrow{p_{2}}
                                               \cdot
                                               \overrightarrow{p_{1}}
                                   }
                              +
                              \frac{\textstyle \overline{{\mathsf{U}}}_{2}
                                                                 \left( \left\langle \overleftarrow{\phi_{1}},
                                                                                            \overrightarrow{\phi_{2}}
                                                                         \right\rangle
                                                                 \right)
                                   }
                                   {\textstyle \overleftarrow{p_{1}}
                                               \cdot
                                               \overrightarrow{p_{2}}
                                   }
                       \right)
           }\,
      \frac{\textstyle 1
           }
           {\textstyle \overleftarrow{p_{1}}
           }
      \cdot
      \overleftarrow{\phi_{1}}
   &                                                                                                                                                                                                                                           \\
= & \frac{\textstyle \overline{{\mathsf{U}}}_{1}
                                         \left( \left\langle \overleftarrow{\phi_{2}},
                                                                    \overrightarrow{\phi_{1}}
                                                  \right\rangle
                                        \right)
           }
           {\textstyle \overleftarrow{p_{1}}
                       \cdot
                       \overline{{\mathsf{U}}}_{1}
                                         \left( \left\langle \overleftarrow{\phi_{2}},
                                                                    \overrightarrow{\phi_{1}}
                                                 \right\rangle
                                        \right)
                       +
                       \frac{\textstyle \overleftarrow{p_{2}}
                                        \cdot                                       
                                        \overrightarrow{p_{1}}
                            }
                            {\textstyle \overrightarrow{p_{2}}
                            }
                       \cdot
                       \overline{{\mathsf{U}}}_{2}
                                         \left( \left\langle \overleftarrow{\phi_{1}},
                                                                     \overrightarrow{\phi_{2}}
                                                 \right\rangle
                                        \right)     
          }
      \cdot
      \overleftarrow{\phi_{1}}
   &                                                                                                                                                                                                                                           \\
= & \frac{\textstyle \overline{{\mathsf{U}}}_{1}
                                         \left( \left\langle \overleftarrow{\phi_{2}},
                                                                    \overrightarrow{\phi_{1}}
                                                 \right\rangle
                                         \right)
           }
           {\textstyle \overleftarrow{p_{1}}
                       \cdot
                       \overline{{\mathsf{U}}}_{1}
                                         \left( \left\langle \overleftarrow{\phi_{2}},
                                                                    \overrightarrow{\phi_{1}}
                                                 \right\rangle
                                        \right)
                       +
                       \frac{\textstyle \overleftarrow{p_{2}}
                                        \cdot                                       
                                        \overrightarrow{p_{1}}
                            }
                            {\textstyle \overrightarrow{p_{2}}
                            }
                       \cdot
                       \frac{\textstyle \overrightarrow{p_{2}}
                            }
                            {\textstyle \overleftarrow{p_{2}}
                            }
                       \cdot     
                       \overline{{\mathsf{U}}}_{1}
                                         \left( \left\langle \overleftarrow{\phi_{2}},
                                                                    \overrightarrow{\phi_{1}}
                                                 \right\rangle
                                        \right)     
          }
      \cdot
      \overleftarrow{\phi_{1}}                                                                       
    &  \mbox{(by (\ref{equation 6}))}                                                                                                                                                                         \\
= & \overleftarrow{\phi_{1}}\, .
    & 
\end{eqnarray*}
}
}
The second component of
the first entry of
$\left\langle \overleftarrow{\varphi_{1}},
                    \overrightarrow{\varphi_{2}}
  \right\rangle           
 \ast
 \left\langle \overleftarrow{\varphi_{2}},
             \overrightarrow{\varphi_{1}}
 \right\rangle$
is
{
\small
\textcolor{black}{
\begin{eqnarray*}
      \frac{\textstyle \overline{{\mathsf{U}}}_{2}
                                         \left( \left\langle \overleftarrow{\varphi_{1}},
                                                                    \overrightarrow{\varphi_{2}}
                                                 \right\rangle
                                         \right)
           }
           {\textstyle \overline{{\mathsf{U}}}_{1}
                                         \left(  \left\langle \overleftarrow{\varphi_{2}},
                                                                      \overrightarrow{\varphi_{1}}
                                                  \right\rangle
                                        \right)
                       +
                       \overline{{\mathsf{U}}}_{2}
                                         \left( \left\langle \overleftarrow{\varphi_{1}},
                                                                    \overrightarrow{\varphi_{2}}
                                                 \right\rangle
                                        \right)
           }
      \cdot     
      \overrightarrow{\varphi_{1}}  
& = & \left( 1
             -
             \frac{\textstyle \overline{{\mathsf{U}}}_{1}
                                                \left( \left\langle \overleftarrow{\varphi_{2}},
                                                                           \overrightarrow{\varphi_{1}}
                                                         \right\rangle
                                               \right)
                  }
                  {\textstyle \overline{{\mathsf{U}}}_{1}
                                                \left( \left\langle \overleftarrow{\varphi_{2}},
                                                                           \overrightarrow{\varphi_{1}}
                                                         \right\rangle
                                               \right)
                              +
                              \overline{{\mathsf{U}}}_{2}
                                                \left( \left\langle \overleftarrow{\varphi_{1}},
                                                                           \overrightarrow{\varphi_{2}}
                                                         \right\rangle
                                                \right)
                  }
      \right)
      \overrightarrow{\varphi_{1}}                                                        \\
& = & \left( 1 - \overleftarrow{p_{1}}
      \right)
      \cdot                  
      \frac{\textstyle 1}
           {\textstyle \overrightarrow{p_{1}}}
      \cdot
      \overrightarrow{\phi_{1}}                                                           \\
& = & \overrightarrow{\phi_{1}}\, .       
\end{eqnarray*}
}
}
\noindent
The first component of
the second entry of
$\langle \overleftarrow{\varphi_{1}},
             \overrightarrow{\varphi_{2}}
  \rangle           
 \ast
 \langle \overleftarrow{\varphi_{2}},
             \overrightarrow{\varphi_{1}}
 \rangle$
is
{
\small
\textcolor{black}{
\begin{eqnarray*}
      \lefteqn{\frac{\textstyle \overline{{\mathsf{U}}}_{1}
                                         \left( \left\langle \overleftarrow{\varphi_{1}},
                                                                    \overrightarrow{\varphi_{2}}
                                                  \right\rangle
                                         \right)
           }
           {\textstyle \overline{{\mathsf{U}}}_{1}
                                         \left( \left\langle \overleftarrow{\varphi_{1}},
                                                                     \overrightarrow{\varphi_{2}}
                                                  \right\rangle
                                        \right)
                       +
                       \overline{{\mathsf{U}}}_{2}
                                         \left( \left\langle \overleftarrow{\varphi_{2}},
                                                                    \overrightarrow{\varphi_{1}}
                                                  \right\rangle
                                        \right)
           }
      \cdot     
      \overleftarrow{\varphi_{2}}}                                                                                                                                                                                  \\     
= & \frac{\textstyle \overline{{\mathsf{U}}}_{1}
                                         \left( \left\langle \overleftarrow{\phi_{1}},
                                                                    \overrightarrow{\phi_{2}}
                                                  \right\rangle
                                        \right)
           }
           {\textstyle \overleftarrow{p_{1}}
                       \cdot
                       \overrightarrow{p_{2}}
                       \cdot
                       \left( \frac{\textstyle \overline{{\mathsf{U}}}_{1}
                                                                 \left( \left\langle \overleftarrow{\phi_{1}},
                                                                                             \overrightarrow{\phi_{2}}
                                                                         \right\rangle
                                                                \right)
                                   }
                                   {\textstyle \overleftarrow{p_{1}}
                                               \cdot
                                               \overrightarrow{p_{2}}
                                   }
                              +
                              \frac{\textstyle \overline{{\mathsf{U}}}_{2}
                                                                 \left( \left\langle \overleftarrow{\phi_{2}},
                                                                                            \overrightarrow{\phi_{1}}
                                                                          \right\rangle
                                                                \right)
                                   }
                                   {\textstyle \overleftarrow{p_{2}}
                                               \cdot
                                               \overrightarrow{p_{1}}
                                   }
                       \right)
           }\,
      \frac{\textstyle 1
           }
           {\textstyle \overleftarrow{p_{2}}
           }
      \cdot
      \overleftarrow{\phi_{2}}
    &                                                                                                                                                                                                                         \\
= & \frac{\textstyle \overline{{\mathsf{U}}}_{1}
                                         \left( \left\langle \overleftarrow{\phi_{1}},
                                                                     \overrightarrow{\phi_{2}}
                                                 \right\rangle
                                        \right)
           }
           {\textstyle \overleftarrow{p_{1}}
                       \cdot
                       \overline{{\mathsf{U}}}_{1}
                                         \left( \left\langle \overleftarrow{\phi_{2}},
                                                                     \overrightarrow{\phi_{1}}
                                                  \right\rangle
                                         \right)
                       +
                       \frac{\textstyle \overleftarrow{p_{2}}
                                        \cdot                                       
                                        \overrightarrow{p_{1}}
                            }
                            {\textstyle \overrightarrow{p_{2}}
                            }
                       \cdot
                       \overline{{\mathsf{U}}}_{2}
                                         \left( \left\langle \overleftarrow{\phi_{2}},
                                                                      \overrightarrow{\phi_{1}}
                                                   \right\rangle
                                         \right)     
          }
      \cdot
      \overleftarrow{\phi_{2}}
    &                                                                                                                                                                                                                     \\
= & \frac{\textstyle \overline{{\mathsf{U}}}_{1}
                                         \left( \left\langle \overleftarrow{\phi_{1}},
                                                                     \overrightarrow{\phi_{2}}
                                                 \right\rangle
                                        \right)
           }
           {\textstyle \overleftarrow{p_{1}}
                       \cdot
                       \widehat{{\mathsf{U}}}_{1}
                                         \left( \left\langle \overleftarrow{\phi_{2}},
                                                                    \overrightarrow{\phi_{1}}
                                                 \right\rangle
                                         \right)
                       +
                       \frac{\textstyle \overleftarrow{p_{2}}
                                        \cdot                                       
                                        \overrightarrow{p_{1}}
                            }
                            {\textstyle \overrightarrow{p_{2}}
                            }
                       \cdot
                       \frac{\textstyle \overrightarrow{p_{2}}
                            }
                            {\textstyle \overleftarrow{p_{2}}
                            }
                       \cdot     
                       \overline{{\mathsf{U}}}_{1}
                                         \left( \left\langle \overleftarrow{\phi_{1}},
                                                                    \overrightarrow{\phi_{2}}
                                                 \right\rangle
                                        \right)     
          }
      \cdot
      \overleftarrow{\phi_{2}} 
    & \mbox{(by (\ref{equation 5}))}                                                                                                                                                          \\
= & \overleftarrow{\phi_{2}}\, .
    & 
\end{eqnarray*}
}
}
The second component of
the second entry of
$\langle \overleftarrow{\varphi_{1}},
             \overrightarrow{\varphi_{2}}
  \rangle           
 \ast
 \langle \overleftarrow{\varphi_{2}},
            \overrightarrow{\varphi_{1}}
 \rangle$
is
{
\small
\textcolor{black}{
\begin{eqnarray*}
      \frac{\textstyle \overline{{\mathsf{U}}}_{2}
                                         (\overleftarrow{\varphi_{2}},
                                          \overrightarrow{\varphi_{1}})
           }
           {\textstyle \overline{{\mathsf{U}}}_{1}
                                         (\overleftarrow{\varphi_{1}},
                                          \overrightarrow{\varphi_{2}})
                       +
                       \overline{{\mathsf{U}}}_{2}
                                         (\overleftarrow{\varphi_{2}},
                                          \overrightarrow{\varphi_{1}})
           }
      \cdot     
      \overrightarrow{\varphi_{2}}  
& = & \left( 1
             -
             \frac{\textstyle \overline{{\mathsf{U}}}_{1}
                                                (\overleftarrow{\varphi_{1}},
                                                 \overrightarrow{\varphi_{2}})
                  }
                  {\textstyle \overline{{\mathsf{U}}}_{1}
                                                (\overleftarrow{\varphi_{1}},
                                                 \overrightarrow{\varphi_{2}})
                              +
                              \overline{{\mathsf{U}}}_{2}
                                                (\overleftarrow{\varphi_{2}},
                                                 \overrightarrow{\varphi_{1}})
                  }
      \right)
      \overrightarrow{\varphi_{2}}                                                        \\
& = & \left( 1 - \overleftarrow{p_{2}}
      \right)
      \cdot                  
      \frac{\textstyle 1}
           {\textstyle \overrightarrow{p_{2}}}
      \cdot
      \overrightarrow{\phi_{2}}                                                           \\
& = & \overrightarrow{\phi_{2}}\, .       
\end{eqnarray*}
}
}
\noindent
The proofs that Condition {\sf (C.2)}
(resp., Condition {\sf (C.3)})
implies Condition {\sf (C'.2)}
(resp., Condition {\sf (C'.3)})
are corresponding.
\end{proof}

\noindent
Recall that
the balanced mixture is an injective map
as long as the supports (but not the probabilities)
are concerned,
Hence,
Proposition~\ref{frombasictosymmetric}
implies that
{
\small
\textcolor{black}{
\begin{eqnarray*}
|{\mathcal{NE}}\left( {\mathsf{GHR}}({\mathsf{G}})
                            \right)|
&  \geq &
  |{\mathcal{NE}}{\mathsf{G}}|^{2}                        
  +
  2 \cdot
  |{\mathcal{NE}}\left( {\mathsf{G}}
                            \right)|\, .
\end{eqnarray*}
}
}  
\noindent
\textcolor{black}{(The squared term comes from
forming \textcolor{black}{the pair}
of balanced mixtures
${\bm{\rho}}
  \ast
  {\bm{\sigma}}$
and
${\bm{\sigma}}
  \ast
  {\bm{\rho}}$,
\textcolor{black}{for each pair}
${\bm{\rho}},
  {\bm{\sigma}}
  \in
  {\mathcal{NE}}\left( \left\langle {\mathsf{S}},
                                                      {\mathsf{S}}^{{\rm{T}}}
                                   \right\rangle
                          \right)$;
the linear term
comes from forming \textcolor{black}{the pair}
of balanced mixtures
${\bm{\rho}}
  \ast
  \left\langle {\mathsf{0}}^{n},
                     {\mathsf{0}}^{n}
  \right\rangle$
and
$\left\langle {\mathsf{0}}^{n},
                     {\mathsf{0}}^{n}
  \right\rangle
  \ast
  {\bm{\rho}}$,
for each
${\bm{\rho}}
  \in
  {\mathcal{NE}}\left( \left\langle {\mathsf{S}},
                                                      {\mathsf{S}}^{{\rm{T}}}
                                   \right\rangle
                          \right)$.)}                                
Proposition~\ref{fromsymmetrictobasic2}
implies that
a Nash equilibrium for
\textcolor{black}{$\langle {\mathsf{S}},
                                           {\mathsf{S}}^{\mbox{\rm{T}}}
 \rangle$}
is induced via the balanced mixture
either by a single Nash equilibrium
(Cases {\sf (C'.2)} and {\sf (C'.3)})
or by a pair of Nash equilibria 
(Case {\sf (C'.1)})
for ${\mathsf{G}}$.
Hence,
Proposition~\ref{fromsymmetrictobasic2}
establishes that
the balanced mixture is a surjective map,
so that
{
\small
\textcolor{black}{
\begin{eqnarray*}
|{\mathcal{NE}}\left( {\mathsf{GHR}}({\mathsf{G}})
                            \right)|
& = &
  |{\mathcal{NE}}\left( {\mathsf{G}}
                                     \right)|^{2}                                    
  +
  2 \cdot
  |{\mathcal{NE}}\left( {\mathsf{G}}
                            \right)|                                                                      \\
& = &
|{\mathcal{NE}}\left( {\mathsf{G}}
                                 \right)|
\cdot                                 
\left( |{\mathcal{NE}}\left( {\mathsf{G}}
                                 \right)|
          +
          2
\right)\, .                 
\end{eqnarray*}
}
} 
\noindent       
Thus,
\textcolor{black}{Theorems}~\ref{frombasictosymmetric}
and~\ref{fromsymmetrictobasic2}
provide together
a complete characterization 
of the Nash equilibria
for the ${\mathsf{GHR}}$-symmetrization
\textcolor{black}{$\langle {\mathsf{S}},
                    {\mathsf{S}}^{\mbox{\rm{T}}}
 \rangle$}
in terms of those for
$\left\langle {\mathsf{R}},
                    {\mathsf{C}}
 \right\rangle$;
so
computing a Nash equilibrium
for a win-lose bimatrix game
and 
computing a Nash equilibrium
for a symmetric win-lose bimatrix game
are polynomially equivalent problems.


\subsection{Complexity of \textcolor{black}{the} Search Problem}
\label{complexity of search problem}

\noindent
Here is an algorithm
to compute a Nash equilibrium
for a win-lose bimatrix game ${\mathsf{G}}$
with the positive utility property,
with a single invocation
of an algorithm
to compute a Nash equilibrium
for a symmetric win-lose bimatrix game
with the positive utility property:

{
\small
\begin{center}
\fbox{
\begin{minipage}{6.0in}
\begin{enumerate}

\item[{\sf (1)}]
Construct the win-lose
${\mathsf{GHR}}$-symmetrization
${\mathsf{GHR}}({\mathsf{G}})$.

\item[{\sf (2)}]
Compute a Nash equilibrium ${\bm{\phi}}$
for the symmetric win-lose game 
${\mathsf{GHR}}({\mathsf{G}})$.

\item[{\sf (3)}]
Recover a Nash equilibrium
for ${\mathsf{G}}$
as follows:
\begin{enumerate}

\item[{\sf (3.1)}]
If
$\overleftarrow{\phi_{1}},
 \overrightarrow{\phi_{1}},
 \overleftarrow{\phi_{2}},
 \overrightarrow{\phi_{2}}
 \neq
 {\mathsf{0}}^n$,
then output
$\langle \overleftarrow{\varphi_{2}},
              \overrightarrow{\varphi_{1}}
  \rangle$
(or  
$\langle \overleftarrow{\varphi_{1}},
             \overrightarrow{\varphi_{2}}
 \rangle$).

\item[{\sf (3.2)}]
If
$\overleftarrow{\phi_{1}}
 =
 \overrightarrow{\phi_{2}}
 =
 {\sf 0}^n$, 
then output
$\langle \overleftarrow{\varphi_{2}},
             \overrightarrow{\varphi_{1}}
  \rangle$.

\item[{\sf (3.3)}]
If
$\overrightarrow{\phi_{1}}
 =
 \overleftarrow{\phi_{2}}
 =
 {\mathsf{0}}^n$,
then output
$\langle \overleftarrow{\varphi_{1}},
              \overrightarrow{\varphi_{2}}
  \rangle$.

\end{enumerate}

\end{enumerate}
\end{minipage}
}
\end{center}
}

\noindent
Correctness 
follows from
Theorem~\ref{fromsymmetrictobasic2}.
Since computing a Nash equilirium
for a win-lose bimatrix game
with the positive utility property
is ${\mathcal{PPAD}}$-hard
(Section~\ref{winlose bimatrix with pup}), 
it immediately follows:

\begin{theorem}
\label{new result}
Computing a Nash equilirium
for a symmetric win-lose bimatrix game
with the positive utility property
is ${\mathcal{PPAD}}$-complete. 
\end{theorem}

\subsection{Complexity of \textcolor{black}{the} Counting Problem}
\label{complexity of counting problem}

\noindent
It was a consequence
of Propositions~\ref{if unsatisfied}
and~\ref{final lemma}
that \textcolor{black}{computing} the number of Nash equilibria
for a win-lose bimatrix game
is $\# {\mathcal{P}}$-hard.
We now \textcolor{black}{extend this result
to symmetric win-lose bimatrix games}.
We show:

\begin{theorem}
\label{sharp pi complete symmetric win-lose}
\textcolor{black}{Computing} the number 
\textcolor{black}{(resp., the parity of the number)}
of Nash equilibria
for a symmetric win-lose bimatrix game
is $\# {\mathcal{P}}$-complete
\textcolor{black}{(resp., $\oplus {\mathcal{P}}$-complete)}.
\end{theorem}

\begin{proof}
\textcolor{black}{Fix a win-lose gadget game
${\widehat{{\mathsf{G}}}}$
\textcolor{black}{and} a {\sf 3SAT} formula ${\mathsf{\phi}}$,
\textcolor{black}{inducing} the win-lose bimatrix game
\textcolor{black}{${\mathsf{G}}
  =
  {\mathsf{G}} ( {\widehat{{\mathsf{G}}}},
                           {\mathsf{\phi}}
                        )$}
and the symmetric win-lose bimatrix game
${\widetilde{{\mathsf{G}}}}
  =
  {\mathsf{GHR}}({\mathsf{G}})$.                           
By Lemma~\ref{f is PNE}
(Condition {\sf (1)})
and Propositions~\ref{if unsatisfied}
and~\ref{final lemma},
{
\small
\textcolor{black}{
\begin{eqnarray*}
          | {\mathcal{NE}}( {\mathsf{G}} ( {\widehat{{\mathsf{G}}}},
                                                                                   {\mathsf{\phi}}
                                                                    )         
                                            )
          |                                        
& = & | {\mathcal{NE}}( {\widehat{{\mathsf{G}}}}
                                            )
          |
          +
          \# {\mathsf{\phi}}\, .       
\end{eqnarray*}
}
}
\noindent
Hence,
by Theorems~\ref{frombasictosymmetric}
and~\ref{fromsymmetrictobasic2},
{
\small
\textcolor{black}{
\begin{eqnarray*}
&     &      | {\mathcal{NE}} ( {\widetilde{{\mathsf{G}}}}
                                            )
                |                                                                                                                                                                                                                  \\
& = & \underbrace{(    | {\mathcal{NE}}( {\widehat{{\mathsf{G}}}}
                                                                )
                                     |
                                     +
                                     \# {\mathsf{\phi}}
                               )^{2}}_{\mbox{balanced mixtures of two Nash equilibria}}
          +
          \underbrace{2
          \cdot
          (    | {\mathcal{NE}}( {\widehat{{\mathsf{G}}}}
                                          )
                |
                +
                \# {\mathsf{\phi}}
          )}_{\mbox{balanced mixtures of a Nash equilibrium and $\left\langle 0^{n}, 0^{n} \right\rangle$}}                                               \\
& = & (    | {\mathcal{NE}}( {\widehat{{\mathsf{G}}}}
                                           )
                |
                +
                \# {\mathsf{\phi}}
          )
          \cdot
          ( | {\mathcal{NE}}( {\widehat{{\mathsf{G}}}}
                                       )
             |
             +
             \# {\mathsf{\phi}}
             +
             2
          )\, ,
\end{eqnarray*}
}
}
which can be solved for $\# {\mathsf{\phi}}$.
Since computing $\# {\mathsf{\phi}}$
is $\# {\mathcal{P}}$-hard~\cite{V79},
the $\# {\mathcal{P}}$-hardness follows.
Note that
\textcolor{black}{
$\oplus  | {\mathcal{NE}}( {\widetilde{{\mathsf{G}}}}
                                         )
             | 
  =
  \oplus (  | {\mathcal{NE}}( {\widehat{{\mathsf{G}}}}
                                           )
                |
                +
                \# {\mathsf{\phi}}
           )$,}
from which
$\oplus {\mathsf{\phi}}$
can be computed.
Since computing $\oplus {\mathsf{\phi}}$
is $\oplus {\mathcal{P}}$-hard~\cite{PZ83},
the $\oplus {\mathcal{P}}$-hardness follows.                                                                          
}
\end{proof}

\section{Complexity \textcolor{black}{Results}}
\label{complexity results}

\noindent
Symmetric win-lose bimatrix games,
win-lose bimatrix games
and win-lose three-player games
are considered in Sections~\ref{symmetric winlose games},
\ref{winlose bimatrix games}
and~\ref{three player games},
respectively.
We first recall an informal summary
of the proof technique
from Section~\ref{introduction complexity results}.
The win-lose reduction
is used
to obtain, 
given a suitable win-lose gadget game
${\widehat{{\mathsf{G}}}}$
and a formula $\phi$, 
a win-lose game
${\mathsf{G}}
  =
  {\mathsf{G}}({\widehat{{\mathsf{G}}}},
                        \phi)$;
Propositions~\ref{if unsatisfied}
and~\ref{final lemma}
are used to relate the properties
of the Nash equilibria for ${\mathsf{G}}$
to the satisfiablity of ${\mathsf{\phi}}$.
For the case of symmetric win-lose bimatrix games,
the ${\mathsf{GHR}}$-symmetrization
from Section~\ref{winlose gkt symmetrization}
is used to obtain from ${\mathsf{G}}$
the symmetric win-lose game
${\widetilde{{\mathsf{G}}}}
  =
  {\mathsf{GHR}}({\mathsf{G}})$;
Theorems~\ref{frombasictosymmetric}
and~\ref{fromsymmetrictobasic2}
are used to relate the properties
of the Nash equilibria
for ${\widetilde{{\mathsf{G}}}}$
to those of the Nash equilibria
for ${\mathsf{G}}$. 
In turn,
this \textcolor{black}{results in} relating
the properties
of the Nash equilibria
for ${\widetilde{{\mathsf{G}}}}$
to the satisfiability of ${\mathsf{\phi}}$. 

\noindent
Recall that
for bimatrix games,
the win-lose reduction sets that
the strategy set of each player 
$i \in [2]$ in
${\mathsf{G}}$ 
is
${\mathsf{\Sigma}}_{i}({\mathsf{G}})
  =
  {\widehat{{\mathsf{\Sigma}}}}_{i}
  \cup
  {\mathsf{L}}
  \cup
  {\mathsf{V}}
  \cup
  {\mathcal{C}}$. 

\subsection{Symmetric Win-Lose Bimatrix Games}
\label{symmetric winlose games}

\noindent
Recall that the \textcolor{black}{win-lose} ${\mathsf{GHR}}$-symmetrization
\textcolor{black}{mirrors}
the strategies of each player
in a bimatrix game.
So,
for each player $i \in [2]$, 
${\mathsf{\Sigma}}_{i}({\widetilde{{\mathsf{G}}}})
  =
  {\widehat{{\mathsf{\Sigma}}}}_{i}
  \cup
  {\mathsf{L}}
  \cup
  {\mathsf{V}}
  \cup
  {\mathcal{C}}
  \cup
  {\widehat{{\mathsf{\Sigma}}}}^{\prime}_{i}
  \cup
  {\mathsf{L}}^{\prime}
  \cup
  {\mathsf{V}}^{\prime}
  \cup
  {\mathcal{C}}^{\prime}$, 
where
the strategy sets 
${\widehat{{\mathsf{\Sigma}}}}^{\prime}_{i}$, 
${\mathsf{L}}^{\prime}$, 
${\mathsf{V}}^{\prime}$ 
and
${\mathcal{C}}^{\prime}$ 
are \textcolor{black}{mirrors of
the strategy sets}
${\widehat{{\mathsf{\Sigma}}}}_{i}$, 
${\mathsf{L}}$, 
${\mathsf{V}}$ 
and
${\mathcal{C}}$, 
respectively. 
We show:

\begin{theorem}
\label{mainextended}
Restricted to symmetric win-lose bimatrix games, 
the following decision problems are ${\mathcal{NP}}$-complete:
\begin{center}
\begin{small}
\begin{tabular}{|l|l|}
\hline
\hline 
  Group I                                                                              & Group III                                                                                                    \\
 \hline
 {\sf $\exists$ NASH WITH SMALL UTILITIES}                    & {\sf $\exists$ NASH WITH LARGE UTILITIES}                                         \\
 {\sf $\exists$ NASH WITH SMALL TOTAL UTILITY}            & {\sf $\exists$ NASH WITH LARGE TOTAL UTILITY}                                 \\
 {\sf $\exists$ NASH WITH LARGE SUPPORTS}                  & {\sf $\exists$ NASH WITH SMALL SUPPORTS}                                        \\
 \cline{2-2}
 \cline{2-2}
 {\sf $\exists$ NASH WITH RESTRICTING SUPPORTS}      &  Group IV                                                                                                   \\
 \cline{2-2}
 {\sf $\exists$ NASH WITH RESTRICTED SUPPORTS}        & {\sf $\exists$ $k+1$ NASH} (with $k \geq 1$)                                        \\
 {\sf $\exists$ NASH WITH SMALL PROBABILITIES}          & {\sf $\exists$ $\neg$ PARETO-OPTIMAL NASH}                                     \\
  {\sf $\exists$ $\neg$ UNIFORM NASH}                           & {\sf $\exists$ $\neg$ STRONGLY PARETO-OPTIMAL NASH}                   \\   
 \cline{1-1}
 Group II                                                                             &  {\sf $\exists$ $\neg$ SYMMETRIC NASH}                                              \\
 \hline
 {\sf $\exists$ UNIFORM NASH}                                        &                                                                                                                  \\
\hline
\hline
\end{tabular}
\end{small}
\end{center}
Furthermore,
their counting versions are
$\# {\mathcal{P}}$-complete;
\textcolor{black}{so is
{\sf $\#$ SYMMETRIC NASH}.}
\textcolor{black}{Except for {\sf $\oplus$ $\neg$ UNIFORM NASH}
(whose $\oplus {\mathcal{P}}$-hardness remains open)
and for {\sf $\oplus$ $\neg$ SYMMETRIC NASH}
(which is in ${\mathcal{P}}$),
their parity versions are $\oplus {\mathcal{P}}$-complete.}       
\end{theorem}

\noindent
For each of the four Groups,
we fix a different gadget game ${\widehat{{\mathsf{G}}}}$.
For {\it Group I,}
$\widehat{{\mathsf{G}}}$ is fixed
to the win-lose cyclic game 
$\widehat{{\sf G}}_{1}[1]$
(Section~\ref{cyclic game}).
For {\it Group II},
$\widehat{{\mathsf{G}}}$
is fixed to the win-lose non-uniform game
$\widehat{{\mathsf{G}}}_{3}$
(Section~\ref{nonuniform game}).
For {\it Group III},
$\widehat{{\mathsf{G}}}$
is fixed to the win-lose cyclic game
$\widehat{{\mathsf{G}}}_{1}[h]$,
for any integer $h$
with $h > \frac{\textstyle n}
                        {\textstyle 2}$
(Section~\ref{cyclic game}).                        
For {\it Group IV,}
${\widehat{{\mathsf{G}}}}$
is fixed to the win-lose cyclic game
${\widehat{{\mathsf{G}}}}_{1}[2]$
(Section~\ref{cyclic game});
we also use
the win-lose diagonal game
${\widehat{{\mathsf{G}}}}_{5}[k]$
as a subgame
in the last stage of the reduction.
Regarding {\it Step 3}
in the \textcolor{black}{three-steps proof plan outlined}
in Section~\ref{introduction complexity results},
the decision problems
in {\it Group I} and {\it Group III}
fall under Case {\sf (1)};
the single decision problem
in {\it Group II}
falls under Case {\sf (2)};
the decision problems
in {\it Group IV}
fall under Case {\sf (3)}.

\noindent
Notation-wise,
we shall use,
for the proof of Theorem~\ref{mainextended}, 
the symbols
{\sf P} and {\sf Q}
to denote properties of \textcolor{black}{Nash equilibria
for} the game
\textcolor{black}{${\mathsf{GHR}}( {\mathsf{G}}( {\widehat{{\mathsf{G}}}},
                                                                    \phi
                                                    )
                           )$}                                           
in the cases where $\phi$
is unsatisfiable and satisfiable,
respectively. 
\textcolor{black}{Loosely speaking,
we shall use,
to establish ${\mathcal{NP}}$-hardness,
a {\sf Q} property
and a \textcolor{black}{contradictory}
{\sf P} property,
\textcolor{black}{which together disentangle satisfiability;}
\textcolor{black}{that is,}
the {\sf Q} property holds
if and only if ${\mathsf{\phi}}$
is satisfiable,
and this implies the ${\mathcal{NP}}$-hardness
of deciding the corresponding property.
To establish $\# {\mathcal{P}}$-hardness,
we shall express the number of Nash equilibria
for \textcolor{black}{${\mathsf{GHR}}( {\mathsf{G}}( {\widehat{{\mathsf{G}}}},
                                                                         \phi
                                                         )
                           )$} 
fulfilling the property {\sf Q}                   
as a function of $\# {\mathsf{\phi}}$;
inverting \textcolor{black}{the formula}
yields $\# {\mathsf{\phi}}$,
and $\# {\mathcal{P}}$-hardness follows.
\textcolor{black}{To establish $\oplus {\mathcal{P}}$-hardness,
we shall express the parity of
the number of Nash equilibria for
\textcolor{black}{${\mathsf{GHR}}( {\mathsf{G}}( {\widehat{{\mathsf{G}}}},
                                                                         \phi
                                                                )
                           )$} 
fulfilling the property {\sf Q}                   
as a function of $\oplus {\mathsf{\phi}}$;
inverting the expression
yields $\oplus {\mathsf{\phi}}$
and $\oplus {\mathcal{P}}$-hardness follows.
We shall use ${\mathsf{U}}$
and ${\widetilde{{\mathsf{U}}}}$
to denote (expected) utilities
for the games 
\textcolor{black}{${\mathsf{G}}
  =
  {\mathsf{G}}( {\widehat{{\mathsf{G}}}},
                                \phi
                      )$}
and
\textcolor{black}{
${\widetilde{{\mathsf{G}}}}
  =
  {\mathsf{GHR}}( {\mathsf{G}}( {\widehat{{\mathsf{G}}}},
                                                                         \phi
                                                  )
                           )$,} 
respectively.        
}                                                            
For a game ${\mathsf{G}}$,
denote as $\kappa ({\mathsf{G}})$
the maximum number of strategies
for each player in the game ${\mathsf{G}}$.}     

\begin{proof}
Consider a {\sf 3SAT} formula $\phi$
with
\textcolor{black}{$n := |{\mathsf{Var}}({\mathsf{\phi}})| \geq 5$}
\textcolor{black}{and}
a win-lose gadget game $\widehat{{\mathsf{G}}}$,
\textcolor{black}{and \textcolor{black}{their} induced}
win-lose game ${\mathsf{G}}
                                =
                                {\mathsf{G}}({\widehat{{\mathsf{G}}}},
                                                      \phi)$
constructed by the win-lose reduction
(Section~\ref{winlose reduction}),
and \textcolor{black}{the} win-lose 
${\mathsf{GHR}}$-symmetrization
$\widetilde{{\mathsf{G}}}
 :=
 {\mathsf{GHR}}\left( {\mathsf{G}}
                           \right)$
(Section~\ref{winlose gkt symmetrization}).
\textcolor{black}{Clearly,
\textcolor{black}{$\kappa ( {\mathsf{G}}({\widehat{{\mathsf{G}}}},
                                              \phi)
              )
  =
  \kappa ( {\widehat{{\mathsf{G}}}}
              )
  +
  2n + |{\mathcal{C}}(\phi)| + n(n+1)$}.
Hence,  
${\mathsf{G}} =
  {\mathsf{G}}({\widehat{{\mathsf{G}}}},
                         \phi)$
has size polynomial                                       
in the size of ${\mathsf{\phi}}$
if and only if 
${\widehat{{\mathsf{G}}}}$
has size polynomial
in the size of ${\mathsf{\phi}}$.
It follows,
by the win-lose ${\mathsf{GHR}}$-symmetrization,
that
$\widetilde{{\mathsf{G}}}$
has size polynomial                                       
in the size of ${\mathsf{\phi}}$
if and only if 
${\widehat{{\mathsf{G}}}}$
has size polynomial
in the size of ${\mathsf{\phi}}$.}

\textcolor{black}{Assume first that
$\phi$
is unsatisfiable.
Proposition~\ref{if unsatisfied} implies that
\textcolor{black}{${\mathcal{NE}}( {\mathsf{G}}({\widehat{{\mathsf{G}}}},
                                                           \phi)
                          )
  =
  {\mathcal{NE}}(  {\widehat{{\mathsf{G}}}}
                          )$};
call each Nash equilibrium
for ${\mathsf{G}}({\widehat{{\mathsf{G}}}},
                               \phi)$
coming from ${\widehat{{\mathsf{G}}}}$
a {\it gadget equilibrium}.}
\textcolor{black}{So there are
$|{\mathcal{NE}}({\widehat{{\mathsf{G}}}})|$
gadget equilibria for ${\mathsf{G}}$.}
\textcolor{black}{By Theorem~\ref{frombasictosymmetric}, 
${\widetilde{{\mathsf{G}}}}
  =
  {\mathsf{GHR}}\left( {\mathsf{G}}
                            \right)$
has the following Nash equilibria,}
\textcolor{black}{which are balanced mixtures 
of the gadget equilibria
for ${\mathsf{G}}$:}
\begin{enumerate}

\item[{\sf (1)}]                                                                                                                                                                                                                                                                                                                                                                                     
\textcolor{black}{${\bm\sigma}^{1}
  =
  {\widehat{\bm{\sigma}}}
  \ast
  {\widehat{\bm{\tau}}}$,
for each ordered pair 
$\langle {\widehat{\bm{\sigma}}},
              {\widehat{\bm{\tau}}}
  \rangle$
of the
gadget equilibria
${\widehat{\bm{\sigma}}}$
and
${\widehat{\bm{\tau}}}$
for ${\mathsf{G}}$.}

\item[{\sf (2)}]
\textcolor{black}{${\bm\sigma}^{2}
  =
  {\widehat{{\bm{\sigma}}}}
  \ast
  \left\langle 0^{\kappa},
                    0^{\kappa}
  \right\rangle
  =
  \left\langle {\widehat{\sigma}}_{1} 
                    \circ
                    {\mathsf{0}}^{\kappa},
                    {\mathsf{0}}^{\kappa} 
                    \circ 
                    {\widehat{\sigma}}_{2}
  \right\rangle$
and
${\bm\sigma}^{3}
  =
  \left\langle {\mathsf{0}}^{\kappa},
                    {\mathsf{0}}^{\kappa}
  \right\rangle                   
  \ast
  {\widehat{{\bm{\sigma}}}}
  =  
  \left\langle {\mathsf{0}}^{\kappa} \circ {\widehat{\sigma}}_{2},
                    {\widehat{\sigma}}_{1} \circ {\mathsf{0}}^{\kappa}
  \right\rangle$,
for each gadget equilibrium
${\widehat{{\bm{\sigma}}}}$
for ${\mathsf{G}}$.}

\end{enumerate}
\noindent
\textcolor{black}{By Theorem~\ref{fromsymmetrictobasic2},
${\widetilde{{\mathsf{G}}}}$
has no other Nash equilibrium.
Hence,
when ${\mathsf{\phi}}$ is unsatisfiable,
there are 
\textcolor{black}{$|{\mathcal{NE}}({\widehat{{\mathsf{G}}}})|^{2}
  +
  2 \cdot |{\mathcal{NE}}({\widehat{{\mathsf{G}}}})|$}
Nash equilibria for ${\widetilde{{\mathsf{G}}}}$.}

\textcolor{black}{Assume now that
${\mathsf{\phi}}$
is satisfiable.                                                                                                                                                                                                                                                                                                                                                                  
Then,
in addition to the gadget equilibria,
it follows,
by Proposition~\ref{final lemma},
that
${\mathsf{G}}({\widehat{{\mathsf{G}}}},
                         \phi)$
has,
for each satisfying assignment ${\mathsf{\gamma}}$
of ${\mathsf{\phi}}$,
a Nash equilibrium
${\bm{\sigma}}
  =
  {\bm{\sigma}}({\mathsf{\gamma}})$;
call ${\bm{\sigma}}$
a {\it literal equilibrium}.}                                                                                                                                                                                                                                                                                                                                                                                                                                                                                                                   
\textcolor{black}{So there are
$\# {\mathsf{\phi}}$
literal equilibria
for ${\mathsf{G}}$.}
\textcolor{black}{By Theorem~\ref{frombasictosymmetric}, 
${\widetilde{{\mathsf{G}}}}
  =
  {\mathsf{GHR}}\left( {\mathsf{G}}
                            \right)$
has,
in addition to the balanced mixtures
of gadget equilibria,}
\textcolor{black}{the following Nash equilibria,
called {\it additional,}
which are either balanced mixtures of
literal equilibria for ${\mathsf{G}}$
or balanced mixtures
of a literal and a gadget equilibrium for ${\mathsf{G}}$:}
\begin{enumerate}

\item[{\sf (1)}]
\textcolor{black}{${\bm\sigma}^{4}
  =
  {\bm{\sigma}}
  \ast
  {\bm{\tau}}$,}
\textcolor{black}{for each ordered pair
$\left\langle {\bm{\sigma}},
                     {\bm{\tau}}
  \right\rangle$
of the literal equilibria
${\bm{\sigma}}$
and
${\bm{\tau}}$
for ${\mathsf{G}}$.}

\item[{\sf (2)}]  
\textcolor{black}{${\bm\sigma}^{5}
  =
  {\bm{\sigma}}
  \ast
  \left\langle 0^{\kappa}, 0^{\kappa}
  \right\rangle$
and
${\bm\sigma}^{6}
  =
  \left\langle 0^{\kappa}, 0^{\kappa}
  \right\rangle
  \ast
  {\bm{\sigma}}$,}
\textcolor{black}{for each literal equilibrium
${\bm{\sigma}}$                                                                                                                                                                                                                                                                                                                                                                                                                                                                                                                                                                                                              
for ${\mathsf{G}}$.}

\item[{\sf (3)}]
\textcolor{black}{${\bm{\sigma}}^{7}
  =
  {\widehat{{\bm{\sigma}}}}
  \ast
  {\bm{\sigma}}$
and
${\bm{\sigma}}^{8}
  =
  {\bm{\sigma}}
  \ast
  {\widehat{{\bm{\sigma}}}}$,}
\textcolor{black}{for each pair  
of a gadget equilibrium ${\widehat{{\bm{\sigma}}}}$
and a literal
\textcolor{black}{equilibrium} ${\bm{\sigma}}$
for ${\mathsf{G}}$.}

\end{enumerate}
\textcolor{black}{Since there are
$|{\mathcal{NE}}({\widehat{{\mathsf{G}}}})|$
gadget equilibria
and
$\# {\mathsf{\phi}}$
literal equilibria,
it follows that
there are
$\left( \# {\mathsf{\phi}} \right)^{2}
  +
  2 \cdot \# {\mathsf{\phi}}
  =
  \# {\mathsf{\phi}}
  \cdot
  \left( \# {\mathsf{\phi}} + 2
  \right)$
balanced mixtures
of literal equilibria
and
$2
  \cdot
  |{\mathcal{NE}}({\widehat{{\mathsf{G}}}})|
  \cdot
  \# {\mathsf{\phi}}$
balanced mixtures of
a gadget equilibrium
and a literal equilibrium,
respectively.}
\textcolor{black}{By Theorem~\ref{fromsymmetrictobasic2},
${\widetilde{{\mathsf{G}}}}$
has no other Nash equilibrium.}

\textcolor{black}{We proceed to establish properties
of the Nash equilibria
for  ${\mathsf{G}}$ and ${\widetilde{{\mathsf{G}}}}$,
respectively.}
\textcolor{black}{Loosely speaking,
these shall be properties disentangling the satisfiability
of ${\mathsf{\phi}}$;
thus, ${\mathcal{NP}}$-hardness follows.} 
\textcolor{black}{We split the proof into four parts,
one for each Group.}

\noindent
\underline{{\it Group I:}}
\textcolor{black}{We start with a particular remark about 
{\sf $\exists$ $\neg$ UNIFORM NASH}.} 
\textcolor{black}{Proposition~\ref{final lemma} ensures that
when ${\mathsf{\phi}}$
is satisfiable,
the literal Nash equilibria
for ${\mathsf{G}}({\widehat{\mathsf{G}}}, {\mathsf{\phi}})$
dismatch non-uniformity
as they are uniform
no matter how ${\widehat{{\mathsf{G}}}}$
were chosen.
So
the win-lose reduction
is inadequate on its own
to ensure the equivalence
of non-uniformity
to the satisfiability of ${\mathsf{\phi}}$.
Although it might seem that
techniques based on the win-lose reduction
could not be adequate for showing
the ${\mathcal{NP}}$-hardness
of deciding the existence of a non-uniform 
Nash equilibrium,
we shall establish
that this is not the case.
In {\it Step 1,}
choose the gadget game ${\widehat{{\mathsf{G}}}}$
to have a single uniform Nash equilibrium.
For {\it Step 2,}
Proposition~\ref{if unsatisfied}
ensures that
${\mathsf{G}}({\widehat{\mathsf{G}}}, {\mathsf{\phi}})$
has no non-uniform
Nash equilibrium
when ${\mathsf{\phi}}$
is unsatisfiable.
\textcolor{black}{Hence,
Theorems~\ref{frombasictosymmetric}
and~\ref{fromsymmetrictobasic2}
ensure that
${\widetilde{{\mathsf{G}}}}$
has no non-uniform
Nash equilibrium
when ${\mathsf{\phi}}$ is unsatisfiable.}
But \textcolor{black}{when ${\mathsf{\phi}}$
is satisfiable,}
in {\it Step 3,}
non-uniform Nash equilibria
are created for the win-lose ${\mathsf{GHR}}$-symmetrization
of ${\mathsf{G}}({\widehat{\mathsf{G}}}, {\mathsf{\phi}})$
as balanced mixtures
of the single uniform Nash equilibrium \textcolor{black}{for} ${\mathsf{G}}$
(coming from ${\widehat{{\mathsf{G}}}}$)
with some \textcolor{black}{literal}
Nash equilibrium
for ${\mathsf{G}}({\widehat{\mathsf{G}}}, {\mathsf{\phi}})$.
So
the win-lose ${\mathsf{GHR}}$-symmetrization
has a non-uniform Nash equilibrium
if and only if ${\mathsf{\phi}}$ is satisfiable,
and ${\mathcal{NP}}$-hardness follows.
\textcolor{black}{We stress that
the ${\mathcal{NP}}$-hardness
of  {\sf $\exists$ $\neg$ UNIFORM NASH}
for symmetric win-lose bimatrix games
exploits the possibility of creating
non-uniform Nash equilibria
as balanced mixtures
of uniform Nash equilibria.
Note that
this result is subsuming the ${\mathcal{NP}}$-hardness
of {\sf $\exists$ $\neg$ UNIFORM NASH}
for the more general class
of win-lose bimatrix games,
for which the win-lose reduction
could accomodate no direct proof.}
We now continue with the formal proof.}

Fix ${\widehat{{\mathsf{G}}}}
        :=
        {\widehat{{\mathsf{G}}}}_{1}[1]$.
\textcolor{black}{Since $\kappa \left( {\widehat{{\mathsf{G}}}}_{1}[1]
                       \right)
           =
           1$,
it follows that                            
each of
${\mathsf{G}} =
  {\mathsf{G}}({\widehat{{\mathsf{G}}}},
                         \phi)$
and
$\widetilde{{\mathsf{G}}}
 =
 {\mathsf{GHR}}\left( {\mathsf{G}}
                           \right)$
has size polynomial                                       
in the size of ${\mathsf{\phi}}$.}   
\textcolor{black}{By Proposition~\ref{gadget1},}
$\widehat{{\sf G}}_{1}[1]$
has a unique Nash equilibrium
$\widehat{\bm{\sigma}}$. 
Assume first that $\phi$
is unsatisfiable.
Since
${\mathcal{NE}}({\mathsf{G}}({\widehat{{\mathsf{G}}}}_{1}[1],
                                                   {\mathsf{\phi}}))
   =
   {\mathcal{NE}}(  {\widehat{{\mathsf{G}}}}_{1}[1])$,
it follows that
${\mathsf{G}}({\widehat{{\mathsf{G}}}}_{1}[1],
                         \phi)$
has a unique 
Nash equilibrium,
the gadget equilibrium
$\widehat{\bm{\sigma}}$,
which has,
by Proposition~\ref{gadget1},
the following properties:
\begin{quote}
\textcolor{black}{For each player $i \in [2]$:
${\mathsf{U}}_{i}({\widehat{\bm{\sigma}}})
 =
 1$,
so that 
$\sum_{i \in [2]}
    {\mathsf{U}}_{i}({\widehat{\bm{\sigma}}})
 =
 2$;
${\mathsf{Supp}}\left( {\widehat{\sigma}}_{i} \right)
  =
  {\widehat{{\mathsf{\Sigma}}}}_{i}$,
with   
$|{\mathsf{Supp}}\left( {\widehat{\sigma}}_{i}
                 \right)|
 =
 1$,
so that
for each strategy
$s \in {\mathsf{Supp}}(\sigma_{i})$,
${\widehat{\sigma}}_{i}(s) = 1$;
${\bm{\sigma}}$
is uniform.}
\end{quote}
\noindent
Hence,
${\widetilde{{\mathsf{G}}}}$
has exactly three Nash equilibria,
${\bm{\sigma}}^{1}$,
${\bm{\sigma}}^{2}$
and
${\bm{\sigma}}^{3}$,
\textcolor{black}{such that:} 
\begin{quote}
\textcolor{black}{For each player $i \in [2]$:
${\widetilde{{\mathsf{U}}}}_{i}(\bm{\sigma}^{1})
 =
 \frac{\textstyle 1}
         {\textstyle 2}$,
so that
$\sum_{i \in [2]}
    {\widetilde{{\mathsf{U}}}}_{i}(\bm{\sigma}^{1})
 =
 1$;
${\mathsf{Supp}}\left( \sigma_{i}^{1} \right)
  =
  {\widehat{{\mathsf{\Sigma}}}}_{i}
  \cup
  {\widehat{{\mathsf{\Sigma}}}}_{i}^{\prime}$,
with   
$|{\mathsf{Supp}}\left( {\widehat{\sigma}}_{i}^{1}
                 \right)|
 =
 2$;
for each strategy
$s \in {\mathsf{Supp}}(\sigma_{i}^{1})$,
${\widehat{\sigma}}_{i}^{1}(s) 
  =
  \frac{\textstyle 1}
          {\textstyle 2}$;
${\bm{\sigma}}^{1}$
is uniform.}
\end{quote}
\begin{quote}
\textcolor{black}{For each 
${\bm{\sigma}} 
   \in
   \left\{ {\bm{\sigma}}^{2},
              {\bm{\sigma}}^{3}
   \right\}$:
For each player $i \in [2]$:
${\widetilde{{\mathsf{U}}}}_{i}(\bm{\sigma})
 =
 1$,
so that
$\sum_{i \in [2]}
    {\widetilde{{\mathsf{U}}}}_{i}(\bm{\sigma})
 =
 2$;
either
${\mathsf{Supp}}\left( \sigma_{i} \right)
  =
  {\widehat{{\mathsf{\Sigma}}}}_{i}$
or
${\mathsf{Supp}}\left( \sigma_{i} \right)
  =
  {\widehat{{\mathsf{\Sigma}}}}_{i}^{\prime}$,
with    
$|{\mathsf{Supp}}\left( {\widehat{\sigma}}_{i}
                 \right)|
 =
 1$;
for each strategy
$s \in {\mathsf{Supp}}(\sigma_{i})$,
${\widehat{\sigma}}_{i}(s) 
  =
  1$;
${\bm{\sigma}}$
is uniform.}
\end{quote}

\noindent
Hence,
when ${\mathsf{\phi}}$
is unsatisfiable,
each Nash equilibrium
${\bm{\sigma}}$
for ${\widetilde{{\mathsf{G}}}}$
has the following properties:
\begin{quote}
\textcolor{black}{For each player $i \in [2]$:
{\sf (P.1)}
${\widetilde{{\mathsf{U}}}}_{i}(\bm{\sigma})
 \geq
 \frac{\textstyle 1}
         {\textstyle 2}$,
so that
{\sf (P.2)}
$\sum_{i \in [2]}
    {\widetilde{{\mathsf{U}}}}_{i}(\bm{\sigma})
 \geq
 1$;
{\sf (P.3)}
${\mathsf{Supp}}\left( \sigma_{i} 
                            \right)
  \subseteq
  {\widehat{{\mathsf{\Sigma}}}}_{i}
  \cup
  {\widehat{{\mathsf{\Sigma}}}}_{i}^{\prime}$,
with
{\sf (P.4)}
$|{\mathsf{Supp}}\left( \sigma_{i} \right)|
  \leq
  2$;
{\sf (P.5)}  
for each strategy
$s \in {\mathsf{Supp}}\left( \sigma_{i} \right)$,
$\sigma_{i}(s)
  \geq
  \frac{\textstyle 1}
          {\textstyle 2}$;
{\sf (P.6)}
${\bm{\sigma}}$ is uniform.}          
\end{quote}

\noindent
Assume now that
$\phi$ is satisfiable.
Then,
by Proposition~\ref{final lemma},
the literal equilibrium
${\bm{\sigma}} = {\bm{\sigma}}({\mathsf{\gamma}})$,
for a satisfying assignement 
${\mathsf{\gamma}}$ of ${\mathsf{\phi}}$,
has the following properties:
\begin{quote}
\textcolor{black}{
For each player $i \in [2]$:
${\mathsf{U}}_{i}(\bm{\sigma})
 =
 \frac{\textstyle 2}
         {\textstyle n}$,
so that
$\sum_{i \in [2]}
    {\mathsf{U}}_{i}(\bm{\sigma})
 =
 \frac{\textstyle 4}
         {\textstyle n}$;
${\mathsf{Supp}}\left( \sigma_{i}
                \right)
 =
 {\mathsf{L}}$,
with
$|{\mathsf{Supp}}\left( \sigma_{i}
                 \right)|
 =
 n$;
for each strategy
$s \in {\mathsf{Supp}}(\sigma_{i})$,
$\sigma_{i}(s)
  =
  \frac{\textstyle 1}
          {\textstyle n}$;  
${\bm{\sigma}}$
is uniform.}        
\end{quote}          
Here are the additional Nash equilibria
for ${\widetilde{{\mathsf{G}}}}$
and their properties:
\begin{itemize}

\item
\underline{The $\left( \# {\mathsf{\phi}} \right)^{2}$ balanced mixtures
of literal equilibria for ${\mathsf{G}}$:}
\begin{quote}
\textcolor{black}{
For each player $i \in [2]$:
${\widetilde{{\mathsf{U}}}}_{i}(\bm{\sigma}^{4})
 =
 \frac{\textstyle 1}
         {\textstyle n}$,
so that         
$\sum_{i \in [2]}
    {\widetilde{{\mathsf{U}}}}_{i}(\bm{\sigma}^{4})
 =
 \frac{\textstyle 2}
         {\textstyle n}$;
${\mathsf{Supp}}\left( \sigma_{i}^{4}
                \right)
 \subset
 {\mathsf{L}}
 \cup
 {\mathsf{L}}^{\prime}$,
with
$|{\mathsf{Supp}}\left( \sigma_{i}^{4}
                              \right)|
 =
 2n$;
for each strategy
$s \in {\mathsf{Supp}}(\sigma_{i}^{4})$,
$\sigma_{i}^{4}(s)
  =
  \frac{\textstyle 1}
          {\textstyle 2n}$;  
${\bm{\sigma}}^{4}$
is uniform.}
\end{quote}          
\begin{quote}
\textcolor{black}{For each
${\bm{\sigma}}
  \in 
  \{ {\bm{\sigma}}^{5},
      {\bm{\sigma}}^{6}
  \}$:
For each player $i \in [2]$:
${\widetilde{{\mathsf{U}}}}_{i}(\bm{\sigma})
 =
 \frac{\textstyle 2}
         {\textstyle n}$,
so that
$\sum_{i \in [2]}
    {\widetilde{{\mathsf{U}}}}_{i}(\bm{\sigma})
 =
 \frac{\textstyle 4}
         {\textstyle n}$;
${\mathsf{Supp}}\left( \sigma_{i}
                             \right)
 \subset
 {\mathsf{L}}
 \cup
 {\mathsf{L}}^{\prime}$,
with
$|{\mathsf{Supp}}\left( \sigma_{i}
                 \right)|
 =
 n$;
for each strategy
$s \in {\mathsf{Supp}}(\sigma_{i})$,
$\sigma_{i}(s)
  =
  \frac{\textstyle 1}
          {\textstyle n}$;
${\bm{\sigma}}$
is uniform.}
\end{quote}

\item
\underline{\textcolor{black}{The $2 \cdot \# {\mathsf{\phi}}$ balanced mixtures of
a gadget equilibrium
and a literal equilibrium for ${\mathsf{G}}$:}}
\textcolor{black}{Note that}
{
\small
\begin{eqnarray*}
           {\bm\sigma}^{7}                                                                                                   
&  = & \left\langle \frac{\frac{\textstyle 2}
                            {\textstyle n}}
                            {\textstyle 1 + \frac{\textstyle 2}
                                                            {\textstyle n}}
                    {\widehat{\sigma}}_{1} 
                    \circ
                    \frac{\textstyle 1}
                            {\textstyle 1 + \frac{\textstyle 2}
                                                             {\textstyle n}}
                    \sigma_{2},\,
                     \frac{\textstyle 1}
                            {\textstyle 1 + \frac{\textstyle 2}
                                                             {\textstyle n}}
                    \sigma_{1} 
                    \circ
                    \frac{\frac{\textstyle 2}
                            {\textstyle n}}
                            {\textstyle 1 + \frac{\textstyle 2}
                                                            {\textstyle n}}
                    {\widehat{\sigma}}_{2}
            \right\rangle                                                                                                                                         \\
& = & \left\langle \frac{\textstyle 2}
                                      {\textstyle n + 2}
                             {\widehat{\sigma}}_{1} 
                             \circ
                             \frac{\textstyle n}
                                     {\textstyle n+2}
                             \sigma_{2},\,
                             \frac{\textstyle n}
                                     {\textstyle n+2}
                             \sigma_{1} 
                             \circ
                             \frac{\textstyle 2}
                                    {\textstyle n+2}
                            {\widehat{\sigma}}_{2}
           \right\rangle\, ,       
\end{eqnarray*}
}
and
{
\small
\begin{eqnarray*}
            {\bm\sigma}^{8}
&  = &  \left\langle \frac{\textstyle n}
                                     {\textstyle n+2}
                               \sigma_{1} 
                               \circ
                               \frac{\textstyle 2}
                                       {\textstyle n+2}
                              {\widehat{\sigma}}_{2},\,
                              \frac{\textstyle 2}
                                      {\textstyle n + 2}
                             {\widehat{\sigma}}_{1} 
                             \circ
                             \frac{\textstyle n}
                                     {\textstyle n+2}
                             \sigma_{2}
            \right\rangle\, ,     
\end{eqnarray*}
}
respectively.
\textcolor{black}{
Consider first
${\bm{\sigma}}^{7} = {\bm{\sigma}}^{7}({\mathsf{\gamma}})$,
for a satisfying assignment ${\mathsf{\gamma}}$
of ${\mathsf{\phi}}$.
Then,
by the definition of
the win-lose ${\mathsf{GHR}}$-symmetrization,
the row (resp., column) player
may get utility $1$
only in two cases:}
\begin{itemize}

\item
\textcolor{black}{
\underline{The row player plays ${\mathsf{L}}^{\prime}$
and the column player
plays ${\mathsf{L}}$.}
By Proposition~\ref{final lemma}
(Condition {\sf (C.4)}),
for each literal $\ell \in {\mathsf{\gamma}}$,
$\sigma_{1}(\ell)
  =
  \sigma_{2}(\ell)
  =
  \frac{\textstyle 1}
          {\textstyle n}$.
Thus,
by the balanced mixture ${\bm{\sigma}}^{7}$,
for each player $i \in [2]$,
for each literal
$\ell \in {\mathsf{\gamma}}$,
{
\small
\begin{eqnarray*}
           \sigma_{i}^{7}(\ell)
& = &
  \frac{\textstyle n}
          {\textstyle n+2} 
  \cdot
  \sigma_{\overline{i}}(\ell)\ \
  =\ \
  \frac{\textstyle n}
          {\textstyle n+2}
  \cdot
  \frac{\textstyle 1}
          {\textstyle n}\ \
  =\ \
  \frac{\textstyle 1}
          {\textstyle n+2}\, .
\end{eqnarray*}
}
\noindent          
By the utility functions,
there are two subcases
in which the row (resp., column) player
gets utility $1$:
The row player plays $\ell$,
the column player plays $\ell'$
with $\ell' \neq \overline{\ell}$,
and
$I(\ell') - I(\ell) \in \{ 0, 1 \}$
(resp,
$I(\ell') - I(\ell) \in \{ 2, 3 \}$).
The two subcases occur with probability
$\sum_{\ell, \ell^{\prime} \in {\mathsf{\gamma}} 
             \mid
             I(\ell') - I(\ell) \in \{ 0, 1 \}
            }
     \sigma_{1}(\ell)
     \cdot
     \sigma_{2}(\ell)
  =
  2n
  \cdot
  \left( \frac{\textstyle 1}
                   {\textstyle n+2}
   \right)^{2}$                          
(resp.,
$\sum_{\ell, \ell^{\prime} \in {\mathsf{\gamma}} 
             \mid
             I(\ell') - I(\ell) \in \{ 2, 3 \}
            }
     \sigma_{1}(\ell)
     \cdot
     \sigma_{2}(\ell)
  =
  2n
  \cdot
  \left( \frac{\textstyle 1}
                   {\textstyle n+2}
   \right)^{2}$).}

\item
\textcolor{black}{
\underline{The row player plays ${\widehat{{\mathsf{\Sigma}}}}_{1}$
and the column player plays
${\widehat{{\mathsf{\Sigma}}}}_{2}^{\prime}$.}
Since
each of 
${\widehat{{\mathsf{\Sigma}}}}_{1}$
and
${\widehat{{\mathsf{\Sigma}}}}_{2}^{\prime}$
consists of a single strategy $s$,
it follows, 
by Proposition~\ref{gadget1},
that
${\widehat{\sigma}}_{1}(s)
  =
  {\widehat{\sigma}}_{2}(s)
  =
  1$.
Thus,
for each player $i \in [2]$,
{
\small
\begin{eqnarray*}
\sigma_{i}^{7}(s)
&  = &
  \frac{\textstyle 2}
          {\textstyle n+2}\,
  {\widehat{\sigma}}_{\overline{i}}(s)\ \
  =\ \
  \frac{\textstyle 2}
           {\textstyle n+2}\, .
\end{eqnarray*}
}            
By the utility functions,
there is a single subcase in which
the row (resp., column) player
gets utility $1$:
Both players choose $s$.
The subcase occurs with probability
$\left(  \frac{\textstyle 2}
                    {\textstyle n+2}
  \right)^{2}$.                                   
}

\end{itemize}
\noindent
\textcolor{black}{Hence,
for each player $i \in [2]$,}
{
\small
\textcolor{black}{
\begin{eqnarray*}
             {\widetilde{{\mathsf{U}}}}_{i}({\bm{\sigma}}^{7})
&  = &
  2n
  \cdot
  \left( \frac{\textstyle 1}
                   {\textstyle n+2}
   \right)^{2}
   +
  \left(  \frac{\textstyle 2}
                    {\textstyle n+2}
  \right)^{2}\ \
  =\ \
  \frac{\textstyle 2}
          {\textstyle n+2}\, .
\end{eqnarray*} 
}
}         
\textcolor{black}{Since ${\bm{\sigma}}^{7}$
and ${\bm{\sigma}}^{8}$
form a symmetric pair of mixed profiles
for the win-lose ${\mathsf{GHR}}$-symmetrization ${\widetilde{{\mathsf{G}}}}$,
it follows that
for each player $i \in [2]$,
${\widetilde{{\mathsf{U}}}}_{i}({\bm{\sigma}}^{8})
  =
  \frac{\textstyle 2}
          {\textstyle n+2}$.        
Thus, we have:}

\begin{quote}
\textcolor{black}{For each
${\bm{\sigma}} \in \left\{ {\bm\sigma}^{7},
                                          {\bm\sigma}^{8}
                               \right\}$:           
For each player $i \in [2]$:
${\widetilde{{\mathsf{U}}}}_{i}(\bm{\sigma})
 =
 \frac{\textstyle 2}
         {\textstyle n+2}$,
so that
$\sum_{i \in [2]}
    {\widetilde{{\mathsf{U}}}}_{i}(\bm{\sigma})
 =
 \frac{\textstyle 4}
         {\textstyle n+2}$;
either
${\mathsf{Supp}}(\sigma_{i})
  \subset
  {\widehat{\mathsf{\Sigma}}}_{i}
  \cup
  {\mathsf{L}}
  \cup
  {\mathsf{L}}^{\prime}$
or
${\mathsf{Supp}}(\sigma_{i})
  \subset
  {\widehat{\mathsf{\Sigma}}}_{i}^{\prime}
  \cup
  {\mathsf{L}}
  \cup
  {\mathsf{L}}^{\prime}$,
with  
$|{\mathsf{Supp}}(\sigma_{i})|
  =
  n+1$;
for each strategy
$s \in {\mathsf{Supp}}(\sigma_{i})$,
$\sigma_{i}(s)
  \leq
  \frac{\textstyle 2}
          {\textstyle n+2}$;
${\bm{\sigma}}$
is non-uniform.
}
\end{quote}

\end{itemize}
\noindent
\textcolor{black}{Hence,
when ${\mathsf{\phi}}$
is satisfiable,
we have:}
\begin{enumerate}

\item[{\sf (1)}]
\textcolor{black}{Each additional Nash equilibrium ${\bm{\sigma}}$
among the
$\left( \# {\mathsf{\phi}} \right)^{2}$
balanced mixtures
of literal equilibria for ${\mathsf{G}}$
has the following properties:}
\begin{quote}
\textcolor{black}{
For each player $i \in [2]$:
{\sf (Q.1)}
${\widetilde{{\mathsf{U}}}}_{i}(\bm{\sigma})
  \leq
 \frac{\textstyle 2}
         {\textstyle n}
 \leq
 \frac{\textstyle 2}
         {\textstyle 5}$,
so that
{\sf (Q.2)}
$\sum_{i \in [2]}
    {\widetilde{{\mathsf{U}}}}_{i}(\bm{\sigma})
 \leq
 \frac{\textstyle 4}
         {\textstyle n}
 \leq
 \frac{\textstyle 4}
         {\textstyle 5}$;
{\sf (Q.3)}         
${\mathsf{Supp}}\left( \sigma_{i}
                             \right)
 \subset
 {\mathsf{L}}
 \cup
 {\mathsf{L}}^{\prime}$,
with
{\sf (Q.4)}
$|{\mathsf{Supp}}\left( \sigma_{i}
                 \right)|
 \geq
 n
 \geq
 5$;
{\sf (Q.5)}
for each strategy
$s \in {\mathsf{Supp}}(\sigma_{i})$,
$\sigma_{i}(s)
  <
  \frac{\textstyle 2}
          {\textstyle n+2}
  \leq
  \frac{\textstyle 2}
          {\textstyle 7}$;
{\sf (Q.6)}          
${\bm{\sigma}}$
is uniform.}
\end{quote}

\item[{\sf (2)}]
\textcolor{black}{Each additional Nash equilibrium ${\bm{\sigma}}$
among the
$2 \cdot \# {\mathsf{\phi}}$
balanced mixtures
of a gadget equilibrium
and a literal equilibrium for ${\mathsf{G}}$
has the following properties:}
\begin{quote}
\textcolor{black}{
For each player $i \in [2]$:
{\sf (Q.7)}
${\widetilde{{\mathsf{U}}}}_{i}(\bm{\sigma})
  <
 \frac{\textstyle 2}
         {\textstyle n}
  \leq
  \frac{\textstyle 2}
          {\textstyle 5}$,
so that
{\sf (Q.8)}
$\sum_{i \in [2]}
    {\widetilde{{\mathsf{U}}}}_{i}(\bm{\sigma})
 <
 \frac{\textstyle 4}
         {\textstyle n}
  \leq
  \frac{\textstyle 4}
          {\textstyle 5}$;
{\sf (Q.9)}         
${\mathsf{Supp}}\left( \sigma_{i}
                             \right)
 \subset
 {\mathsf{L}}
 \cup
 {\mathsf{L}}^{\prime}
 \cup
 {\widehat{{\mathsf{\Sigma}}}}_{i}
 \cup
 {\widehat{{\mathsf{\Sigma}}}}_{i}^{\prime}$,
with
{\sf (Q.10)}
$|{\mathsf{Supp}}\left( \sigma_{i}
                 \right)|
 >
 n
 \geq
 5$;
{\sf (Q.11)}
for each strategy
$s \in {\mathsf{Supp}}(\sigma_{i})$,
$\sigma_{i}(s)
  \leq
  \frac{\textstyle 2}
          {\textstyle n+2}
  \leq
  \frac{\textstyle 2}
          {\textstyle 7}$;
{\sf (Q.12)}          
${\bm{\sigma}}$
is non-uniform.}
\end{quote}

\end{enumerate}

\noindent
\textcolor{black}{Hence,
we derive ${\mathcal{NP}}$-hardness
from the following table:}
\begin{center}
\begin{small}
\begin{tabular}{|l|l|l|l|}
\hline
\hline
${\mathcal{NP}}$-hard decision problem:                           & \multicolumn{3}{l|} {By properties disentangling the satisfiability of ${\mathsf{\phi}}$:}                                                        \\
\cline{2-4}
                                                                                             & Unsat.:                                                       & \multicolumn{2}{l|}{Sat.:}                                                                                    \\
\hline
\hline
{\sf $\exists$ NASH WITH SMALL UTILITIES}, 
                                                                                            & \textcolor{black}{{\sf (P.1)}}                        & \textcolor{black}{{\sf (Q.1)}, {\sf (Q.7)}}      
                                                                                                                                                                     & \textcolor{black}{$\# {\mathsf{\phi}}
                                                                                                                                                                             \cdot
                                                                                                                                                                             \left( \# {\mathsf{\phi}} + 2
                                                                                                                                                                             \right)$}                                                                                  \\
\textcolor{black}{~with $\frac{\textstyle 2}
                                             {\textstyle 5} 
         \leq
         u
         <
         \frac{\textstyle 1}
                 {\textstyle 2}$}                                                 &                                                                      &                 
                                                                                                                                                                    &                                                                                                        \\
\hline
{\sf $\exists$ NASH WITH SMALL TOTAL UTILITY},          & \textcolor{black}{{\sf (P.2)}}                        & \textcolor{black}{{\sf (Q.2)}, {\sf (Q.8)}}
                                                                                                                                                                   & \textcolor{black}{$\# {\mathsf{\phi}}
                                                                                                                                                                             \cdot
                                                                                                                                                                             \left( \# {\mathsf{\phi}} + 2
                                                                                                                                                                             \right)$}                                                                                  \\
\textcolor{black}{~with $\frac{\textstyle 4}
                                             {\textstyle 5} 
                                     \leq
                                     u
                                    <
                                    1$}                                                &                                                                      &
                                                                                                                                                                  &                                                                                                        \\
\hline
{\sf $\exists$ NASH WITH LARGE SUPPORTS},              & \textcolor{black}{{\sf (P.4)}}                        & \textcolor{black}{{\sf (Q.4)}, {\sf (Q.10)}}
                                                                                                                                                                 & \textcolor{black}{$\# {\mathsf{\phi}}
                                                                                                                                                                             \cdot
                                                                                                                                                                             \left( \# {\mathsf{\phi}} + 2
                                                                                                                                                                             \right)$}                                                                                 \\
\textcolor{black}{~with $2 < k \leq 5$}                              &                                                                     &
                                                                                                                                                                &                                                                                                          \\ 
\hline
{\sf $\exists$ NASH WITH RESTRICTING SUPPORTS},     &  \textcolor{black}{{\sf (P.3)}}                       & \textcolor{black}{{\sf (Q.3)}, {\sf (Q.9)}}
                                                                                                                                                                    & \textcolor{black}{$\# {\mathsf{\phi}}
                                                                                                                                                                             \cdot
                                                                                                                                                                             \left( \# {\mathsf{\phi}} + 2
                                                                                                                                                                             \right)$}                                                                                 \\
\textcolor{black}{~with
$T_{i}
  =
  \{ \ell
  \}$
and
$T_{i}
  =
  \{ \overline{\ell} \}$  , $i \in [n]$,
for any $\ell \in {\mathsf{L}}$}                                         &                                                                      &
                                                                                                                                                                   &                                                                                                          \\
\hline
{\sf $\exists$ NASH WITH RESTRICTED SUPPORTS},     & \textcolor{black}{{\sf (P.3)}}                         & \textcolor{black}{{\sf (Q.3)}}
                                                                                                                                                                   & \textcolor{black}{$\left( \# {\mathsf{\phi}}
                                                                                                                                                                             \right)^{2}$}                                                                          \\
\textcolor{black}{~with 
${\mathsf{T}}_{i}
  =
  {\mathsf{L}}
  \cup
  {\mathsf{L}}^{\prime}$,
  $i \in [n]$}                                                                       &                                                                        &
                                                                                                                                                                    &                                                                                                      \\  
\hline
{\sf $\exists$ NASH WITH SMALL PROBABILITIES}         &  \textcolor{black}{{\sf (P.5)}}                         & \textcolor{black}{{\sf (Q.5)}, {\sf (Q.11)}}
                                                                                                                                                                      & \textcolor{black}{$\# {\mathsf{\phi}}
                                                                                                                                                                             \cdot
                                                                                                                                                                             \left( \# {\mathsf{\phi}} + 2
                                                                                                                                                                             \right)$}                                                                                 \\
\hline
{\sf $\exists$ $\neg$ UNIFORM NASH}                           &  \textcolor{black}{{\sf (P.6)}}                         & \textcolor{black}{{\sf (Q.12)}}
                                                                                                                                                                      & \textcolor{black}{$2 \cdot \# {\mathsf{\phi}}$}                      \\ 
\hline
\hline
\end{tabular}
\end{small}
\end{center}
\noindent
\textcolor{black}{The formulas in the rightmost column
can be solved for $\# {\mathsf{\phi}}$;
by the $\# {\mathcal{P}}$-hardness
of computing $\# {\mathsf{\phi}}$~\cite{V79},
this yields the $\# {\mathcal{P}}$-hardness
of the seven counting problems.
Furthermore,
except for $2 \cdot \# {\mathsf{\phi}}$,
these formulas preserve the parity 
$\oplus {\mathsf{\phi}}$;
by the $\oplus {\mathcal{P}}$-hardness
of computing $\oplus {\mathsf{\phi}}$~\cite{PZ83},
this yields the $\oplus {\mathcal{P}}$-hardness
of the six parity problems
other than
{\sf $\oplus$ $\neg$ UNIFORM NASH}.    
}

\noindent
\underline{{\it Group II}:}
Fix ${\widehat{{\mathsf{G}}}}
        :=
        {\widehat{\mathsf{G}}}_{3}$. 
\textcolor{black}{Since $\kappa ({\widehat{{\mathsf{G}}}}_{3})
           =
           4$,
it follows that
each of
${\mathsf{G}}\left( {\widehat{\mathsf{G}}}_{3},
                                {\mathsf{\phi}}
                       \right)$
and
${\mathsf{GHR}}({\mathsf{G}})$                       
has size polynomial
in the size of ${\mathsf{\phi}}$.} 
\textcolor{black}{By Proposition~\ref{barrage},
${\widehat{{\mathsf{G}}}}_{3}$
has no uniform Nash equilibrium.}                                           
Assume first that ${\mathsf{\phi}}$
is unsatisfiable.
\textcolor{black}{Since
${\mathcal{NE}}({\mathsf{G}}({\widehat{{\mathsf{G}}}}_{3},
                                                   {\mathsf{\phi}}))
   =
   {\mathcal{NE}}(  {\widehat{{\mathsf{G}}}}_{3})$,
it follows that}
\textcolor{black}{${\mathsf{G}}( {\widehat{\mathsf{G}}}_{3},
                                {\mathsf{\phi}}
                       )$}
has no uniform Nash equilibrium.
Assume,
by way of contradiction,
that
\textcolor{black}{${\widetilde{{\mathsf{G}}}}
  =
  {\mathsf{GHR}}({\mathsf{G}})$}
has a uniform Nash equilibrium
${\bm{\sigma}}$.
By Proposition~\ref{fromsymmetrictobasic2},
exactly one of the following conditions holds
for ${\bm{\sigma}}$:
\begin{enumerate}

\item[{\sf (C'.1)}]
$\overleftarrow{\sigma_{1}},
 \overrightarrow{\sigma_{1}},
 \overleftarrow{\sigma_{2}},
 \overrightarrow{\sigma_{2}}
 \neq
 {\mathsf{0}}^n$,
with
${\bm\sigma}
 =
 \langle \overleftarrow{\varsigma_{1}},
             \overrightarrow{\varsigma_{2}}
 \rangle             
  \ast
  \langle \overleftarrow{\varsigma_{2}},
              \overrightarrow{\varsigma_{1}}
  \rangle$
and
$\langle \overleftarrow{\varsigma_{2}},
              \overrightarrow{\varsigma_{1}}
  \rangle,
 \langle \overleftarrow{\varsigma_{1}},
             \overrightarrow{\varsigma_{2}}
 \rangle            
 \in
 {\mathcal{NE}}({\mathsf{G}})$.

\item[{\sf (C'.2)}]
$\overleftarrow{\sigma_{1}}
 =
 \overrightarrow{\sigma_{2}}
 =
 {\sf 0}^n$, 
with 
${\bm\sigma}
 =
 \langle {\mathsf{0}}^{n},
             {\mathsf{0}}^{n}
 \rangle            
  \ast
  \langle \overleftarrow{\varsigma_{2}},
             \overrightarrow{\varsigma_{1}}
  \rangle$
and
$\langle \overleftarrow{\varsigma_{2}},
             \overrightarrow{\varsigma_{1}}
  \rangle           
  \in
  {\mathcal{NE}}({\mathsf{G}})$.

\item[{\sf (C'.3)}]
$\overrightarrow{\sigma_{1}}
 =
 \overleftarrow{\sigma_{2}}
 =
 {\mathsf{0}}^n$,
with
${\bm\sigma}
 =
 \langle {\mathsf{0}}^{n},
             {\mathsf{0}}^{n}
 \rangle            
 \ast
 \langle \overleftarrow{\varsigma_{1}},
             \overrightarrow{\varsigma_{2}}
 \rangle$
and
$\langle \overleftarrow{\varsigma_{1}},
              \overrightarrow{\varsigma_{2}}
  \rangle             
  \in
  {\mathcal{NE}}({\mathsf{G}})$.

\end{enumerate}
\noindent
By Observation~\ref{facebook observation},
the uniformity of ${\bm{\sigma}}$
implies that
exactly one of these conditions holds:
\begin{enumerate}

\item[{\sf (C''.1)}]
$\overleftarrow{\sigma_{1}},
 \overrightarrow{\sigma_{1}},
 \overleftarrow{\sigma_{2}},
 \overrightarrow{\sigma_{2}}
 \neq
 {\mathsf{0}}^n$,
and both
$\langle \overleftarrow{\varsigma_{2}},
              \overrightarrow{\varsigma_{1}}
  \rangle$
and  
$\langle \overleftarrow{\varsigma_{1}},
             \overrightarrow{\varsigma_{2}}
 \rangle$
are uniform Nash equilibria
for ${\mathsf{G}}$.

\item[{\sf (C''.2)}]
$\overleftarrow{\sigma_{1}}
 =
 \overrightarrow{\sigma_{2}}
 =
 {\sf 0}^n$, 
and
$\langle \overleftarrow{\varsigma_{2}},
             \overrightarrow{\varsigma_{1}}
  \rangle$           
is a uniform Nash equilibrium for
${\mathsf{G}}$.

\item[{\sf (C''.3)}]
$\overrightarrow{\sigma_{1}}
 =
 \overleftarrow{\sigma_{2}}
 =
 {\mathsf{0}}^n$,
and
$\langle \overleftarrow{\varsigma_{1}},
              \overrightarrow{\varsigma_{2}}
  \rangle$             
is a uniform Nash equilibrium for
${\mathsf{G}}$.

\end{enumerate}
\noindent
So,
in all cases,
${\mathsf{G}}$
has a uniform Nash equilibrium.
A contradiction.
\textcolor{black}{This implies the following property:
{\sf (P.1)}
${\widetilde{{\mathsf{G}}}}$
has no uniform Nash equilibrium.}

Assume now that ${\mathsf{\phi}}$
is  satisfiable.
\textcolor{black}{Then,
by Proposition~\ref{final lemma}  
(Condition {\sf (4)}),
the literal equilibrium
${\bm{\sigma}}
  =
  {\bm{\sigma}}({\mathsf{\gamma}})$,
for a satisfying assignment ${\mathsf{\gamma}}$
of ${\mathsf{\phi}}$,
is uniform.}  
Here are the additional Nash equilibria for ${\widetilde{{\mathsf{G}}}}$
and their properties:
\begin{itemize}

\item
\underline{\textcolor{black}{The \textcolor{black}{$\left( \# {\mathsf{\phi}} \right)^{2}$} 
                                           balanced mixtures of literal equilibria
                                           for ${\mathsf{G}}$:}}
${\bm{\sigma}}^{4}$
is uniform.
By Observation~\ref{facebook observation},
${\bm{\sigma}}^{5}$
and
${\bm{\sigma}}^{6}$
are non-uniform.

\item
\underline{\textcolor{black}{The $2 \cdot \# {\mathsf{\phi}}$
                  balanced mixtures
of a gadget equilibrium
and a literal equilibrium
for ${\mathsf{G}}$:}} 
\textcolor{black}{Since
${\widehat{{\bm{\sigma}}}}$
is non-uniform,
it follows,
by Observation~\ref{facebook observation},
that
${\bm{\sigma}}^{7}$
and
${\bm{\sigma}}^{8}$
are non-uniform.}

\end{itemize}
\noindent
\textcolor{black}{Hence,
when ${\mathsf{\phi}}$
is satisfiable,
we have:}
\begin{enumerate}

\item[{\sf (1)}]
\textcolor{black}{Each additional Nash equilibrium
${\bm{\sigma}}$
among the $\left( {\mathsf{\phi}} \right)^{2}$
balanced mixtures of literal equilibria
for ${\mathsf{G}}$
has the following property:
{\sf (Q.1)}
${\bm{\sigma}}$ is uniform.}

\item[{\sf (2)}]
\textcolor{black}{Each additional Nash equilibrium
${\bm{\sigma}}$
among the $2 \cdot {\mathsf{\phi}}$
balanced mixtures of 
a gadget equilibrium
and a literal equilibrium
for ${\mathsf{G}}$
has the following property:
{\sf (Q.2)}
${\bm{\sigma}}$ is non-uniform.}

\end{enumerate}
\noindent
\textcolor{black}{Hence,
we derive ${\mathcal{NP}}$-hardness
from the following table:}
\begin{center}
\begin{small}
\begin{tabular}{|l|l|l|l|}
\hline
\hline
${\mathcal{NP}}$-hard decision problem:                           & \multicolumn{3}{l|} {By properties disentangling the satisfiability of ${\mathsf{\phi}}$:}                                \\
\cline{2-4}
                                                                                             & Unsat.:            & \multicolumn{2}{l|}{Sat.:}                                                                                                        \\
\hline
\hline
{\sf $\exists$ UNIFORM NASH}                                           & \textcolor{black}{{\sf (P.1)}}                        & \textcolor{black}{{\sf (Q.1)}}      
                                                                                                                                                                     & \textcolor{black}{$\left( \# {\mathsf{\phi}}
                                                                                                                                                                                                     \right)^{2}$}                                                  \\                                                                                                                                                                                                                                                                                                                                                                       \hline
\hline
\end{tabular}
\end{small}
\end{center}
\noindent
\textcolor{black}{The formula
$\left( \# {\mathsf{\phi}} \right)^{2}$
in the rightmost column
can be solved for $\# {\mathsf{\phi}}$;
by the $\# {\mathcal{P}}$-hardness
of computing $\# {\mathsf{\phi}}$~\cite{V79},
this yields the $\# {\mathcal{P}}$-hardness
of the counting problem.
Furthermore,
the formula $\left( \# {\mathsf{\phi}} \right)^{2}$
preserves the parity $\oplus {\mathsf{\phi}}$;
by the $\oplus {\mathcal{P}}$-hardness
of computing $\oplus {\mathsf{\phi}}$~\cite{PZ83},
this yields the $\oplus {\mathcal{P}}$-hardness
of the parity problem.
}

\noindent
\underline{{\it Group III}:}
Fix
${\widehat{{\mathsf{G}}}}
  :=
  {\widehat{{\mathsf{G}}}}_{1}[h]$, 
where $h$ is polynomial in $n$ with
\textcolor{black}{$h > 2n$.}
\textcolor{black}{Since $\kappa \left( {\mathsf{G}}_{1}[h]
                       \right)
           =
           h$,}
where
$h$ is polynomial in $n$,                                         
it follows that each of
${\mathsf{G}}
 =
  {\mathsf{G}}({\widehat{{\mathsf{G}}}}_{1}[h],
                         \phi)$ 
and 
${\widetilde{{\mathsf{G}}}}
 =
  {\mathsf{GHR}}({\mathsf{G}})$ 
has size polynomial
in the size of ${\mathsf{\phi}}$. 
\textcolor{black}{By Proposition~\ref{gadget1},
${\widehat{{\mathsf{G}}}}_{1}[h]$
has a unique Nash equilibrium
$\widehat{\bm{\sigma}}$.} 
Assume first that
$\phi$ is unsatisfiable.
\textcolor{black}{Since
${\mathcal{NE}}({\mathsf{G}}({\widehat{{\mathsf{G}}}}_{1}[h],
                                                   {\mathsf{\phi}}))
   =
   {\mathcal{NE}}(  {\widehat{{\mathsf{G}}}}_{1}[h])$,
it follows that}
${\mathsf{G}}
     ({\widehat{{\mathsf{G}}}}_{1}[h],
                        \phi)$
has a unique 
Nash equilibrium,
\textcolor{black}{the gadget equilibrium $\widehat{\bm{\sigma}}$,}
which has,
by Proposition~\ref{gadget1},
the following properties:
\begin{quote}
\textcolor{black}{For each player $i \in [2]$:
${\mathsf{U}}_{i}({\widehat{\bm{\sigma}}})
 =
 \frac{\textstyle 1}
         {\textstyle h}$,
so that
$\sum_{i \in [2]}
    {\mathsf{U}}_{i}({\widehat{\bm{\sigma}}})
 =
 \frac{\textstyle 2}
         {\textstyle h}$;
$|{\mathsf{Supp}}\left( {\widehat{\sigma}}_{i}
                 \right)|
 =
 h$.}
\end{quote}
\noindent
\textcolor{black}{Hence,
${\widetilde{{\mathsf{G}}}}$
has exactly three Nash equilbria,
${\bm\sigma}^{1}$,
${\bm\sigma}^{2}$
and
${\bm\sigma}^{3}$,
such that:}
\begin{quote}
\textcolor{black}{For each player $i \in [2]$:
${\widetilde{{\mathsf{U}}}}_{i}(\bm{\sigma}^{1})
 =
 \frac{\textstyle 1}
         {\textstyle 2h}$,
so that
$\sum_{i \in [2]}
    {\widetilde{{\mathsf{U}}}}_{i}(\bm{\sigma}^{1})
 =
 \frac{\textstyle 1}
         {\textstyle h}$;
$|{\mathsf{Supp}}\left( \sigma_{i}^{1}
                 \right)|
 =
 2h$.}
\end{quote}
\begin{quote}
\textcolor{black}{For each ${\bm{\sigma}}
                                            \in
                                            \left\{ {\bm{\sigma}}^{2},
                                                       {\bm{\sigma}}^{3}
                                            \right\}$:            
For each player $i \in [2]$:
${\widetilde{{\mathsf{U}}}}_{i}(\bm{\sigma})
 =
 \frac{\textstyle 1}
         {\textstyle h}$,
so that        
$\sum_{i \in [2]}
    {\widetilde{{\mathsf{U}}}}_{i}(\bm{\sigma})
 =
 \frac{\textstyle 2}
         {\textstyle h}$;
$|{\mathsf{Supp}}\left( \sigma_{i}
                              \right)|
 =
 h$.}
\end{quote}
\noindent
\textcolor{black}{Hence,
when ${\mathsf{\phi}}$ is unsatisfiable,
each Nash equilibrium
${\bm{\sigma}}$
for ${\widetilde{{\mathsf{G}}}}$
has the following properties:}
\begin{quote}
\textcolor{black}{For each player $i \in [2]$:
{\sf (P.1)}
${\widetilde{{\mathsf{U}}}}_{i}(\bm{\sigma})
 \leq
 \frac{\textstyle 1}
         {\textstyle h}
 <
 \frac{\textstyle 1}
         {\textstyle 2n}$,
so that
{\sf (P.2)}        
$\sum_{i \in [2]}
    {\widetilde{{\mathsf{U}}}}_{i}(\bm{\sigma})
 \leq
 \frac{\textstyle 2}
         {\textstyle h}
 <
 \frac{\textstyle 1}
         {\textstyle n}$;
{\sf (P.3)}         
$|{\mathsf{Supp}}\left( \sigma_{i}
                              \right)|
 \geq
 h
 > 2n$.}
\end{quote}

Assume now that $\phi$ is satisfiable.
Then,
by Proposition~\ref{final lemma},
\textcolor{black}{the literal equilibrium
\textcolor{black}{$\bm{\sigma} = \bm{\sigma}({\mathsf{\gamma}})$,}
for a satisfying assignment ${\mathsf{\gamma}}$
of ${\mathsf{\phi}}$,
has the following properties:} 
\begin{quote}
\textcolor{black}{For each player $i \in [2]$:
${\mathsf{U}}_{i}(\bm{\sigma})
 =
 \frac{\textstyle 2}
         {\textstyle n}$,
so that
$\sum_{i \in [2]}
   {\mathsf{U}}_{i}(\bm{\sigma})
 =
 \frac{\textstyle 4}
         {\textstyle n}$;
$|{\mathsf{Supp}}\left( \sigma_{i}
                 \right)|
 =
 n$.}
\end{quote}
\noindent
\textcolor{black}{Here are the additional Nash equilibria
of ${\widetilde{{\mathsf{G}}}}$
and their properties:}
\begin{itemize}

\item
\underline{\textcolor{black}{The \textcolor{black}{$\left( \# {\mathsf{\phi}} \right)^{2}$}
                                            balanced mixtures
of literal equilibria for ${\mathsf{G}}$:}}                                                                                                                                                       
\begin{quote}
\textcolor{black}{
For each player $i \in [2]$:
${\widetilde{{\mathsf{U}}}}_{i}(\bm{\sigma}^{4})
 =
 \frac{\textstyle 1}
         {\textstyle n}$,
so that
$\sum_{i \in [2]}
   {\widetilde{{\mathsf{U}}}}_{i}(\bm{\sigma}^{4})
 =
 \frac{\textstyle 2}
         {\textstyle n}$;
$|{\mathsf{Supp}}\left( \sigma^{4}_{i}
                              \right)|
 =
 2n$.}
\end{quote}                                                                                                                                                                                                                                                                                                                                                     
\begin{quote}
\textcolor{black}{
For each
${\bm{\sigma}}
  \in
  \left\{ {\bm{\sigma}}^{5},
             {\bm{\sigma}}^{6}
  \right\}$:           
For each player $i \in [2]$:
${\widetilde{{\mathsf{U}}}}_{i}(\bm{\sigma})
 =
 \frac{\textstyle 2}
         {\textstyle n}$,
so that        
$\sum_{i \in [2]}
   {\widetilde{{\mathsf{U}}}}_{i}(\bm{\sigma})
 =
 \frac{\textstyle 4}
        {\textstyle n}$;
$|{\mathsf{Supp}}\left( \sigma_{i}
                 \right)|
 =
 n$.}
\end{quote}

\item
\underline{\textcolor{black}{The $2 \cdot \# {\mathsf{\phi}}$
balanced mixtures
of a gadget equilibrium
and a literal equilibrium for ${\mathsf{G}}$:}}
\textcolor{black}{Note that}
{
\small
\textcolor{black}{
\begin{eqnarray*}
         {\bm{\sigma}}^{7}
& = & \left\langle \frac{\textstyle \frac{\textstyle 2}
                                                               {\textstyle n}}
                                     {\textstyle \frac{\textstyle 1}
                                                               {\textstyle h}
                                                        +
                                                        \frac{\textstyle 2}
                                                                {\textstyle n}}
                             {\widehat{\sigma}}_{1}
                             \circ
                             \frac{\textstyle \frac{\textstyle 1}
                                                               {\textstyle h}}
                                     {\textstyle \frac{\textstyle 1}
                                                               {\textstyle h}
                                                        +
                                                        \frac{\textstyle 2}
                                                                {\textstyle n}}
                             \sigma_{2},
                             \frac{\textstyle \frac{\textstyle 1}
                                                               {\textstyle h}}
                                     {\textstyle \frac{\textstyle 1}
                                                               {\textstyle h}
                                                        +
                                                        \frac{\textstyle 2}
                                                                {\textstyle n}}
                             \sigma_{1}
                             \circ
                             \frac{\textstyle \frac{\textstyle 2}
                                                               {\textstyle n}}
                                     {\textstyle \frac{\textstyle 1}
                                                               {\textstyle h}
                                                        +
                                                        \frac{\textstyle 2}
                                                                {\textstyle n}}
                             {\widehat{\sigma}}_{2}                                                                     
          \right\rangle                                                                                                      \\
& = &   \left\langle \frac{\textstyle 2h}
                                       {\textstyle n+2h}
                              {\widehat{\sigma}}_{1}
                             \circ
                             \frac{\textstyle n}
                                     {\textstyle n+2h}
                             \sigma_{2},
                             \frac{\textstyle n}
                                     {\textstyle n+2h}
                             \sigma_{1}
                             \circ
                             \frac{\textstyle 2h}
                                     {\textstyle n+2h}
                            {\widehat{\sigma}}_{2}                                                                                             
          \right\rangle\, ,                      
\end{eqnarray*}
}
}
\textcolor{black}{and}
{
\small
\textcolor{black}{
\begin{eqnarray*}                                                      
           {\bm{\sigma}}^{8}           
& = &    \left\langle \frac{\textstyle n}
                                       {\textstyle n+2h}
                              \sigma_{1}
                             \circ
                             \frac{\textstyle 2h}
                                     {\textstyle n+2h}
                             {\widehat{\sigma}}_{2},
                             \frac{\textstyle 2h}
                                     {\textstyle n+2h}
                             {\widehat{\sigma}}_{1}
                             \circ
                             \frac{\textstyle n}
                                     {\textstyle n+2h}
                            \sigma_{2}                                                                                             
          \right\rangle\, ,                                                                                                           
\end{eqnarray*}
}
}
\textcolor{black}{respectively.
Consider first ${\bm{\sigma}}^{7}$,
with ${\bm{\sigma}}^{7} = {\bm{\sigma}}^{7}({\mathsf{\gamma}})$
for a satisfying assignment ${\mathsf{\gamma}}$
of ${\mathsf{\phi}}$.
Then,
by the definition of
the win-lose ${\mathsf{GHR}}$-symmetrization,
the row (resp., column) player
may get utility $1$
only in the following two cases:}
\begin{itemize}

\item
\textcolor{black}{
\underline{The row player plays ${\mathsf{L}}^{\prime}$
and the column player
plays ${\mathsf{L}}$.}
By Proposition~\ref{final lemma}
(Condition {\sf (C.4)}),
for each literal $\ell \in {\mathsf{\gamma}}$,
$\sigma_{1}(\ell)
  =
  \sigma_{2}(\ell)
  =
  \frac{\textstyle 1}
          {\textstyle n}$.
Thus,
by the balanced mixture ${\bm{\sigma}}^{7}$,
for each player $i \in [2]$,
for each literal
$\ell \in {\mathsf{\gamma}}$,}
{
\small
\textcolor{black}{
\begin{eqnarray*}
\sigma_{i}^{7}(\ell)
&  = &
  \frac{\textstyle n}
          {\textstyle n+2h} 
  \cdot
  \sigma_{\overline{i}}(\ell)
  =\ \
  \frac{\textstyle n}
          {\textstyle n+2h}
  \cdot
  \frac{\textstyle 1}
          {\textstyle n}\ \
  =\ \
  \frac{\textstyle 1}
          {\textstyle n+2h}\, .
\end{eqnarray*}
}
}
\noindent          
\textcolor{black}{By the utility functions,
there are two subcases
in which the row (resp., column) player
gets utility $1$:
The row player chooses $\ell$,
the column player chooses $\ell'$
with $\ell' \neq \overline{\ell}$,
and
$I(\ell') - I(\ell) \in \{ 0, 1 \}$
(resp,
$I(\ell') - I(\ell) \in \{ 2, 3 \}$).
The two subcases occur with probability
$\sum_{\ell, \ell^{\prime} \in {\mathsf{\gamma}} 
             \mid
             I(\ell') - I(\ell) \in \{ 0, 1 \}
            }
     \sigma_{1}(\ell)
     \cdot
     \sigma_{2}(\ell)
  =
  2n
  \cdot
  \left( \frac{\textstyle 1}
                   {\textstyle n+2h}
   \right)^{2}$.                          
(resp.,
$\sum_{\ell, \ell^{\prime} \in {\mathsf{\gamma}} 
             \mid
             I(\ell') - I(\ell) \in \{ 2, 3 \}
            }
     \sigma_{1}(\ell)
     \cdot
     \sigma_{2}(\ell)
  =
  2n
  \cdot
  \left( \frac{\textstyle 1}
                   {\textstyle n+2h}
   \right)^{2}$).}

\item
\textcolor{black}{
\underline{The row player plays ${\widehat{{\mathsf{\Sigma}}}}_{1}$
and the column player plays
${\widehat{{\mathsf{\Sigma}}}}_{2}^{\prime}$.}
By Proposition~\ref{gadget1},
it follows that
for each strategy
$s \in [h]$,
${\widehat{\sigma}}_{1}(s)
  =
  {\widehat{\sigma}}_{2}(s)
  =
  \frac{\textstyle 1}
          {\textstyle h}$.
Thus,
by the balanced mixture
${\bm{\sigma}}^{7}$,
for each player $i \in [2]$,
for each strategy
$s \in [h]$,
$\sigma_{i}^{7}(s)
  =
  \frac{\textstyle 2h}
          {\textstyle n+2h}\,
  {\widehat{\sigma}}_{\overline{i}}(s)
  =
   \frac{\textstyle 2}
           {\textstyle n+2h}$. 
By the utility functions,
there is a single subcase in which
the row (resp., column) player
gets utility $1$:
The row player chooses $s_{l}$
and the column player chooses
$s_{l}$
(resp., 
the row player chooses $s_{l}$
and the column player chooses
$s_{l+1}$).
The subcase occurs with probability
$\sum_{s_{l} \in [h]}
    \sigma_{1}^{7}(s_{l})
    \cdot
    \sigma_{2}^{7}(s_{l})
  =
  h
  \cdot
  \left( \frac{\textstyle 2}
                   {\textstyle n+2h}
  \right)^{2}$
(resp.,
$\sum_{s_{l} \in [h]}
    \sigma_{1}^{7}(s_{l})
    \cdot
    \sigma_{2}^{7}(s_{l+1})
  =
  h
  \cdot
  \left( \frac{\textstyle 2}
                   {\textstyle n+2h}
  \right)^{2}$).}

\end{itemize}
\noindent
\textcolor{black}{It follows that
for each player $i \in [2]$,}
{
\small
\textcolor{black}{
\begin{eqnarray*}
{\widetilde{{\mathsf{U}}}}_{i}({\bm{\sigma}}^{7})
&  =  &
  2n
  \cdot
  \left( \frac{\textstyle 1}
                   {\textstyle n+2h}
   \right)^{2}
   +
   h
   \cdot
  \left(  \frac{\textstyle 2}
                    {\textstyle n+2h}
  \right)^{2}\ \
  =\ \
  \frac{\textstyle 2}
          {\textstyle n+2h}\, .
\end{eqnarray*}
}
}          
\textcolor{black}{Since ${\bm{\sigma}}^{7}$
and ${\bm{\sigma}}^{8}$
form a symmetric pair of mixed profiles
for the win-lose ${\mathsf{GHR}}$-symmetrization ${\widetilde{{\mathsf{G}}}}$,
it follows that
for each player $i \in [2]$,
${\widetilde{{\mathsf{U}}}}_{i}({\bm{\sigma}}^{8})
  =
  \frac{\textstyle 2}
          {\textstyle n+2h}$.        
Thus, 
it follows:}                                                                                                                                                     
\begin{quote}
\textcolor{black}{For each
${\bm{\sigma}}
  \in
  \left\{ {\bm{\sigma}}^{7},
             {\bm{\sigma}}^{8}
  \right\}$:                                                                                                                                         
For each player $i \in [2]$:}
\textcolor{black}{${\widetilde{{\mathsf{U}}}}_{i}\left( {\bm{\sigma}}
                             \right)
  =
  \frac{\textstyle 2}
          {\textstyle n+2h}$,
so that          
$\sum_{i \in [2]}
    {\widetilde{{\mathsf{U}}}}_{i}({\bm{\sigma}})
    =
     \frac{\textstyle 4}
             {\textstyle n+2h}$;
$|{\mathsf{Supp}} (\sigma_{i})|
  =
  n+h$.}
\end{quote}                                                                                                                                              
\noindent                                                                                                                                                                                                  
                                                                                                                                                                                                                                                                                                                                                                                                   \end{itemize}
\noindent
\textcolor{black}{Hence,
when ${\mathsf{\phi}}$
is satisfiable,
we have:}
\begin{itemize}

\item[\textcolor{black}{{\sf (1)}}]
\textcolor{black}{Each additional Nash equilibrium
among the $\left( \# {\mathsf{\phi}}
                    \right)^{2}$
balanced mixtures
of literal equilibria
for ${\mathsf{G}}$
has the following properties:}                    
\begin{quote}
\textcolor{black}{For each player $i \in [2]$:
{\sf (Q.1)}
${\widetilde{{\mathsf{U}}}}_{i}({\bm{\sigma}})
  \geq
  \frac{\textstyle 1}
          {\textstyle n}$,
so that
{\sf (Q.2)}
$\sum_{i \in [2]}
    {\widetilde{{\mathsf{U}}}}_{i}({\bm{\sigma}})
  \geq
  \frac{\textstyle 2}
          {\textstyle n}$;
{\sf (Q.3)}          
${\mathsf{Supp}}(\sigma_{i})
  \leq
  2n
  <
  h$.}                          
\end{quote}

\item[\textcolor{black}{{\sf (2)}}]
\textcolor{black}{Each additional Nash equilibrium
among the $2 \cdot \# {\mathsf{\phi}}$
balanced mixtures
of a gadget equilibrium
and a literal equilibrium
for ${\mathsf{G}}$
has the following properties:}  
\begin{quote}
\textcolor{black}{For each player $i \in [2]$:
{\sf (Q.4)}
${\widetilde{{\mathsf{U}}}}_{i}({\bm{\sigma}})
  =
  \frac{\textstyle 2}
          {\textstyle n+2h}$,
so that
{\sf (Q.5)}
$\sum_{i \in [2]}
    {\widetilde{{\mathsf{U}}}}_{i}({\bm{\sigma}})
  =
  \frac{\textstyle 4}
          {\textstyle n+2h}$;
{\sf (Q.6)}          
${\mathsf{Supp}}(\sigma_{i})
  =
  n+h$.}          
\end{quote}
\noindent
\textcolor{black}{Note that
the choice
$h > 2n$
implies that
$\frac{\textstyle 2}
          {\textstyle n+2h}
  <
  \frac{\textstyle 1}
          {\textstyle 2h}$        
and
$h < n+h < 2h$.
Hence,
the properties
{\sf (Q.4)},
{\sf (Q.5)}
and
{\sf (Q.6)}
are not contradictory to the properties
{\sf (P.1)},
{\sf (P.2)}
and
{\sf (P.3)},
respectively,
of Nash equilibria,
holding when ${\mathsf{\phi}}$
is unsatisfiable.}

\end{itemize}

\noindent
\textcolor{black}{Hence,
we derive ${\mathcal{NP}}$-hardness
from the following table:}
\begin{center}
\begin{small}
\begin{tabular}{|l|l|l|l|}
\hline
\hline
${\mathcal{NP}}$-hard decision problem:                           & \multicolumn{3}{l|} {By properties disentangling the satisfiability of ${\mathsf{\phi}}$:}                                                       \\
\cline{2-4}
                                                                                             & Unsat.:                                                       & \multicolumn{2}{l|}{Sat.:}                                                                                   \\
\hline
\hline  
{\sf $\exists$ NASH WITH LARGE UTILITIES},                     & {\sf (P.1)}                                                   & \textcolor{black}{{\sf (Q.1)}}                    &  $\left( \# {\mathsf{\phi}}
                                                                                                                                                                                                                                             \right)^{2}$                                  \\
~\textcolor{black}{with $\frac{\textstyle 1}   
                    {\textstyle 2n}
            \leq
            u
            <
            \frac{\textstyle 1}
                    {\textstyle n}$}                                              &                                                                    &                                                                   &                                                       \\        
\hline
{\sf $\exists$ NASH WITH LARGE TOTAL UTILITY},           & {\sf (P.2)}                                                   & \textcolor{black}{{\sf (Q.2)}}                     & $\left( \# {\mathsf{\phi}}
                                                                                                                                                                                                                                           \right)^{2}$                                  \\                                                                                                                                                                                                                                       
~\textcolor{black}{with $\frac{\textstyle 1}   
                    {\textstyle n}
            \leq
            u
            <
            \frac{\textstyle 2}
                    {\textstyle n}$}                                              &                                                                    &                                                                   &                                                       \\        
\hline
{\sf $\exists$ NASH WITH SMALL SUPPORTS},                  & {\sf (P.3)}                                                   & \textcolor{black}{{\sf (Q.3)}}                    & $\left( \# {\mathsf{\phi}}
                                                                                                                                                                                                                                           \right)^{2}$                              \\
~\textcolor{black}{with $2n \leq k < h$}                           &                                                                    &                                                                   &                                                    \\
\hline
\hline
\end{tabular}
\end{small}
\end{center}
\noindent
\textcolor{black}{The formula
                           $( \# {\mathsf{\phi}}
                             )^{2}$}
in the rightmost column
can be solved for $\# {\mathsf{\phi}}$;
by the $\# {\mathcal{P}}$-hardness
of computing $\# {\mathsf{\phi}}$~\cite{V79},
this yields the $\# {\mathcal{P}}$-hardness
of the three counting problems.
Furthermore,
this formula preserves the parity $\oplus {\mathsf{\phi}}$;
by the $\oplus {\mathcal{P}}$-hardness
of computing $\oplus {\mathsf{\phi}}$~\cite{PZ83},
this yields the $\oplus {\mathcal{P}}$-hardness
of the three parity problems.

\noindent
\underline{{\it Group IV}:}
We start with an informal outline of the proof.
First note that
the problems
{\sf $\exists$ 2 NASH}
and {\sf $\exists$ 3 NASH}
``escape'' the technique
using the ${\mathsf{GHR}}$-symmetrization
of ${\mathsf{G}}
       =
       {\mathsf{G}} ( {\widehat{{\mathsf{G}}}},
                                       {\mathsf{\phi}}
                             )$ since,
by Theorem~\ref{fromsymmetrictobasic2},
each Nash equilibrium of ${\mathsf{G}}$
\textcolor{black}{gives rise}
to exactly three, distinct Nash equilibria
of ${\widetilde{{\mathsf{G}}}}=\mathsf{GHR}(\mathsf{G})$;
thus, 
even if $\widehat{\mathsf G}$ 
had a unique Nash equilibrium,
this approach could only help
proving the ${\mathcal{NP}}$-hardness
of {\sf $\exists$} $k+1$ {\sf NASH} with $k\geq 3$.
Instead,
to prove the ${\mathcal{NP}}$-hardness of
{\sf $\exists$ $k+1$ NASH}
for \textcolor{black}{all integers} $k \geq 1$,
we ``embed'' the win-lose diagonal game
${\widehat{{\mathsf{G}}}}_{5}[k]$
(Section~\ref{diagonal game})
as a subgame of ${\widetilde{{\mathsf{G}}}}$.
Denote the resulting game as
${\widetilde{{\mathsf{G}}}}
  \parallel
  {\widehat{{\mathsf{G}}}}_{5}[k]$,
which is still \textcolor{black}{symmetric and win-lose}.
We shall establish that
when $\phi$ is unsatisfiable,
all Nash equilibria
of ${\widetilde{{\mathsf{G}}}}$
which are ``inherited'' from $\widehat{\mathsf G}$
are ``destroyed'' in
${\widetilde{{\mathsf{G}}}}
  \parallel
  {\widehat{{\mathsf{G}}}}_{5}[k]$
as long as $\widehat{\mathsf{G}}$
\textcolor{black}{has been chosen 
so that
the players' expected utilities 
are smaller than $1$
in a Nash equilibrium;}
thus,
in this case,
${\widetilde{{\mathsf{G}}}}
  \parallel
  {\widehat{{\mathsf{G}}}}_{5}[k]$
has exactly $k$ Nash equilibria,
which are ``inherited'' from
the win-lose diagonal game
${\widehat{{\mathsf{G}}}}_{5}[k]$;
instead,
when  $\phi$ is satisfiable,
all $k$ Nash equilibria
inherited from ${\widehat{{\mathsf{G}}}}_{5}[k]$ survive,
and a new Nash equilibrium
is created for each satisfying assignment of ${\mathsf{\phi}}$.
Besides \textcolor{black}{the decision problem}
{\sf $\exists$} $k+1$ {\sf NASH}, with $k\geq 1$,
the same idea goes through
for the last three decision problems
in {\it Group IV}.
We now continue with the details
of the formal proof.

\noindent
Fix
${\widehat{{\mathsf{G}}}}
  :=
  {\widehat{{\mathsf{G}}}}_{1}[2]$.
\textcolor{black}{Since $\kappa ( {\widehat{{\mathsf{G}}}}_{1}[2]
                       )
           =
           2$,            
it follows that
each of 
${\mathsf{G}}$
and
${\widetilde{{\mathsf G}}}$
has size polynomial
in the size of ${\mathsf{\phi}}$.}
Denote as
${\widetilde{{\mathsf{G}}}}
  \parallel
  {\widehat{{\mathsf{G}}}}_{5}[k]$
the win-lose game
resulting from
${\widetilde{{\mathsf{G}}}}$
by \textcolor{black}{adding the
$k$ strategies} from
${\mathsf{T}}
  :=
  \{ t_{1}, \ldots, t_{k} \}$
to the strategy set of each player,
so that \textcolor{black}{${\mathsf{\Sigma}}(  {\widetilde{{\mathsf{G}}}}
                                                      \parallel
                                                      {\widehat{{\mathsf{G}}}}_{5}[k]
                                            )
              =
             {\mathsf{\Sigma}}( {\widetilde{{\mathsf{G}}}}
                                          )
             \cup
             {\mathsf{T}}$,}
and setting:
\begin{itemize}

\item[$(\mathsf{C.1})$]
${\mathsf{U}}(\langle s, t\rangle)
 :=
 \langle 0,1 \rangle$
and
${\mathsf{U}}(\langle t, s \rangle)
 :=
 \langle 1, 0 \rangle$,
when
\textcolor{black}{
$s \in {\mathsf{\Sigma}}({\widetilde{{\mathsf{G}}}})
          \setminus
          ({\mathsf{L}} 
            \cup
            {\mathsf{L}}^{\prime})$} 
and $t \in {\mathsf{T}}$.

\item[$(\mathsf{C.2})$]
${\mathsf{U}}(\langle s, t\rangle)
  :=
  {\mathsf{U}}(\langle t, s\rangle)
  :=
  \langle 0,0 \rangle$,
when $s \in {\mathsf{L}}
                    \cup
                    {\mathsf{L}}^{\prime}$
and
$t \in {\mathsf{T}}$.

\item[$(\mathsf{C.3})$]
${\mathsf{U}}(\langle t_{j}, t_{l}
                        \rangle)
 :=
 \langle {\mathsf{D}}_{k}[j, l],
             {\sf D}^{{\rm T}}_{k}[j, l]
 \rangle$,
when $j, l \in [k]$;
 thus,
 ${\widehat{{\mathsf{G}}}}_{5}[k]$
 is ``embedded'' into
 $\widetilde{\mathsf{G}}$
as a subgame.

\end{itemize}

\noindent
Note that
${\widetilde{{\mathsf{G}}}}
  \parallel
  {\widehat{{\mathsf{G}}}}_{5}[k]$
is a symmetric game,
which has size polynomial
in the size of ${\mathsf{\phi}}$.
Furthermore,  
by Proposition~\ref{very recent},
for each Nash equilibrium
\textcolor{black}{${\bm{\sigma}}
  \in
  {\mathcal{NE}}( {\widehat{{\mathsf{G}}}}_{5}[k]
                           )$,}
for each player $i \in [2]$,
${\mathsf{U}}_{i}({\bm{\sigma}})
  =
  1$.
Since ${\widetilde{{\mathsf{G}}}}
           \parallel
           {\widehat{{\mathsf{G}}}}_{5}[k]$
is win-lose and
\textcolor{black}{${\mathsf T}
  \subseteq
  {\mathsf{\Sigma}}(  {\widetilde{{\mathsf{G}}}}
                                          \parallel
                                          {\widehat{{\mathsf{G}}}}_{5}[k]
                               )$,}
it follows that
${\bm{\sigma}}$
is a Nash equilibrium for
${\widetilde{{\mathsf{G}}}}
  \parallel
  {\widehat{{\mathsf{G}}}}_{5}[k]$.
Thus,
\textcolor{black}{${\mathcal{NE}}( {\widehat{{\mathsf{G}}}}_{5}[k]
                          )
  \subseteq
  {\mathcal{NE}}(  {\widetilde{{\mathsf{G}}}}
                                     \parallel
                                     {\widehat{{\mathsf{G}}}}_{5}[k]
                          )$}.
We continue to prove
two simple properties
of the Nash equilibria for
${\widetilde{\mathsf{G}}}
  \parallel
  {\widehat{{\mathsf{G}}}}_{5}[k]$.

\begin{claim}
\label{claim1}
Fix a Nash equilibrium
\textcolor{black}{${\bm{\sigma}}
  \in
  {\mathcal{NE}}( {\widetilde{\mathsf{G}}}
                             \parallel
                             {\widehat{{\mathsf{G}}}}_{5}[k]
                           )$}
with
${\mathsf{Supp}}(\sigma_{i})
  \subseteq
  {\mathsf{T}}$
for some player $i \in [2]$.
Then,
\textcolor{black}{${\bm{\sigma}}
  \in
  {\mathcal{NE}}( {\widehat{{\mathsf{G}}}}_{5}[k]
                          )$}.
\end{claim}

\begin{proof}
By Cases $(\mathsf{C.1})$ and $(\mathsf{C.2})$
from the definition of the utility functions of
${\widetilde{\mathsf{G}}}
  \parallel
  {\widehat{{\mathsf{G}}}}_{5}[k]$,
${\mathsf{U}}_{\overline i}
                 ({\bm{\sigma}}_{-\overline i}\diamond s)=0$
for each strategy
\textcolor{black}{$s \in {\mathsf{\Sigma}}({\widetilde{{\mathsf{G}}}}
                                       )$}.
Since ${\widehat{{\mathsf{G}}}}_{5}[k]$
has the positive utility property,
there is a strategy $t \in {\mathsf{T}}$
such that
${\mathsf{U}}_{\overline i}
                 ({\bm{\sigma}}_{-\overline i} \diamond t)
  >
  0$.
Thus,
by Lemma~\ref{basic property of mixed nash equilibria},
every best-response strategy for player $\overline i$
is contained in ${\mathsf{T}}$, 
which
implies that 
$\mathsf{Supp}(\sigma_{\overline{i}})\subseteq{\mathsf T}$.
Hence, 
${\bm{\sigma}}$
maps to a mixed profile for 
${\widehat{{\mathsf{G}}}}_{5}[k]$,
and the claim follows.
\end{proof}

\noindent
Fix now a mixed profile
${\bm\sigma}$ for
${\widetilde{\mathsf{G}}}
  \parallel
  {\widehat{{\mathsf{G}}}}_{5}[k]$
such that
for each player $i \in [2]$,
$\sigma_{i}({\mathsf{T}}) < 1$;
so,
for each player $i \in [2]$,
${\mathsf{Supp}}(\sigma_{i})
  \not\subseteq
  {\mathsf{T}}$.
Construct 
from ${\bm{\sigma}}$
the mixed profile
$\bm{\varsigma}
 =
 \bm{\varsigma}(\bm{\sigma})$
for ${\widetilde{{\mathsf{G}}}}$
such that
for each player $i \in [2]$
and strategy
\textcolor{black}{$s \in {\mathsf{\Sigma}}( {\widetilde{{\mathsf{G}}}}
                                       )$,}
{
\small
\begin{eqnarray*}
          \varsigma_{i}(s)
&:=& \frac{\textstyle \sigma_{i} (s)}
                  {1 - \sigma_{i}({\mathsf{T}})}\, ;
\end{eqnarray*}
}
\noindent
that is, 
${\bm{\varsigma}}$
is the projection of 
${\bm{\sigma}}$ to
${\mathsf{\Sigma}}({\widetilde{\mathsf{G}}})$.
Note that,
by the construction,
for each player $i \in [2]$,
${\mathsf{Supp}}(\varsigma_{i})
  \subseteq
  {\mathsf{Supp}}(\sigma_{i})$.
We now prove:

\begin{claim}
\label{claim2}
Fix a Nash equilibrium
\textcolor{black}{${\bm{\sigma}}
   \in
   {\mathcal{NE}}({\widetilde{\mathsf{G}}}
                             \parallel
                             {\widehat{{\mathsf{G}}}}_{5}[k])
   \setminus
   {\mathcal{NE}}({\widehat{{\mathsf{G}}}}_{5}[k])$}.
Then,
${\bm{\varsigma}}({\bm{\sigma}})$
is a Nash equilibrium for ${\widetilde{{\mathsf{G}}}}$.
\end{claim}

\begin{proof}
Assume,
by way of contradiction,
that
${\bm{\varsigma}}({\bm{\sigma}})$
is not a Nash equilibrium
for ${\widetilde{{\mathsf{G}}}}$.
Then,
by Lemma~\ref{basic property of mixed nash equilibria},
there is a player $i \in [2]$
and a pair of strategies
\textcolor{black}{$s^{\prime}
 \in
 {\mathsf{\Sigma}}( {\widetilde{{\mathsf{G}}}}
                              )$}
and 
$s \in {\mathsf{Supp}}(\varsigma_{i})$
such that
{
\small
\begin{equation}
\label{equa1}
{\mathsf{ U}}_{i}\left( {\bm\varsigma}_{-i}
                                        \diamond
                                        s'
                             \right)\ \
  >\ \
  {\mathsf{U}}_{i}\left( {\bm{\varsigma}}_{-i}
                                       \diamond
                                       s
                            \right).
\end{equation}
}
\noindent
Since
${\mathsf{Supp}}(\varsigma_{i})
  \subseteq
  {\mathsf{Supp}}(\sigma_{i})$,
it follows that
$s \in {\mathsf{Supp}}(\sigma_{i})$.
Since 
${\widetilde{\mathsf{G}}}
  \parallel
  {\widehat{{\mathsf{G}}}}_{5}[k]$ is a symmetric game, 
\textcolor{black}{we may} assume, without loss of generality, 
that $i=1$.
Clearly,
{
\small
\begin{eqnarray*}
           {\mathsf{ U}}_{1}\left( {\bm\varsigma}_{-1}
                                                  \diamond
                                                   s^{\prime}
                                        \right)
& = & \sum_{t
                     \in
                     {\mathsf{\Sigma}}\left( {\widetilde{{\mathsf{G}}}}
                                                  \right)}
             {\mathsf{ U}}_{1}
                              \left( \langle s^{\prime},
                                                   t
                                       \rangle
                              \right)
             \cdot
             \varsigma_{2}\left( t
                                    \right)                                                                                                                                                                 \\
& = & \frac{\textstyle 1}
                  {\textstyle 1-\sigma_{2}({\mathsf{T}})}\,
          \sum_{t
                     \in
                     {\mathsf{\Sigma}}\left( {\widetilde{{\mathsf{G}}}}
                                                   \right)}
            {\mathsf{U}}_{1}\left( \langle s^{\prime},
                                                              t
                                                  \rangle
                                        \right)
            \cdot
            \sigma_{2}\left( t
                                                 \right)\, ,
\end{eqnarray*}
}
and
{
\small
\begin{eqnarray*}
          {\mathsf{U}}_{1}({\bm{\varsigma}}_{-1}
                                       \diamond
                                       s)
& = & \sum_{t \in {\mathsf{\Sigma}}\left({\widetilde{{\mathsf{G}}}}\right)}
             {\mathsf{U}}_{1}(\langle s, t
                                          \rangle)
             \cdot
             \varsigma_{2}(t)                                                                                                  \\
& = & \frac{\textstyle 1}
                  {\textstyle 1 - \sigma_{2}({\mathsf T})}
          \sum_{t \in {\mathsf{\Sigma}}\left({\widehat{{\mathsf{G}}}}\right)}
          {\mathsf{U}}_{1}\left( \langle  s, t
                                              \rangle
                                     \right)
          \cdot
          \sigma_{2}(t)\, .
\end{eqnarray*}
}
Hence, 
by (\ref{equa1}),
it follows that
{
\small
\begin{equation}
\label{equa2}
          \sum_{t \in {\mathsf{\Sigma}}\left({\widehat{{\mathsf{G}}}}\right)}
              {\mathsf{U}}_{1}(\langle s', t
                                           \rangle)
              \cdot
              \sigma_{2}(t)\ \
>\ \
 \sum_{t \in {\mathsf{\Sigma}}\left({\widehat{\mathsf{G}}}\right)}
              {\mathsf{U}}_{1}(\langle s,
                                                      t
                                          \rangle)
              \cdot
              \sigma_{2}(t)\, .
\end{equation}
}
\noindent
Thus,
{
\small
\begin{eqnarray*}
          \lefteqn{{\mathsf{U}}_{1}\left( {\bm{\sigma}}_{-1}
                                                              \diamond
                                                              s^{\prime}
                                                    \right)}                                                                                                                                                       \\
= & \sum_{t \in {\mathsf{\Sigma}}\left( {\widetilde{{\mathsf{G}}}}
                                                                    \parallel
                                                                    {\widehat{{\mathsf{G}}}}_{5}[k]
                                                          \right)}
             {\mathsf{U}}_{1}\left( \langle s', t
                                                  \rangle
                                        \right)
             \cdot
             \sigma_{2}(t)
    &                                                                                                                                                                                                             \\
= & \sum_{t \in {{\mathsf{\Sigma}}}\left( {\widetilde{{\mathsf{G}}}}
                                                              \right)}
             {\mathsf{U}}_{1}\left( \langle s', t
                                                 \rangle
                                        \right)
             \cdot
             \sigma_{2}(t)   
    & \mbox{(since
                    \textcolor{black}{${\mathsf{\Sigma}}( {\widetilde{{\mathsf{G}}}}
                                                             \parallel
                                                             {\widehat{{\mathsf{G}}}}_{5}[k]
                                                  )
                      =
                      {\mathsf{\Sigma}}({\widetilde{{\mathsf{G}}}}
                                                   )
                      \cup
                     {\mathsf{T}}$}
                     \&
                     ${\mathsf{U}}_{1}(s, t)
                            =
                            0$
                    for \textcolor{black}{$s \in {\mathsf{\Sigma}}(\widetilde{{\sf G}})$,}
                    $t \in {\mathsf{T}}$)}                                                                                                                                                     \\
> & \sum_{t \in {{\mathsf{\Sigma}}}\left( {\widetilde{{\mathsf{G}}}}
                                                              \right)}
             {\mathsf{U}}_{1}\left( \langle s, t
                                                 \rangle
                                        \right)
             \cdot
             \sigma_{2}(t)
    & \mbox{(by (\ref{equa2}))}                                                                                                                                                        \\
= & \sum_{t \in {\mathsf{\Sigma}}\left( {\widetilde{{\mathsf{G}}}}
                                                                    \parallel
                                                                    {\widehat{{\mathsf{G}}}}_{5}[k]
                                                          \right)}
             {\mathsf{U}}_{1}\left( \langle s, t
                                                 \rangle
                                        \right)
             \cdot
             \sigma_{2}(t)                                                              
    &  \mbox{(since
                    \textcolor{black}{${\mathsf{\Sigma}}( {\widetilde{{\mathsf{G}}}}
                                                             \parallel
                                                             {\widehat{{\mathsf{G}}}}_{5}[k]
                                                  )
                      =
                      {\mathsf{\Sigma}}({\widetilde{{\mathsf{G}}}}
                                                   )
                      \cup
                     {\mathsf{T}}$}
                     \&
                     ${\mathsf{U}}_{1}(s, t)
                            =
                            0$
                    for \textcolor{black}{$s \in {\mathsf{\Sigma}}(\widetilde{{\sf G}})$,}
                    $t \in {\mathsf{T}}$)}                                                                                                                                              \\
= & {\mathsf{U}}_{1}\left( {\bm{\sigma}}_{-1}
                                              \diamond
                                              s
                                   \right)\, .
    &                               
\end{eqnarray*}
}
Since ${\bm{\sigma}}$
is a Nash equilibrium for
${\widetilde{{\mathsf{G}}}}
  \parallel
  {\widehat{{\mathsf{G}}}}_{5}[k]$
and $s \in{\mathsf{Supp}}(\sigma_{1})$, 
Lemma~\ref{basic property of mixed nash equilibria} (Condition (1))
implies that
${\mathsf{U}}_{1}\left( {\bm{\sigma}}_{-1}
                                               \diamond
                                               s
                                     \right)\geq
                                     {\mathsf{U}}_{1}\left( {\bm{\sigma}}_{-1}
                                               \diamond
                                               s^{\prime}
                                     \right)$. 
A contradiction.
\end{proof}

\noindent
We now prove:

\begin{lemma}
\label{group3.1}
Assume that $\phi$ is unsatisfiable.
Then,
\textcolor{black}{${\mathcal{NE}}( {\widetilde{\mathsf{G}}}
                                    \parallel
                                   {\widehat{{\mathsf{G}}}}_{5}[k]
                          )
  =
  {\mathcal{NE}}({\widehat{{\mathsf{G}}}}_{5}[k])$}.
\end{lemma}

\begin{proof}
Since
\textcolor{black}{${\mathcal{NE}}( {\widehat{{\mathsf{G}}}}_{5}[k]
                          )
  \subseteq
  {\mathcal{NE}}(  {\widetilde{{\mathsf{G}}}}
                                     \parallel
                                     {\widehat{{\mathsf{G}}}}_{5}[k]
                          )$,}
it remains to prove that
\textcolor{black}{${\mathcal{NE}}(  {\widetilde{{\mathsf{G}}}}
                                     \parallel
                                     {\widehat{{\mathsf{G}}}}_{5}[k]
                          )
  \setminus
  {\mathcal{NE}}( {\widehat{{\mathsf{G}}}}_{5}[k]
                          )
  =
  \emptyset$}.
Assume,
by way of contradiction,
that there is
a Nash equilibrium
\textcolor{black}{${\bm{\sigma}}
  \in
  {\mathcal{NE}}( {\widetilde{\mathsf{G}}}
                                    \parallel
                                    {\widehat{{\mathsf{G}}}}_{5}[k]
                          )
  \setminus
  {\mathcal{NE}}({\widehat{{\mathsf{G}}}}_{5}[k])$}.

Since $\phi$ is unsatisfiable, 
it follows,
by Proposition~\ref{if unsatisfied},
that
\textcolor{black}{${\mathcal{NE}}({\mathsf{G}})
  =
  {\mathcal{NE}}( {\widehat{{\mathsf{G}}}}_{1}[2]
                          )$}.
Thus,
by Proposition~\ref{fromsymmetrictobasic2},
for each 
\textcolor{black}{${\bm{\tau}}
   \in
   {\mathcal{NE}}( {\widetilde{{\mathsf{G}}}}
                           )$},
for each player $i \in [2]$,
${\mathsf{Supp}}(\tau_i)
  \subseteq
  {\widehat{{\mathsf{\Sigma}}}}_{i}
  \cup
  {\widehat{{\mathsf{\Sigma}}}}_{i}^{\prime}$.
Since, 
by Claim \ref{claim2}, 
${\bm{\varsigma}}({\bm{\sigma}})$
is a Nash equilibrium for
${\widetilde{{\mathsf{G}}}}$,
it follows that
for each player $i \in [2]$,
${\mathsf{Supp}}(\sigma_{i})
  \subseteq
  {\widehat{{\mathsf{\Sigma}}}}_{i}
  \cup
  {\widehat{{\mathsf{\Sigma}}}}_{i}^{\prime}
  \cup
  {\mathsf{{T}}}$.
Since
\textcolor{black}{${\bm{\sigma}}
  \not\in
  {\mathcal{NE}}( {\widehat{{\mathsf{G}}}}_{5}[k]
                          )$}, 
Claim~\ref{claim1}
implies that
for each player $i \in [2]$,
${\mathsf{Supp}}(\sigma_{i})
  \not\subseteq
  {\mathsf{T}}$.
Hence,
it follows that
for each player $i \in [2]$,  
\textcolor{black}{${\mathsf{Supp}}(\sigma_{i})
  \cap
  ( {\widehat{{\mathsf{\Sigma}}}}_{i}
           \cup
           {\widehat{{\mathsf{\Sigma}}}}_{i}^{\prime}
  )         
  \neq
  \emptyset$}.
There are two cases:
\begin{enumerate}

\item
Assume first
that 
for each player $i \in [2]$,
${\mathsf{Supp}}(\sigma_{i})
  \subseteq
  {\widehat{{\mathsf{\Sigma}}}}_{i}
  \cup
  {\widehat{{\mathsf{\Sigma}}}}_{i}^{\prime}$.
By the utility functions of 
${\widehat{{\mathsf{G}}}}_{1}[2]$
and by the definition of the {\sf GHR}-symmetrization, 
it follows that
for any profile
${\bf s}$ supported in ${\bm{\sigma}}$,
${\mathsf U}_{1}({\bf s})
 +
 {\mathsf U}_{2}({\bf s})
 \leq
 1$.
This implies that
there is a player $i \in [2]$ 
with
${\mathsf U}_{i}({\bf s}) < 1$.\footnote{We prepare the reader
                that this is the only property 
                of the gadget game 
                ${\widehat{{\mathsf{G}}}}_{1}[2]$
                that we shall use for the \textcolor{black}{proof}
                regarding the decision problems in \textcolor{black}{{\it Group IV}}.
                Thus,
                there could be used,
                as a gadget game,
                any win-lose game 
                other than $\widehat{\mathsf G}_1[2]$,
                such that
                there is,
                for each Nash equilibrium of it,
                a player receiving utility less than $1$.}
This implies that
${\mathsf{U}}_{i}({\bm{\sigma}})
  <
  1$.
Consider now
player $i$
switching to a pure strategy
$t \in {\mathsf{T}}$.
Note that only profiles falling into case ({\sf C.1})
are supported in
${\bm{\sigma}}_{-i} \diamond t$.
Hence,
by the utility functions of
${\widetilde{\mathsf{G}}}
  \parallel
  {\widehat{{\mathsf{G}}}}_{5}[k]$,
${\mathsf U}_{i}\left( {\bm\sigma}_{-i}
                                    \diamond
                                    t
                          \right)
  = 1$.
Since ${\mathsf U}_i({\bm\sigma})<1$, 
Lemma \ref{basic property of mixed nash equilibria} (Condition (1))
yields a contradiction.

\item
Assume now
that there is a player 
$i \in [2]$ with
${\mathsf{Supp}}(\sigma_{i})
  \cap
  {\mathsf{T}}
  \neq
  \emptyset$.
Since
${\widetilde{\mathsf{G}}}
  \parallel
  {\widehat{{\mathsf{G}}}}_{5}[k]$  
is a symmetric game, 
we may assume,
without loss of generality,
that 
$i=2$. 
Since
${\mathsf{Supp}}(\sigma_{i})
  \cap
  \left( {\widehat{{\mathsf{\Sigma}}}}_{i}
            \cup
            {\widehat{{\mathsf{\Sigma}}}}_{i}^{\prime}
  \right)          
  \neq
  \emptyset$
for each player $i \in [2]$,  
there is a strategy 
\textcolor{black}{$s \in {\mathsf{Supp}}(\sigma_{1})
         \cap
         ( {\widehat{{\mathsf{\Sigma}}}}_{1}
                  \cup
                  {\widehat{{\mathsf{\Sigma}}}}_{1}^{\prime}
         )$}.
Thus,
{
\small
\begin{eqnarray*}
         \lefteqn{{\sf U}_1({\bm\sigma}_{-1}\diamond s)}                                                                                                                                               \\
=   & \sum_{s' \in {\widehat{{\mathsf{\Sigma}}}}_{2}
                             \cup
                             {\widehat{{\mathsf{\Sigma}}}}_{2}^{\prime}}
         {\sf U}_{1}(\langle s, s'
                            \rangle)
         \cdot
         \sigma_{2}(s')
         +
         \sum_{t' \in\mathsf{T}}{\sf U}_1(\langle s,t'\rangle)\cdot\sigma_2(t')
      &                                                                                                                                                                                                                              \\
=   & \sum_{s' \in {\widehat{{\mathsf{\Sigma}}}}_{2}
                             \cup
                              {\widehat{{\mathsf{\Sigma}}}}_{2}^{\prime}}
         {\sf U}_1(\langle s,s'\rangle)\cdot\sigma_2(s')
      & \mbox{(from Case {\sf (C.1)} 
                      in the utility functions)}                                                                                                                                                                           \\
\leq& \sum_{s' \in {\widehat{{\mathsf{\Sigma}}}}_{2}
                              \cup
                              {\widehat{{\mathsf{\Sigma}}}}_{2}^{\prime}}
            \sigma_{2}(s')
      & \mbox{(since ${\widetilde{\mathsf{G}}}
                                 \parallel
                                 {\widehat{{\mathsf{G}}}}_{5}[k]$ is win-lose)}\, .
\end{eqnarray*}
}
\noindent
Since $\widehat{\sf G}_5[k]$ has the positive utility property
and
${\mathsf{Supp}}(\sigma_{2})
  \cap
  {\mathsf{T}}
  \neq
  \emptyset$,
it follows that
there is a strategy
$t \in {\mathsf{T}}$
such that 
{
\small
\begin{equation}
\label{equa3}
\sum_{t' \in \mathsf{T}}
   {\mathsf{U}}_{1}(\langle t, t'
                                 \rangle)
   \cdot
   \sigma_{2}(t')\ \
   >\ \
   0\, .
\end{equation}
}
\noindent
So consider player $1$
switching to
the pure strategy $t$. 
Then,
{
\small
\begin{eqnarray*}
       \lefteqn{{\sf U}_1({\bm{\sigma}}_{-1}
                                       \diamond 
                                       t)}                                                                                                                                  \\ 
= & \sum_{s' \in {\widehat{{\mathsf{\Sigma}}}}_{2}
                           \cup
                           {\widehat{{\mathsf{\Sigma}}}}_{2}^{\prime}}
          {\mathsf{U}}_{1}\left( \langle t, s'
                                                \rangle
                                      \right)
          \cdot
          \sigma_{2}(s')
       +
       \sum_{t'  \in {\mathsf{T}}}
           {\mathsf{U}}_{1}\left( \langle t, t'
                                                 \rangle
                                       \right)
           \cdot
           \sigma_{2}(t')
    &                                                                                                                                                                      \\  
> & \sum_{s' \in {\widehat{{\mathsf{\Sigma}}}}_{2}
                           \cup
                           {\widehat{{\mathsf{\Sigma}}}}_{2}^{\prime}}
          {\sf U}_1\left( \langle t, s'
                                  \rangle
                        \right)          
          \cdot
          \sigma_{2}(s')
    &  \mbox{(from (\ref{equa3}))}                                                                                                                   \\
= & \sum_{s' \in {\widehat{{\mathsf{\Sigma}}}}_{2}
                            \cup
                            {\widehat{{\mathsf{\Sigma}}}}_{2}^{\prime}}
          \sigma_{2}(s')
    & \mbox{(from Case ({\sf C.1}) in the utility functions)}\, .     
\end{eqnarray*}
}
It follows that
{
\small
\begin{eqnarray*}
           {\mathsf{U}}_{1}\left( {\bm{\sigma}}_{-1}
                                                 \diamond 
                                                 s
                                       \right)
& < & {\mathsf{U}}_{1}\left( {\bm{\sigma}}_{-1}
                                                \diamond 
                                                t
                                       \right)\, .
\end{eqnarray*}
}
Since $s \in {\mathsf{Supp}}(\sigma_{1})$,
Lemma~\ref{basic property of mixed nash equilibria} (Condition (1))
implies that
{
\small
\begin{eqnarray*}
           {\mathsf{U}}_{1}\left( {\bm{\sigma}}_{-1}
                                                 \diamond 
                                                 s
                                       \right)
& \geq  & {\mathsf{U}}_{1}\left( {\bm{\sigma}}_{-1}
                                                \diamond 
                                                t
                                       \right)\, .
\end{eqnarray*}
}
A contradiction.

\end{enumerate}
The claim now follows.
\end{proof}

\noindent
\textcolor{black}{By Proposition~\ref{very recent},
Lemma~\ref{group3.1}
immediately implies:}

\begin{corollary}
\label{xouany}
\textcolor{black}{Assume that $\phi$ is unsatisfiable.
Then:
{\sf (P.0)}
${\widetilde{{\mathsf{G}}}}
  \parallel
  {\widehat{{\mathsf{G}}}}_{5}[k]$
has exactly $k$ Nash equilibria,
and each is
{\sf (P.1)}
Pareto-Optimal, 
{\sf (P.2)}
Strongly Pareto-Optimal,
and
{\sf (P.3)}
symmetric.}
\end{corollary}

\noindent
We continue to prove:

\begin{lemma}
\label{group3.2}
\textcolor{black}{Assume that $\phi$ is satisfiable.
Then,
{\sf (Q.1)}
$\# {\mathsf{\phi}}
  \cdot
  \left(  \# {\mathsf{\phi}}
            +
            4
  \right)$
of the additional
Nash equilibria for
${\widetilde{{\mathsf{G}}}}
  \parallel
  {\widehat{{\mathsf{G}}}}_{5}[k]$
are neither Pareto-Optimal
nor Strongly Pareto-Optimal,
{\sf (Q.2)}
$\# {\mathsf{\phi}} \cdot (\# {\mathsf{\phi}} + 1)$
of them
are non-symmetric,
and
{\sf (Q.3)}
$\# {\mathsf{\phi}}$
of them 
are symmetric.}
\end{lemma}

\begin{proof}
Since $\phi$ is satisfiable, 
it follows,
by Proposition~\ref{final lemma}, 
that 
\textcolor{black}{for each satisfying assignment ${\mathsf{\gamma}}$
of ${\mathsf{\phi}}$,}
there is
a Nash equilibrium
${\bm{\sigma}}$
for ${\mathsf{G}}$
such that
for each player $i \in [2]$,
${\mathsf{U}}_{i}({\bm{\sigma}})
  =
  \frac{\textstyle 2}
          {\textstyle n}
  < 1$.
\textcolor{black}{Here are the additional Nash equilibria for ${\widetilde{{\mathsf{G}}}}$
and their properties:}  
\begin{itemize}

\item
\underline{\textcolor{black}{The $\left( {\mathsf{\phi}} \right)^{2}$
balanced mixtures
of literal equilibria for ${\mathsf{G}}$:}} 
\textcolor{black}{By Lemma~\ref{symmetry observation},
${\bm{\sigma}}^{4}$
is symmetric
if and only if
${\bm{\sigma}} = {\bm{\tau}}$,
while
${\bm{\sigma}}^{5}$
and
${\bm{\sigma}}^{6}$
are non-symmetric.
Furthermore:}
\begin{quote}           
\textcolor{black}{For each player $i \in [2]$,
${\mathsf{U}}_{i}({\bm{\sigma}}^{4})
  =
  \frac{\textstyle 1}
          {\textstyle n}
  < 1$.}
\end{quote}
\begin{quote}  
\textcolor{black}{For each
${\bm{\sigma}}
  \in
  \left\{ {\bm{\sigma}}^{5},
             {\bm{\sigma}}^{6}
  \right\}$:    
$\sum_{i \in [2]}
   {\mathsf{U}}_{i}({\bm{\sigma}})
  =
  \frac{\textstyle 2}
          {\textstyle n}
  < 2$.}
\end{quote}

\item
\underline{\textcolor{black}{The $2 \cdot \# {\mathsf{\phi}}$
          balanced mixtures
          of a literal equilibrium
          and a gadget equilibrium
          for ${\mathsf{G}}$:}}
\textcolor{black}{By Lemma~\ref{symmetry observation},
it follows that
${\bm{\sigma}}^{7}$
and
${\bm{\sigma}}^{8}$
are non-symmetric.}
\textcolor{black}{Furthermore,
by Lemma~\ref{vittorio simplification for group iv}:}
\begin{quote}
\textcolor{black}{For each
${\bm{\sigma}}
  \in
  \left\{ {\bm{\sigma}}^{5},
             {\bm{\sigma}}^{6}
  \right\}$:    
For each player $i \in [2]$:
${\mathsf{U}}_{i}({\bm{\sigma}}^{7})
  <
  \max \{ {\mathsf{U}}_{i}({\bm{\sigma}}),
               {\mathsf{U}}_{i}({\widehat{{\bm{\sigma}}}}) 
           \}
   <
   1$.}
\end{quote}

\end{itemize}
\noindent
\textcolor{black}{Hence,
there are $\# {\mathsf{\phi}}$ additional Nash equilibria
\textcolor{black}{for ${\widetilde{{\mathsf{G}}}}$}
which are symmetric
and 
\textcolor{black}{$\# {\mathsf{\phi}}
 \cdot 
 \left( \# {\mathsf{\phi}}
         + 2 
 \right)
 -
 \# {\mathsf{\phi}}
 =
 \# {\mathsf{\phi}}
 \cdot 
 \left( \# {\mathsf{\phi}} + 1
 \right)$}
additional Nash equilibria 
\textcolor{black}{for ${\widetilde{{\mathsf{G}}}}$ 
which are non-symmetric.}}

\textcolor{black}{Consider now any
\textcolor{black}{additional Nash equilibrium
${\bm{\sigma}}$
for ${\widetilde{{\mathsf{G}}}}$.}}  
Since
$\widetilde{{\mathsf{G}}}
  \parallel
  {\widehat{{\mathsf{G}}}}_{5}[k]$
is constructed 
from ${\widetilde{\mathsf{G}}}$
by adding ${\mathsf{T}}$
to the strategy set of each player,
it suffices,
in order to prove that
\textcolor{black}{${\bm{\sigma}}$
is a Nash equilibrium for
${\widetilde{{\mathsf{G}}}}
 \parallel
 {\widehat{{\mathsf{G}}}}_{5}[k]$,}
to check that no player can improve her utility
\textcolor{black}{in ${\bm{\sigma}}$}
by switching to a strategy in ${\mathsf{T}}$. 
By Case {\sf (C.2}) 
in the utility functions, 
it holds that
${\mathsf{U}}\left( \langle s_{1},
                                            s_{2}
                                \rangle
                       \right)
  =
  \langle 0, 0
  \rangle$
for $s_{1} \in {\mathsf{L}}
                       \cup
                       {\mathsf{L}}^{\prime}$ 
(resp., $s_{1} \in {\mathsf{T}}$) 
and
$s_{2} \in {\mathsf{T}}$ 
(resp., $s_{2} \in {\mathsf{L}}
                             \cup
                             {\mathsf{L}}^{\prime}$). 
Hence,
for each player $i \in [2]$,
\textcolor{black}{${\mathsf{U}}_{i}\left( {\bm{\sigma}}_{-i}
                                       \diamond
                                       t
                             \right)
  =0$} 
for each strategy
$t \in {\mathsf T}$. 
Thus,
\textcolor{black}{${\bm{\sigma}}
  \in
  {\mathcal{NE}}( {\widetilde{{\mathsf{G}}}}
                                    \parallel
                                    {\widehat{{\mathsf{G}}}}_{5}[k]
                           )$}
with
\textcolor{black}{${\mathsf{U}}_{i}({\bm{\sigma}})
  <1$}
for each player $i \in [2]$. 
\textcolor{black}{Hence,
${\widetilde{{\mathsf{G}}}}
  \parallel
  {\widehat{{\mathsf{G}}}}_{5}[k]$
has
$\# {\mathsf{\phi}}$ symmetric Nash equilibria
and
$\# {\mathsf{\phi}}
  \cdot
  \left( \# {\mathsf{\phi}} + 1
  \right)$
non-symmetric Nash equilibria.
Properties {\sf (Q.2)}
and {\sf (Q.3)}
follow.
}

By Proposition~\ref{very recent},
each player $i \in [2]$ has utility $1$
in any Nash equilibrium
for ${\widehat{{\mathsf{G}}}}_{5}[k]$.
Since
\textcolor{black}{${\mathcal{NE}}( {\widehat{{\mathsf{G}}}}_{5}[k]
                                                                  )
  \subseteq
  {\mathcal{NE}}(  {\widetilde{{\mathsf{G}}}}
                                     \parallel
                                     {\widehat{{\mathsf{G}}}}_{5}[k]
                          )$},
this implies that
${\bm{\sigma}}$
is neither Pareto-Otimal
nor Strongly Pareto-Optimal.
\textcolor{black}{So
there are
$\left( \# {\mathsf{\phi}}
  \right)^{2}
  +
  2 \cdot \# {\mathsf{\phi}}
  +
  2 \cdot \# {\mathsf{\phi}}
  \# {\mathsf{\phi}}
  \cdot
  \left( {\mathsf{\phi}} + 4
  \right)$
additional Nash equilibria
which are neither
Pareto-Otimal
nor Strongly Pareto-Optimal.
Property
${\mathsf{(Q.1)}}$ 
follows.}
\end{proof}

\noindent
\textcolor{black}{Hence,
using Lemma~\ref{group3.2},
we derive ${\mathcal{NP}}$-hardness
from the following table:}

\begin{center}
\begin{small}
\begin{tabular}{|l|l|l|l|}
\hline
\hline
Decision Problem is ${\mathcal{NP}}$-hard due to:            & \multicolumn{3}{l|} {Properties disentangling the satisfiability of ${\mathsf{\phi}}$:}                                                       \\
\cline{2-4}
                                                                                             & Unsat.:                                                       & \multicolumn{2}{l|}{Sat.:}                                                                              \\
\hline
\hline 
{\sf $\exists$ $k+1$ NASH}                                                   & {\sf (P.0)}                                               & ---  
                                                                                                                                                                  & ---                                                                                                                    \\                                                                                                            
\hline
{\sf $\exists$ $\neg$ PARETO-OPTIMAL NASH}                   & {\sf (P.1)}                                               & \textcolor{black}{{\sf (Q.1)}}  
                                                                                                                                                                 & $\# {\mathsf{\phi}}
                                                                                                                                                                      \cdot
                                                                                                                                                                      \left( \# {\mathsf{\phi}} + 4
                                                                                                                                                                      \right)$                                                                                                          \\
\hline
{\sf $\exists$ $\neg$ STRONGLY PARETO-OPTIMAL NASH} & {\sf (P.2)}                                              & \textcolor{black}{{\sf (Q.1)}}  
                                                                                                                                                                &  $\# {\mathsf{\phi}}
                                                                                                                                                                      \cdot
                                                                                                                                                                      \left( \# {\mathsf{\phi}} + 4
                                                                                                                                                                      \right)$                                                                                                         \\
\hline
{\sf $\exists$ $\neg$ SYMMETRIC NASH}                             &  {\sf (P.3)}                                             &  \textcolor{black}{{\sf (Q.2)}} 
                                                                                                                                                                &  $\# {\mathsf{\phi}}
                                                                                                                                                                      \cdot
                                                                                                                                                                      \left( \# {\mathsf{\phi}} + 1
                                                                                                                                                                      \right)$                                                                                                         \\ 
\hline
\hline
\end{tabular}
\end{small}
\end{center}
\noindent
\textcolor{black}{Both formulas in the rightmost column
can be solved for $\# {\mathsf{\phi}}$;
by the $\# {\mathcal{P}}$-hardness
of computing $\# {\mathsf{\phi}}$~\cite{V79},
this yields the $\# {\mathcal{P}}$-hardness
of the three counting problems.
Furthermore,
the formula $\# {\mathsf{\phi}} \cdot (\# {\mathsf{\phi}} + 1)$
preserves the parity 
$\oplus {\mathsf{\phi}}$;
by the $\oplus {\mathcal{P}}$-hardness
of computing $\oplus {\mathsf{\phi}}$~\cite{PZ83},
this yields the $\oplus {\mathcal{P}}$-hardness
of the parity problems
{\sf $\oplus$ $\neg$ PARETO-OPTIMAL NASH}
and
{\sf $\oplus$ $\neg$ STRONGLY PARETO-OPTIMAL NASH}.    
}

\noindent
\textcolor{black}{By 
Corollary~\ref{xouany}
and Lemma~\ref{group3.2}
(Condition {\sf (Q.3)}),
$\mbox{{\sf $\#$ SYMMETRIC NASH}}
 = 
 k + \# {\mathsf{\phi}}$.
This formula can be solved for $\# {\mathsf{\phi}}$;
by the $\# {\mathcal{P}}$-hardness
of computing $\# {\mathsf{\phi}}$~\cite{V79},
this yields the $\# {\mathcal{P}}$-hardness
of 
{\sf $\#$ SYMMETRIC NASH}.
Furthermore,
this formula allows computing
$\oplus {\mathsf{\phi}}$
from
{\sf $\oplus$ SYMMETRIC NASH};
by the $\oplus {\mathcal{P}}$-hardness
of computing $\oplus {\mathsf{\phi}}$~\cite{PZ83},
this yields the $\oplus {\mathcal{P}}$-hardness
of
{\sf $\oplus$ SYMMETRIC NASH}.}
\end{proof}

\noindent
We now show:

\begin{theorem}
\label{maintheorem symmetric}
Fix a win-lose bimatrix game ${\widehat{{\mathsf{G}}}}$
with the positive utility property.
Then,
restricted to symmetric win-lose bimatrix games,
{\sf NASH-EQUIVALENCE} \textcolor{black}{$( {\mathsf{GHR}}(\widehat{\mathsf{G}})
                                            )$}
is co-${\mathcal{NP}}$-complete.
\end{theorem}

\begin{proof}
Here is a polynomial time algorithm
to decide
the non-satisfiability
of an input formula ${\mathsf{\phi}}$,
with access to an oracle for
{\sf NASH-EQUIVALENCE} \textcolor{black}{$( {\mathsf{GHR}}(\widehat{\mathsf{G}})
                                           )$}:
\begin{quote}
\textcolor{black}{Construct the symmetric win-lose bimatrix games
${\mathsf{GHR}}({\widehat{{\mathsf{G}}}})$
and} 
\textcolor{forestgreen}{${\widetilde{{\mathsf{G}}}}
  =
  {\mathsf{GHR}}( {\mathsf{G}} ( {\widehat{{\mathsf{G}}}}, {\mathsf{\phi}}
                                                     )
                           )$}.
Query the oracle
for 
{\sf NASH-EQUIVALENCE} \textcolor{black}{$( {\mathsf{GHR}}(\widehat{\mathsf{G}})
                                                                                    )$}
on input ${\widetilde{{\mathsf{G}}}}$
and return the answer of the oracle.
\end{quote}
\noindent
To establish the correctness of the algorithm,
we prove:

\begin{lemma}
\label{zousou}
\textcolor{black}{${\mathcal{NE}}( {\mathsf{G}}\left( {\widehat{{\mathsf{G}}}}, {\mathsf{\phi}}
                                                         \right)
                         )                                
  =
  {\mathcal{NE}}( {\widehat{{\mathsf{G}}}}
                          )$}
if and only if
\textcolor{black}{${\mathcal{NE}}( {\mathsf{GHR}} ( {\mathsf{G}}( {\widehat{{\mathsf{G}}}}, {\mathsf{\phi}}
                                                                                )
                                                        )                              
                         )                                
  =
  {\mathcal{NE}}( {\mathsf{GHR}} ( {\widehat{{\mathsf{G}}}}
                                                        )
                          )$}.                          
\end{lemma}

\begin{proof}
Assume first that
\textcolor{black}{${\mathcal{NE}}( {\mathsf{GHR}} ( {\mathsf{G}}( {\widehat{{\mathsf{G}}}}, {\mathsf{\phi}}
                                                                                                                       )
                                                               )                              
                          )                                
  =
  {\mathcal{NE}}( {\mathsf{GHR}} ( {\widehat{{\mathsf{G}}}}
                                                        )
                          )$}.     
Since every Nash equilibrium for
${\widehat{{\mathsf{G}}}}$
is a gadget equilibrium,
the win-lose ${\mathsf{GHR}}$-symmetrization
implies that
every Nash equilibrium for
\textcolor{black}{${\mathsf{GHR}} ( {\widehat{{\mathsf{G}}}}
                             )$}
is also a gadget equilibrium.
Hence,
the Nash-equivalence of
\textcolor{black}{${\mathsf{GHR}} ( {\widehat{{\mathsf{G}}}}
                             )$}                             
and
\textcolor{black}{${\mathsf{GHR}} ( {\mathsf{G}}( {\widehat{{\mathsf{G}}}}, {\mathsf{\phi}}
                                                                     )
                             )$}
implies that
every Nash equilibrium for
\textcolor{black}{${\mathsf{GHR}} ({\mathsf{G}}( {\widehat{{\mathsf{G}}}}, {\mathsf{\phi}}
                                                  )
                             )$}                              
is a gadget equilibrium.
Hence,
the win-lose ${\mathsf{GHR}}$-symmetrization implies that
every Nash equilibrium for
\textcolor{black}{${\mathsf{G}}( {\widehat{{\mathsf{G}}}}, {\mathsf{\phi}}
                       )$}
is a gadget equilibrium.
By Propositions~\ref{if unsatisfied}
and~\ref{final lemma},
it follows that                       
${\mathsf{\phi}}$
is unsatisfiable
(since otherwise
\textcolor{black}{${\mathsf{G}}( {\widehat{{\mathsf{G}}}}, {\mathsf{\phi}}
                       )$}
has a literal equilibrium).
Hence,
Proposition~\ref{if unsatisfied}
implies that
\textcolor{black}{${\mathcal{NE}}( {\mathsf{G}}( {\widehat{{\mathsf{G}}}}, {\mathsf{\phi}}
                                                   )
                         )                                
  =
  {\mathcal{NE}}( {\widehat{{\mathsf{G}}}}
                          )$}.

Assume now that
\textcolor{black}{${\mathcal{NE}}( {\mathsf{GHR}} ( {\mathsf{G}}( {\widehat{{\mathsf{G}}}}, {\mathsf{\phi}}
                                                                                )
                                                               )                              
                         )                                
  \neq
  {\mathcal{NE}}( {\mathsf{GHR}} ( {\widehat{{\mathsf{G}}}}
                                                        )
                          )$}.   
Since
\textcolor{black}{${\mathcal{NE}}( {\widehat{{\mathsf{G}}}}
                          )
  \subseteq
  {\mathcal{NE}}( {\mathsf{G}}( {\widehat{{\mathsf{G}}}}, {\mathsf{\phi}}
                                                  )
                          )$,}
every gadget equilibrium
is a Nash equilibrium
for
\textcolor{black}{${\mathsf{G}}( {\widehat{{\mathsf{G}}}}, {\mathsf{\phi}}
                       )$}.                                                           
Hence,
Proposition~\ref{frombasictosymmetric}
implies that
\textcolor{black}{${\mathcal{NE}}( {\mathsf{GHR}} ( {\widehat{{\mathsf{G}}}}
                                                        )  
                          )
  \subseteq
  {\mathcal{NE}}( {\mathsf{GHR}} ( {\mathsf{G}} ( {\widehat{{\mathsf{G}}}}, {\mathsf{\phi}}
                                                                                 )
                                                        )
                          )$}.
Hence,
the assumption implies that
\textcolor{black}{${\mathsf{GHR}} ( {\mathsf{G}} ( {\widehat{{\mathsf{G}}}}, {\mathsf{\phi}}
                                                      )
                            )$}
has a literal equilibrium.
By the ${\mathsf{GHR}}$-symmetrization,
this implies that
\textcolor{black}{${\mathsf{G}} ( {\widehat{{\mathsf{G}}}}, {\mathsf{\phi}}
                        )$}
has a literal equilibrium.
By Propositions~\ref{if unsatisfied}
and~\ref{final lemma},
it follows that                       
${\mathsf{\phi}}$
is satisfiable
(since otherwise
\textcolor{black}{${\mathsf{G}}( {\widehat{{\mathsf{G}}}}, {\mathsf{\phi}}
                       )$}
has no literal equilibrium).                        
Hence,
Propositions~\ref{final lemma}
implies that
\textcolor{black}{
${\mathcal{NE}}( {\mathsf{G}}( {\widehat{{\mathsf{G}}}}, {\mathsf{\phi}}
                                                  )
                         )                                
  \neq
  {\mathcal{NE}}( {\widehat{{\mathsf{G}}}}
                          )$,}  
and this completes the proof.   
\end{proof}
\noindent
By Propositions~\ref{if unsatisfied}
and~\ref{final lemma},
${\mathsf{\phi}}$
is unsatisfiable
if and only if
\textcolor{black}{
${\mathcal{NE}}( {\mathsf{G}}( {\widehat{{\mathsf{G}}}}, {\mathsf{\phi}}
                                                  )
                         )                                
  =
  {\mathcal{NE}}( {\widehat{{\mathsf{G}}}}
                          )$}.
Hence,
by Lemma~\ref{zousou},
it follows that
${\mathsf{\phi}}$
is unsatisfiable
if and only if      
\textcolor{black}{                    
${\mathcal{NE}}( {\mathsf{GHR}} ( {\mathsf{G}}( {\widehat{{\mathsf{G}}}}, {\mathsf{\phi}}
                                                                                )
                                                       )                              
                         )                                
  =
  {\mathcal{NE}}( {\mathsf{GHR}} ( {\widehat{{\mathsf{G}}}}
                                                        )
                          )$}.   
Hence,
{\sf NASH-EQUIVALENCE}\textcolor{black}{$( {\mathsf{GHR}}(\widehat{\mathsf{G}})
                                             )$}
is co-${\mathcal{NP}}$-hard.
\end{proof}

\subsection{Win-Lose Bimatrix Games}
\label{winlose bimatrix games}

\noindent
We show:

\begin{theorem}
\label{second last minute theorem}
Restricted to win-lose bimatrix games,
{\sf $\exists$ SYMMETRIC NASH}
is ${\mathcal{NP}}$-complete.
\end{theorem}

\begin{proof}
\textcolor{black}{Fix 
${\widehat{{\mathsf{G}}}}
   :=
   {\widehat{\mathsf{G}}}_{4}$ 
(Section~\ref{nonsymmetric game}).}
Assume first that
${\mathsf{\phi}}$ is unsatisfiable.
Then,
by Proposition~\ref{if unsatisfied},
\textcolor{black}{
${\mathcal{NE}} ( {\mathsf{G}}( {\widehat{\mathsf{G}}}_{4},
                                                                   {\mathsf{\phi}}
                                                   )
                           )
   =
   {\mathcal{NE}} ( {\widehat{\mathsf{G}}}_{4}
                             )$}.
Hence,
by Proposition~\ref{gadget5},
\textcolor{black}{${\mathsf{G}}( {\widehat{\mathsf{G}}}_{4},
                                 {\mathsf{\phi}}
                       )$}
has no symmetric Nash equilibrium.
Assume now that
${\mathsf{\phi}}$ is satisfiable.
Then,
by Proposition~\ref{final lemma}
(Condition {\sf (4)}),
\textcolor{black}{${\mathsf{G}}( {\widehat{\mathsf{G}}}_{4},
                                 {\mathsf{\phi}}
                       )$}
has a symmetric Nash equilibrium. 
So
\textcolor{black}{${\mathsf{G}}( {\widehat{\mathsf{G}}}_{4},
                                 {\mathsf{\phi}}
                       )$}
has a symmetric Nash equilibrium
if and only if
$\phi$ is satisfiable,
\textcolor{black}{and the ${\mathcal{NP}}$-hardness 
of {\sf $\exists$ SYMMETRIC NASH}
follows.}
\end{proof}

\noindent
\textcolor{black}{We remark that
the $\# {\mathcal{P}}$-hardness
of {\sf $\#$ SYMMETRIC NASH}
for win-lose bimatrix games
is already implied by the $\# {\mathcal{P}}$-hardness
of {\sf $\#$ SYMMETRIC NASH}
for symmetric win-lose bimatrix games
(Theorem~\ref{mainextended}).}

\subsection{Win-Lose \textcolor{black}{3-Player} Games}
\label{three player games}

\noindent
We show:

\begin{theorem}
\label{last minute theorem}
Restricted to win-lose \textcolor{black}{3-player} games,
{\sf $\exists$ RATIONAL NASH}
is ${\mathcal{NP}}$-complete.
\textcolor{black}{Furthermore,
{\sf $\#$ RATIONAL NASH}
is $\# {\mathcal{P}}$-complete}
\textcolor{black}{and
{\sf $\oplus$ RATIONAL NASH}
is $\oplus {\mathcal{P}}$-complete.}
\end{theorem}

\begin{proof}
Fix 
${\widehat{{\mathsf{G}}}}
   :=
   {\widehat{\mathsf{G}}}_{2}$ 
(Section~\ref{irrational game}).
Assume first that
${\mathsf{\phi}}$ is unsatisfiable.
Then,
by Proposition~\ref{if unsatisfied},
\textcolor{black}{${\mathcal{NE}} ( {\mathsf{G}}( {\widehat{\mathsf{G}}}_{2},
                                                                   {\mathsf{\phi}}
                                                   )
                           )
   =
   {\mathcal{NE}} ( {\widehat{\mathsf{G}}}_{2}
                             )$}.
Hence,
by Proposition~\ref{gadget2},
\textcolor{black}{${\mathsf{G}}( {\widehat{\mathsf{G}}}_{2},
                                 {\mathsf{\phi}}
                       )$}
has a single Nash equilibrium,
which is irrational.
Assume now that
${\mathsf{\phi}}$ is satisfiable.
Then,
by Proposition~\ref{final lemma}
(Condition {\sf (4)}),
\textcolor{black}{${\mathsf{G}}( {\widehat{\mathsf{G}}}_{2},
                                 {\mathsf{\phi}}
                       )$}
has,
\textcolor{black}{for each satisfying assignment ${\mathsf{\gamma}}$
of ${\mathsf{\phi}}$,}
a rational Nash equilibrium
\textcolor{black}{${\bm{\sigma}} = {\bm{\sigma}}({\mathsf{\gamma}})$.} 
Thus,
\textcolor{black}{${\mathsf{G}}( {\widehat{\mathsf{G}}}_{2},
                                                                  {\mathsf{\phi}}
                       )$}
has a rational Nash equilibrium
if and only if
$\phi$ is satisfiable,
and the ${\mathcal{NP}}$-hardness
of {\sf $\exists$ RATIONAL NASH}
follows.
By Propositions~\ref{if unsatisfied}
and \ref{final lemma},
the number of rational Nash equilibria
for ${\mathsf{G}}$
is
$\# {\mathsf{\phi}}$.
\textcolor{black}{Since computing $\# {\mathsf{\phi}}$
is $\# {\mathcal{P}}$-hard~\cite{V79}, 
the $\# {\mathcal{P}}$-hardness
of {\sf $\#$ RATIONAL NASH}
follows.
Since computing $\oplus {\mathsf{\phi}}$
is $\oplus {\mathcal{P}}$-hard~\cite{PZ83},
$\oplus {\mathcal{P}}$-hardness
of {\sf $\oplus$ RATIONAL NASH}
follows.}
\end{proof}

\section{Discussion and Open Problems}
\label{epilogue}

\noindent
We have established that symmetric win-lose bimatrix games
are as complex as general bimatrix games
with respect to a handful
of decision problems about Nash equilibria, 
which either were previously studied in~\cite{BM11,BDL08,CS05,CS08,GZ89,MT10,MT10a,MVY15} or are new
(Theorems~\ref{mainextended} \textcolor{black}{and~\ref{maintheorem symmetric}}).
\textcolor{black}{Furthermore,
deciding the existence of a symmetric Nash equilibrium
for win-lose bimatrix games
is as hard as for general bimatrix games
(Theorem~\ref{second last minute theorem});}
deciding the existence of a rational Nash equilibrium
for win-lose 3-player games
is as hard as
for general 3-player games
(Theorem~\ref{last minute theorem}).
Figure~\ref{table1} 
provides a tabular summary of 
these ${\mathcal{NP}}$-hardness results,
which improve and extend 
previous corresponding ${\mathcal{NP}}$-hardness results
from~\cite{BM11,BDL08,CS05,CS08,GZ89,MVY15}--- 
there \textcolor{black}{remain} just two
of those decision problems
whose ${\mathcal{NP}}$-hardness
is still open;
\textcolor{black}{we conjecture:}

\begin{conjecture}
\label{pareto open problem}
{\sf $\exists$ PARETO-OPTIMAL NASH}
and 
{\sf $\exists$ STRONGLY PARETO-OPTIMAL NASH} 
are ${\mathcal{NP}}$-hard
for symmetric win-lose bimatrix games. 
\end{conjecture}

\noindent
\textcolor{black}{Furthermore,
the counting versions 
of these ${\mathcal{NP}}$-hard problems
(Theorems~\ref{mainextended},~\ref{second last minute theorem} and~\ref{last minute theorem})
are all $\# {\mathcal{P}}$-hard.}
\textcolor{black}{Except for two
of the parity versions of
those ${\mathcal{NP}}$-hard decision problems,
namely
{\sf $\oplus \neg$ UNIFORM NASH}
and
{\sf $\oplus \neg$ SYMMETRIC NASH},
their parity versions
are all $\oplus {\mathcal{P}}$-hard.
Of these two,
recall that
{\sf $\oplus \neg$ SYMMETRIC NASH}
is in ${\mathcal{P}}$;
the $\oplus {\mathcal{P}}$-hardness 
of 
{\sf $\oplus \neg$ UNIFORM NASH}
remains open.
We conjecture:}

\begin{conjecture}
\label{parity open}
\textcolor{black}{
{\sf $\oplus$ $\neg$ UNIFORM NASH}
is $\oplus {\mathcal{P}}$-hard.}
\end{conjecture}

\noindent
The following decision problem
inquires about the irrationality
of Nash equilibria:

\vspace{0.2cm}
\noindent
{\sf $\exists$ IRRATIONAL NASH}

\begin{tabular}{lp{13.7cm}l}
\hline
{\sc I.:} & A game ${\mathsf{G}}$.                                 \\
{\sc Q.:} & Does ${\mathsf{G}}$ have an irrational Nash equilibrium?  \\
\hline
\end{tabular}
\vspace{0.2cm}

\noindent
{\sf $\exists$ IRRATIONAL NASH} 
is a trivial problem for \textcolor{black}{{\it non-degenerate}} bimatrix games,\footnote{\textcolor{black}{For bimatrix games,
                                                                                                                                            the existence of an irrational Nash equilibrium
                                                                                                                                            is equivalent to the existence of infinitely many Nash equilibria.
                                                                                                                                            Non-degenerate bimatrix games
                                                                                                                                            have a finite number of Nash equilibria.
                                                                                                                                            Hence, they have no irrational Nash equilibrium.}}
while it is ${\mathcal{NP}}$-hard
for general 3-player games~\cite[Theorem 4]{BM11}.\footnote{The proof used an alternative approach,
\textcolor{black}{employing} a reduction from the decision problem
{\sf NASH-REDUCTION}~\cite[Section 2.3.2]{BM11},
which is not considered here.}

\begin{conjecture}
\label{nonuniform irrational}
{\sf $\exists$ IRRATIONAL NASH}
is ${\mathcal{NP}}$-hard
for win-lose 3-player games.
\end{conjecture}

\noindent
\textcolor{black}{Obtaining {\em tight} hardness results
for the complexity of 
the decision problems
about Nash equilibria 
(resp., symmetric Nash equilibria)
in 3-player win-lose games
(resp., symmetric 3-player win-lose games)
remains a tantalizing open problem.
In two very recent works~\cite{BM16,BM17}
extending~\cite{GMVY15},
the present authors
established tight $\exists {\mathbb{R}}$-hardness results
for this complexity
in} \textcolor{black}{general 3-player games
(resp., general symmetric 3-player games);
$\exists {\mathbb{R}}$
is the complexity class capturing the
{\it Existential Theory of the Reals}~\cite{SS15}.}

\begin{openproblem}
Determine whether or not
the $\exists {\mathbb{R}}$-complete
decision problems about Nash equilibria
(resp., symmetric Nash equilibria)
in 3-player games
(resp., symmetric 3-player games)
from~{\em \cite{BM16,BM17}}
remain $\exists {\mathbb{R}}$-hard
when the corresponding games
are restricted to win-lose.
\end{openproblem}

\begin{figure}[tp]
\begin{center}
\begin{small}
\begin{tabular}{||l|c|c|c|c||}
\hline
& \multicolumn{4}{c|}{{\sf $\cal NP$-hard for bimatrix games:}}                                                                                                                                                                           \\
\hline
{\sf Decision problem:}                                                       & {\bf G}        & {\bf S}                                                                   & {\bf WL}                                                  & {\bf S+WL}                                      \\
\hline
\hline
{\sf $\exists$ $k+1$ NASH} (with $k \geq 1$)                  & $\Leftarrow$  & \cite{CS08,GZ89} ($k=1$) \& $\Leftarrow$ & \cite{CS05} ($k=1$) \& $\Leftarrow$ &  (\ref{mainextended}), ${\widehat{{\mathsf{G}}}}_{1}[2]$
                                                                                                                                                                                                                                                                   \& ${\widehat{{\mathsf{G}}}}_{5}[k]$                                 \\
\hline
{\sf $\exists$ NASH WITH LARGE UTILITIES}                   & $\Leftarrow$  & \cite{CS08,GZ89} \& $\Leftarrow$               &  $\Leftarrow$                                  &  (\ref{mainextended}), ${\widehat{{\mathsf{G}}}}_{1}[h]$ \\
\hline
{\sf $\exists$ NASH WITH SMALL UTILITIES}                   & $\Leftarrow$  & $\Leftarrow$                                                 &  $\Leftarrow$                                  &  (\ref{mainextended}), ${\widehat{{\mathsf{G}}}}_{1}[1]$ \\
\hline
{\sf $\exists$ NASH WITH LARGE TOTAL UTILITY}           & $\Leftarrow$  & \cite{CS08} \& $\Leftarrow$                         &  $\Leftarrow$                                  &  (\ref{mainextended}), ${\widehat{{\mathsf{G}}}}_{1}[h]$ \\      
\hline
{\sf $\exists$ NASH WITH SMALL TOTAL UTILITY}           & $\Leftarrow$  & $\Leftarrow$                                                 &  $\Leftarrow$                                   &  (\ref{mainextended}), ${\widehat{{\mathsf{G}}}}_{1}[1]$ \\
\hline
{\sf $\exists$ NASH WITH LARGE SUPPORTS}                 & $\Leftarrow$  & \cite{CS08,GZ89} \& $\Leftarrow$                &  $\Leftarrow$                                  & (\ref{mainextended}), ${\widehat{{\mathsf{G}}}}_{1}[1]$ \\
\hline
{\sf $\exists$ NASH WITH SMALL SUPPORTS}                  & $\Leftarrow$  & \cite{GZ89} \& $\Leftarrow$                         &  $\Leftarrow$                                  & (\ref{mainextended}), ${\widehat{{\mathsf{G}}}}_{1}[h]$ \\
\hline
{\sf $\exists$ NASH WITH RESTRICTED SUPPORTS}        & $\Leftarrow$  & \cite{CS08,GZ89} \& $\Leftarrow$               &  $\Leftarrow$                                  & (\ref{mainextended}), ${\widehat{{\mathsf{G}}}}_{1}[1]$   \\
\hline
{\sf $\exists$ NASH WITH RESTRICTING SUPPORTS}      & $\Leftarrow$  & \cite{CS08,GZ89} \& $\Leftarrow$               &  $\Leftarrow$                                  & (\ref{mainextended}), ${\widehat{{\mathsf{G}}}}_{1}[1]$   \\
\hline
{\sf $\exists$ PARETO-OPTIMAL NASH}                            & $\Leftarrow$  & \cite{CS08}                                                   &  {\bf ?}                                            & {\bf ?}                                                                                       \\
\hline
{\sf $\exists$ $\neg$ PARETO-OPTIMAL NASH}                   & $\Leftarrow$   & $\Leftarrow$                                            &  $\Leftarrow$                                  & (\ref{mainextended}), ${\widehat{{\mathsf{G}}}}_{1}[2]$ 
                                                                                                                                                                                                                                                              \& ${\widehat{{\mathsf{G}}}}_{5}[k]$                                  \\
\hline
{\sf $\exists$ STRONGLY PARETO-OPTIMAL NASH}         & $\Leftarrow$  & \cite{CS08}                                                   & {\bf ?}                                              & {\bf ?}                                                                                      \\
\hline
{\sf $\exists$ $\neg$ STRONGLY PARETO-OPTIMAL NASH}& $\Leftarrow$ & $\Leftarrow$                                              &  $\Leftarrow$                                   &  (\ref{mainextended}), ${\widehat{{\mathsf{G}}}}_{1}[2]$
                                                                                                                                                                                                                                                              \&  ${\widehat{{\mathsf{G}}}}_{5}[k]$                                 \\
\hline
{\sf $\exists$ NASH WITH SMALL PROBABILITIES}           & $\Leftarrow$  & $\Leftarrow$                                               &  $\Leftarrow$                                   & (\ref{mainextended}), ${\widehat{{\mathsf{G}}}}_{1}[1]$  \\
\hline
{\sf $\exists$ UNIFORM NASH}                                         & $\Leftarrow$ & $\Leftarrow$                                                 & \cite{BDL08} \& $\Leftarrow$         & (\ref{mainextended}), ${\widehat{{\mathsf{G}}}}_{3}$       \\
\hline
\hline
{\sf $\exists$ $\neg$ UNIFORM NASH}                            & $\Leftarrow$      & $\Leftarrow$                                            &  $\Leftarrow$                                   & (\ref{mainextended}), ${\widehat{{\mathsf{G}}}}_{1}[1]$  \\
\hline
{\sf $\exists$ SYMMETRIC NASH}                                     & $\Leftarrow$      & $\bot$                                                      & (\ref{second last minute theorem}),
                                                                                                                                                                                                ${\widehat{{\mathsf{G}}}}_{4}$   & $\bot$                                                                                       \\
\hline
{\sf $\exists$ $\neg$ SYMMETRIC NASH }                        & $\Leftarrow$   & \cite{MVY15} \& $\Leftarrow$                    &  $\Leftarrow$                                  & (\ref{mainextended}), ${\widehat{{\mathsf{G}}}}_{1}[2]$ 
                                                                                                                                                                                                                                                            \&  ${\widehat{{\mathsf{G}}}}_{5}[k]$                                   \\
\hline
\hline
${\sf NASH\mbox{-}EQUIVALENCE(\widehat{\sf G})}$    & \cite{BM11}  & $\Leftarrow$                                                 &  $\Leftarrow$                                   & (\ref{maintheorem symmetric}), ${\widehat{{\mathsf{G}}}}$ \\
\hline
${\sf NASH\mbox{-}EQUIVALENCE}$                                & \cite{BM11}  & $\Downarrow$                                              &  $\Downarrow$                                & $\Downarrow$                                                                          \\
\hline
\hline
\end{tabular}
\end{small}
\end{center}
\begin{center}
\begin{small}
\begin{tabular}{||l|c|c||}
\hline
                                                        & \multicolumn{2}{c||}{{\sf $\cal NP$-hard for 3-player games}}                                \\
\hline
{\sf Decision problem:}                   & {\bf G}          & {\bf WL}                                                                                             \\
\hline
\hline
{\sf $\exists$ RATIONAL NASH}     & \cite{BM11}  & (\ref{last minute theorem}), ${\widehat{{\mathsf{G}}}}_{2}$       \\
\hline
{\sf $\exists$ IRRATIONAL NASH} & \cite{BM11}  & {\bf ?}                                                                                                \\
\hline
\hline
\end{tabular}
\caption{Summary of the ${\mathcal{NP}}$-hardness
              results for decision problems
              about Nash equilibria,
              and comparison to previous related work in~\cite{BM11,BDL08,CS05,CS08,GZ89,MVY15}.
              Columns correspond to classes of bimatrix games;
              specifically,
              the symbols {\bf G}, {\bf S}, {\bf WL}
              and {\bf S+WL} denote general, symmetric, win-lose 
              and simultaneously symmetric and win-lose bimatrix games,
              respectively. 
              The rightmost column gives the number
              of the corresponding theorem (in parenthesis),
              followed by the gadget used
              in the ${\mathcal{NP}}$-hardness proof.
              The symbol "{\bf ?}" 
              indicates 
              that the ${\mathcal{NP}}$-hardness
              of the decision problem
              remains open
              for the corresponding class. 
              The symbol $\bot$
              indicates
              that the decision problem is trivial
              for the corresponding class. 
              An arrow $\Leftarrow$ in some entry
              indicates that ${\mathcal{NP}}$-hardness
              follows directly
              from a corresponding ${\mathcal{NP}}$-hardness result
              for some entry in the same row and to its right.
              An arrow $\Downarrow$ in some entry
              indicates that ${\mathcal{NP}}$-hardness follows directly
              from the
              ${\mathcal{NP}}$-hardness result
              for the entry in the same column and just above it.
              The first table applies to bimatrix games,
              while the second table applies to 3-player games. 
              }
\label{table1}
\end{small}
\end{center}
\end{figure}

\subsubsection*{Acknowledgments:}
\textcolor{black}{We would like to thank Burkhard Monien
and Pino Persiano
for some helpful comments.}

\newpage



\end{document}